\documentclass[aps,prx,superscriptaddress,footinbib,notitlepage,twocolumn,10pt]{revtex4-2}

\usepackage{tensor}
\usepackage{amsmath}
\usepackage{amsthm}
\usepackage{amssymb}
\usepackage{graphicx}
\usepackage{comment}
\usepackage{amsfonts}
\usepackage{bm}
\usepackage{braket}
\usepackage{xcolor}

\usepackage{thmtools}
\usepackage{tikz}
\usetikzlibrary{arrows.meta,positioning,calc,fit}

\definecolor{steelblue}{HTML}{5083C5}
\definecolor{medgreen}{HTML}{397A3E}

\usepackage[colorlinks=true,linkcolor=medgreen,urlcolor=teal,citecolor=steelblue,hypertexnames=false]{hyperref}

\usepackage{enumerate}
\usepackage{mathtools}
\usepackage{bbm}

\usepackage{mathrsfs}
\usepackage{dsfont}

\newcommand{\ii}{\mathrm{i}}
\newcommand{\dd}{\mathrm{d}}

\newtheorem{theorem}{Theorem}

\newtheorem{proposition}[theorem]{Proposition}
\newtheorem{corollary}[theorem]{Corollary}
\newtheorem{definition}[theorem]{Definition}
\newtheorem{lemma}[theorem]{Lemma}

\newtheorem*{lemma*}{Lemma}
\newtheorem{remark}[theorem]{Remark}
\newtheorem{fact}[theorem]{Fact}

\DeclareMathOperator{\End}{\mathcal{L}}
\DeclareMathOperator{\Hom}{\mathcal{L}}
\DeclareMathOperator{\Tr}{Tr}
\DeclareMathOperator{\Ad}{Ad}

\DeclareMathOperator{\Sym}{Sym}
\DeclareMathOperator{\rank}{rank}

\DeclareMathOperator{\Wt}{Wt}
\DeclareMathOperator{\spec}{spec}
\DeclareMathOperator{\Ran}{Ran}

\DeclareMathOperator{\Span}{span}
\DeclareMathOperator{\GL}{GL}
\DeclareMathOperator{\pr}{pr}
\DeclareMathOperator{\Var}{Var}

\newcommand{\hor}{\mathrm{hor}}
\newcommand{\ver}{\mathrm{ver}}

\newcommand{\sym}{\mathrm{sym}}
\newcommand{\HS}{\mathrm{HS}}
\newcommand{\op}{\mathrm{op}}
\newcommand{\typ}{\mathrm{typ}}
\newcommand{\atyp}{\mathrm{atyp}}
\newcommand{\lin}{\mathrm{lin}}
\newcommand{\tail}{\mathrm{tail}}
\newcommand{\fwd}{\mathrm{fwd}}
\newcommand{\rev}{\mathrm{rev}}

\newcommand{\cre}{\mathsf{C}_\lambda}
\newcommand{\ann}{\mathsf{A}_\lambda}
\newcommand{\Lmat}{\mathsf{L}_\lambda}
\newcommand{\Umat}{\mathsf{U}^{\mathrm{lin}}_\lambda}
\newcommand{\Qmat}{\mathsf{Q}_\lambda}

\newcommand{\Dmat}{\mathsf{D}_\lambda}
\newcommand{\fdisgen}{\mathsf{B}}
\newcommand{\sdisgen}{\mathsf{K}}

\newcommand{\word}{\mathsf{V}}
\newcommand{\symword}{\mathsf{J}}
\newcommand{\gram}{\mathsf{G}}

\newcommand{\annlin}{\mathsf{A}^{\mathrm{lin}}}

\newcommand{\cins}{C_{\mathrm{ins}}}

\newcommand{\crej}{\mathsf{C}_{\lambda j}}
\newcommand{\annj}{\mathsf{A}_{\lambda j}}

\DeclarePairedDelimiter\ceil{\lceil}{\rceil}
\DeclarePairedDelimiter\floor{\lfloor}{\rfloor}

\renewcommand{\Re}{\operatorname{Re}}
\renewcommand{\Im}{\operatorname{Im}}

\makeatletter

\let\orig@addcontentsline\addcontentsline

\newif\ifapx@tocactive
\apx@tocactivefalse

\def\apx@tocfile{apxa}

\def\addcontentsline#1#2#3{%
  \def\temp{#1}%
  \def\tocname{toc}%
  \ifapx@tocactive
    \ifx\temp\tocname
      \orig@addcontentsline{\apx@tocfile}{#2}{#3}%
    \else
      \orig@addcontentsline{#1}{#2}{#3}%
    \fi
  \else
    \orig@addcontentsline{#1}{#2}{#3}%
  \fi
}

\newcommand{\PrintAppendixTOC}[2]{%
  \apx@tocactivetrue
  \def\apx@tocfile{#1}%
  \begingroup
    \def\addcontentsline##1##2##3{}%
    \section*{#2}%
  \endgroup
  \@starttoc{#1}%
}

\newcommand{\AppendixTOCOne}[1]{%
  \PrintAppendixTOC{apxa}{#1}%
}

\newcommand{\AppendixTOCTwo}[1]{%
  \PrintAppendixTOC{apxb}{#1}%
}

\newcommand{\AppendixTOCThree}[1]{%
  \PrintAppendixTOC{apxc}{#1}%
}

\makeatother

\begin{document}

\title{Quantifying Symmetry Breaking with Metric Adjusted Quantum Geometric Tensors}

\author{Koji Yamaguchi}
\email{koji.yamaguchi@uwaterloo.ca}
\affiliation{Department of Physics, University of Waterloo, Waterloo, ON N2L 3G1, Canada}
\affiliation{Perimeter Institute for Theoretical Physics, Waterloo, Ontario N2L 2Y5, Canada}

\author{Hiroyasu Tajima}
\email{hiroyasu.tajima@inf.kyushu-u.ac.jp}
\affiliation{Department of Informatics, Faculty of Information Science and Electrical Engineering,
Kyushu University, 744 Motooka, Nishi-ku, Fukuoka, 819-0395, Japan}
\affiliation{JST, FOREST, 4-1-8 Honcho, Kawaguchi, Saitama, 332-0012, Japan}

\begin{abstract}
Quantifying properties of quantum states through the limits of their manipulation is a central goal of quantum resource theories. For symmetry breaking, the quantum geometric tensor (QGT) governs asymptotic pure-state conversion, but a complete characterization for general mixed states has remained elusive.
Here we fully resolve this problem for finite-dimensional systems under compact Lie group symmetries in the i.i.d. asymptotic regime. Specifically, we establish a single-letter formula for the optimal conversion rate between arbitrary states, with vanishing trace-distance error, in the resource theory of asymmetry. The rate is determined by a one-parameter family of metric adjusted QGTs, obtained by rescaling quantum Fisher information (QFI) matrices that interpolate between the symmetric logarithmic derivative and right logarithmic derivative (RLD) QFIs. For pure states, the RLD endpoint recovers the standard QGT. No state-independent finite subset of this family suffices in general, even for $U(1)$ symmetry, revealing a qualitative distinction from pure-state conversion. Our formula further yields an exact pure-state distillation-rate formula in terms of a single asymmetry measure, characterizes asymptotically reversible interconversion, and identifies bound asymmetry for quantum clocks. 
Complementarity among members of the metric adjusted QGT family also reveals an activation mechanism: a state with zero conversion rate to a mixed clock state can still enhance another input's yield under joint processing.
Our proof relies on two developments of independent interest. First, we extend quantum local asymptotic normality to unitary models with arbitrary rank and spectral degeneracy. Second, we characterize convertibility between quantum Gaussian shift models in terms of the same one-parameter family of QFIs. 
\end{abstract}

\maketitle

%\textbf{\textit{Introduction.}}---
\section{Introduction}
Quantification is fundamental to physics, turning qualitative concepts into objects that can be analyzed theoretically and tested experimentally. A particularly successful approach is to characterize such concepts through the ultimate limits of physical operations. Thermodynamic entropy and entanglement entropy exemplify this operational perspective: the former quantifies irreversibility through constraints on thermodynamic processes~\cite{liebPhysicsMathematicsSecond1999}, while the latter quantifies entanglement through fundamental limitations imposed by local operations and classical communication~\cite{bennettConcentratingPartialEntanglement1996}.

Quantum resource theories provide a versatile framework for such operational quantification by treating properties of quantum states as resources and translating physical constraints into rules for their manipulation~\cite{chitambarQuantumResourceTheories2019}. One of their central goals is to identify measures that fully capture the ultimate limits of resource manipulation. Measures with this property are referred to as complete measures~\cite{sagawaEntropyDivergenceMajorization2022,dattaThereFiniteComplete2023,yamaguchi_QuantumGeometricTensorDeterminesPureState_2026}. In the standard i.i.d.~conversion setting, where many identical copies of one resource state are converted into many copies of another, complete measures have been established in several representative settings, including entanglement~\cite{bennettConcentratingPartialEntanglement1996,haydenAsymptoticEntanglementCost2001,devetakDistillationSecretKey2005}, coherence~\cite{winter_OperationalResourceTheoryCoherence_2016,chitambar_DephasingcovariantOperationsEnableasymptoticreversibility_2018,lami_CompletingGrandTourAsymptoticQuantum_2020}, and athermality~\cite{brandaoResourceTheoryQuantum2013} in their respective resource theories.

Symmetry and its breaking provide a unifying language across diverse fields of physics. Establishing an operational quantification of symmetry breaking is therefore a fundamental challenge in quantum resource theories. The resource theory of asymmetry~\cite{bartlettReferenceFramesSuperselection2007,gourResourceTheoryQuantum2008,gourMeasuringQualityQuantum2009,marvianmashhadSymmetryAsymmetryQuantum2012,korzekwaResourceTheoryAsymmetry2013,marvianAsymmetryPropertiesPure2014} provides a framework for this problem by treating symmetry breaking as a quantum resource, with applications in quantum clocks~\cite{marvianCoherenceDistillationMachines2020,marvianOperationalInterpretationQuantum2022,kazi_OptimalDistillationQubitClocks_2025}, quantum thermodynamics~\cite{lostaglioQuantumCoherenceTimeTranslation2015,faist_GibbspreservingMapsOutperformthermaloperations_2015,marvianCoherenceDistillationMachines2020,tajima_GibbsPreservingOperationsRequiringInfiniteAmount_2025}, measurements~\cite{ahmadi_WignerArakiYanasetheoremquantum_2013,marvian_InformationtheoreticAccountWignerArakiYanasetheorem_2012,korzekwaResourceTheoryAsymmetry2013,tajimaCoherencevarianceUncertaintyRelation2019,tajimaUniversalTradeoffStructure2022,emori_ErrorDisturbanceIrreversibilityApplicationsUnified_2026,hokkyo_QuantitativeWignerArakiYanaseTheoremsUnitaryAntiunitary_2026}, quantum computing~\cite{tajimaUncertaintyRelationsImplementation2018,tajimaCoherenceCostViolating2020,tajimaUniversalLimitationQuantum2021,tajimaUniversalTradeoffStructure2022}, error-correcting codes~\cite{kubicaUsingQuantumMetrological2021,zhouNewPerspectivesCovariant2021,yangOptimalUniversalQuantum2022,tajimaUniversalLimitationQuantum2021,liuQuantumErrorCorrection2023,tajimaUniversalTradeoffStructure2022,liu_ApproximateSymmetriesQuantumerrorcorrection_2023}, black hole physics~\cite{tajimaUniversalLimitationQuantum2021,tajimaUniversalTradeoffStructure2022}, 
(non)-Gaussianity~\cite{yadavalli_OptimalDistillationCoherentStatesPhaseInsensitive_2025,koukoulekidis_SymmetryAsymmetryBosonicGaussianSystems_2025}, 
and the quantum Mpemba effect~\cite{summer_ResourceTheoreticalUnificationMpembaEffectsClassical_2026,yamashika_QuantumFisherInformationMeasureSymmetry_2025,kusuki_ResourceTheoreticQuantifiersWeakStrongSymmetry_2026}.
% quantum clocks~\cite{marvianCoherenceDistillationMachines2020,marvianOperationalInterpretationQuantum2022,kazi_OptimalDistillationQubitClocks_2025}
% quantum thermodynamics~\cite{lostaglioQuantumCoherenceTimeTranslation2015,Faist2015Gibbs-preserving,marvianCoherenceDistillationMachines2020,Tajima_Takagi2025}, measurements~\cite{ahmadi_WignerArakiYanasetheoremquantum_2013,WAY_RTA2,korzekwaResourceTheoryAsymmetry2013,tajimaCoherencevarianceUncertaintyRelation2019, tajima_universal_2022,emori_ErrorDisturbanceIrreversibilityApplicationsUnified_2026,hokkyo_QuantitativeWignerArakiYanaseTheoremsUnitaryAntiunitary_2026}, quantum computing~\cite{tajima_uncertainty_2018,tajima_coherence_2020,tajimaUniversalLimitationQuantum2021,tajima_universal_2022}, error-correcting codes~\cite{kubica_using_2021,zhou_new_2021,yang_optimal_2022,tajimaUniversalLimitationQuantum2021,liu_QuantumErrorCorrectionmeetscontinuous_2023,tajima_universal_2022,liu_ApproximateSymmetriesQuantumerrorcorrection_2023}, black hole physics \cite{tajimaUniversalLimitationQuantum2021,tajima_universal_2022}, (non)-Gaussianity~\cite{yadavalli_OptimalDistillationCoherentStatesPhaseInsensitive_2025,koukoulekidis_SymmetryAsymmetryBosonicGaussianSystems_2025}, and the quantum Mpemba effect~\cite{summer_ResourceTheoreticalUnificationMpembaEffectsClassical_2026,yamashika_QuantumFisherInformationMeasureSymmetry_2025,kusuki_ResourceTheoreticQuantifiersWeakStrongSymmetry_2026}. 
For pure-state asymmetry under continuous symmetries, complete measures had been known only for the $U(1)$ group~\cite{gourResourceTheoryQuantum2008,marvianOperationalInterpretationQuantum2022}, while results for other continuous groups were restricted to specific classes of states~\cite{gourResourceTheoryQuantum2008,yangUnitsRotationalInformation2017}. This limitation was recently overcome by a general theory establishing the quantum geometric tensor (QGT)~\cite{provostRiemannianStructureManifolds1980,berryQuantumPhaseFive1989} as a complete measure of pure-state symmetry breaking for arbitrary compact Lie groups~\cite{yamaguchi_QuantumGeometricTensorDeterminesPureState_2026}. 

However, a comprehensive theory of symmetry breaking requires an understanding of mixed states, which are essential for describing quantum systems in realistic settings. The asymptotic manipulation of mixed-state asymmetry remains largely unexplored: apart from pure-to-mixed conversion under $U(1)$ symmetry~\cite{marvianOperationalInterpretationQuantum2022}, existing results are limited to bounds and no-go theorems for convertibility~\cite{marvianCoherenceDistillationMachines2020,yamaguchi_QuantumGeometricTensorDeterminesPureState_2026}.

In this work, we resolve this problem by establishing a general theory of the asymptotic manipulation of mixed-state asymmetry. Specifically, for any symmetry described by a compact Lie group, we prove a single-letter formula for the optimal i.i.d. asymptotic conversion rate between arbitrary states on finite-dimensional systems within the standard formulation of the resource theory of asymmetry. This result identifies a family of metric adjusted QGTs as a complete family of asymmetry measures. We define these tensors by rescaling quantum Fisher information (QFI) matrices interpolating between the symmetric logarithmic derivative (SLD) and right logarithmic derivative (RLD) QFIs. For pure states, the RLD endpoint of the metric adjusted QGT family yields the standard QGT, thereby directly connecting to the prior result~\cite{yamaguchi_QuantumGeometricTensorDeterminesPureState_2026} for pure states. 

Our general formula has three further consequences. First, we show that no state-independent finite subset of this family suffices to determine all conversion rates, even for $U(1)$ symmetry. 
Second, for mixed-to-pure distillation, the full-family formula collapses to a single metric adjusted QGT at the RLD endpoint. This tensor coincides with the extension of the QGT proposed in Ref.~\cite{yamaguchi_QuantumGeometricTensorDeterminesPureState_2026}, and completely determines the distillation rate.
Third, we derive a necessary and sufficient condition for asymptotic reversibility, identifying when asymmetry can be preserved without asymptotic loss (see Fig.~\ref{fig:conversion_rate_law_comparison} in Section~\ref{sec:main_results} for a summary of these results and comparison with prior studies).

The necessity of a family of measures, rather than a single one, also reveals a new mechanism for enhancing asymptotic conversion rates. We show that, even for $U(1)$ symmetry, a state with zero conversion rate to a given target can nevertheless increase the yield obtainable from another state when the two are processed jointly. In analogy with activation phenomena for channel capacities, we call this phenomenon asymmetry activation. This activation arises from complementarity among different members of the metric adjusted QGT family: distinct members impose different constraints on the conversion rate, and jointly processing input states can enhance the rate beyond what either state alone can achieve. Thus, the family not only completely characterizes asymptotic conversion rates but also reveals a distinct mechanism for asymmetry activation.

The proof of our conversion formula rests on three theoretical advances. First, we establish quantum local asymptotic normality (QLAN) for general finite-dimensional unitary models, allowing arbitrary rank and degeneracies among the positive eigenvalues~\cite{gutaLocalAsymptoticNormality2006,kahnQuantumLocalAsymptotic2008,kahnLocalAsymptoticNormality2009,lahiryMinimaxEstimationLowrank2024}. Together with the equivalence between convertibility in the resource theory of asymmetry and convertibility of the associated statistical models~\cite{marvianmashhadSymmetryAsymmetryQuantum2012,yamaguchi_QuantumGeometricTensorDeterminesPureState_2026}, our QLAN result reduces the asymptotic conversion problem to a convertibility problem for quantum Gaussian shift models. Second, we fully characterize the convertibility between the resulting Gaussian models and show that it is governed precisely by inequalities between the entire QFI family interpolating between the SLD and RLD, which are equivalent to metric adjusted QGT inequalities.
Third, to establish the optimality of our achievable rate, we derive quantitative bounds on the QFI matrices of states close to i.i.d. models, overcoming the asymptotic discontinuity of QFIs~\cite{gourMeasuringQualityQuantum2009,marvianCoherenceDistillationMachines2020,marvianOperationalInterpretationQuantum2022,yamaguchiSmoothMetricAdjusted2023,yamaguchi_QuantumGeometricTensorDeterminesPureState_2026}. 
Together, these advances establish our general conversion formula and the completeness of the metric adjusted QGT family for mixed-state asymmetry.

\section{Preliminaries}
\subsection{Resource theory of asymmetry}
%\textbf{\textit{Resource theory of asymmetry.}}---
We study symmetries described by a compact Lie group $G$. Consider a quantum system with a finite-dimensional Hilbert space $\mathcal{H}$. The action of each $g \in G$ is represented by a unitary operator $U(g)$ on $\mathcal{H}$. Consistency under successive transformations requires
$U(g_2)U(g_1)=\omega(g_2,g_1)U(g_2g_1)$ for all $g_1,g_2\in G$, where $\omega$ is a complex-valued function of modulus one. Such a map $U$ is called a projective unitary representation of $G$ on $\mathcal{H}$. When $\omega(g_2,g_1)=1$ for all $g_1,g_2\in G$, $U$ is called a nonprojective unitary representation. In what follows, we assume that the map $U$ is differentiable to simplify the arguments. The main results can be extended~\cite{sm} to any continuous $U$ using the methods in Refs.~\cite{shitara_IidStateConvertibilityResourceTheory_2025,yamaguchi_QuantumGeometricTensorDeterminesPureState_2026}. 

We employ the standard formulation of the resource theory of asymmetry~\cite{bartlettReferenceFramesSuperselection2007,gourResourceTheoryQuantum2008,gourMeasuringQualityQuantum2009,marvianmashhadSymmetryAsymmetryQuantum2012,korzekwaResourceTheoryAsymmetry2013,marvianAsymmetryPropertiesPure2014}, a framework for studying symmetry and its breaking. As in other quantum resource theories, it is specified by free states and free operations. The free states are the $G$-symmetric states, namely, those invariant under every symmetry transformation: $\mathcal{U}_g(\rho)=\rho$ for all $g\in G$, where $\mathcal{U}_g(\cdot)\coloneqq U(g)(\cdot)U(g)^\dag$. The free operations are the $G$-covariant channels, which are compatible with the symmetry transformations in the sense that $\mathcal{E}\circ\mathcal{U}_g=\mathcal{U}'_g\circ\mathcal{E}$ for all $g\in G$. Throughout, we generally use unprimed and primed symbols for input and output quantities, respectively.

The choice of $G$-covariant channels as free operations is operationally justified because they can be implemented without an external source of asymmetry. Indeed, the covariant Stinespring dilation theorem~\cite{keylOptimalCloningPure1999,marvianmashhadSymmetryAsymmetryQuantum2012,yamaguchi_QuantumGeometricTensorDeterminesPureState_2026} guarantees that any $G$-covariant channel can be realized using a symmetry-preserving unitary interaction with an ancillary system prepared in a $G$-symmetric state. Thus, if a state $\rho$ is convertible into another state $\rho'$ via a $G$-covariant channel, then $\rho$ exhibits at least as much symmetry breaking as $\rho'$.

Resource measures provide a quantitative way to make such comparisons. We call a function $M$ an asymmetry measure if
\begin{align}
    M(\rho)\geq M(\rho')
    \label{eq:monotonicity_measure}
\end{align}
whenever $\rho$ can be converted into $\rho'$ via a $G$-covariant channel, and $M(\sigma)=0$ for any $G$-symmetric state $\sigma$. Recent developments in the resource theory of asymmetry~\cite{gaoSufficientStatisticRecoverability2024,kudoFisherInformationMatrix2023,yamaguchi_QuantumGeometricTensorDeterminesPureState_2026} have highlighted the importance of matrix-valued measures that generalize conventional scalar-valued measures. For matrix-valued measures, the inequality in Eq.~\eqref{eq:monotonicity_measure} is understood in the positive-semidefinite order: $A\geq B$ means that $A-B$ is positive semidefinite. Various asymmetry measures have been studied, including the $G$-asymmetry~\cite{vaccaroTradeoffExtractableMechanical2008}, the relative entropy of $G$-asymmetry~\cite{gourMeasuringQualityQuantum2009}, quantum Fisher information matrices~\cite{marvianOperationalInterpretationQuantum2022,marvianCoherenceDistillationMachines2020,gaoSufficientStatisticRecoverability2024,kudoFisherInformationMatrix2023,yamaguchi_QuantumGeometricTensorDeterminesPureState_2026}, and the quantum geometric tensor~\cite{yamaguchi_QuantumGeometricTensorDeterminesPureState_2026}. Monotonicity in Eq.~\eqref{eq:monotonicity_measure}, however, generally provides only a necessary condition for state conversion.

A central question in quantum resource theories is therefore which resource measures fully characterize state convertibility. A set of measures is called complete if their monotonicity condition is both necessary and sufficient for state conversion, thereby providing a complete quantitative characterization of convertibility. A fundamental setting for addressing this question is the i.i.d. asymptotic regime, where many identical copies $\rho^{\otimes n}$ of an input state are converted into many copies of an output state $\rho'$.

For multiple copies, the symmetry acts collectively; hence, a channel $\mathcal{E}$ mapping $n$ input copies to $m$ output copies is $G$-covariant if $\mathcal{E}\circ\mathcal{U}_g^{\otimes n}=\mathcal{U}_g'^{\otimes m}\circ\mathcal{E}$ for all $g\in G$. In the asymptotic regime, a rate $r>0$ is said to be achievable for the conversion from $\rho$ to $\rho'$ if there exists a sequence of $G$-covariant channels $\{\mathcal{E}_n\}_n$ such that
\begin{align}
    \lim_{n\to\infty}
    \left\|
        \mathcal{E}_n(\rho^{\otimes n})-
        \rho'^{\otimes \lfloor rn\rfloor}
    \right\|_1
    =0.
\end{align}
This is equivalent to the convergence in the trace distance $T(\rho,\sigma)\coloneqq \frac{1}{2}\|\rho-\sigma\|_1$, which has operational meaning in state distinguishability~\cite{helstromQuantumDetectionEstimation1969,holevoStatisticalDecisionTheory1973}. 
The optimal conversion rate $R(\rho\to\rho')$ is defined as the supremum of all achievable rates.

As a basic characteristic of a state $\rho$, we introduce the symmetry subgroup~\cite{marvianmashhadSymmetryAsymmetryQuantum2012} by
\begin{align}
    \Sym_G(\rho)\coloneqq \{g\in G\colon \mathcal{U}_g(\rho)=\rho\}.
\end{align}
It is known that $\Sym_G(\rho)\subset \Sym_G(\rho')$ is a necessary condition for converting $\rho$ into $\rho'$ in the one-shot exact setting~\cite{marvianmashhadSymmetryAsymmetryQuantum2012} and in the asymptotic setting with vanishing error~\cite{yamaguchi_QuantumGeometricTensorDeterminesPureState_2026}. In what follows, we say that a pair of states satisfies the symmetry-subgroup condition if the above inclusion holds.

\subsection{Quantum Fisher information}
%\textbf{\textit{Quantum Fisher information.}}---
Quantum Fisher information (QFI) is associated with metrics on the quantum state space that contract under data processing~\cite{morozova_MarkovInvariantGeometrymanifoldsstates_1991a,petzMonotoneMetricsMatrix1996}, which are in one-to-one correspondence with operator monotone functions. A function $f:(0,\infty)\to(0,\infty)$ is called operator monotone if $A\leq B$ implies $f(A)\leq f(B)$ for any positive definite matrices. Among others, we consider a one-parameter family of such functions~\cite{yamaguchi_QuantumGeometricTensorDeterminesPureState_2026}
\begin{align}
    f_q(x)\coloneqq (1-q)+qx,\qquad q\in(0,1).
\end{align}

For a smooth $p$-parameter family of states $\rho_{\theta}$ with $\theta\in\mathbb{R}^p$, the associated $f_q$-QFI $\mathcal{F}_{\rho}^{f_q}$ at $\rho\coloneqq\rho_{\theta_0}$ is given by a $p\times p$ matrix~\cite{sm}
\begin{align}
    \left(\mathcal{F}_{\rho}^{f_q}\right)_{ij}\coloneqq \sum_{\substack{k,l=1;\\(1-q)\mu_l+q\mu_k>0} }^d\frac{\braket{l|\partial_i\rho|k}\braket{k|\partial_j\rho|l}}{(1-q)\mu_l+q\mu_k}
\end{align}
for $i,j=1,\ldots,p$, where $\partial_i\rho\coloneqq \partial_{\theta_i}\rho_{\theta}|_{\theta=\theta_0}$ and $\rho=\sum_{i=1}^d \mu_i\ket{i}\bra{i}$ denotes the eigenvalue decomposition of $\rho$. 

The choices at $q=0,1/2,1$ correspond to the left logarithmic derivative (LLD), symmetric logarithmic derivative (SLD), and right logarithmic derivative (RLD), respectively~\cite{shitara_DeterminingContinuousFamilyquantumFisher_2016}. 
The LLD and RLD boundary limits of the QFI are finite for full-rank states but can diverge for rank-deficient states.
We refer to these as the LLD, SLD, and RLD QFIs, respectively. The QFI family obeys a transpose relation under the interchange $q\leftrightarrow(1-q)$, expressed as 
\begin{align}
    \mathcal{F}_{\rho}^{f_{1-q}}=\left(\mathcal{F}_{\rho}^{f_q}\right)^\top,\label{eq:QFI_transpose_relation}
\end{align}
where $\top$ denotes the transpose. Therefore, the matrix inequalities imposed for all $q\in(0,1)$ can equivalently be checked on $q\in[1/2,1)$, since transposition preserves the positive-semidefinite order, i.e., $A\geq B\Leftrightarrow A^\top\geq B^\top$ for positive-semidefinite matrices $A$ and $B$.

In the context of the resource theory of asymmetry we consider a family of states parameterized by group elements. Concretely, we first parametrize group elements near the identity by $\theta \in\mathbb{R}^{\dim G}$ as $g(\theta)=\exp(\ii \sum_{i=1}^{\dim G} \theta_i A_i)$, using a basis $\{A_i\}_{i=1}^{\dim G}$ of the Lie algebra $\mathfrak{g}$. Then, for a state $\rho$, let $\rho_\theta\coloneqq \mathcal{U}_{g(\theta)}(\rho)$ near $\theta=0$. The tangent at $\rho=\rho_{\theta=0}$ is given by
\begin{align}
    \partial_i \rho=\ii [X_i,\rho],\label{eq:tangent_unitary_model}
\end{align}
where $ X_i\coloneqq -\ii \partial_{i}U(g(\theta))|_{\theta=0}$ denote the infinitesimal generators of symmetry transformations. Thus, the $f_q$-QFI in the RTA is given by
\begin{align}
    \left(\mathcal{F}_{\rho}^{f_q}\right)_{ij}=\sum_{\substack{k,l=1;\\(1-q)\mu_l+q\mu_k>0} }^d\frac{(\mu_k-\mu_l)^2\braket{l|X_i|k}\braket{k|X_j|l}}{(1-q)\mu_l+q\mu_k}.\label{eq:QFI_RTA}
\end{align}
For rank-deficient states, endpoint divergences occur when the generators have nonzero matrix elements between the support and the kernel of $\rho$.

For $U(1)$ symmetry, the SLD QFI determines the asymptotic coherence cost~\cite{marvianOperationalInterpretationQuantum2022}, while the RLD QFI, also called purity of coherence~\cite{marvianCoherenceDistillationMachines2020}, governs the optimal asymptotic infidelity of many-to-one distillation~\cite{marvianCoherenceDistillationMachines2020,yadavalli_OptimalDistillationCoherentStatesPhaseInsensitive_2025,kazi_OptimalDistillationQubitClocks_2025}.

\subsection{Metric adjusted quantum geometric tensor}
The standard QGT of a smooth family of normalized pure states $\ket{\psi_\theta}$, evaluated at a point $\theta_0$, is given by
\begin{align}
    \left(\mathcal{Q}_\psi\right)_{ij}\coloneqq \braket{\partial_i\psi|(I-\ket{\psi}\bra{\psi})|\partial_j\psi},
\end{align}
where $\ket{\psi}\coloneqq\ket{\psi_{\theta_0}}$ and $\ket{\partial_i\psi}\coloneqq \partial_{\theta_i}\ket{\psi_\theta}|_{\theta=\theta_0}$~\cite{provostRiemannianStructureManifolds1980,berryQuantumPhaseFive1989}. 

For a general state model, we introduce the metric adjusted QGT family
\begin{align}
    \mathcal{Q}_\rho^q\coloneqq q(1-q)\mathcal{F}_{\rho}^{f_q},\quad q\in(0,1).\label{eq:definition_Q_F}
\end{align}
The Supplemental Material provides the construction for general operator monotone functions and shows that, for symmetric functions and unitary models, the metric adjusted QGT is directly related to the metric adjusted skew information~\cite{sm}. The normalization removes the divergences at LLD and RLD endpoints, thereby enabling continuous extensions to the endpoints $q=0,1$: $\mathcal{Q}_\rho^0\coloneqq \lim_{q\to 0^+}\mathcal{Q}_\rho^q$ and $\mathcal{Q}_\rho^1\coloneqq \lim_{q\to 1^-}\mathcal{Q}_\rho^q$. In particular, for a pure-state model $\psi_\theta=\ket{\psi_\theta}\bra{\psi_\theta}$,
\begin{align}
    \mathcal{Q}_\psi^q=q\mathcal{Q}_\psi+(1-q)\mathcal{Q}_\psi^\top,\label{eq:QGT_q_pure}
\end{align}
and hence the standard QGT is recovered at the RLD endpoint: $\mathcal{Q}_\psi^1=\mathcal{Q}_\psi$. 

For a unitary model whose tangent is given by Eq.~\eqref{eq:tangent_unitary_model}, the endpoint tensors are given by
\begin{align}
    \left(\mathcal{Q}_\rho^1\right)_{ij}=\Tr(\rho X_i(I-\Pi_\rho)X_j)
\end{align}
and $\mathcal{Q}_\rho^0=\left(\mathcal{Q}_\rho^1\right)^\top$, where $\Pi_\rho$ denotes the projector onto the support of $\rho$. The RLD endpoint $\mathcal{Q}_\rho^1$ coincides with the extension of the QGT proposed in Ref.~\cite{yamaguchi_QuantumGeometricTensorDeterminesPureState_2026}. 

Each member of this family is positive semidefinite, monotone under covariant channels, and additive under tensor products with collective symmetry action~\cite{sm}. Moreover,
\begin{align}
    \mathcal{Q}_\rho^{1-q}=\left(\mathcal{Q}_\rho^q\right)^\top, \quad q\in[0,1].\label{eq:transpose_relation_QGT}
\end{align}

\section{Main results}\label{sec:main_results}
%\textbf{\textit{Main result.}}---
The main theorem of this paper is the following conversion rate formula:
\begin{theorem}[Completeness of the metric adjusted QGT family]\label{thm:main_theorem_conversion_rate}
    For any states $\rho$ and $\rho'$ on finite-dimensional Hilbert spaces, if $\Sym_G(\rho)\subset \Sym_{G}(\rho')$, then
    \begin{align}
        R(\rho\to\rho')&=\sup\{r\geq 0\colon \forall q\in[1/2,1],\mathcal{Q}_{\rho}^{q}\geq r \mathcal{Q}_{\rho'}^{q}\},\label{eq:main_theorem_conversion_rate}\\
        &=\inf_{q\in[1/2,1]}\sup\{r\geq 0\colon\mathcal{Q}_{\rho}^{q}\geq r \mathcal{Q}_{\rho'}^{q}\}\label{eq:main_theorem_conversion_rate_inf_q}
    \end{align}
    otherwise, $R(\rho\to\rho')=0$. 
\end{theorem}
In words, Eq.~\eqref{eq:main_theorem_conversion_rate} shows that, under the symmetry-subgroup condition, the metric adjusted QGT family is a complete family of measures of symmetry breaking. The proof sketch is provided in Section~\ref{sec:proof_sketch}, while the detailed full proof can be found in the Supplemental Material~\cite{sm}. Moreover, for any $0<s<R(\rho\to\rho')$ and any $\kappa\in(0,1/2)$, the constructive proof in the Supplemental Material gives an error of $O(n^{-1/2+\kappa})$ at conversion rate $s$~\cite{sm}. 
The transpose relation in Eq.~\eqref{eq:transpose_relation_QGT} shows that the range of $q$ in the conversion-rate formula can alternatively be written as either $[0,1/2]$ or $[0,1]$. 
By Eq.~\eqref{eq:definition_Q_F} and the continuity in $q$, the conditions are equivalent to $\mathcal{F}_\rho^{f_q}\geq r \mathcal{F}_{\rho'}^{f_q}$ for all $q\in [1/2,1)$. 
We remark that Ref.~\cite{shitara_DeterminingContinuousFamilyquantumFisher_2016} provides a method to measure QFI associated with any operator monotone function through linear-response theory. Since Eq.~\eqref{eq:definition_Q_F} merely rescales each interior QFI by a known factor, Theorem~\ref{thm:main_theorem_conversion_rate} connects the conversion rate to these experimentally accessible quantities. 
Also, by using the quantum max-relative entropy~\cite{dattaMinMaxRelativeEntropies2009,tomamichelQuantumInformationProcessing2016}, defined for positive-semidefinite matrices $A$ and $B$ by
\begin{align}
    D_{\mathrm{max}}(A\|B)\coloneqq \inf\{\lambda\in\mathbb{R}\colon A\leq 2^\lambda B\}\in[-\infty,\infty],
\end{align}
Eq.~\eqref{eq:main_theorem_conversion_rate} can also be expressed as
\begin{align}
    R(\rho\to\rho')&=\inf_{q\in[\frac{1}{2},1]}2^{-D_{\mathrm{max}}(\mathcal{Q}^{q}_{\rho'}\|\mathcal{Q}^{q}_{\rho})}.
\end{align}

Theorem~\ref{thm:main_theorem_conversion_rate} recovers previously known asymptotic conversion laws in the resource theory of asymmetry~ \cite{gourResourceTheoryQuantum2008,marvianOperationalInterpretationQuantum2022,shitara_IidStateConvertibilityResourceTheory_2025,yangUnitsRotationalInformation2017,yamaguchi_QuantumGeometricTensorDeterminesPureState_2026}.
It is worth emphasizing that, in each of these settings, the conversion rate is determined by a single asymmetry measure. 
This naturally raises the question of whether the continuous metric adjusted QGT family appearing in Theorem~\ref{thm:main_theorem_conversion_rate} is genuinely necessary, or whether it can be replaced by one or finitely many fixed values of $q$. We show that such a universal finite reduction is impossible, revealing a qualitative departure from the known single-measure laws:
\begin{proposition}[No universal finite sampling]
\label{prop:finite_q_is_insufficient}
For every nonempty finite set $S\subset[1/2,1]$, there exist states $\rho$ and $\rho'$ such that $\Sym_{U(1)}(\rho)=\Sym_{U(1)}(\rho')$ and $U(1)$-covariant conversion rate satisfies
\begin{align}
    R(\rho\to\rho')<\min_{q\in S}\sup\{r\geq 0\colon\mathcal{Q}_{\rho}^{q}\geq r \mathcal{Q}_{\rho'}^{q}\}.
\end{align}
\end{proposition}

The proof is given in the Supplemental Material~\cite{sm}. Comparing this proposition with Eq.~\eqref{eq:main_theorem_conversion_rate_inf_q}, we see that no state-independent finite set of $q$ values is sufficient to determine the conversion rate for arbitrary states.
Figure~\ref{fig:conversion_rate_law_comparison} summarizes the resulting qualitative distinction between the previously known single-measure laws and the general mixed-state conversion law established here.

\begin{figure*}[t]
\centering
\begingroup

\definecolor{figKnown}{HTML}{5083C5}
\definecolor{figNew}{HTML}{397A3E}

\begin{tikzpicture}[
  >=latex, x=1cm, y=1cm,
  line cap=round, line join=round,
  every node/.style={inner sep=0pt,outer sep=0pt,align=center},
  panel/.style={fill=white,rounded corners=2pt,line width=.55pt},
  grouphead/.style={font=\fontsize{9.3}{10.8}\selectfont\bfseries},
  rowhead/.style={font=\fontsize{8.5}{9.9}\selectfont\bfseries},
  panelhead/.style={
    font=\fontsize{8.7}{10.0}\selectfont\bfseries,
    text width=3.15cm},
  note/.style={font=\fontsize{8.1}{9.35}\selectfont,text width=3.15cm},
  eqn/.style={font=\fontsize{9.5}{11.1}\selectfont},
  smallref/.style={
    font=\fontsize{7.7}{8.8}\selectfont\bfseries,
    text width=3.15cm},
  widehead/.style={
    font=\fontsize{8.7}{10.0}\selectfont\bfseries,
    text width=6.7cm},
  widenote/.style={font=\fontsize{8.1}{9.35}\selectfont,text width=6.7cm},
  wideeqn/.style={font=\fontsize{9.1}{10.6}\selectfont},
  wideref/.style={
    font=\fontsize{7.7}{8.8}\selectfont\bfseries,
    text width=6.7cm},
  % Single measure: regular text, no background.
  measure/.style={
    font=\fontsize{8.0}{9.2}\selectfont\mdseries\upshape,
    text width=3.15cm,
    fill=none,draw=none},
  % Family of measures: same size, bold, no background.
  family/.style={
    measure,
    font=\fontsize{8.0}{9.2}\selectfont\bfseries\upshape}
]

% ========================================================
% Main theorem
% ========================================================
\draw[draw=figNew,fill=figNew!4,rounded corners=3pt,line width=1pt]
  (.05,0) rectangle (15.60,-1.52);
\node[font=\fontsize{9.0}{10.4}\selectfont\bfseries,text=figNew]
  at (7.825,-.26)
  {Main result --- Theorem~\ref{thm:main_theorem_conversion_rate}};
\node[font=\fontsize{10.0}{11.8}\selectfont\bfseries]
  at (7.825,-.69)
  {Complete family of asymmetry measures for compact Lie group asymmetry};
\node[font=\fontsize{9.45}{11.1}\selectfont]
  at (7.825,-1.18)
  {$R(\rho\!\to\!\rho')=
    \sup\left\{r\ge0:\;
    \mathcal{Q}_\rho^{q}\geq r\mathcal{Q}_{\rho'}^{q}
    \quad\forall q\in[1/2,1]\right\}$};

% Branching
\draw[line width=.65pt] (7.825,-1.52) -- (7.825,-1.77);
\draw[line width=.65pt] (3.85,-1.77) -- (11.80,-1.77);
\draw[->,line width=.65pt] (3.85,-1.77) -- (3.85,-2.02);
\draw[->,line width=.65pt] (11.80,-1.77) -- (11.80,-2.02);

% ========================================================
% Group frames: known / new
% ========================================================
\draw[draw=figKnown,fill=figKnown!3,rounded corners=3pt,
      line width=.9pt,dash pattern=on 4pt off 2pt]
  (.05,-2.02) rectangle (7.65,-9.25);
\draw[draw=figNew,fill=figNew!3,rounded corners=3pt,line width=.9pt]
  (8.00,-2.02) rectangle (15.60,-9.25);
\node[grouphead,text=figKnown] at (3.85,-2.28) {Known results};
\node[grouphead,text=figNew] at (11.80,-2.28) {This work};

% Row headings
\node[rowhead] at (3.85,-2.63) {Optimal conversion rates};
\node[rowhead] at (11.80,-2.63) {Optimal conversion rates};
\node[rowhead] at (3.85,-6.50) {Asymptotic reversibility};
\node[rowhead] at (11.80,-6.50) {Asymptotic reversibility};

% ========================================================
% Upper panels
% ========================================================
\draw[panel,draw=figKnown!75] (.25,-2.88) rectangle (3.73,-6.08);
\draw[panel,draw=figKnown!75] (3.97,-2.88) rectangle (7.45,-6.08);
\draw[panel,draw=figNew!75] (8.20,-2.88) rectangle (11.68,-6.08);
\draw[panel,draw=figNew!75] (11.92,-2.88) rectangle (15.40,-6.08);

% Labels inside the panels; only Family of measures is bold.
\node[measure,text=figKnown] at (1.99,-3.12) {Single measure};
\node[measure,text=figKnown] at (5.71,-3.12) {Single measure};
\node[family,text=figNew] at (9.94,-3.12) {Family of measures};
\node[measure,text=figNew] at (13.66,-3.12) {Single measure};

% (a) Pure-to-mixed: established U(1) formation law
\node[panelhead] at (1.99,-3.54) {(a) Pure-to-mixed};
\node[note] at (1.99,-3.93) {$U(1)$ group};
\node[font=\fontsize{9.1}{10.7}\selectfont] at (1.99,-4.77)
  {$R(\psi\!\to\!\rho')=
    \dfrac{\mathcal{Q}_\psi^{1/2}}
          {\mathcal{Q}_{\rho'}^{1/2}}$};
\node[note] at (1.99,-5.56) {SLD point $(q=1/2)$};
\node[smallref,text=figKnown] at (1.99,-5.87)
  {Ref.~\cite{marvianOperationalInterpretationQuantum2022}};

% (b) Pure-to-pure: adjacent to the general-state panel (c)
\node[panelhead] at (5.71,-3.54) {(b) Pure-to-pure};
\node[note] at (5.71,-3.93) {Compact Lie group};
\node[eqn] at (5.71,-4.77)
  {$\mathcal Q_\psi\geq r\mathcal Q_\phi$};
\node[note] at (5.71,-5.56) {standard QGT};
\node[smallref,text=figKnown] at (5.71,-5.87)
  {Ref.~\cite{yamaguchi_QuantumGeometricTensorDeterminesPureState_2026}};

% (c) General states: universal finite sampling is insufficient
\node[panelhead] at (9.94,-3.54) {(c) General states};
\node[note] at (9.94,-3.93) {Already for $U(1)$};
\node[note] at (9.94,-4.77)
  {No fixed finite $q$-grid\\[-1pt] suffices universally};
\node[note] at (9.94,-5.56) {QGT family};
\node[smallref,text=figNew] at (9.94,-5.87)
  {Proposition~\ref{prop:finite_q_is_insufficient}};

% (d) Mixed-to-pure: a single-measure reduction of the general law
\node[panelhead] at (13.66,-3.54) {(d) Mixed-to-pure};
\node[note] at (13.66,-3.93) {Compact Lie group};
\node[eqn] at (13.66,-4.77)
  {$\mathcal Q_\rho^1\geq r\mathcal Q_\phi^1$};
\node[note] at (13.66,-5.56) {QGT endpoint $(q=1)$};
\node[smallref,text=figNew] at (13.66,-5.87)
  {Corollary~\ref{cor:distillation}};

% ========================================================
% Lower panels: reversibility
% ========================================================
\draw[panel,draw=figKnown!75] (.25,-6.77) rectangle (7.45,-9.03);
\draw[panel,draw=figNew!75] (8.20,-6.77) rectangle (15.40,-9.03);

\node[measure,text=figKnown] at (3.85,-7.00) {Single measure};
\node[family,text=figNew] at (11.80,-7.00) {Family of measures};

% (e) Pure-to-pure reversibility
\node[widehead] at (3.85,-7.37) {(e) Pure-to-pure reversibility};
\node[widenote] at (3.85,-7.71) {standard QGT};
\node[wideeqn] at (3.85,-8.08)
  {$\operatorname{Sym}_G(\psi)=\operatorname{Sym}_G(\phi)$};
\node[wideeqn] at (3.85,-8.46)
  {$\exists!\,r>0:\quad\mathcal Q_\psi=r\mathcal Q_\phi$};
\node[wideref,text=figKnown] at (3.85,-8.80)
  {Ref.~\cite{yamaguchi_QuantumGeometricTensorDeterminesPureState_2026};
   conjectured in Ref.~\cite{marvianAsymmetryPropertiesPure2014}};

% (f) General-state reversibility
\node[widehead] at (11.80,-7.37) {(f) General-state reversibility};
\node[widenote] at (11.80,-7.71) {Metric adjusted QGT family};
\node[wideeqn] at (11.80,-8.08)
  {$\operatorname{Sym}_G(\rho)=\operatorname{Sym}_G(\rho')$};
\node[wideeqn] at (11.80,-8.46)
  {$\exists!\,r>0:\quad
    \mathcal{Q}_\rho^{q}=r\mathcal{Q}_{\rho'}^{q}\quad\forall q$};
\node[wideref,text=figNew] at (11.80,-8.80)
  {Corollary~\ref{cor:reversibility}};

\end{tikzpicture}
\endgroup

\caption{
\textbf{Single measure versus a family of measures:
conversion rates and reversibility.}
The upper box gives the complete conversion-rate formula of Theorem~\ref{thm:main_theorem_conversion_rate} under the symmetry-subgroup inclusion condition. Blue dashed and green solid frames distinguish known results from this work.
\emph{Optimal conversion rates:} (a) $U(1)$ pure-to-mixed formation is determined by the ratio of SLD QFIs~\cite{marvianOperationalInterpretationQuantum2022}, which equals the ratio of metric adjusted QGTs at $q=1/2$;
(b) pure-to-pure conversion is characterized by a single QGT~\cite{yamaguchi_QuantumGeometricTensorDeterminesPureState_2026}, which unifies earlier pure-state conversion results~\cite{gourResourceTheoryQuantum2008,marvianOperationalInterpretationQuantum2022,shitara_IidStateConvertibilityResourceTheory_2025,yangUnitsRotationalInformation2017};
(c) the general conversion law cannot be reduced to a state-independent fixed finite set of $q$ values (Proposition~\ref{prop:finite_q_is_insufficient});
and (d) a pure target reduces the characterization to the endpoint $\mathcal{Q}^1$ (Corollary~\ref{cor:distillation}).
\emph{Asymptotic reversibility:}
(e) pure-state reversibility is characterized by equal symmetry subgroups and proportional QGTs;
(f) general-state reversibility is characterized by equal symmetry subgroups and proportionality of the full metric adjusted QGT family with a common factor (Corollary~\ref{cor:reversibility}).
In (e) and (f), $\exists!\,r>0$ denotes a unique positive proportionality factor.
Throughout, $q\in[1/2,1]$, and a ``single measure'' may be matrix-valued.
}
\label{fig:conversion_rate_law_comparison}
\end{figure*}
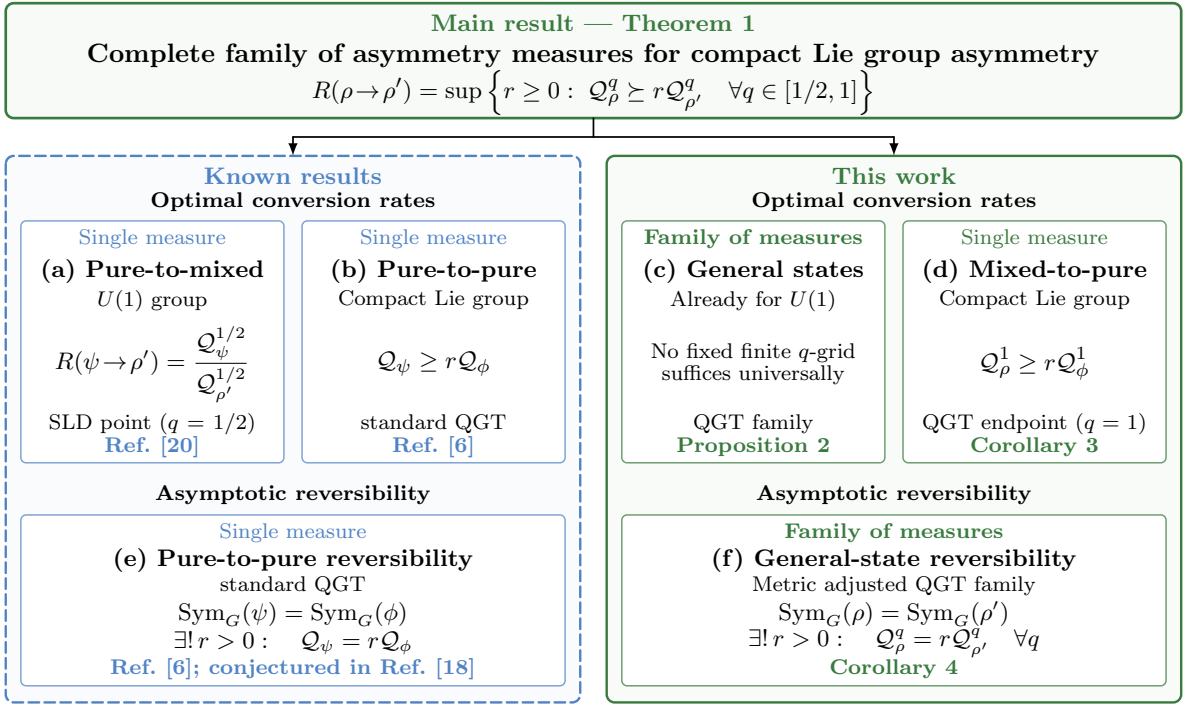

We emphasize, however, that the necessity of the continuous metric adjusted QGT family does not imply that the conversion rate must be computed by solving matrix inequalities over a continuum of $q$ values.
Indeed, the Supplemental Material~\cite{sm} provides a semidefinite-programming (SDP) formulation of the conversion rate,
without explicitly evaluating the metric adjusted QGT family. 
The SDP solves the computational problem of evaluating the rate, whereas the central resource-theoretic result is the identification of the metric adjusted QGT family as the complete set of asymmetry measures that determine asymptotic conversion rates.

\section{Distillation and reversibility}
This section presents two immediate consequences of Theorem~\ref{thm:main_theorem_conversion_rate}: a distillation-rate formula determined by the QGT endpoint and a necessary and sufficient condition for asymptotic reversibility.

\subsection{Reduction to a single complete measure in distillation}
% Let $\Pi_\rho$ be the projector onto the support of $\rho$, and define the generalized QGT
% \begin{align}
%     \left(\mathcal{Q}_\rho\right)_{ij}\coloneqq \Tr(\rho X_i (I-\Pi_\rho)X_j).
% \end{align}
% This quantity was introduced in Ref.~\cite{yamaguchi_QuantumGeometricTensorDeterminesPureState_2026}. For pure states, the generalized QGT reduces to the standard QGT~\cite{provostRiemannianStructureManifolds1980,berryQuantumPhaseFive1989}. 

Separating transitions involving the kernel of $\rho$ from those within its support gives
\begin{align}
    \mathcal{Q}_{\rho}^{q}=q\mathcal{Q}_\rho^1+(1-q)(\mathcal{Q}_\rho^1)^\top+\tilde{\mathcal{Q}}_\rho^{q}.\label{eq:decomposition_F_q_zeroevs}
\end{align}
Here, the within-support contribution is 
\begin{align}
    \left(\tilde{\mathcal{Q}}_\rho^{q}\right)_{ij}\coloneqq \sum_{\substack{k,l;\\\mu_l>0,\mu_k>0} }\frac{q(1-q)(\mu_k-\mu_l)^2\braket{l|X_i|k}\braket{k|X_j|l}}{(1-q)\mu_l+q\mu_k}.
\end{align}
It is positive semidefinite, vanishes identically for pure states, and tends to zero as $q\to0^+$ and $q\to 1^-$ for any fixed state $\rho$. 

For a pure output, the complete family in Theorem~\ref{thm:main_theorem_conversion_rate} collapses to a single endpoint. The resulting formula promotes the upper bound on the distillation rate in Ref.~\cite{yamaguchi_QuantumGeometricTensorDeterminesPureState_2026} to an achievable rate:
\begin{corollary}\label{cor:distillation}
    For an arbitrary state $\rho$ and a pure state $\phi$ such that $\Sym_G(\rho)\subset\Sym_G(\phi)$, 
    \begin{align}
        R(\rho\to\phi)=\sup\{r\geq 0\colon \mathcal{Q}_\rho^1\geq r \mathcal{Q}_\phi^1\},\label{eq:distillation_rate}
    \end{align}
\end{corollary}
\begin{proof}
    Necessity is the $q=1$ condition in Theorem~\ref{thm:main_theorem_conversion_rate}.
    Conversely, let $\Delta\coloneqq \mathcal{Q}_\rho^1- r \mathcal{Q}_\phi^1\geq 0$. Then we also have $\Delta^\top\geq 0$, and since $\tilde{\mathcal{Q}}_\phi^{q}=0$ for any pure state $\phi$, 
    \begin{align}
        &\mathcal{Q}_{\rho}^{q}-r\mathcal{Q}_{\phi}^{q}=q\Delta+(1-q)\Delta^\top+\tilde{\mathcal{Q}}_{\rho}^{q}\geq0,
    \end{align}
    for every $q\in[1/2,1]$. Theorem~\ref{thm:main_theorem_conversion_rate} then gives Eq.~\eqref{eq:distillation_rate}. 
\end{proof}
When the input state is also pure, the endpoint tensors are the standard QGTs, and Corollary~\ref{cor:distillation} reduces to the pure-state conversion law of Ref.~\cite{yamaguchi_QuantumGeometricTensorDeterminesPureState_2026}, which in turn recovers earlier special cases of pure-state i.i.d. asymmetry conversion~\cite{gourResourceTheoryQuantum2008,marvianOperationalInterpretationQuantum2022,yangUnitsRotationalInformation2017,shitara_IidStateConvertibilityResourceTheory_2025,marvianmashhadSymmetryAsymmetryQuantum2012,marvianAsymmetryPropertiesPure2014}.

For $G=U(1)$, all the matrix-valued quantities reduce to scalars because $\dim G=1$. A canonical application of $G=U(1)$ asymmetry theory is quantum clocks. Let $\phi$ be a pure asymmetric state serving as a reference clock, and let $\rho$ be a state with the same symmetry subgroup, or equivalently, the same period. Define the formation cost and distillation yield per copy of $\rho$ by
\begin{align}
    C_\phi(\rho)\coloneqq \frac{1}{R(\phi\to \rho)},\qquad D_\phi(\rho)\coloneqq R(\rho\to\phi). 
\end{align}
Corollary~\ref{cor:distillation} gives 
\begin{align}
    D_\phi(\rho)=\frac{\mathcal{Q}_\rho^1}{\mathcal{Q}_\phi}=\frac{\mathcal{Q}_\rho^1}{\Var_\phi(H_\phi)}\label{eq:distillation_U1},
\end{align}
where we used $\mathcal{Q}_\phi^q=\Var_\phi(H_\phi)\coloneqq \braket{\phi|H_\phi^2|\phi}-\braket{\phi|H_\phi|\phi}^2$ for any $q\in[0,1]$. 
This formula yields a necessary and sufficient criterion for bound asymmetry with respect to pure-clock distillation. 
Here, by analogy with bound entanglement~\cite{horodecki_SeparabilityCriterionInseparablemixedstates_1997,horodecki_MixedStateEntanglementDistillationthereBound_1998}, we say that an asymmetric state $\rho$ has bound asymmetry if $D_\phi(\rho)=0$. Let $H$ denote the $U(1)$ generator, i.e., the Hamiltonian, on the system of $\rho$. Equation~\eqref{eq:distillation_U1} gives
\begin{align}
    D_\phi(\rho)=0\Longleftrightarrow \mathcal{Q}_\rho^1=0\Longleftrightarrow  [\Pi_\rho,H]=0.\label{eq:distillation_condition_U1}
\end{align}
The implication $[\Pi_\rho,H]=0\,\implies\,D_\phi(\rho)=0$ was established in Ref.~\cite{marvianCoherenceDistillationMachines2020}; Eq.~\eqref{eq:distillation_U1} proves the converse. 
An asymmetric state serves as a quantum reference frame~\cite{bartlettReferenceFramesSuperselection2007,gourResourceTheoryQuantum2008,gourMeasuringQualityQuantum2009}. Thus, the above condition also implies that a quantum reference frame $\rho$ satisfying Eq.~\eqref{eq:distillation_condition_U1} cannot yield a pure-state quantum reference frame $\phi$ from $\rho$ at a nonvanishing rate, even when it is asymmetric.

On the other hand, the formation cost for $U(1)$ asymmetry is determined by the SLD QFI~\cite{marvianOperationalInterpretationQuantum2022}, or equivalently by the metric adjusted QGT at the SLD point:
\begin{align}
    C_\phi(\rho)=\frac{\mathcal{F}_\rho^{f_{1/2}}}{\mathcal{F}_\phi^{f_{1/2}}}=\frac{\mathcal{Q}_\rho^{1/2}}{\mathcal{Q}_\phi^{1/2}}=\frac{\mathcal{Q}_\rho^{1/2}}{\Var_\phi(H_\phi)}.
\end{align}
This formula also follows from Theorem~\ref{thm:main_theorem_conversion_rate} using a scalar inequality for the $U(1)$ group
\begin{align}
    \mathcal{Q}_{\sigma}^{q}\leq \mathcal{Q}_{\sigma}^{{1/2}},\label{eq:Fq_SLD_inequality}
\end{align}
which is proven in the Supplemental Material~\cite{sm}.

The formation cost and distillation yield are thus respectively measured by the two endpoints $q=1/2$ and $q=1$, in the same units set by the pure clock $\phi$. Their difference is
\begin{align}
    C_\phi(\rho)-D_\phi(\rho)=\frac{\mathcal{Q}_\rho^{1/2}-\mathcal{Q}_\rho^1}{\Var_\phi(H_\phi)}=\frac{\tilde{\mathcal{Q}}_\rho^{1/2}}{\Var_\phi(H_\phi)}.
\end{align}
The within-support contribution $\tilde{\mathcal{Q}}_\rho^{1/2}$ quantifies the formation cost that is not recovered by pure-clock distillation. Such irreversibility is also captured by the corresponding round-trip efficiency 
\begin{align}
    \eta_{\mathrm{rt}}(\rho)\coloneqq R(\phi\to\rho)R(\rho\to\phi)=\frac{\mathcal{Q}_\rho^1}{\mathcal{Q}_\rho^{1/2}}=1-\frac{\tilde{\mathcal{Q}}_\rho^{1/2}}{\mathcal{Q}_\rho^{1/2}},
\end{align}
which gives the fraction of pure-clock resource recovered after an optimal formation-and-distillation cycle. The unrecoverable fraction, $\frac{\tilde{\mathcal{Q}}_\rho^{1/2}}{\mathcal{Q}_\rho^{1/2}}$, is precisely the fraction of the QGT at the SLD point arising from transitions within the support of $\rho$.

\subsection{Necessary and sufficient condition for asymptotic reversibility}
We call the asymptotic conversion between $\rho$ and $\rho'$ reversible if $R(\rho\to\rho')R(\rho'\to\rho)=1$. 
Theorem~\ref{thm:main_theorem_conversion_rate} also yields a necessary and sufficient condition for this reversibility.
\begin{corollary}\label{cor:reversibility}
    The asymptotic conversion between $\rho$ and $\rho'$ is reversible if and only if $\Sym_G(\rho)=\Sym_G(\rho')$ and there exists a unique $r>0$ such that
    \begin{align}
        \forall q\in[1/2,1],\quad \mathcal{Q}^{q}_\rho= r \mathcal{Q}^{q}_{\rho'}.\label{eq:reversibility_F}
    \end{align}
\end{corollary}
\begin{proof}
    Suppose $R(\rho\to\rho')R(\rho'\to\rho)=1$. Since $R(\rho\to\rho')$ and $R(\rho'\to\rho)$ must be positive, we have $\Sym_G(\rho)\subset\Sym_G(\rho')$ and $\Sym_G(\rho')\subset\Sym_G(\rho)$, and hence $\Sym_G(\rho)=\Sym_G(\rho')$. 
    Moreover, writing $r\coloneqq R(\rho\to\rho')$, we have $R(\rho'\to\rho)=1/r$. The conversion-rate formula for each conversion direction yields $\mathcal{Q}^{q}_\rho\geq r \mathcal{Q}^{q}_{\rho'}$ and $r \mathcal{Q}^{q}_{\rho'}\geq \mathcal{Q}^{q}_\rho$ for all $q\in [1/2,1]$. Indeed, although the rate is defined as a supremum, the inequalities remain valid at the optimal rate by taking a sequence of achievable rates that converges to the supremum. 
    Consequently, we obtain Eq.~\eqref{eq:reversibility_F}. 

    Moreover, the proportionality factor $r$ must be unique; otherwise, if two distinct proportionality factors satisfied Eq.~\eqref{eq:reversibility_F}, their difference would imply $\mathcal{Q}^{q}_{\rho'}=0$, and hence $\mathcal{Q}^{q}_\rho=0$ for all $q\in(0,1)$. In this case Theorem~\ref{thm:main_theorem_conversion_rate} would give an unbounded conversion rate, contradicting $R(\rho\to\rho')R(\rho'\to\rho)=1$.
    
    Conversely, if $\Sym_G(\rho)=\Sym_G(\rho')$ and there exists a unique $r>0$ satisfying Eq.~\eqref{eq:reversibility_F}, then Theorem~\ref{thm:main_theorem_conversion_rate} gives $R(\rho\to\rho')=r$ and $R(\rho'\to\rho)=1/r$, and hence $R(\rho\to\rho')R(\rho'\to\rho)=1$. 
\end{proof}

For pure states $\psi,\phi$, Eq.~\eqref{eq:reversibility_F} reduces to proportionality of the standard QGT. Indeed, necessity follows by evaluating the condition at $q=1$:
\begin{align}
    \mathcal{Q}_\psi=r\mathcal{Q}_\phi
\end{align}
Conversely, this equality and its transpose imply $\mathcal{Q}_\psi^q=r\mathcal{Q}_\phi^q$ for every $q\in[1/2,1]$ by Eq.~\eqref{eq:QGT_q_pure}.
A condition equivalent to this QGT criterion, formulated without using the QGT, was conjectured by Marvian and Spekkens~\cite{marvianAsymmetryPropertiesPure2014}, and was later proven from the pure-state QGT conversion-rate formula in Ref.~\cite{yamaguchi_QuantumGeometricTensorDeterminesPureState_2026}.

For general states, Corollary~\ref{cor:reversibility} extends this characterization from a single QGT to the full metric adjusted QGT family: reversible conversion requires the same proportionality factor for every $q$, together with the equality of their symmetry subgroups. This is the same single-measure-versus-family distinction highlighted in Figure~\ref{fig:conversion_rate_law_comparison}. To the authors' best knowledge, no necessary and sufficient condition
for asymptotic reversibility had previously been proposed or established for general mixed states in the resource theory of asymmetry.

\subsection{Example: distillation and reversible conversion in quantum clocks}\label{sec:Example_quantum_clock}

To illustrate both distillation and reversible conversion, consider a qutrit system with an orthonormal basis $\{\ket{0},\ket{1},\ket{2}\}$, whose Hamiltonian is given by $H=\ket{1}\bra{1}+2\ket{2}\bra{2}$. Introducing another orthonormal basis
\begin{align}
    \ket{v_{\pm}}&\coloneqq \frac{1}{2}\left(\ket{0}\pm\sqrt{2}\ket{1}+\ket{2}\right),\\
    \ket{v_0}&\coloneqq \frac{1}{\sqrt{2}}(\ket{0}-\ket{2}),
\end{align}
we have
\begin{align}
    \braket{v_+|H|v_-}&=0,\\
    |\braket{v_+|H|v_0}|^2&=|\braket{v_0|H|v_-}|^2=\frac{1}{2},
\end{align}
implying that the Hamiltonian couples only adjacent states $\ket{v_+}\leftrightarrow\ket{v_0}\leftrightarrow\ket{v_-}$. Consider a state
\begin{align}
    \sigma\coloneqq \frac{3}{4}\ket{v_+}\bra{v_+}+\frac{1}{4}\ket{v_0}\bra{v_0},
\end{align}
which satisfies $\Sym_{U(1)}(\sigma)=\{e\}$.
Equation~\eqref{eq:decomposition_F_q_zeroevs} decomposes the metric adjusted QGT into the support-kernel and support-support contributions as
\begin{align}
    \mathcal{Q}_{\sigma}^{q}=\mathcal{Q}_\sigma^1+\tilde{\mathcal{Q}}^{q}_{\sigma}\label{eq:QFI_sigma}
\end{align}
with
\begin{align}
    \mathcal{Q}_\sigma^1=\frac{1}{8},\qquad \tilde{\mathcal{Q}}_\sigma^{q}=\frac{2q(1-q)}{(1+2q)(3-2q)}.
\end{align}

As a standard pure reference clock~\cite{yamaguchiIidResourceTheory2023,marvianOperationalInterpretationQuantum2022}, we adopt the coherence bit $\ket{\phi}\coloneqq (\ket{0}+\ket{1})/\sqrt{2}$ with Hamiltonian $H_\phi\coloneqq \ket{1}\bra{1}$, which satisfies $\Sym_{U(1)}(\ket{\phi}\bra{\phi})=\{e\}$. As Corollary~\ref{cor:distillation} shows, only the support-kernel term contributes to the pure-state distillation, and the distillation rate is given by
\begin{align}
    D_\phi(\sigma)=R(\sigma\to\phi)=\frac{\mathcal{Q}_\sigma^1}{\mathcal{Q}_\phi^1}=\frac{1}{2},
\end{align}
where we used $\mathcal{Q}_\phi^1=\Var_\phi(H_\phi)=1/4$ for the coherence bit. 
Moreover, since $\mathcal{Q}_{\sigma}^{1/2}=1/4$ and $\tilde{\mathcal{Q}}_{\sigma}^{1/2}=1/8$, the round-trip efficiency is given by 
\begin{align}
    \eta_{\mathrm{rt}}(\sigma)=R(\phi\to\sigma)R(\sigma\to\phi)=\frac{1}{2}.
\end{align}

We now investigate reversible conversion. Consider a two-qubit system with Hamiltonian $H_1\otimes I+I\otimes H_2$ and $H_1=H_2=\ket{1}\bra{1}$, and a state $\rho\coloneqq\rho_1\otimes \rho_2$ with
\begin{align}
    \rho_1\coloneqq \frac{1}{2}\left(I+\frac{1}{2}X+\frac{\sqrt{3}}{2}Z\right),\, \rho_2\coloneqq \frac{1}{2}\left(I+\frac{1}{2}X\right).
\end{align}
The state $\rho_1$ is pure, while $\rho_2$ is full-rank. Their symmetry subgroups are trivial: $\Sym_{U(1)}(\rho_1)=\Sym_{U(1)}(\rho_2)=\{e\}$, and hence $\Sym_{U(1)}(\rho)=\{e\}$. 
A direct calculation yields
\begin{align}
    \mathcal{Q}_{\rho_1}^{q}=\frac{1}{16},\quad \mathcal{Q}_{\rho_2}^{q}=\frac{q(1-q)}{(1+2q)(3-2q)},\label{eq:QFI_rho1_rho2}
\end{align}
implying that
\begin{align}
    \mathcal{Q}_{\rho}^{q}=\mathcal{Q}_{\rho_1}^{q}+\mathcal{Q}_{\rho_2}^{q}=\frac{1}{2} \mathcal{Q}_{\sigma}^{q}
\end{align}
for all $q\in[0,1]$. Thus, Corollary~\ref{cor:reversibility} ensures that the conversion between $\rho$ and $\sigma$ is asymptotically reversible, and indeed the conversion rates are given by
\begin{align}
    R(\sigma\to\rho_1\otimes\rho_2)=2,\qquad R(\rho_1\otimes\rho_2\to\sigma)=\frac{1}{2}.\label{eq:reversiblity_example}
\end{align}

This example highlights the different roles of the two contributions in Eq.~\eqref{eq:decomposition_F_q_zeroevs}. Pure-state distillation $\sigma\to\phi$ depends only on the support-kernel contribution $\mathcal Q_\sigma^1$. By contrast, reversibility between the mixed state $\sigma$ and $\rho=\rho_1\otimes\rho_2$ is governed by the full metric adjusted QGT family. In the present example, both contributions satisfy the proportionality condition:
\begin{align}
    \mathcal{Q}_\rho^1=\frac{1}{2}\mathcal{Q}_\sigma^1,\qquad \tilde{\mathcal{Q}}_\rho^{q}=\frac{1}{2}\tilde{\mathcal{Q}}_\sigma^{q}\text{  for all }q\in(0,1).
\end{align}
Thus, the support-support contribution, which is not recovered by pure-state distillation, can still contribute to reversible mixed-state conversion. This distinction also underlies the activation mechanism discussed next.

\section{Activation of asymmetry}
From the definition of the asymptotic conversion rate, parallel execution of separate conversion protocols implies
\begin{align}
    R(\rho_1\otimes\rho_2\to\sigma)\geq R(\rho_1\to\sigma)+R(\rho_2\to\sigma).
\end{align}
A strict inequality therefore signals a genuine advantage of joint processing over separate conversion. Such a strict superadditivity can arise when combining inputs removes a symmetry-subgroup obstruction, as in the synchronization protocol of Ref.~\cite{yamaguchi_QuantumGeometricTensorDeterminesPureState_2026}. 

The complete metric adjusted QGT characterization in Theorem~\ref{thm:main_theorem_conversion_rate} reveals a new mechanism, in which the inputs compensate for each other's deficits associated with different members of the metric adjusted QGT family. In what follows, we take $G=U(1)$ and only consider states satisfying the symmetry-subgroup condition $\Sym_G(\rho_i)\subset \Sym_G(\sigma)$.

Suppose first that the conversion rate were solely determined by an additive scalar measure $\mathcal{M}$, with $\mathcal{M}(\sigma)>0$, namely, 
\begin{align}
    R(\tau\to \sigma)=\frac{\mathcal{M}(\tau)}{\mathcal{M}(\sigma)}
\end{align}
for $\tau=\rho_1,\rho_2$ and $\rho_1\otimes\rho_2$. 
Then, strict superadditivity is impossible since
\begin{align}
    R(\rho_1\otimes \rho_2\to\sigma)&=\frac{\mathcal{M}(\rho_1\otimes \rho_2)}{\mathcal{M}(\sigma)}\\
    &=\frac{\mathcal{M}(\rho_1)+\mathcal{M}(\rho_2)}{\mathcal{M}(\sigma)}\\
    &=R(\rho_1\to\sigma)+R(\rho_2\to\sigma).
\end{align}
For the $U(1)$ group, the endpoint $\mathcal{Q}^1$ provides such a measure when the output is pure, and so does the metric adjusted QGT at the SLD point $\mathcal{Q}^{1/2}$ when the inputs $\rho_i$ are pure. Thus, these additive single-measure laws disable improvement via joint processing. 

For an asymmetric mixed output $\sigma$ and general inputs, however, Theorem~\ref{thm:main_theorem_conversion_rate} and additivity of the metric adjusted QGT give
\begin{align}
    R(\rho_1\otimes\rho_2\to \sigma)=\inf_{q\in[1/2,1)}\left\{r_{1}(q)+r_2(q)\right\},\label{eq:rate_joint_proc}
\end{align}
where 
\begin{align}
    r_{i}(q)\coloneqq \frac{\mathcal{Q}_{\rho_i}^{q}}{\mathcal{Q}_{\sigma}^{q}}.
\end{align}
On the other hand, by Theorem~\ref{thm:main_theorem_conversion_rate}, the individual rate is given by $R(\rho_i\to\sigma)=\inf_{q\in[1/2,1)}r_i(q)$, and hence
\begin{align}
    R(\rho_1\to\sigma)+R(\rho_2\to\sigma)=\inf_{q\in[1/2,1)}r_1(q)+\inf_{q\in[1/2,1)}r_2(q).\label{eq:rate_separate_proc}
\end{align}

Comparison between Eqs.~\eqref{eq:rate_joint_proc} and~\eqref{eq:rate_separate_proc} shows that joint processing adds their metric adjusted QGT contributions before taking the infimum, allowing one input to compensate for the constraint that limits the other. The resulting rate can therefore exceed the sum of the separate rates. This ``resource complementarity'' arises across the metric adjusted QGT family, even though each individual QGT remains additive. 

The states $\rho_1,\rho_2,\sigma$ in Section~\ref{sec:Example_quantum_clock} provide a concrete example. Indeed, from Eqs.~\eqref{eq:QFI_sigma} and~\eqref{eq:QFI_rho1_rho2}, we obtain their ratios
\begin{align}
    r_1(q)&=\frac{(1+2q)(3-2q)}{2(3+20q(1-q))},\quad r_2(q)=\frac{8q(1-q)}{3+20q(1-q)},
\end{align}
Figure~\ref{fig:fq_activation} shows the behavior of these quantities.
The ratio $r_1(q)$ attains its minimum $1/4$ at the SLD point $q=1/2$, while the second ratio $r_2(q)$ approaches its infimum $0$ as $q\to 1^-$. Nevertheless, their sum is constant:
\begin{align}
    r_1(q)+r_2(q)=\frac{1}{2},\qquad \forall q\in[1/2,1).
\end{align}
Therefore,
\begin{align}
    R(\rho_1\to\sigma)=\frac{1}{4}&,\qquad R(\rho_2\to\sigma)=0,\\
    R(\rho_1\otimes\rho_2\to\sigma)&=\frac{1}{2},
\end{align}
implying that the zero-rate mixed state $\rho_2$ doubles the yield from $\rho_1$. This is analogous to activation of channel capacities~\cite{duan_SuperActivationZeroErrorCapacityNoisyQuantum_2009,li_PrivateCapacityQuantumChannelsNot_2009,smith_ExtensiveNonadditivityPrivacy_2009}, and therefore, we call it asymmetry activation. 

The metric adjusted QGT family clarifies the mechanism of this improvement. The pure input $\rho_1$ supplies the support-kernel contribution required by the output state, but its individual rate is limited at the SLD point. The full-rank input $\rho_2$ does not supply the support-kernel contribution, i.e., $\mathcal{Q}_{\rho_2}^1=0$, while $\mathcal{Q}_\sigma^1=1/8>0$. Thus, the endpoint condition in Theorem~\ref{thm:main_theorem_conversion_rate} forces its individual conversion rate to vanish. Nevertheless, the support-support contribution of $\rho_2$ supplies the missing SLD-point QGT. The two inputs compensate for each other's deficits, thereby satisfying every metric adjusted QGT constraint at rate $1/2$. A contribution that cannot be recovered by pure-clock distillation thus increases the optimal yield of a mixed-state conversion, resulting in the asymptotic reversibility in Eq.~\eqref{eq:reversiblity_example}.

\begin{figure}
    \centering
    \includegraphics[width=0.85\linewidth]{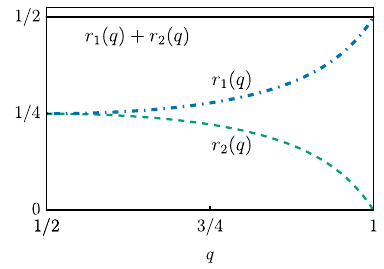}
    \caption{\textbf{Asymmetry activation via metric adjusted QGT complementarity.}
    The infima of the individual ratios $r_1(q)\coloneqq  \frac{\mathcal{Q}_{\rho_1}^{q}}{\mathcal{Q}_{\sigma}^{q}}$ and $r_2(q)\coloneqq \frac{\mathcal{Q}_{\rho_2}^{q}}{\mathcal{Q}_{\sigma}^{q}}$ are $1/4$ and $0$, which equal $R(\rho_1\to\sigma)$ and $R(\rho_2\to\sigma)$, respectively. The sum, $r_1(q)+r_2(q)$, is $1/2$ for all $q\in[1/2,1)$, hence the joint conversion rate $R(\rho_1\otimes\rho_2\to\sigma)=1/2$. Thus, the zero-rate mixed state $\rho_2$ doubles the resulting yield when jointly processed with $\rho_1$.} 
    \label{fig:fq_activation}
\end{figure}

\section{Proof sketch}\label{sec:proof_sketch}
The proof of Theorem~\ref{thm:main_theorem_conversion_rate} consists of two parts: a converse part and a direct part. The converse part provides an upper bound of the optimal conversion rate, while the direct part proves its tightness.

\subsection{Converse part}
Suppose that $\rho$ is convertible to $\rho'$ at rate $r$ with a sequence of channels $\mathcal{E}_n$. The actual output $\chi_n\coloneqq \mathcal{E}_n(\rho^{\otimes n})$ satisfies $\mathcal{Q}_{\rho^{\otimes n}}^{q}\geq\mathcal{Q}_{\chi_n}^{q}$ due to the monotonicity of metric adjusted QGTs. By the additivity of metric adjusted QGT, $n\mathcal{Q}_{\rho}^{q}=\mathcal{Q}_{\rho^{\otimes n}}^{q}$, which implies
\begin{align}
    \mathcal{Q}_{\rho}^{q}\geq \frac{1}{n}\mathcal{Q}_{\chi_n}^{q}.\label{eq:monotonicity_fq_chi}
\end{align}
Since the trace-distance error $\epsilon_n\coloneqq T(\chi_n,\rho'^{\otimes \floor{rn}})$ converges to zero, one may naively expect that $\mathcal{Q}_{\chi_n}^{q}$ can directly be replaced by $\mathcal{Q}_{\rho'^{\otimes \floor{rn}}}^{q}=m_n\mathcal{Q}_{\rho'}^{q}$ with $m_n\coloneqq \floor{rn}$, thereby obtaining $\mathcal{Q}_{\rho}^{q}\geq r \mathcal{Q}_{\rho'}^{q}$ in the limit of $n\to\infty$. 

However, a trace-distance perturbation of size $\epsilon_n$ can change the $f_q$-QFI, and hence the metric adjusted QGT, by order $\epsilon_nn^2$. Consequently, vanishing trace-distance error alone does not justify replacing the per-copy tensor by its i.i.d. value, and a more careful analysis is required~\cite{gourMeasuringQualityQuantum2009,marvianCoherenceDistillationMachines2020,marvianOperationalInterpretationQuantum2022,yamaguchiSmoothMetricAdjusted2023,yamaguchi_QuantumGeometricTensorDeterminesPureState_2026}.

To overcome this difficulty, we show that there is a real-valued function $h_q$ such that
\begin{align}
    \frac{1}{ m_n}\mathcal{Q}_{\chi_n}^{q}&\geq\mathcal{Q}_{\rho'}^{q}-q(1-q)h_q(\epsilon_n)\mathcal{X},\label{eq:continuity_iid}\\
    \lim_{\epsilon\to0^+}h_q(\epsilon)&=0
\end{align}
for $q\in(0,1)$, where $(\mathcal{X})_{ij}\coloneqq \Tr(\tilde{X}_i'\tilde{X}_j')$ with $\tilde{X}_i'\coloneqq X_i'-(\Tr(X_i')/d')I$. Combining Eqs.~\eqref{eq:monotonicity_fq_chi} and~\eqref{eq:continuity_iid}, and taking the limit of $n\to\infty$, we obtain
\begin{align}
   \mathcal{Q}_{\rho}^{q}\geq r \mathcal{Q}_{\rho'}^{q},
\end{align}
for $q\in(0,1)$ and continuity then gives the endpoint inequalities at $q=0,1$. Thus,
\begin{align}
    R(\rho\to\rho')\leq \sup\{r\geq 0\colon \forall q\in [1/2,1],\,\mathcal{Q}_{\rho}^{q}\geq r \mathcal{Q}_{\rho'}^{q}\}.
\end{align}

\subsection{Direct part}

In the direct part, assuming the symmetry-subgroup condition, we show that every rate $r$ strictly below the metric adjusted QGT bound in Eq.~\eqref{eq:main_theorem_conversion_rate} is asymptotically achievable. We use the equivalence between convertibility from $\rho$ to $\rho'$ via a $G$-covariant channel and convertibility from a statistical model $\{\mathcal{U}_g(\rho)\}_{g\in G}$ to $\{\mathcal{U}_g'(\rho')\}_{g\in G}$ via a parameter-independent quantum channel~\cite{marvianmashhadSymmetryAsymmetryQuantum2012,yamaguchi_QuantumGeometricTensorDeterminesPureState_2026}. In the latter setting, $g\in G$ is treated as an unknown parameter and the desired convertibility means $\mathcal{E}_n(\mathcal{U}_g(\rho)^{\otimes n})\approx \mathcal{U}_g'(\rho')^{\otimes m_n}$ uniformly over $g\in G$, where $\mathcal{E}_n$ is not necessarily a $G$-covariant channel.

To achieve such a conversion, we adopt the estimate-and-convert strategy~\cite{yamaguchi_QuantumGeometricTensorDeterminesPureState_2026}.
We first use a sublinear (in $n$) number of copies of $\mathcal{U}_g(\rho)$ to obtain an estimate of $g\in G$, which identifies the parameter region within an error of order $\sim n^{-1/2}$ with an exponentially small failure probability.
Since $r$ is strictly below the metric adjusted QGT bound, the local conversion can be performed at a slightly larger intermediate rate; hence, the sublinear number of copies used for estimation does not affect the final rate $r$.
After this parameter localization, QLAN applies and reversibly converts each i.i.d. statistical model into a Gaussian shift model:
\begin{align}
    \mathcal{U}_{g(un^{-1/2})}(\rho)^{\otimes n}\mathrel{\overset{T_{n}}{\underset{S_{n}}{\rightleftarrows}}}\mathcal{G}_u,\\
    \mathcal{U}'_{g(un^{-1/2})}(\rho')^{\otimes m_n}\mathrel{\overset{T_{m_n}'}{\underset{S_{m_n}'}{\rightleftarrows}}}\mathcal{G}'_{\sqrt{r}u},
\end{align}
with vanishing error as $n\to\infty$ in trace distance. Here, $u$ is the remaining unknown parameter, $T_n,S_n,T'_{m_n},S'_{m_n}$ are quantum channels, and $\mathcal{G}_u$ and $\mathcal{G}_u'$ are the corresponding Gaussian shift models whose displacement is parametrized by $u\in\mathbb{R}^{\dim G}$.
Since QLAN has been established under the assumption that positive eigenvalues are nondegenerate~\cite{gutaLocalAsymptoticNormality2006,kahnQuantumLocalAsymptotic2008,kahnLocalAsymptoticNormality2009,lahiryMinimaxEstimationLowrank2024}, we extend it to general finite-dimensional unitary models, allowing arbitrary rank and degeneracies among the positive eigenvalues, in the Supplemental Material. 
The resulting Gaussian shift models have the same family of $f_q$-QFI matrices as the corresponding unitary models.
We then prove that the family of $f_q$-QFI inequalities is equivalent to the exact convertibility between the Gaussian shift models. Using Eq.~\eqref{eq:definition_Q_F} and continuous extension to the endpoints $q=0,1$, this gives
\begin{align}
    &\mathcal{Q}_{\rho}^{q}\geq r\mathcal{Q}_{\rho'}^{q}\qquad \forall q\in[1/2,1]\nonumber\\
    &\Longleftrightarrow \exists\text{a quantum channel } \Lambda\text{ independent of $u$ such that }\nonumber\\
    &\qquad \qquad \Lambda(\mathcal{G}_u)=\mathcal{G}'_{\sqrt{r}u},\quad \forall u\in\mathbb{R}^{\dim G}.
\end{align}
Consequently, the composite channels $\Lambda_n\coloneqq S_{m_n}' \circ \Lambda \circ T_n$ achieve asymptotic local conversion
\begin{align}
    \Lambda_n\left(\mathcal{U}_{g(un^{-1/2})}(\rho)^{\otimes n}\right)\approx \mathcal{U}'_{g(un^{-1/2})}(\rho')^{\otimes m_n}
\end{align}
with vanishing error. Concretely, also taking into account the failure probability, we derive a uniform upper bound on the conversion error of order $n^{-1/2+\kappa}$ for any $\kappa\in (0,1/2)$. 
Combining these local channels with the estimation step and using the equivalence between $G$-covariant state conversion and conversion of the associated statistical models, we conclude that the rate $r$ is achievable with error of order $n^{-1/2+\kappa}$.

\section{Conclusion}

This paper establishes a formula for i.i.d. asymptotic conversion between arbitrary states on finite-dimensional systems in the resource theory of asymmetry for any compact Lie group. This identifies a complete family of quantities characterizing symmetry breaking in this setting: a one-parameter family of metric adjusted QGTs, defined from QFI matrices interpolating between the SLD and RLD QFIs. Our formula is single-letter, i.e., it requires neither optimization over pure-state ensembles nor regularization over multiple copies, in contrast to formulas for mixed states that arise in some other resource theories~\cite{haydenAsymptoticEntanglementCost2001,devetakDistillationSecretKey2005}.
Combined with the connection between QFIs and linear-response theory~\cite{shitara_DeterminingContinuousFamilyquantumFisher_2016}, our formula expresses the conversion rate in terms of experimentally accessible quantities. 

Our formula for general states highlights the qualitative difference from the pure-state case. Unlike in the pure-state case, the rate is governed not by a single matrix-valued measure but by the family of metric adjusted QGTs; indeed, no state-independent finite subset of this family reproduces all conversion rates, even for $U(1)$ symmetry. Despite this continuum of constraints, the rate admits a computationally tractable finite-dimensional semidefinite programming formulation. Asymptotically reversible interconversion between asymmetric states occurs precisely when their symmetry subgroups coincide, and all QGTs in the family are proportional by a unique positive factor. The full family thus determines which conversions are reversible and which incur an unavoidable loss.

For pure-state distillation, the rate is governed by the endpoint $\mathcal{Q}^1$, which depends on matrix elements of the symmetry generators connecting the state's support to its kernel. For quantum clocks with a common period, the within-support contribution at the SLD point quantifies the gap between the cost of formation from pure clock states and the distillation yield into pure clock states. We demonstrated asymmetry activation phenomena arising from this contribution: a mixed state with zero asymptotic yield to a given mixed clock state increases the yield obtainable from another input when the two are processed jointly. Asymmetry that cannot be recovered through distillation into pure clock states can therefore retain value in mixed-state conversion.

\begin{acknowledgments}
The authors thank Yui Kuramochi, Yosuke Mitsuhashi, and Tomohiro Shitara for valuable discussions. 
K.Y. acknowledges support through the Dieter Schwarz Foundation, a Discovery Grant of the Natural Sciences and Engineering Research Council of Canada (NSERC), and is grateful for the hospitality of Perimeter Institute, where part of this work was carried out. Research at Perimeter Institute is supported in part by the Government of Canada through the Department of Innovation, Science and Economic Development Canada and by the Province of Ontario through the Ministry of Colleges, Universities, Research Excellence and Security.
H.T. was supported by JSPS Grants-in-Aid for Scientific Research No. JP25K00924, MEXT KAKENHI Grant-in-Aid for Transformative
Research Areas B ``Quantum Energy Innovation” Grant Numbers 24H00830 and 24H00831, JST FOREST No. JPMJFR2365, JST MOONSHOT No. JPMJMS256E and Royal Society International Collaboration Awards 2025 Flexigrant number ICA/R2/252240.

\end{acknowledgments}

\clearpage

\appendix
\widetext

\makeatletter
\ltx@footnote@pop
\makeatother

\setcounter{footnote}{0}
\setcounter{page}{1}
\setcounter{figure}{0}
\setcounter{table}{0}
\setcounter{equation}{0}
\setcounter{secnumdepth}{3}

% Appendix sections: A, B, C, ...
\renewcommand{\thesection}{\Alph{section}}
\renewcommand{\thesubsection}{\arabic{subsection}}
\renewcommand{\thesubsubsection}{\alph{subsubsection}}

% Reset theorem counter at every appendix section
\makeatletter
\@addtoreset{theorem}{section}
\makeatother

% Theorem-like environments: A.1, A.2, ...
\renewcommand{\thetheorem}{\thesection.\arabic{theorem}}
\renewcommand{\theconjecture}{\thetheorem}
\renewcommand{\theproposition}{\thetheorem}
\renewcommand{\thecorollary}{\thetheorem}
\renewcommand{\thedefinition}{\thetheorem}
\renewcommand{\thelemma}{\thetheorem}
\renewcommand{\theassumption}{\thetheorem}
\renewcommand{\theremark}{\thetheorem}
\renewcommand{\thefact}{\thetheorem}

% Figures / tables / equations
\renewcommand{\thefigure}{S\arabic{figure}}
\renewcommand{\thetable}{S\arabic{table}}
\renewcommand{\theequation}{\thesection.\arabic{equation}}
\renewcommand{\theHequation}{\theHsection.\arabic{equation}}

% hyperref anchors
\renewcommand{\theHtheorem}{\theHsection.\arabic{theorem}}
\renewcommand{\theHconjecture}{\theHtheorem}
\renewcommand{\theHproposition}{\theHtheorem}
\renewcommand{\theHcorollary}{\theHtheorem}
\renewcommand{\theHdefinition}{\theHtheorem}
\renewcommand{\theHlemma}{\theHtheorem}
\renewcommand{\theHassumption}{\theHtheorem}
\renewcommand{\theHremark}{\theHtheorem}
\renewcommand{\theHfact}{\theHtheorem}

\begin{center}
{\large \bf Supplemental Material for\\
``Quantifying Symmetry Breaking with Metric Adjusted Quantum Geometric Tensors''}\\
\vspace*{0.3cm}
Koji Yamaguchi$^{1,2}$ and Hiroyasu Tajima$^{3,4}$\\
\vspace*{0.1cm}

$^{1}${\small \it Department of Physics, University of Waterloo, Waterloo, ON N2L 3G1, Canada}
\\
$^{2}${\small \it Perimeter Institute for Theoretical Physics, Waterloo, Ontario N2L 2Y5, Canada}

$^{3}${\small \it Department of Informatics, Faculty of Information Science and Electrical Engineering,
Kyushu University, 744 Motooka, Nishi-ku, Fukuoka, 819-0395, Japan}

$^{4}${\small \it JST, FOREST, 4-1-8 Honcho, Kawaguchi, Saitama, 332-0012, Japan}

\end{center}

\vspace{1cm}

Throughout this Supplemental Material, all Hilbert spaces are complex, and their inner product $\braket{\cdot,\cdot}$ is anti-linear in the first argument and linear in the second. For a vector $v$, its standard norm is denoted by $\|v\|\coloneqq \sqrt{\braket{v,v}}$.

For a linear operator $T$ on a vector space $V$, we use three different norms; the trace norm: $\|T\|_1\coloneqq \Tr_V|T|=\Tr\sqrt{T^\dag T}$, the Hilbert--Schmidt norm: $\|T\|_\HS\coloneqq \sqrt{\Tr_V(T^\dag T)}$, and the operator norm: $\|T\|_\op\coloneqq \sup_{\|v\|=1;\,v\in V}\|Tv\|$. For notational simplicity, the operator norm is simply written as $\|T\|$ when no confusion can arise. 

\vspace{1cm}

% =========================
% Appendix Part I
% =========================

\begin{center}    
\textbf{{\large \underline{Part I: Asymptotic Conversion Rate in the Resource Theory of Asymmetry}}}
\end{center}
In this Part, we prove the optimal asymptotic conversion rate in the resource theory of asymmetry for arbitrary states on finite-dimensional systems for any compact Lie group. We postpone the proofs of two key results to Parts II and III: the convertibility between Gaussian shift models and the unitary-orbit QLAN.

\AppendixTOCOne{Contents of Part I}

\section{Metric adjusted quantum geometric tensor}

\subsection{Monotone metric and its variational expression} 
A function $f:(0,\infty)\to (0,\infty)$ is called an operator-monotone function if and only if for any positive operators $A$ and $B$ of the same size, $0<A\leq B$ implies $f(A)\leq f(B)$. By rescaling the function, we may assume $f(1)=1$ without loss of generality. For an invertible quantum state $\rho$ on a finite-dimensional Hilbert space $\mathcal{H}$ of dimension $d$, we define a superoperator $\mathbb{J}^f_\rho\coloneqq f(\mathbb{L}_\rho \mathbb{R}_\rho^{-1})\mathbb{R}_\rho$, where $\mathbb{L}_\rho$ and $\mathbb{R}_\rho$ are left- and right-multiplications defined by 
\begin{align}
    \mathbb{L}_\rho A=\rho A,\qquad \mathbb{R}_\rho A=A\rho.
\end{align}
More explicitly, using an eigenvalue decomposition of $\rho=\sum_{i=1}^d\mu_i\ket{i}\bra{i}$, we have
\begin{align}
    \mathbb{J}^f_\rho A=\sum_{i,j=1}^df(\mu_i/\mu_j)\mu_j\braket{i|A|j}\ket{i}\bra{j}.
\end{align}
We then define 
\begin{align}
    \braket{A,B}_{f,\rho}\coloneqq \braket{A,(\mathbb{J}^{f}_\rho)^{-1}B}
   =\sum_{i,j=1}^d \frac{1}{f(\mu_i/\mu_j)\mu_j}\braket{j|A^\dag|i}\braket{i|B|j},\label{eq:innerproduct_invertible}
\end{align}
where $\braket{A,B}\coloneqq \Tr(A^\dag B)$ denotes the Hilbert--Schmidt inner product. The norm $\|\cdot\|_{f,\rho}$ is then defined by
\begin{align}
    \|A\|_{f,\rho}^2\coloneqq \braket{A,(\mathbb{J}^{f}_\rho)^{-1}A}
   =\sum_{i,j=1}^d \frac{|\braket{i|A|j}|^2}{f(\mu_i/\mu_j)\mu_j}.\label{eq:norm_invertible_state}
\end{align}
For a completely positive and trace-preserving (CPTP) map $\mathcal{E}$ such that $\mathcal{E}(\rho)$ is also invertible, this norm is non-increasing~\cite{petzMonotoneMetricsMatrix1996}, i.e., 
\begin{align}
     \|A\|_{f,\rho}^2\geq  \|\mathcal{E}(A)\|_{f,\mathcal{E}(\rho)}^2. \label{eq:monotonicity_norm}
\end{align}

When $\rho$ is not invertible, the definition of the quadratic form in Eq.~\eqref{eq:norm_invertible_state} requires a careful treatment since the definition of $\mathbb{J}^f_\rho$ includes $\mathbb{R}_\rho^{-1}$. Following the standard boundary convention (e.g., in Refs.~\cite{hiaiDifferentQuantumFdivergences2017,hiai_QuantumFDivergencesNeumannAlgebrasReversibility_2021}), we extend $yf(x/y)$ by
\begin{align}
    m_f(x,y)\coloneqq \lim_{\epsilon\to 0^+}(y+\epsilon)f\left(\frac{x+\epsilon}{y+\epsilon}\right)=
    \begin{cases}
        yf(x/y)&(x>0\land y>0)\\
        c_0y&(x=0\land y>0)\\
        c_\infty x&(x>0\land y=0)\\
        0&(x=y=0)
    \end{cases},\label{eq:definition_m_f}
\end{align}
where the boundary values are denoted by
\begin{align}
    c_0\coloneqq f(0)\coloneqq \lim_{\epsilon\to 0^+}f(\epsilon),\qquad c_\infty&\coloneqq \lim_{t\to\infty}\frac{f(t)}{t}.\label{eq:bdry_values_monotone_fct}
\end{align}
We then define $\mathbb{J}^f_\rho \coloneqq m_f(\mathbb{L}_\rho,\mathbb{R}_\rho)$, or equivalently,
\begin{align}
    \mathbb{J}^f_\rho A=\sum_{i,j=1}^dm_f(\mu_i,\mu_j)\braket{i|A|j}\ket{i}\bra{j},
\end{align}
using an eigenvalue decomposition of $\rho=\sum_{i=1}^d\mu_i\ket{i}\bra{i}$. 

Since $\mathbb{J}^f_\rho $ may have a nontrivial kernel, we use the Moore--Penrose inverse $(\mathbb{J}^f_\rho)^+$ instead of $(\mathbb{J}^f_\rho)^{-1}$ and define~\cite{petzIntroductionQuantumFisher2011}
\begin{align}
    \braket{A,B}_{f,\rho}&\coloneqq \braket{A,(\mathbb{J}^{f}_\rho)^{+}B}=\sum_{\substack{i,j=1\\m_f(\mu_i,\mu_j)>0}}^d \frac{1}{m_f(\mu_i,\mu_j)}\braket{j|A^\dag|i}\braket{i|B|j}\\
\|A\|_{f,\rho}^2&\coloneqq \braket{A,A}_{f,\rho}=\sum_{\substack{i,j=1\\m_f(\mu_i,\mu_j)>0}}^d \frac{1}{m_f(\mu_i,\mu_j)}|\braket{i|A|j}|^2.
\end{align}
On the full operator space, this is generally a seminorm, while it is a norm if restricted to $\Ran \mathbb{J}_{\rho}^f$.

Let $\rho_{\theta}$ be a family of states smoothly parameterized by $p$ real parameters $\theta\in\Theta$, where $\Theta\subset\mathbb{R}^p$ is open. Then, the $f$-QFI $\mathcal{F}^f_{\rho_\theta}$ is defined as a $p\times p$ matrix whose elements are given by
\begin{align}
    \left(\mathcal{F}^f_{\rho_\theta}\right)_{ij}\coloneqq \braket{\partial_{\theta_i}\rho_\theta,\partial_{\theta_j}\rho_\theta}_{f,\rho_\theta}=\braket{\partial_{\theta_i}\rho_\theta,(\mathbb{J}^{f}_{\rho_\theta})^{+}\partial_{\theta_j}\rho_\theta}.\label{eq:definition_QFI_matrix}
\end{align}

\subsection{Metric adjusted skew information}
For an operator monotone function satisfying the following conditions:
\begin{enumerate}[(i)]
    \item Regularity condition:
    \begin{align}
        c_0\coloneqq f(0)\coloneqq \lim_{\epsilon\to 0^+}f(\epsilon)>0,\label{eq:regularity_condition_operator_monotone}
    \end{align}
    \item Symmetry condition: 
    \begin{align}
        f(t)=tf(t^{-1})\qquad \text{for all } t>0,\label{eq:symmetry_condition_operator_monotone}
    \end{align}
\end{enumerate}
metric adjusted skew information is defined~\cite{hansenMetricAdjustedSkew2008} by
\begin{align}
    I_{\rho}^f(A)&\coloneqq\frac{f(0)}{2} \|\ii[\rho,A]\|_{f,\rho}^2=\frac{f(0)}{2}\sum_{\substack{i,j=1\\m_f(\mu_i,\mu_j)>0}}^d \frac{(\mu_i-\mu_j)^2}{m_f(\mu_i,\mu_j)}|\braket{i|A|j}|^2.
\end{align}
The term ``skew information'' takes its name from the Wigner--Yanase skew information~\cite{wignerInformationContentsDistributions1963}, which is a special case of the metric-adjusted skew information corresponding to $f(t)=(1+\sqrt{t})^2/4$~\cite{hansenMetricAdjustedSkew2008}.

The qualifier ``metric adjusted'' refers to the fact that it is defined in terms of the corresponding monotone metric, rescaled by the factor $f(0)/2$. Indeed, for a unitary model $\rho_t\coloneqq e^{\ii t A} \rho e^{-\ii tA}$, its first derivative is given by $\partial_t \rho_t|_{t=0}=\ii [A,\rho]$, implying that 
\begin{align}
    I_{\rho}^f(A)=\frac{f(0)}{2} \mathcal{F}_{\rho_t}^f\biggl|_{t=0}.\label{eq:skew_info_and_QFI}
\end{align}

The normalization is chosen so that it coincides with the variance for pure states. Indeed, since the symmetry condition implies
\begin{align}
    c_\infty=\lim_{t\to\infty}\frac{f(t)}{t}=\lim_{t\to \infty}f(t^{-1})=f(0)=c_0\label{eq:c_inf_equal_c_0}
\end{align}
and hence $m_f(t,0)=m_f(0,t)=tf(0)$ for $t>0$, for a pure state $\psi$, we have
\begin{align}
    I_{\psi}^f(A)&=\frac{f(0)}{2}\left(\frac{1}{f(0)}\sum_{\substack{i=1\\\mu_i\neq 1}}^d|\braket{\psi|A|i}|^2+\frac{1}{f(0)}\sum_{\substack{i=1\\\mu_i\neq 1}}^d|\braket{i|A|\psi}|^2\right)=\braket{\psi|A(I-\psi)A|\psi}\\
    &=\braket{\psi|A^2|\psi}-\braket{\psi|A|\psi}^2.
\end{align}

\subsection{Metric adjusted quantum geometric tensor}

Here we introduce the metric adjusted QGT associated with an operator monotone function $f$. 
We consider operator monotone functions satisfying two boundary regularity conditions. In addition to the usual regularity condition in Eq.~\eqref{eq:regularity_condition_operator_monotone}, namely
\begin{align}
    c_0\coloneqq f(0)\coloneqq \lim_{\epsilon\to 0^+}f(\epsilon)>0,
\end{align}
we also impose
\begin{align}
    c_\infty=\lim_{t\to\infty}\frac{f(t)}{t}>0.
\end{align}
In contrast to the metric adjusted skew information, we do not impose the symmetry condition in Eq.~\eqref{eq:symmetry_condition_operator_monotone}. 

For a family of states $\rho_{\theta}$ smoothly parameterized by $p$ real parameters $\theta\in\Theta\subset \mathbb{R}^p$, we define the metric adjusted QGT by
\begin{align}
    \mathcal{Q}^f\coloneqq \frac{c_0c_\infty}{c_0+c_\infty}\mathcal{F}^f_{\rho_\theta}\biggl|_{\theta=0}
\end{align}
where $\mathcal{F}^f_{\rho_\theta}$ is defined in \eqref{eq:definition_QFI_matrix}. Since $\mathcal{Q}^f$ is Hermitian, its real and imaginary parts are given by
\begin{align}
    \Re \mathcal{Q}^f=\frac{1}{2}\left(\mathcal{Q}^f+\left(\mathcal{Q}^f\right)^\top\right),\qquad \Im \mathcal{Q}^f=\frac{1}{2\ii}\left(\mathcal{Q}^f-\left(\mathcal{Q}^f\right)^\top\right).
\end{align}

When $f$ is symmetric, this definition reduces to the metric adjusted skew information for unitary models. Indeed, the symmetry condition implies $c_0=c_\infty$ as shown in Eq.~\eqref{eq:c_inf_equal_c_0}, and therefore
\begin{align}
    \frac{c_0c_\infty}{c_0+c_\infty}=\frac{f(0)}{2}.
\end{align}
For symmetric $f$, $m_f(x,y)=m_f(y,x)$ holds for $x,y\in\mathbb{R}_{\geq 0}$, implying that
\begin{align}
    \left( \mathcal{Q}^{f}\right)^\top= \mathcal{Q}^{f}
\end{align}
and hence 
\begin{align} 
    \Re \mathcal{Q}^{f} =\mathcal{Q}^{f},\qquad \Im \mathcal{Q}^{f}=0.
\end{align}
Therefore, for a unitary model $\rho_\theta=e^{\ii \sum_{i=1}^p\theta_iX_i}\rho e^{-\ii \sum_{i=1}^p\theta_iX_i}$, we obtain
\begin{align}
    \forall u\in\mathbb{R}^p, \qquad u^\top \mathcal{Q}_\rho^f u= u^\top \Re \mathcal{Q}_\rho^f u=I_{\rho}^f\left(\sum_{i=1}^pu_iX_i\right).
\end{align}

Dropping the symmetry condition allows the metric adjusted QGT to have a nonzero imaginary part. This is essential for recovering the full QGT, rather than only its real part. To make this connection explicit, consider the one-parameter family
\begin{align}
    f_q(t)\coloneqq (1-q)+qt.
\end{align}
For $q\in(0,1)$, $f_q$ satisfies both boundary conditions with $c_0=1-q$ and $c_\infty=q$. The corresponding metric adjusted QGT is denoted by
\begin{align}
    \mathcal{Q}_\rho^q\coloneqq q(1-q)\mathcal{F}^{f_q}_{\rho_\theta}\biggl|_{\theta =0}.\label{eq:def_QGT_QFI_relation}
\end{align}
We remark that, although we adopt the same notation $\mathcal{Q}$ as in Ref.~\cite{yamaguchi_QuantumGeometricTensorDeterminesPureState_2026}, our convention differs from that used there only by an overall factor of $q$ and the placement of the indices. This factor makes the symmetry under the interchange $q\leftrightarrow 1-q$ transparent, as will be shown in Eq.~\eqref{eq:maQGT_transpose}. In either convention, the endpoint limit as $q\to1^-$ is the same.

To establish the connection with the standard QGT for pure states, let $\ket{\psi_\theta}$ be a smoothly parameterized family of pure states and define
\begin{align}
    \ket{\partial_i\psi}\coloneqq \frac{\partial}{\partial\theta_i}\ket{\psi_{\theta}}\biggl|_{\theta=0}.
\end{align}
Then, the standard QGT~\cite{provostRiemannianStructureManifolds1980,berryQuantumPhaseFive1989}, defined for pure states, is 
\begin{align}
    \left(\mathcal{Q}_\psi\right)_{ij}\coloneqq \braket{\partial_i \psi|\partial_j\psi}-\braket{\partial_i\psi|\psi}\braket{\psi|\partial_j\psi}.\label{eq:QGT_definition}
\end{align}
The metric adjusted QGT for the pure state $\psi_\theta\coloneqq \ket{\psi_\theta}\bra{\psi_\theta}$ is given by
\begin{align}
    \mathcal{Q}_\psi^q=q(1-q)\left(\frac{1}{1-q}\mathcal{Q}_\psi+\frac{1}{q}\mathcal{Q}_\psi^\top\right)=q\mathcal{Q}_\psi+(1-q)\mathcal{Q}_\psi^\top=\Re\mathcal{Q}_\psi+\ii(2q-1)\Im\mathcal{Q}_\psi.
\end{align}
Thus, for pure states, the metric adjusted QGT connects to the standard QGT as $q\to 1^-$, i.e.,
\begin{align}
    \lim_{q\to 1^-}\mathcal{Q}^q_\psi=\mathcal{Q}_\psi.
\end{align}

The metric adjusted QGT in Eq.~\eqref{eq:def_QGT_QFI_relation} also regularizes the possible endpoint divergences of the $f_q$-QFI associated with the LLD and RLD limits. Indeed, for a general state $\rho=\sum_{k=1}^d\mu_k\ket{k}\bra{k}$, we have
\begin{align}
    \left(\mathcal{F}_{\rho}^{f_q}\right)_{ij}=\sum_{\substack{k,l=1;\\(1-q)\mu_l+q\mu_k>0} }^d\frac{\braket{l|\partial_i\rho|k}\braket{k|\partial_j\rho|l}}{(1-q)\mu_l+q\mu_k}=\frac{1}{1-q}\left(\mathcal{Q}_\rho\right)_{ij}+\frac{1}{q}\left(\mathcal{Q}_\rho^\top\right)_{ij}+\sum_{\substack{k,l=1;\\\mu_l>0,\,\mu_k>0} }^d\frac{\braket{l|\partial_i\rho|k}\braket{k|\partial_j\rho|l}}{(1-q)\mu_l+q\mu_k}.
\end{align}
Here, denoting by $\Pi_\rho$ the projector onto the support of $\rho$, we defined 
\begin{align}
    \left(\mathcal{Q}_\rho\right)_{ij}\coloneqq \sum_{\substack{k,l=1;\\\mu_l>0,\,\mu_k=0} }^d\frac{\braket{l|\partial_i\rho|k}\braket{k|\partial_j\rho|l}}{\mu_l}=\sum_{\substack{l=1;\mu_l>0} }^d\frac{\braket{l|\partial_i\rho(I-\Pi_\rho)\partial_j\rho|l}}{\mu_l},
\end{align}
which reduces to the standard QGT in Eq.~\eqref{eq:QGT_definition} when $\rho$ is pure. Unless $\mathcal{Q}_\rho=0$, the first two terms exhibit $1/(1-q)$ and $1/q$ divergences as  $q\to 1^-$ and $q\to 0^+$, respectively. The prefactor $q(1-q)$ in Eq.~\eqref{eq:def_QGT_QFI_relation} removes both endpoint divergences. Indeed,
\begin{align}
    \left(\mathcal{Q}_{\rho}^q\right)_{ij}=q\left(\mathcal{Q}_\rho\right)_{ij}+(1-q)\left(\mathcal{Q}_\rho^\top\right)_{ij}+q(1-q)\sum_{\substack{k,l=1;\\\mu_l>0,\,\mu_k>0} }^d\frac{\braket{l|\partial_i\rho|k}\braket{k|\partial_j\rho|l}}{(1-q)\mu_l+q\mu_k}.
\end{align}
Therefore, $\mathcal{Q}_\rho^q$ admits continuous extensions to both endpoints,
\begin{align}
    \mathcal{Q}_{\rho}^1\coloneqq \lim_{q\to 1^-}\mathcal{Q}_\rho^q=\mathcal{Q}_\rho,\qquad \mathcal{Q}_{\rho}^0\coloneqq \lim_{q\to 0^+}\mathcal{Q}_\rho^q=\mathcal{Q}_\rho^\top.
\end{align}
Here, the implication from the interior $q\in(0,1)$ to the endpoints follows by taking the continuous limits of $\mathcal{Q}_\rho^q-r\mathcal{Q}_\sigma^q$. 
Since $q(1-q)>0$ for $q\in(0,1)$, we have
\begin{align}
    \forall q\in[0,1],\,\mathcal{Q}_\rho^q\geq r\mathcal{Q}_\sigma^q\quad \Longleftrightarrow\quad \forall q\in(0,1),\,\mathcal{F}_\rho^{f_q}\geq r \mathcal{F}_\sigma^{f_q}
\end{align}
Moreover, since $f_{1-q}(t)=tf_q(t^{-1})$, we have
\begin{align}
    \mathcal{Q}_\rho^{1-q}=\left(\mathcal{Q}_\rho^q\right)^\top.\label{eq:maQGT_transpose}
\end{align}
Thus, whenever the property under consideration is invariant under transposition, it is sufficient to consider $q\in[1/2,1]$ for the metric adjusted QGT (and $q\in[1/2,1)$ for the unregularized $f_q$-QFI). An important example is the matrix ordering via positive-semidefiniteness: for Hermitian matrices $A$ and $B$, 
\begin{align}
    A\geq B\quad \Longleftrightarrow\quad A^\top \geq B^\top.
\end{align}
In summary, we obtain
\begin{align}
    \forall q\in[0,1],\,\mathcal{Q}_\rho^q\geq r\mathcal{Q}_\sigma^q
    \quad &\Longleftrightarrow\quad \forall q\in[1/2,1],\,\mathcal{Q}_\rho^q\geq r\mathcal{Q}_\sigma^q\\
    & \Longleftrightarrow\quad \forall q\in(0,1),\,\mathcal{F}_\rho^{f_q}\geq r \mathcal{F}_\sigma^{f_q}
    \quad \Longleftrightarrow\quad \forall q\in[1/2,1),\,\mathcal{F}_\rho^{f_q}\geq r \mathcal{F}_\sigma^{f_q}.\label{eq:QGTordering_QFIordering}
\end{align}

\subsection{Variational representation and monotonicity}\label{sec:variational_rep_monotone_metric}
For later convenience, we introduce a variational formulation. Let $M\geq 0$ be a positive semidefinite operator on a finite-dimensional Hilbert space $\mathcal{H}$ with inner product $\braket{\cdot,\cdot}$. Since $M$ is positive semidefinite, its Moore--Penrose inverse $M^+$ is also positive semidefinite. Fix $a\in \Ran M$. For $x\in\mathcal{H}$, we have
\begin{align}
    0\leq \braket{a-Mx, M^+ (a-Mx)}
    &=\braket{a,M^+a}-2\Re\braket{x,MM^+a}+\braket{x,Mx},
\end{align}
which implies 
\begin{align}
    2\Re\braket{x,MM^+a}-\braket{x,Mx}\leq \braket{a,M^+a}
\end{align}
Since $MM^+$ is the orthogonal projector onto $\Ran M$, we have $MM^+a=a$. Moreover, equality is attained at $x=M^+a$, since $Mx=MM^+a=a$. Therefore, 
\begin{align}
    \braket{a,M^+a}=\sup_{x}\{2\Re\braket{x,a}-\braket{x,Mx}\}\qquad\qquad (a\in \Ran M).
\end{align}
Applying this relation with $M=\mathbb{J}_{\rho}^f$, we obtain a variational formulation
\begin{align}
    \|A\|_{f,\rho}^2= \sup_{X}\{2\Re \Tr(X^\dag A)-\braket{X,\mathbb{J}_{\rho}^f X}\}\qquad (A\in \Ran \mathbb{J}_{\rho}^f ). \label{eq:norm_variational_formula}
\end{align}

We remark that the monotonicity in Eq.~\eqref{eq:monotonicity_norm} is often proven only for invertible states. 
For a one-parameter family of operator monotone functions
\begin{align}
    f_q(t)\coloneqq (1-q)+qt,\qquad q\in (0,1),
\end{align}
the corresponding superoperator has a direct extension to an arbitrary quantum state $\rho$, given by 
\begin{align}
    \mathbb{J}_\rho^{f_q}=q\mathbb{L}_\rho+(1-q)\mathbb{R}_\rho,\qquad \mathbb{J}_\rho^{f_q}(A)=q\rho A+(1-q)A\rho.
\end{align}
The monotonicity of the $f_q$-norm has been established~\cite{yamaguchi_QuantumGeometricTensorDeterminesPureState_2026} for general states as a corollary of Theorem~6.1 in \cite{hayashiQuantumInformationTheory2017}. 
We provide an alternative proof using the variational formula.
\begin{lemma}\label{lem:fq_norm_monotonicity_general}
    Let $\rho$ be a quantum state. For a quantum channel $\mathcal{E}$, we define $\rho'\coloneqq \mathcal{E}(\rho)$. If $A\in\Ran\mathbb{J}_{\rho}^{f_q} $, then $\|A\|_{f_q,\rho}^2\geq  \|\mathcal{E}(A)\|_{f_q,\mathcal{E}(\rho)}^2$. 
\end{lemma}
\begin{proof}
    We first show $\mathcal{E}(A) \in \Ran  \mathbb{J}_{\mathcal{E}(\rho)}^{f_q}$ for $A\in\Ran\mathbb{J}_{\rho}^{f_q} $. 
    Since $ \mathbb{J}_\rho^{f_q}(\ket{i}\bra{j})=(q\mu_i +(1-q)\mu_j)\ket{i}\bra{j}$, we have
    \begin{align}
        A\in \Ran \mathbb{J}_\rho^{f_q}\Longleftrightarrow PAP=0,\label{eq:kernel_rangeJ}
    \end{align}
    where $P$ denotes the projector onto $\ker \rho$. Let $\{K_\alpha\}$ be the Kraus operators of the quantum channel $\mathcal{E}$. Then, for any $\ket{\psi}\in\ker \mathcal{E}(\rho)$, we have
    \begin{align}
        0=\braket{\psi|  \mathcal{E}(\rho)|\psi}=\sum_\alpha\braket{\psi|K_\alpha\rho K_\alpha^\dag |\psi}=\sum_\alpha\|\rho^{1/2}K_\alpha^\dag \ket{\psi}\|^2,
    \end{align}
    and hence $\rho^{1/2}K_\alpha^\dag \ket{\psi}=0$ for every $\alpha$, implying that $K_\alpha^\dag \ket{\psi}\in\ker \rho$ for all $\alpha$. Therefore, for any $\ket{\psi},\ket{\psi'}\in\ker\mathcal{E}(\rho)$ and $A\in\Ran\mathbb{J}_{\rho}^{f_q}$, we have
    \begin{align}
        \braket{\psi|\mathcal{E}(A)|\psi'}=\sum_\alpha\braket{K_\alpha^\dag \psi|A|K_\alpha^\dag \psi'}\overset{\text{Eq.~\eqref{eq:kernel_rangeJ}}}{=}0.
    \end{align}
    Therefore, $\mathcal{E}(A)\in\Ran \mathbb{J}_{\mathcal{E}(\rho)}^{f_q}$. 

    By using the dual map $\Lambda \coloneqq \mathcal{E}^*$ of $\mathcal{E}$, we have
    \begin{align}
        \braket{Y,\mathcal{E}\circ\mathbb{J}_\rho^{f_q}\circ\Lambda  (Y)}= \braket{\Lambda(Y),\mathbb{J}_\rho^{f_q}(\Lambda (Y))}=q\Tr(\rho \Lambda(Y)\Lambda(Y)^\dag)+(1-q)\Tr(\rho\Lambda(Y)^\dag \Lambda(Y))
    \end{align}
    for any $Y$. 
    Since $\Lambda$ is a CP unital map, the Schwarz inequality for unital CP maps (see e.g., Ref.~\cite{choi_SchwarzInequalityPositivelinearmaps_1974_math} and Proposition~3.3 in Ref.~\cite{paulsen_CompletelyBoundedMapsOperatorAlgebras_2003}) gives
    \begin{align}
        \Lambda(Y)^\dag \Lambda(Y)\leq \Lambda (Y^\dag Y),\qquad \Lambda(Y) \Lambda(Y)^\dag\leq \Lambda (YY^\dag).
    \end{align}
    Therefore, 
    \begin{align}
        \braket{Y,\mathcal{E}\circ\mathbb{J}_\rho^{f_q}\circ\Lambda (Y)}&\leq q\Tr(\rho \Lambda(YY^\dag))+(1-q)\Tr(\rho\Lambda(Y^\dag Y))\\
        &=q\Tr(\rho' YY^\dag )+(1-q)\Tr(\rho' Y^\dag Y)=\braket{Y,\mathbb{J}_{\rho'}^{f_q}(Y)}.
    \end{align}
    Since this inequality holds for any $Y$, we get
    \begin{align}
        \mathcal{E}\circ\mathbb{J}_\rho^{f_q}\circ\Lambda \leq \mathbb{J}_{\rho'}^{f_q}.
    \end{align}
    Since $\mathcal{E}(A)\in \Ran \mathbb{J}_{\rho'}^{f_q}$ for $A\in\Ran\mathbb{J}_{\rho}^{f_q} $, we obtain
    \begin{align}
        \|\mathcal{E}(A)\|^2_{f_q,\rho'}&=\sup_{Y}\{2\Re \braket{Y, \mathcal{E}(A)}-\braket{Y,\mathbb{J}_{\rho'}^{f_q} Y}\}\\
        &\leq \sup_{Y}\{2\Re \braket{\Lambda(Y),A}-\braket{Y,\mathcal{E}\circ\mathbb{J}_\rho^{f_q}\circ\Lambda (Y)}\}=\sup_{Y}\{2\Re \braket{\Lambda(Y),A}-\braket{\Lambda(Y),\mathbb{J}_\rho^{f_q}(\Lambda(Y))}\}\\
        &\leq \sup_{X}\{2\Re \braket{X,A}-\braket{X,\mathbb{J}_\rho^{f_q}(X)}\}=\|A\|^2_{f_q,\rho}.
    \end{align}
\end{proof}

%\subsection{Quantum Fisher information matrix}
We now explain the connection between the monotonicity of the quadratic form and its associated quantum Fisher information matrix. 
Let $\rho_{\theta}$ be a family of states smoothly parameterized by $p$ real parameters $\theta\in\Theta$, where $\Theta\subset\mathbb{R}^p$ is open. Then, the associated quantum Fisher information matrix $\mathcal{F}^f_{\rho_\theta}$ is defined as a $p\times p$ matrix whose elements are given by
\begin{align}
    \left(\mathcal{F}^f_{\rho_\theta}\right)_{ij}\coloneqq \braket{\partial_{\theta_i}\rho_\theta,\partial_{\theta_j}\rho_\theta}_{f,\rho_\theta}=\braket{\partial_{\theta_i}\rho_\theta,(\mathbb{J}^{f}_{\rho_\theta})^{+}\partial_{\theta_j}\rho_\theta}
\end{align}
For $z\in\mathbb{C}^p$, we have
\begin{align}
    z^\dag \mathcal{F}^f_{\rho_\theta} z =\left\|\sum_{i=1}^pz_i\partial_{\theta_i}\rho_\theta\right\|_{f,\rho_\theta}^2. 
\end{align}
Introducing the generalized logarithmic derivative $L_i\coloneqq (\mathbb{J}^{f}_{\rho_\theta})^{+}\partial_{\theta_i}\rho_\theta$, we have 
\begin{align}
    \mathbb{J}^{f}_{\rho_\theta}L_i=P_{\Ran \mathbb{J}^{f}_{\rho_\theta}}\partial_{\theta_i}\rho_\theta,
\end{align}
where $P_{\Ran\mathbb{J}^{f}_{\rho_\theta}}$ denotes the orthogonal projector onto $\Ran\mathbb{J}^{f}_{\rho_\theta}$. In particular, if $\partial_{\theta_i}\rho_\theta\in \Ran\mathbb{J}^{f}_{\rho_\theta}$, then $\partial_{\theta_i}\rho_\theta=\mathbb{J}^{f}_{\rho_\theta}L_i$. In either case, one may also write
\begin{align}
    \left(\mathcal{F}^f_{\rho_\theta}\right)_{ij}=\braket{\partial_{\theta_i}\rho_\theta,L_j}=\braket{L_i,\mathbb{J}^{f}_{\rho_\theta}L_j}.
\end{align}

Let $\mathcal{E}$ be a parameter-independent quantum channel and define $\rho'_\theta\coloneqq \mathcal{E}(\rho_\theta)$. Then, since $\mathcal{E}$ is linear, we have $\partial_{\theta_i}\rho'_\theta=\mathcal{E}(\partial_{\theta_i}\rho_\theta)$. Therefore, the quantum Fisher information $\mathcal{F}^f_{\rho'_\theta}$ for $\rho'_\theta$ satisfies
\begin{align}
    z^\dag \mathcal{F}^f_{\rho'_\theta}z =\left\|\sum_{i=1}^pz_i\partial_{\theta_i}\rho'_\theta\right\|_{f,\rho'_\theta}^2=\left\|\mathcal{E}\left(\sum_{i=1}^pz_i\partial_{\theta_i}\rho_\theta\right)\right\|_{f,\rho_\theta'}^2.\label{eq:monotonicity_QFI_oneshot}
\end{align}
Thus, as long as the quadratic form monotonically decreases under a quantum channel, $z^\dag \mathcal{F}^f_{\rho_\theta} z \geq z^\dag \mathcal{F}^f_{\rho'_\theta}z $ for all $z\in\mathbb{C}^p$, or equivalently, $\mathcal{F}^f_{\rho_\theta}\geq \mathcal{F}^f_{\rho'_\theta}$. 
Therefore, under any parameter-independent quantum channel, the quantum Fisher information matrix is non-increasing in the sense of matrix inequality.

In particular, for $f_q(x)=(1-q) +qx$, Lemma~\ref{lem:fq_norm_monotonicity_general} proves the monotonicity of the corresponding quadratic form for any $A\in\Ran \mathbb{J}_{\rho_\theta}^{f_q}$. We now show that for a smooth state family on the open parameter domain $\Theta$, the condition $\partial_{\theta_i}\rho_\theta \in\Ran \mathbb{J}_{\rho_\theta}^{f_q}$ is automatically satisfied at every interior point $\theta_0$. Let $P_{\theta_0}$ be the projector onto $\ker \rho_{\theta_0}$ and fix $\ket{v}\in\ker\rho_{\theta_0}$. Define
\begin{align}
    g(\theta)\coloneqq \braket{v|\rho_\theta|v}.
\end{align}
Since $\ket{v}\in\ker\rho_{\theta_0}$, we have $g(\theta_0)=0$. Also, the positivity of $\rho_{\theta}$ implies $g(\theta)\geq 0$ for all $\theta\in\Theta$. Thus, $g(\theta)$ attains the minimum at $\theta=\theta_0$, implying that $\partial_{\theta_i}g(\theta)|_{\theta=\theta_0}=0$.  
Consequently, we get
\begin{align}
    \braket{v|T_i|v}=0,\qquad \forall v\in \ker\rho_{\theta_0}
\end{align}
for any $i=1,\ldots,p$, where $T_i\coloneqq\partial_{\theta_i}\rho_\theta|_{\theta=\theta_0}$. Moreover, since $T_i$ is Hermitian, and since $\ket{u}+\ket{v},\ket{u}+\ii\ket{v}\in\ker\rho_{\theta_0}$ for any $\ket{u},\ket{v}\in\ker\rho_{\theta_{0}}$, we also have
\begin{align}
    0&=\braket{u+v|T_i|u+v}=2\Re \braket{u|T_i|v}\\
    0&=\braket{u+\ii v|T_i|u+\ii v}=-2\Im \braket{u|T_i|v}
\end{align}
and hence $P_{\theta_0}T_iP_{\theta_0}=0$. Since $\theta_0\in\Theta$ was arbitrary, Eq.~\eqref{eq:kernel_rangeJ} implies 
\begin{align}
    \partial_{\theta_i}\rho_\theta\in \Ran \mathbb{J}_{\rho_\theta}^{f_q},\quad \forall \theta\in\Theta.\label{eq:tangent_in_range}
\end{align}
Because $\Ran \mathbb{J}_{\rho_\theta}^{f_q}$ is a linear subspace, it follows that
\begin{align}
    \sum_{i=1}^pz_i\partial_{\theta_i}\rho_\theta\in \Ran \mathbb{J}_{\rho_\theta}^{f_q},\quad \forall z\in\mathbb{C}^p.
\end{align}
Therefore, by Lemma~\ref{lem:fq_norm_monotonicity_general}, we establish the following: 
\begin{proposition}\label{prop:fq_QFI_monotonicity_general}
    Let $q\in(0,1)$, and let $\rho_\theta$ be a smooth state family on an open parameter domain $\Theta\subset \mathbb{R}^p$. Then, the corresponding quantum Fisher information matrix is monotonic under any parameter-independent quantum channel $\mathcal{E}$, namely
    \begin{align}
        \forall q\in(0,1),\quad \mathcal{F}^{f_q}_{\rho_\theta}\geq \mathcal{F}^{f_q}_{\mathcal{E}(\rho_\theta)}
    \end{align}
    for any $\theta\in\Theta$. 
    Multiplying this relation by $q(1-q)$, we obtain the same inequality for the metric adjusted QGTs. Continuously extending to endpoints, we obtain
    \begin{align}
        \forall q\in[0,1],\quad \mathcal{Q}^{q}_{\rho_\theta}\geq \mathcal{Q}^{q}_{\mathcal{E}(\rho_\theta)} ,\qquad 
    \end{align}
\end{proposition}

The additivity of $f_q$-QFI immediately implies additivity of metric adjusted QGTs, i.e., for independent models $\rho_\theta$ and $\sigma_\theta$ with the same parameter,
\begin{align}
    \mathcal{Q}_{\rho_\theta\otimes\sigma_\theta}^q=\mathcal{Q}_{\rho_\theta}^q+\mathcal{Q}_{\sigma_\theta}^q,\qquad q\in[0,1].
\end{align}

\subsection{Basic properties of the metric adjusted QGTs}

For convenience, we summarize basic properties of metric adjusted QGTs as resource measures:
\begin{proposition}[Properties of $\mathcal{Q}^q$]
    \leavevmode\par
    \begin{enumerate}[(i)]
        \item Positivity: For any state $\rho$ and every $q\in[0,1]$, $\mathcal{Q}_{\rho}^q\geq 0$. 
        \item Additivity: For any pair of states $\rho$ and $\sigma$ and every $q\in[0,1]$, $\mathcal{Q}_{\rho\otimes \sigma}^q=\mathcal{Q}_{\rho}^q+\mathcal{Q}_{\sigma}^q$.
        \item Convexity: For any probability distribution $\{p_k\}_k$ and set of states $\{\rho_k\}_k$ and every $q\in[0,1]$, $\sum_kp_k \mathcal{Q}_{\rho_k}^q\geq  \mathcal{Q}_{\sum_{k}p_k\rho_k}^q$.
        \item Monotonicity: For any $G$-covariant quantum channel $\mathcal{E}$ and state $\rho$ and every $q\in[0,1]$, $\mathcal{Q}_{\rho}^q\geq \mathcal{Q}_{\mathcal{E}(\rho)}^q$. 
        \item Strong monotonicity (also referred to as selective monotonicity): Let $\{\mathcal{E}_k\}_k$ be any $G$-covariant instrument, meaning that each $\mathcal{E}_k$ is a $G$-covariant, completely-positive, trace-nonincreasing map and $\sum_k\mathcal{E}_k$ is trace-preserving. For an input state $\rho$, define $p_k\coloneqq \mathrm{Tr}(\mathcal{E}_k(\rho))$ and, whenever $p_k>0$, $\rho_k\coloneqq \frac{\mathcal{E}_k(\rho)}{p_k}$. Then, for every $q\in[0,1]$, 
        \begin{align}
            \mathcal{Q}_\rho^q\geq \sum_{k;\,p_k>0}p_k\mathcal{Q}_{\rho_k}^q.
        \end{align}
        \item Faithfulness: If $G$ is connected, then, for every fixed $q\in(0,1)$, $\mathcal{Q}^q$ is a faithful measure, i.e., $\mathcal{Q}_\rho^q=0$ if and only if $\rho$ is $G$-symmetric.
    \end{enumerate}
\end{proposition}
\begin{proof}
    For $q\in(0,1)$, as noted below Eq.~\eqref{eq:def_QGT_QFI_relation}, our tensor $\mathcal{Q}_\rho^q$ equals $q$ times the tensor defined in Ref.~\cite{yamaguchi_QuantumGeometricTensorDeterminesPureState_2026}. Since $q>0$, properties~(i)--(v) follow from Proposition~I~1 of Ref.~\cite{yamaguchi_QuantumGeometricTensorDeterminesPureState_2026}. The endpoint cases $q=0,1$ follow by taking the corresponding continuous limits. 

    We now prove~(vi). Fix $q\in (0,1)$. We first prove
    \begin{align}
        \mathcal{Q}_\rho^q=0\quad \Longleftrightarrow\quad  [X_i,\rho]=0\quad \text{for all } i=1,\ldots,\dim G,\label{eq:QGT_zero_iff_commutators}
    \end{align}
    without assuming that $G$ is connected. Let $\rho=\sum_k\mu_k\ket{k}\bra{k}$ be an eigenvalue decomposition. Since $m_{f_q}(\mu_k,\mu_l)=(1-q)\mu_l+q\mu_k$, we have
    \begin{align}
        m_{f_q}(\mu_k,\mu_l)>0\quad \Longleftrightarrow\quad  \mu_k+\mu_l>0.
    \end{align}
    The diagonal entries of $\mathcal{Q}_\rho^q$ are therefore
    \begin{align}
        \left(\mathcal{Q}_\rho^q\right)_{ii}=q(1-q)\sum_{\substack{k,l\\ \mu_k+\mu_l>0}}\frac{|\braket{k|[\rho,X_i]|l}|^2}{m_{f_q}(\mu_k,\mu_l)},
    \end{align}
    Every summand in the right-hand side is nonnegative. Since $\braket{k|[\rho,X_i]|l}=0$ whenever $\mu_k=\mu_l=0$, we obtain
    \begin{align}
         \left(\mathcal{Q}_\rho^q\right)_{ii}=0&\quad \Longleftrightarrow\quad  |\braket{k|[\rho,X_i]|l}|^2=0\text{ for all }k,l \text{ with }\mu_k+\mu_l>0\\
         &\quad \Longleftrightarrow\quad  [\rho,X_i]=0,
    \end{align}
    Thus, if $\mathcal{Q}_\rho^q=0$, then $[\rho,X_i]=0$ for all $i=1,\ldots,\dim G$. The converse is trivial. Thus, we obtain Eq.~\eqref{eq:QGT_zero_iff_commutators}. 

    Suppose now that $G$ is connected. Let $\{A_i\}_{i=1}^p$ be the Lie algebra basis used in the parametrization $g(\theta)=\exp(\ii\sum_i\theta_iA_i)$. For $\theta\in\mathbb{R}^{\dim G}$, define
    \begin{align}
        A(\theta)\coloneqq \sum_{i=1}^{\dim G}\theta_iA_i,\qquad X(\theta)\coloneqq \sum_{i=1}^{\dim G}\theta_iX_i.
    \end{align}
    When $[X_i,\rho]=0$ for all $i$, $ [X(\theta),\rho]=0$ for any $\theta\in\mathbb{R}^{\dim G}$. 
    Since
    \begin{align}
        \frac{\dd}{\dd t}\mathcal{U}_{\exp(\ii t A(\theta))}(\rho)=\mathcal{U}_{\exp(\ii t A(\theta))}\left(\ii[X(\theta),\rho]\right)=0,
    \end{align}
    and $\mathcal{U}_{\exp(\ii t A(\theta))}(\rho)|_{t=0}=\rho$, we obtain
    \begin{align}
        \mathcal{U}_{\exp(\ii t A(\theta))}(\rho)=\rho,\qquad\forall t\in\mathbb{R},\quad \theta\in \mathbb{R}^{\dim G}.
    \end{align}
    A connected Lie group is generated by the image of its exponential map. Therefore, $\mathcal{U}_g(\rho)=\rho$ for every $g\in G$.

    Conversely, if $\rho$ is $G$-symmetric, differentiating $\mathcal{U}_{g(\theta)}(\rho)=\rho$ at $\theta=0$
    gives $\ii[X_i,\rho]=0$ for every $i$. Equation~\eqref{eq:QGT_zero_iff_commutators} then implies $\mathcal{Q}_{\rho}^q=0$, which completes the proof.    
\end{proof}

\clearpage
\section{Converse part}

\subsection{Quantum Fisher information near i.i.d. states}
We fix a real-valued smooth function $\eta$ with compact support, satisfying
\begin{align}
    \eta(t)=t\quad (|t|\leq 1),\qquad |\eta(t)|\leq R\qquad (t\in\mathbb{R})
\end{align}
for some finite $R>0$. For example, a concrete choice is obtained by
\begin{align}
    u(x)&\coloneqq 
    \begin{cases}
        e^{-\frac{1}{x}}&\quad (x> 0)\\
        0&\quad (x\leq 0)
    \end{cases},\qquad
     \theta(x)\coloneqq \frac{u(x)}{u(x)+u(1-x)}.
\end{align}
and 
\begin{align}
    \eta(t)\coloneqq t\theta(2(t+2))\theta(2(2-t)).
\end{align}
Then, $\eta(t)=t $ for all $t\in [-1,1]$, $\mathrm{supp}\,\eta\subset[-2,2]$, and thus one may take $R=2$. Since $\eta$ has a compact support, its Fourier transformation 
\begin{align}
    \tilde{\eta}(s)\coloneqq \int_{-\infty}^\infty \eta(t) e^{-\ii st}\dd t
\end{align}
decays faster than any polynomial of $s$. Consequently,
\begin{align}
    c_\eta\coloneqq \frac{1}{2\pi}\int_{\mathbb{R}}|s||\tilde{\eta}(s)|\dd s
\end{align}
is finite. We also define 
\begin{align}
        \tau_{n,x}(t)\coloneqq x\sqrt{n}\eta\left(\frac{t}{x\sqrt{n}}\right)\label{eq:cutoff_renorm_sqrt_n}
\end{align}
for a positive number $x>0$ and a positive integer $n$. 
The cutoff in Eq.~\eqref{eq:cutoff_renorm_sqrt_n} is a smooth analogue of the $x\sqrt{n}$-scale core-tail decomposition used in the proof of Lemma~C5 in Ref.~\cite{yamaguchi_QuantumGeometricTensorDeterminesPureState_2026}. 

For a single-site operator $O$ on $\mathcal{H}$, we introduce operators on $\mathcal{H}^{\otimes n}$ by
\begin{align}
    O^{[n]}\coloneqq \sum_{i=1}^nO^{(i)},\qquad O^{(i)}\coloneqq I^{\otimes(i-1)}\otimes O\otimes I^{\otimes (n-i)}.\label{eq:iid_extension_definition}
\end{align}

In order to analyze the asymptotic behavior of quantum Fisher information near i.i.d. states, let us first derive two lemmas that we later use to bound tails. 
\begin{lemma}\label{lem:tail_N}
    Let $\sigma$ be a state on a finite-dimensional Hilbert space $\mathcal{H}$, let $C$ be a Hermitian operator on $\mathcal{H}$ satisfying $\mathrm{Tr}(\sigma C)=0$. Then, we have
    \begin{align}
        \mathrm{Tr}\left(\sigma^{\otimes n}(\tau_{n,x}(C^{[n]})-C^{[n]})^2\right)\leq \frac{\kappa_C}{x^2}n
        ,
    \end{align}
    where $\kappa_C$ is given by
    \begin{align}
        \kappa_C\coloneqq (R+1)^2\left(3\left(\mathrm{Tr}(\sigma C^2)\right)^2+\left|\mathrm{Tr}(\sigma C^4)-3\left(\mathrm{Tr}(\sigma C^2)\right)^2\right|\right),\label{eq:kappa_definition}
    \end{align}
    which is a finite constant, independent of $n$ and $x$. In particular, 
    \begin{align}
        \kappa_C\leq 5(R+1)^2\|C\|^4.\label{eq:kappa_bound}
    \end{align}
\end{lemma}
\begin{proof}
    Let $C=\sum_jc_jP_j$ be the spectral decomposition of $C$, and define a real random variable $X$ by $\mathbb{P}(X=c_j)=\mathrm{Tr}(\sigma P_j)$.
    Then, 
    \begin{align}
        \mathbb{E}X=\sum_jc_j\mathrm{Tr}(\sigma P_j)=\mathrm{Tr}(\sigma C)=0.
    \end{align}
    Let $X_1,...,X_n$ be i.i.d. copies of $X$, and define $ S_n\coloneqq \sum_{k=1}^{n}X_k$. 
    Since $C^{(1)},\cdots, C^{(n)}$ commute with each other, we have
    \begin{align}
        \mathrm{Tr}\left(\sigma^{\otimes n}\left(\tau_{n,x}(C^{[n]})-C^{[n]}\right)^2\right)=\mathbb{E}\left[(\tau_{n,x}(S_n)-S_n)^2\right].
    \end{align}

    On the event $|S_n|\leq x\sqrt{n}$, since $\tau_{n,x}(t)=t$ for $|t|\leq x\sqrt{n}$, we have $\tau_{n,x}(S_n)-S_n=0$.
    On the event $|S_n|>x\sqrt{n}$, since $|\tau_{n,x}(t)|\leq Rx\sqrt{n}$, we have
    \begin{align}
        |\tau_{n,x}(S_n)-S_n|\leq |\tau_{n,x}(S_n)|+|S_n|\leq Rx\sqrt{n}+|S_n|<(R+1)|S_n|.
    \end{align}
    Therefore,
    \begin{align}
        \mathbb{E}\left[(\tau_{n,x}(S_n)-S_n)^2\right]&=\mathbb{E}\left[(\tau_{n,x}(S_n)-S_n)^2\textbf{1}_{|S_n|\leq  x\sqrt{n}}\right]+\mathbb{E}\left[(\tau_{n,x}(S_n)-S_n)^2\textbf{1}_{|S_n|>  x\sqrt{n}}\right]\\
        &\leq (R+1)^2 \mathbb{E}\left[S_n^2\textbf{1}_{|S_n|>  x\sqrt{n}}\right]\label{eq:tail_core_decomposition}
    \end{align}

    Since $S_n^2\leq \frac{S_n^4}{x^2n}$ on the event $|S_n|>x\sqrt{n}$, we have
    \begin{align}
        \mathbb{E}\left[S_n^2\textbf{1}_{|S_n|>  x\sqrt{n}}\right]\leq \mathbb{E}\left[\frac{S_n^4}{x^2n}\textbf{1}_{|S_n|>  x\sqrt{n}}\right]\leq \mathbb{E}\left[\frac{S_n^4}{x^2n}\right]= \frac{\mathbb{E}S_n^4}{x^2n}.
    \end{align}
    Since $X$ is centered, i.e., $\mathbb{E}X=0$, we have
    \begin{align}
        \mathbb{E}S_n^4&=n\mathbb{E}X^4+\binom{4}{2}\binom{n}{2}(\mathbb{E}X^2)^2=3n^2(\mathbb{E}X^2)^2+n\left(\mathbb{E}X^4-3(\mathbb{E}X^2)^2\right)\\
        &\leq 3n^2(\mathbb{E}X^2)^2+n|\mathbb{E}X^4-3(\mathbb{E}X^2)^2|.
    \end{align}
    Since $n\leq n^2$ for any positive integer $n$, we have
    \begin{align}
        \mathbb{E}S_n^4\leq K_C n^2,\label{eq:K_C_bound}
    \end{align}
    where 
    \begin{align}
       K_C\coloneqq 3(\mathbb{E}X^2)^2 +|\mathbb{E}X^4-3(\mathbb{E}X^2)^2|=3\left(\mathrm{Tr}(\sigma C^2)\right)^2+\left|\mathrm{Tr}(\sigma C^4)-3\left(\mathrm{Tr}(\sigma C^2)\right)^2\right|,
    \end{align}
    which is independent of $n$ and $x$.

    Therefore, from Eqs.~\eqref{eq:tail_core_decomposition} and~\eqref{eq:K_C_bound}, we obtain
    \begin{align}
        \mathbb{E}\left[(\tau_{n,x}(S_n)-S_n)^2\right]\leq \frac{\kappa_C}{x^2} n,
    \end{align}
    with $\kappa_C\coloneqq (R+1)^2 K_C$. 
    
    Equation~\eqref{eq:kappa_bound} is derived as follows: First, note that $\mathrm{Tr}(\sigma C^4)\geq \left(\mathrm{Tr}(\sigma C^2)\right)^2$ since
    \begin{align}
        0\leq \mathrm{Tr}(\sigma \left(C^2-\mathrm{Tr}(\sigma C^2)\right)^2)=\mathrm{Tr}(\sigma C^4)- \left(\mathrm{Tr}(\sigma C^2)\right)^2.
    \end{align}
    If $\mathrm{Tr}(\sigma C^4)\leq 3\left(\mathrm{Tr}(\sigma C^2)\right)^2$, then
    \begin{align}
        K_C= 6\left(\mathrm{Tr}(\sigma C^2)\right)^2-\mathrm{Tr}(\sigma C^4)\leq 5\left(\mathrm{Tr}(\sigma C^2)\right)^2\leq 5\|C\|^4,
    \end{align}
    while if $\mathrm{Tr}(\sigma C^4)> 3\left(\mathrm{Tr}(\sigma C^2)\right)^2$, then
    \begin{align}
        K_C=\mathrm{Tr}(\sigma C^4)\leq \|C\|^4\leq 5\|C\|^4.
    \end{align}
\end{proof}

\begin{lemma}\label{lem:commutator_tau_norm}
    For any Hermitian operator $H$, any operator $K$, any positive integer $n$, and any $x>0$, we have
    \begin{align}
        \|[K,\tau_{n,x}(H)]\|\leq c_\eta \|[K,H]\|.
    \end{align}
\end{lemma}

\begin{proof}
For $a\coloneqq x\sqrt{n}>0$ and $\tau_a(t)\coloneqq a\eta(t/a)$, we have
\begin{align}
    \tilde{\tau_a}(s)&=a\int_{\mathbb{R}} \eta(t/a)e^{-\ii s t}\dd t=a^2\tilde{\eta}(as)
\end{align}
and hence
\begin{align}
    \frac{1}{2\pi}\int_{\mathbb{R}}|s ||\tilde{\tau_a}(s)|\dd s= \frac{1}{2\pi}\int_{\mathbb{R}}|s ||\tilde{\eta}(s)|\dd s=c_\eta.
\end{align}

We derive an upper bound of the operator norm of the following: 
\begin{align}
    [K,\tau_a(H)]=\frac{1}{2\pi}\int_{\mathbb{R}}\tilde{\tau_a}(s)[K,e^{\ii s H}]\dd s.
\end{align}
To this end, we use an identity:
\begin{align}
    [K,e^{\ii s H}]=-\ii\int_{0}^s e^{\ii t H}[H,K]e^{\ii(s-t)H}\dd t.
\end{align}
Indeed, introducing a function $F(t)\coloneqq e^{\ii t H}Ke^{\ii (s-t) H}$, we have $F(0)=Ke^{\ii s H}$ and $F(s)=e^{\ii sH}K$. Since
\begin{align}
    \frac{\dd}{\dd t}F(t)=\ii e^{\ii t H}[H,K]e^{\ii (s-t) H},
\end{align}
we have
\begin{align}
    [K,e^{\ii s H}]=-\ii \int_0^s e^{\ii t H}[H,K]e^{\ii (s-t) H} \dd t.
\end{align}
Since $e^{\ii tH}$ is unitary for Hermitian $H$, we have $\|e^{\ii t H}[H,K]e^{\ii (s-t) H}\|= \|[H,K]\|$, and hence
\begin{align}
    \|[K,e^{\ii s H}]\| \leq\int_0^{|s|} \dd t \|e^{\ii t H}[H,K]e^{\ii (s-t) H}\|=\int_0^{|s|} \dd t \|[H,K]\|=|s| \|[H,K]\|.
\end{align}
Therefore,
\begin{align}
    \|[K,\tau_a(H)]\|&\leq \frac{1}{2\pi}\int_{\mathbb{R}}|\tilde{\tau_a}(s)|\|[K,e^{\ii s H}]\| \dd s\leq  \frac{1}{2\pi}\int_{\mathbb{R}}|\tilde{\tau_a}(s)||s| \|[H,K]\|\dd s=c_\eta\|[H,K]\| .
\end{align}
\end{proof}

We now derive a lower bound on the quantum Fisher information of a $p$-parameter family of states near its i.i.d. extension:
\begin{lemma}\label{lem:master_bound}
    Fix $q\in(0,1)$.
    Let $\{\sigma_\theta\}_{\theta\in\Theta}$ be a smooth $p$-parameter family of states on $\mathcal{H}$, where $\Theta\subset\mathbb{R}^p$ is open. We fix $\theta_0\in\Theta$ and write
    \begin{align}
        \sigma&\coloneqq \sigma_{\theta_0},\qquad \partial_{i}\sigma\coloneqq \partial_{\theta_i}\sigma_\theta|_{\theta=\theta_0},\qquad L_i\coloneqq( \mathbb{J}_\sigma^{f_q})^{+} \partial_i\sigma,\\
        \sigma_n&\coloneqq \sigma^{\otimes n},\qquad \sigma_{n,\theta}\coloneqq \sigma_{\theta}^{\otimes n},\qquad \partial_{i}{\sigma}_n\coloneqq \partial_{\theta_i}\sigma_{n,\theta}|_{\theta=\theta_0}=\sum_{k=1}^n\sigma^{\otimes(k-1)}\otimes\partial_i\sigma\otimes\sigma^{\otimes(n-k)}.
    \end{align}
    Let $\{\rho_{n,\theta}\}_{\theta\in\Theta}$ be a smooth state family on $\mathcal{H}^{\otimes n}$ and write
    \begin{align}
        \rho_n\coloneqq \rho_{n,\theta_0},\qquad \partial_i\rho_n\coloneqq \partial_{\theta_i}\rho_{n,\theta}|_{\theta=\theta_0}
    \end{align}
    
    For $ z\in\mathbb{C}^p$, we write 
    \begin{align}
         \partial_z\sigma\coloneqq \sum_{i=1}^pz_i\partial_{i}\sigma,\qquad L_z\coloneqq \sum_{i=1}^p z_iL_i,\qquad 
         \partial_z\sigma_n\coloneqq \sum_{i=1}^pz_i\partial_{i}\sigma_n,\qquad \partial_z\rho_n\coloneqq \sum_{i=1}^pz_i\partial_i\rho_n.
    \end{align}
    Then, for each $n$, 
    \begin{align}
        \|\partial_z\rho_n\|_{f_q,\rho_n}^2\geq n\|\partial_z\sigma\|_{f_q,\sigma}^2-n\left(\frac{20(R+1)^2\|L_z\|^4}{x^2}+8R^2\epsilon_n x^2+d_{n,x,z}\right),\label{eq:master_bound}
    \end{align}
    for any $x>0$, where 
    \begin{align}
        \epsilon_n&\coloneqq \frac{1}{2}\left\|\rho_n-\sigma^{\otimes n}\right\|_1,\\
        d_{n,x,z}&\coloneqq \frac{2}{n}|\Re \Tr(C_{n,x,z}^\dag(\partial_z\rho_n-\partial_z\sigma_n))|,\qquad  C_{n,x,z}\coloneqq \tau_{n,x}\left(\frac{L_z^{[n]}+(L_z^{[n]})^\dag}{2}\right)+\ii \tau_{n,x}\left(\frac{L_z^{[n]}-(L_z^{[n]})^\dag}{2\ii}\right).
    \end{align}
\end{lemma}

\begin{proof}
    We fix $z\in\mathbb{C}^p$.
    From Eq.~\eqref{eq:tangent_in_range}, $\partial_z\sigma\in\Ran \mathbb{J}_\sigma^{f_q}$ and hence $\mathbb{J}_\sigma^{f_q}(L_z)=\partial_z\sigma$. Moreover,
    \begin{align}
        \Tr(\sigma L_z)=\Tr (\mathbb{J}_\sigma^{f_q}(L_z))=\Tr(\partial_z\sigma)=0.
    \end{align}
    We define $L^{[n]}$ as in Eq.~\eqref{eq:iid_extension_definition}. Since $L_z^{[n]}\coloneqq \sum_{k=1}^nL_z^{(k)}$ satisfies
    \begin{align}
        \mathbb{J}_{\sigma^{\otimes n}}^{f_q}(L_z^{[n]})=q\sigma^{\otimes n}L_z^{[n]}+(1-q)L_z^{[n]}\sigma^{\otimes n}= \partial_z\sigma_n,
    \end{align}
    we have
    \begin{align}
        \braket{L_z^{(i)},\mathbb{J}_{\sigma^{\otimes n}}^{f_q}(L_{z}^{(j)})}=\Tr(L_z^\dag \sigma)\Tr(\partial_{z}\sigma)=0\qquad (i\neq j).
    \end{align}
    Therefore,
    \begin{align}
        \|\partial_z\sigma_n\|_{f_q,\sigma^{\otimes n}}^2=\braket{L_z^{[n]},\mathbb{J}_{\sigma^{\otimes n}}^{f_q}(L_z^{[n]})}=n\braket{L_z,\mathbb{J}_{\sigma}^{f_q}L_z}=n\|\partial_z\sigma\|_{f_q,\sigma}^2.\label{eq:additivitiy_Ln}
    \end{align}

    We decompose the logarithmic derivative into its Hermitian and anti-Hermitian parts,
    \begin{align}
        L_z=A_z+\ii B_z,\qquad A_z\coloneqq \frac{L_z+L_z^\dag}{2},\qquad B_z\coloneqq\frac{L_z-L_z^\dag}{2\ii},
    \end{align}
    which satisfies
    \begin{align}
        A_z^\dag =A_z,\quad B_z^\dag =B_z,\quad \Tr(\sigma A_z)=\Tr(\sigma B_z)=0,\quad \|A_z\|\leq \|L_z\|,\quad \|B_z\|\leq \|L_z\|,\quad L_z^{[n]}=A_z^{[n]}+\ii B_z^{[n]}.
    \end{align}
    We define
    \begin{align}
        C_{n,x,z}\coloneqq \tau_{n,x}(A_z^{[n]})+\ii \tau_{n,x}(B_z^{[n]}).
    \end{align}
    By Eq.~\eqref{eq:cutoff_renorm_sqrt_n}, we have 
    \begin{align}
        \|C_{n,x,z}\|\leq\| \tau_{n,x}(A_z^{[n]})\|+\| \tau_{n,x}(B_z^{[n]})\|\leq 2R \sqrt{n}x.
    \end{align}
    
    For a general state $\rho$, since $\braket{X\pm Y,\mathbb{J}_{\rho}^{f_q}(X\pm Y)}=\braket{X,\mathbb{J}_{\rho}^{f_q}X}+\braket{Y,\mathbb{J}_{\rho}^{f_q}Y}\pm 2\Re \braket{X,\mathbb{J}_{\rho}^{f_q}Y}$, we have
    \begin{align}
        \braket{X+ Y,\mathbb{J}_{\rho}^{f_q}(X+Y)}&=2\braket{X,\mathbb{J}_{\rho}^{f_q}X}+2\braket{Y,\mathbb{J}_{\rho}^{f_q}Y}-\braket{X- Y,\mathbb{J}_{\rho}^{f_q}(X-Y)}\\
        &\leq 2\braket{X,\mathbb{J}_{\rho}^{f_q}X}+2\braket{Y,\mathbb{J}_{\rho}^{f_q}Y}.
    \end{align}
    Therefore, by Lemma~\ref{lem:tail_N}, we obtain
    \begin{align}
        \braket{C_{n,x,z}-L_z^{[n]},\mathbb{J}_{\sigma^{\otimes n}}^{f_q}(C_{n,x,z}-L_z^{[n]})}&\leq  2\Tr \left(\sigma^{\otimes n}\left(\tau_{n,x}(A_z^{[n]})-A_z^{[n]}\right)^2\right)+2 \Tr \left(\sigma^{\otimes n}\left(\tau_{n,x}(B_z^{[n]})-B_z^{[n]}\right)^2\right)\\
        &\leq 2\frac{\kappa_A+\kappa_B}{x^2}n\leq  \frac{\kappa}{x^2}n\label{eq:C_minus_L_bound},
    \end{align}
    where $\kappa\coloneqq 20(R+1)^2\|L_z\|^4$.

    Since $\mathbb{J}_{\sigma^{\otimes n}}^{f_q}L_z^{[n]}=\partial_z\sigma_n$, we have
    \begin{align}
        \braket{C_{n,x,z}-L_z^{[n]},\mathbb{J}_{\sigma^{\otimes n}}^{f_q}(C_{n,x,z}-L_z^{[n]})}=\braket{C_{n,x,z},\mathbb{J}_{\sigma^{\otimes n}}^{f_q}(C_{n,x,z})}+\braket{L_z^{[n]},\mathbb{J}_{\sigma^{\otimes n}}^{f_q}(L_z^{[n]})}-2\Re\Tr(C_{n,x,z}^\dag\partial_z\sigma_n).
    \end{align}
    Therefore, by Eqs.~\eqref{eq:additivitiy_Ln} and~\eqref{eq:C_minus_L_bound}, we get
    \begin{align}
        2\Re\Tr(C_{n,x,z}^\dag\partial_z\sigma_n)-\braket{C_{n,x,z},\mathbb{J}_{\sigma^{\otimes n}}^{f_q}(C_{n,x,z})}\geq n\|\partial_z\sigma\|^2_{f_q,\sigma}-\frac{\kappa}{x^2}n.\label{eq:sigma_bound}
    \end{align}

    By Eq.~\eqref{eq:tangent_in_range}, applied to the smooth family $\rho_{n,\theta}$, we have $\partial_z\rho_n\in\Ran\mathbb{J}_{\rho_n}^{f_q}$. 
    Thus, the variational formula in Eq.~\eqref{eq:norm_variational_formula} gives
    \begin{align}
        \|\partial_z\rho_n\|_{f_q,\rho_n}^2&\geq 2\Re \Tr(C_{n,x,z}^\dag\partial_z\rho_n)-\braket{C_{n,x,z},\mathbb{J}_{\rho_n}^{f_q}C_{n,x,z}}\\
        &\geq n\|\partial_z\sigma\|^2_{f_q,\sigma}-\frac{\kappa}{x^2}n+2\Re \Tr(C_{n,x,z}^\dag(\partial_z\rho_n-\partial_z\sigma_n))-(\braket{C_{n,x,z},\mathbb{J}_{\rho_n}^{f_q}C_{n,x,z}}-\braket{C_{n,x,z},\mathbb{J}_{\sigma^{\otimes n}}^{f_q}C_{n,x,z}}),
    \end{align}
    where we used Eq.~\eqref{eq:sigma_bound}. Since
    \begin{align}
        |\braket{C_{n,x,z},\mathbb{J}_{\rho_n}^{f_q}C_{n,x,z}}-\braket{C_{n,x,z},\mathbb{J}_{\sigma^{\otimes n}}^{f_q}C_{n,x,z}}|&=|\Tr((\rho_n-\sigma^{\otimes n})(qC_{n,x,z}C_{n,x,z}^\dag+(1-q)C_{n,x,z}^\dag C_{n,x,z} ))| \\
        &\leq \|\rho_n-\sigma^{\otimes n}\|_1\|C_{n,x,z}\|^2\leq 8 R^2 \epsilon_n x^2 n,
    \end{align}
    we obtain Eq.~\eqref{eq:master_bound}.     
\end{proof}

The above lemma is applicable to any smooth models $\sigma_\theta$ and $\rho_{n,\theta}$. If we particularize it to unitary models, the contribution $d_{n,x,z}$, originating from tangents $\partial_z\rho_n$ and $\partial_z\sigma_n$, is directly related to the trace-norm error $\epsilon_n$, thereby obtaining the following:
\begin{theorem}\label{thm:QFI_monotonicity_unitary}
    We retain the notation from Lemma~\ref{lem:master_bound}.
    Let $X_1,\ldots,X_p$ be Hermitian operators on $\mathcal{H}$. If 
    \begin{align}
        \partial_i\sigma =\ii[\sigma,X_i],\qquad \partial_i\rho_n=\ii[\rho_n,X_i^{[n]}],\label{eq:commutator_tangent}
    \end{align}
    then for each $n$, 
    \begin{align}
        \frac{1}{n}\mathcal{F}_{\rho_n}^{f_q}\geq \mathcal{F}_{\sigma}^{f_q}-h_{q,\eta}(\epsilon_n)\mathcal{X},\label{eq:QFI_bound_each_n}
    \end{align}
    where $\mathcal{X}$ is a $p\times p$ matrix defined by
    \begin{align}
        (\mathcal{X})_{ij}\coloneqq \Tr\left(\left(X_i-\frac{\Tr X_i}{d}I\right)\left(X_j-\frac{\Tr X_j}{d}I\right)\right)
    \end{align}
    and $h_{q,\eta}$ is a real-valued function defined by
    \begin{align}
        h_{q,\eta}(\epsilon)\coloneqq \frac{8\sqrt{10}R(R+1)}{\delta_q^2}\sqrt{\epsilon}+\frac{16c_\eta}{\delta_q}\epsilon,\qquad \delta_q\coloneqq \min\{q,1-q\}>0.
    \end{align}
\end{theorem}
\begin{proof}
    By shifting a scalar, we define $\tilde{X}_i\coloneqq X_i-\frac{\Tr X_i}{d}I$ so that $\Tr \tilde{X}_i=0$. Note that Eqs.~\eqref{eq:commutator_tangent} and~\eqref{eq:QFI_bound_each_n} remain invariant under the shift $X_i\mapsto \tilde{X}_i$. 
    We define $O_z\coloneqq \sum_{i=1}^pz_i\tilde{X}_i$.
    Since
    \begin{align}
        \partial_z\rho_n-\partial_z\sigma_n=\ii[\rho_n-\sigma^{\otimes n},O_z^{[n]}],
    \end{align}
    we have
    \begin{align}
         2\Re \Tr(C_{n,x,z}^\dag(\partial_z\rho_n-\partial_z\sigma_n))= 2\Re \Tr((\rho_n-\sigma^{\otimes n})\ii [O_z^{[n]},C_{n,x,z}^\dag])=\Tr((\rho_n-\sigma^{\otimes n})K_{n,x,z}),
    \end{align}
    where
    \begin{align}
        K_{n,x,z}\coloneqq \ii [O_z^{[n]},C_{n,x,z}^\dag]+\ii [(O_z^{[n]})^\dag,C_{n,x,z}].
    \end{align}
    
    Since $C_{n,x,z}=\tau_{n,x}(A_z^{[n]})+\ii \tau_{n,x}(B_z^{[n]})$, by Lemma~\ref{lem:commutator_tau_norm}, we obtain
    \begin{align}
       \|K_{n,x,z}\|&\leq \|[O_z^{[n]},\tau_{n,x}(A_z^{[n]})]\|+\|[O_z^{[n]},\tau_{n,x}(B_z^{[n]})]\|+\|[O_z^{[n]\dag},\tau_{n,x}(A_z^{[n]})]\|+\|[O_z^{[n]\dag},\tau_{n,x}(B_z^{[n]})]\|\\
       &\leq c_\eta \left( \|[O_z^{[n]},A_z^{[n]}]\|+\|[O_z^{[n]},B_z^{[n]}]\|+\|[O_z^{[n]\dag},A_z^{[n]}]\|+\|[O_z^{[n]\dag},B_z^{[n]}]\|\right).
    \end{align}
    Note that the right-hand side is independent of $x$. 
    Since
    \begin{align}
        [O_z^{[n]},A_z^{[n]}]=\sum_{k=1}^nI^{\otimes (k-1)}\otimes \left([O_z,A_z]\right)\otimes I^{\otimes (n-k)},
    \end{align}
    we have $\|[O_z^{[n]},A_z^{[n]}]\|\leq n \|[O_z,A_z]\|$. 
    Applying the same argument to each term, we get
    \begin{align}
        \|K_{n,x,z}\|&\leq n c_\eta \left( \|[O_z,A_z]\|+\|[O_z,B_z]\|+\|[O_z^{\dag},A_z]\|+\|[O_z^{\dag},B_z]\|\right)\\
        &\leq 8nc_\eta \|O_z\|\|L_z\|\leq 8nc_\eta \|O_z\|_\HS\|L_z\|,
    \end{align}
    where in the last inequality, we used $\|O_z\|\leq \|O_z\|_\HS$. Therefore,
    \begin{align}
        d_{n,x,z}=\frac{2}{n}|\Re \Tr(C_{n,x,z}^\dag(\partial_z\rho_n-\partial_z\sigma_n))|=\frac{1}{n}|\Tr((\rho_n-\sigma^{\otimes n})K_{n,x,z})|\leq  16c_\eta \|O_z\|_\HS\|L_z\|\epsilon_n.
    \end{align}

    For any $a,b>0$, 
    \begin{align}
        \inf_{x>0}\left\{\frac{a}{x^2}+bx^2\right\}=2\sqrt{ab}.
    \end{align}
    Therefore, by Lemma~\ref{lem:master_bound}, 
    \begin{align}
        \|\partial_z\rho_n\|_{f_q,\rho_n}^2%&\geq n\|\partial_z\sigma\|_{f_q,\sigma}^2-n\left(\frac{20(R+1)^2\|L_z\|^4}{x^2}+8R^2\epsilon_n x^2+d_{n,x,z}\right)\\
        &\geq  n\|\partial_z\sigma\|_{f_q,\sigma}^2-n\left(8\sqrt{10}R(R+1)\|L_z\|^2\sqrt{\epsilon_n}+16c_\eta \|O_z\|_\HS\|L_z\|\epsilon_n\right).\label{eq:bound_L_z_O_z}
    \end{align}

    Since $L_i=(\mathbb{J}_{\sigma}^{f_q})^+(\ii[\sigma,X_i])$, by using eigenvalue decomposition $\sigma=\sum_j\mu_j\ket{j}\bra{j}$, we have
    \begin{align}
        \braket{k|L_i|l}=
        \begin{cases}
            \frac{1}{q\mu_k+(1-q)\mu_l}\braket{k|\ii[\sigma,X_i]|l}=\frac{\ii(\mu_k-\mu_l)}{q\mu_k+(1-q)\mu_l}\braket{k|X_i|l}\qquad&(\text{if }q\mu_k+(1-q)\mu_l\neq 0)\\
            0\qquad &(\text{otherwise})
        \end{cases}
    \end{align}
    Since
    \begin{align}
        q\mu_k+(1-q)\mu_l\geq \delta_q(\mu_k+\mu_l)\geq \delta_q|\mu_k-\mu_l|,
    \end{align}
    whenever $q\mu_k+(1-q)\mu_l>0$, we have
    \begin{align}
        |\braket{k|L_z|l}| =\frac{|\mu_k-\mu_l|}{q\mu_k+(1-q)\mu_l}\left|\braket{k|O_z|l}\right|\leq \frac{1}{\delta_q}|\braket{k|O_z|l}|.
    \end{align}
    If $q\mu_k+(1-q)\mu_l=0$, then $\mu_k=\mu_l=0$ and $\braket{k|L_z|l}=0$, hence the same bound holds.
    Therefore,
    \begin{align}
        \|L_z\|\leq \|L_z\|_\HS=\sqrt{\sum_{k,l}|\braket{k|L_z|l}|^2}\leq \frac{1}{\delta_q}\sqrt{\sum_{k,l}|\braket{k|O_z|l}|^2}=\frac{1}{\delta_q}\|O_z\|_\HS.
    \end{align}
    Substituting this bound into Eq.~\eqref{eq:bound_L_z_O_z}, 
    \begin{align}
        \|\partial_z\rho_n\|_{f_q,\rho_n}^2%&\geq n\|\partial_z\sigma\|_{f_q,\sigma}^2-n\left(\frac{20(R+1)^2\|L_z\|^4}{x^2}+8R^2\epsilon_n x^2+d_{n,x,z}\right)\\
        &\geq  n\|\partial_z\sigma\|_{f_q,\sigma}^2-n h_{q,\eta}(\epsilon_n)\|O_z\|_\HS^2,
    \end{align}
    i.e.,
    \begin{align}
        z^\dag \left( \frac{1}{n}\mathcal{F}_{\rho_n}^{f_q}- \left(\mathcal{F}_{\sigma}^{f_q}-h_{q,\eta}(\epsilon_n)\mathcal{X}\right)\right)z\geq 0. 
    \end{align}
    Since this inequality holds for any $z\in\mathbb{C}^p$, we obtain Eq.~\eqref{eq:QFI_bound_each_n}. 
\end{proof}

\subsection{Proof of converse part}

We remark that Appendix~B.4 of Ref.~\cite{yamaguchi_QuantumGeometricTensorDeterminesPureState_2026} shows that $R(\rho\to\rho')=0$ if $\Sym_G(\rho)\not\subset\Sym_G(\rho')$ for arbitrary states $\rho$ and $\rho'$ on finite-dimensional Hilbert spaces. Here, we prove the following theorem, which serves as the converse part of the conversion rate formula. 
\begin{theorem}
    Let $U$ and $U'$ be (nonprojective) unitary representations of a compact Lie group $G$ on finite-dimensional Hilbert spaces $\mathcal{H}$ and $\mathcal{H}'$. Let $\rho$ and $\rho'$ be arbitrary states on $\mathcal{H}$ and $\mathcal{H}'$, respectively. 
    Suppose that for some $r>0$, there exist $G$-covariant channels $\{\mathcal{E}_n\}_n$ such that
    \begin{align}
       \lim_{n\to\infty} \left\|\mathcal{E}_n(\rho^{\otimes n})-\rho^{\prime \otimes \floor{rn}}\right\|_1=0.\label{eq:conversion_error_assumption}
    \end{align}
    
    Then
    \begin{align}
        \forall q\in [1/2,1],\qquad \mathcal{Q}^{q}_\rho\geq r\mathcal{Q}^{q}_{\rho'}.
    \end{align}
    Consequently,
    \begin{align}
        R(\rho\to \rho')\leq \sup\{r\geq 0\mid \forall q\in [1/2,1],\,\mathcal{Q}^{q}_\rho\geq r\mathcal{Q}^{q}_{\rho'}\}
    \end{align}
\end{theorem}
\begin{proof}
    By Eq.~\eqref{eq:QGTordering_QFIordering}, it suffices to prove the inequalities for $f_q$-QFI for every fixed $q\in(0,1)$.
    In a neighborhood of the identity element $e\in G$, we parametrize $g(\theta)=\exp(\ii\sum_i\theta_iA_i)$, where $\{A_i\}_{i=1}^{\dim G}$ is a basis of the Lie algebra. We set $m_n\coloneqq \floor{rn}$ and define models by
    \begin{align}
        \rho_{n,\theta}\coloneqq \mathcal{U}_{g(\theta)}\left(\rho \right)^{\otimes n},\qquad \rho'_{m_n,\theta}\coloneqq \mathcal{U}'_{g(\theta)}\left(\rho'\right)^{\otimes m_n},
    \end{align}
    and 
    \begin{align}
        \chi_n\coloneqq \mathcal{E}_n(\rho^{\otimes n}),\qquad \chi_{n,\theta}\coloneqq\mathcal{E}_n(\rho_{n,\theta}).
    \end{align}
    Since $\mathcal{E}_n$ is $G$-covariant, we have
    \begin{align}
        \chi_{n,\theta}=\mathcal{E}_n(\rho_{n,\theta})=\mathcal{U}'^{\otimes m_n}_{g(\theta)}(\mathcal{E}_n(\rho^{\otimes n})).
    \end{align}
    Due to this covariance, the conversion error
    \begin{align}
        \epsilon_n\coloneqq \frac{1}{2}\left\|\chi_{n}-\rho'^{\otimes m_n}\right\|_1=\frac{1}{2}\left\|\chi_{n,\theta}-\rho'_{m_n,\theta}\right\|_1
    \end{align}
    is independent of $\theta$ and $\lim_{n\to\infty}\epsilon_n=0$ follows from Eq.~\eqref{eq:conversion_error_assumption}. 
    
    By Proposition~\ref{prop:fq_QFI_monotonicity_general}, we have
    \begin{align}
        n\mathcal{F}_{\rho}^{f_q}\geq \mathcal{F}_{\chi_{n}}^{f_q}.
    \end{align}
    By Theorem~\ref{thm:QFI_monotonicity_unitary}, we have
    \begin{align}
        \mathcal{F}_{\chi_{n}}^{f_q}\geq m_n\left(\mathcal{F}_{\rho'}^{f_q}-h_{q,\eta}(\epsilon_n)\mathcal{X}\right).
    \end{align}
    Therefore,
    \begin{align}
        \mathcal{F}_{\rho}^{f_q}\geq \frac{m_n}{n}\left(\mathcal{F}_{\rho'}^{f_q}-h_{q,\eta}(\epsilon_n)\mathcal{X}\right).
    \end{align}
    In the limit $n\to\infty$, we obtain
    \begin{align}
        \mathcal{F}_{\rho}^{f_q}\geq r\mathcal{F}_{\rho'}^{f_q}.
    \end{align}
    Multiplying by $q(1-q)$ yields the metric adjusted QGT inequalities, which continuously extends to the endpoint $q=1$. 
\end{proof}

\clearpage
\section{Direct part}

\subsection{Direct part for nonprojective unitary representations}
We here prove the direct part, which we extend to projective unitary representations in Section~\ref{sec:extension_proj_rep}.
\begin{theorem}\label{thm:nonproj_unitary_rep_direct_part}
    Let $U$ and $U'$ be (nonprojective) unitary representations of a compact Lie group $G$ on finite-dimensional Hilbert spaces $\mathcal{H}$ and $\mathcal{H}'$. Let $\rho$ and $\rho'$ be arbitrary states on $\mathcal{H}$ and $\mathcal{H}'$, respectively. 
    Suppose that $\Sym_G(\rho)\subset \Sym_G(\rho')$, and for some $r>0$, 
    \begin{align}
        \forall q\in [1/2,1],\qquad \mathcal{Q}^{q}_\rho\geq r\mathcal{Q}^{q}_{\rho'}.
    \end{align}
    
    Then, for any $s\in(0,r)$ and any $\kappa>0$, there exist $G$-covariant channels $\{\mathcal{E}_n\}_n$ such that
    \begin{align}
        \left\|\mathcal{E}_n(\rho^{\otimes n})-\rho^{\prime \otimes \floor{sn}}\right\|_1=O(n^{-1/2+\kappa}).
    \end{align}
    Consequently,
    \begin{align}
        R(\rho\to \rho')\geq \sup\{r\geq 0\mid \forall q\in [1/2,1],\,\mathcal{Q}^{q}_\rho\geq r\mathcal{Q}^{q}_{\rho'}\}
    \end{align}
\end{theorem}

We first prove the following lemma.
\begin{lemma}[Convertibility between local models]\label{lemma:local_model_convertibility}
    Let $\rho$ and $\rho'$ be quantum states, and $\{X_i\}_{i=1}^p$ and $\{X_i'\}_{i=1}^p$ be sets of Hermitian operators on finite-dimensional Hilbert spaces $\mathcal{H}$ and $\mathcal{H}'$, respectively. We define
    \begin{align}
        \rho_{u,n}&\coloneqq\left(\exp\left(\ii n^{-1/2}\sum_{i=1}^pu_iX_i\right)\rho\exp\left(-\ii n^{-1/2}\sum_{i=1}^pu_iX_i\right)\right)^{\otimes n}\\
        \rho'^{(r)}_{u,n}&\coloneqq\left(\exp\left(\ii n^{-1/2}\sum_{i=1}^pu_iX_i'\right)\rho'\exp\left(-\ii n^{-1/2}\sum_{i=1}^pu_iX_i'\right)\right)^{\otimes \floor{rn}},
    \end{align}
    where $r>0$. 
    Assume that 
    \begin{align}
        \forall q\in (0,1),\qquad \mathcal{F}^{f_q}_{\rho}\geq r\mathcal{F}^{f_q}_{\rho'}.\label{eq:assumption_QFI_condition}
    \end{align} 
    For $\epsilon\in[0,1/12)$, take any $\kappa$ such that $6\epsilon<\kappa<1/2$. Then, there exist quantum channels $\Lambda_n:\mathcal{T}_1(\mathcal{H}^{\otimes n})\to \mathcal{T}_1(\mathcal{H}'^{\otimes \floor{rn}})$ such that
    \begin{align}
        \sup_{\|u\|\leq n^{\epsilon}}\left\|\Lambda_n\left(\rho_{u,n}\right) -  \rho'^{(r)}_{u,n}\right\|_1=O(n^{-1/2+\kappa}).
    \end{align}
\end{lemma}

We prove this lemma using the following two theorems, which will be proven in Parts~II and~III.
\begin{theorem}[Theorem~\ref{thm:QLAN_explicit_expression} in \protect\hyperlink{Part3}{Part~III}]\label{thm:RTA_sec_QLAN}
    Let $\rho$ be a quantum state on a finite-dimensional Hilbert space $\mathcal{H}$ of dimension $d$, and let $\{X_i\}_{i=1}^p$ be a set of Hermitian operators on $\mathcal{H}$. Consider a unitary model at rate $r>0$
    \begin{align}
         \rho_{u,n}^{(r)}\coloneqq \left(\exp\left(\ii n^{-1/2}\sum_{i=1}^p u_i X_i \right)\rho \exp\left(-\ii n^{-1/2}\sum_{i=1}^pu_i X_i \right)\right)^{\otimes \floor{rn}},\qquad u\in\mathbb{R}^p.
    \end{align}
    Using the eigenvalue decomposition $\rho=\sum_{k=1}^d\mu_k\ket{k}\bra{k}$, we define a set of pairs of labels 
    \begin{align}
        \mathcal{J}\coloneqq \{(k,l)\colon \mu_k>\mu_l\}.
    \end{align}
    If $\rho=I/d$, the unitary orbit of $\rho$ consists of a single point, so the QLAN statement is trivial. We therefore restrict attention to $\rho\neq I/d$; equivalently, $\mathcal{J}\neq\emptyset$. 
    Let $\mathcal{F}$ be the $|\mathcal{J}|$-mode bosonic Fock space, whose creation and annihilation operators satisfy the canonical commutation relation
    \begin{align}
        [a_{kl},a_{k'l'}^\dag]=\delta_{kk'}\delta_{ll'}I,\qquad [a_{kl},a_{k'l'}]=0,\qquad [a_{kl}^\dag ,a_{k'l'}^\dag]=0.
    \end{align}
    We define the Gaussian shift family consisting of $|\mathcal{J}|$ modes by
    \begin{align}
        \Phi_{z}^\rho\coloneqq D(z) \Phi^\rho D(z)^\dag, 
    \end{align}
    where the reference Gaussian state is given by
    \begin{align}
        \Phi^\rho \coloneqq \bigotimes_{(k,l)\in\mathcal{J}}\phi_{\beta_{kl}}^{(kl)},\qquad \phi_{\beta_{kl}}^{(kl)}\coloneqq \frac{e^{-\beta_{kl}a^{\dag}_{kl}a_{kl}}}{\Tr e^{-\beta_{kl}a^{\dag}_{kl}a_{kl}}},\qquad \beta_{kl}\coloneqq -\ln\left(\frac{\mu_l}{\mu_k}\right),
    \end{align}
    with the conventions $-\ln 0\coloneqq \infty$ and $\phi_{\infty}^{(kl)}=\ket{0}\bra{0}$, and the displacement operator is defined by
    \begin{align}
        D(z)\coloneqq \exp\left(\sum_{(k,l)\in\mathcal{J}}\left(z_{kl}a^\dag_{kl}-\overline{z_{kl}}a_{kl}\right)\right),\qquad z=(z_{kl})_{(k,l)\in\mathcal{J}}\in\mathbb{C}^{|\mathcal{J}|}.
    \end{align}
    We define a matrix $C\in\mathbb{C}^{|\mathcal{J}|\times p}$ whose matrix elements are given by
    \begin{align}
         C_{(k,l),i}\coloneqq \ii\sqrt{\mu_k-\mu_l}\braket{l|X_i|k}.
    \end{align}
    
    Fix $\epsilon\in[0,1/12)$, and take any $\kappa$ such that $6\epsilon<\kappa<1/2$. Then there exist quantum channels
    \begin{align}
        T_{n}^{(r)}:\mathcal{T}_1(\mathcal{H}^{\otimes \floor{rn}})\to \mathcal{T}_1(\mathcal{F}),\qquad S_{n}^{(r)}: \mathcal{T}_1(\mathcal{F})\to\mathcal{T}_1(\mathcal{H}^{\otimes \floor{rn}})
    \end{align}
    such that
    \begin{align}
        \sup_{\|u\|\leq n^\epsilon}\left\|T_{n}^{(r)}\left( \rho_{u,n}^{(r)}\right)-\Phi_{\sqrt{r}Cu}^\rho\right\|_1&=O(n^{-1/2+\kappa}),\qquad
        \sup_{\|u\|\leq n^\epsilon}\left\|\rho_{u,n}^{(r)}-S_{n}^{(r)}\left(\Phi_{\sqrt{r}Cu}^\rho\right)\right\|_1=O(n^{-1/2+\kappa}).
    \end{align}
\end{theorem}

\begin{theorem}[Conditions~(i) and~(ii) of Theorem~\ref{thm:convertibility_gaussian_shift} in \protect\hyperlink{Part2}{Part~II}]\label{thm:RTA_Gaussian_convertibility}
    Let $\Phi$ and $\Phi'$ be $m$- and $m'$-mode Gaussian states with covariance matrices $\sigma$ and $\sigma'$, and symplectic forms $\Omega$ and $\Omega'$, respectively. Using two complex matrices $C\in\mathbb{C}^{m\times p}$ and $C'\in\mathbb{C}^{m'\times p}$, define the two $p$-parameter Gaussian shift models $\mathcal{G}(\Phi,C)\coloneqq \{\Phi_{Cu}\}_{u\in\mathbb{R}^p}$ and $\mathcal{G}(\Phi',C')\coloneqq \{\Phi'_{C'u}\}_{u\in\mathbb{R}^p}$. The following conditions are equivalent:
    \begin{enumerate}[(i)]
        \item %The Gaussian shift model $\mathcal{G}(\Phi,C)$ is convertible to $\mathcal{G}(\Phi',C')$ without error. That is, 
        There exists a quantum channel $\mathcal{E}$, independent of $u\in\mathbb{R}^p$, such that $\mathcal{E}(\Phi_{Cu})=\Phi'_{C'u}$ for all $u\in\mathbb{R}^p$.
        \item For all $q\in(0,1)$, $\mathcal{F}^{f_q}_{\Phi,C}\geq \mathcal{F}^{f_q}_{\Phi',C'}$, where the two matrices denote the $f_q$-QFI matrices for $\mathcal{G}(\Phi,C)$ and $\mathcal{G}(\Phi',C')$.
            % \begin{align}
                
            % \end{align}
            % Equivalently, $\forall t\in(-1,1)$, 
            % \begin{align}
            %     K^\top \left(\sigma-\ii t\Omega\right)^{-1}K\geq K^{\prime\top} \left(\sigma'-\ii t\Omega'\right)^{-1}K',
            % \end{align}
            % where $K\in\mathbb{R}^{2m\times p}$ and $K'\in\mathbb{R}^{2m'\times p}$ are defined from $C$ and $C'$ via Eq.~\eqref{eq:definition_K_from_C}. 
        % \item There exists a real matrix $X\in \mathbb{R}^{2m'\times 2m}$ such that 
        % \begin{align}
        %     XK=K',\qquad \sigma'+\ii\Omega'\geq X(\sigma +\ii \Omega)X^\top.
        % \end{align}
        % \item There exists a Gaussian channel $\mathcal{E}_{\mathrm{G}}$, independent of $u\in\mathbb{R}^p$, such that $\mathcal{E}_{\mathrm{G}}(\Phi_{Cu})=\Phi'_{C'u}$ for all $u\in\mathbb{R}^p$.
    \end{enumerate}
\end{theorem}

Before proving Lemma~\ref{lemma:local_model_convertibility}, we calculate $f_q$-QFI for the limit Gaussian model. 
\begin{lemma}[$f_q$-QFI for the limit Gaussian model]\label{lem:limit_model_QFI_conservation}
    We retain the notation of Theorem~\ref{thm:RTA_sec_QLAN} with $r=1$. Then for each $q\in(0,1)$, $\mathcal{F}_{\Phi^{\rho},C}^{f_q}=\mathcal{F}_\rho^{f_{q}}$. Moreover, for $r\geq 0$, $\mathcal{F}_{\Phi^{\rho},\sqrt{r}C}^{f_q}=r\mathcal{F}_{\Phi^{\rho},C}^{f_q}$. 
\end{lemma}
\begin{proof}
    The covariance matrix of the $(k,l)\in\mathcal{J}$ mode of $\Phi^\rho$ is 
    \begin{align}
        \sigma_{kl}=\nu_{kl}I_2,\qquad \nu_{kl}\coloneqq \frac{1+e^{-\beta_{kl}}}{1-e^{-\beta_{kl}}}=\frac{\mu_k+\mu_l}{\mu_k-\mu_l},
    \end{align}
    where we used the standard relation of the symplectic eigenvalue $\nu_{kl}$ and the inverse temperature (see, e.g., Eq.~(3.57) of Ref.~\cite{serafiniQuantumContinuousVariables2017}). As shown in Eq.~\eqref{eq:QFI_formula_Fock}, the $f_q$-QFI for $\mathcal{G}(\Phi^{\rho},C)$ is given by
    \begin{align}
        \left(\mathcal{F}_{\Phi^{\rho},C}^{f_q}\right)_{ij}=2\sum_{(k,l)\in\mathcal{J}}(k_{i}^{kl})^\top (\nu_{kl}I_2-\ii t \Omega_1)^{-1}k_{j}^{kl},
    \end{align}
    where
    \begin{align}
        \Omega_1\coloneqq 
        \begin{pmatrix}
            0&1\\
            -1&0
        \end{pmatrix},\qquad k_i^{kl}\coloneqq \sqrt{2}
        \begin{pmatrix}
            \Re C_{(k,l),i}\\
            \Im C_{(k,l),i}
        \end{pmatrix},\qquad t\coloneqq 2q-1.
    \end{align}
    Since
    \begin{align}
        (\nu I_2-\ii t \Omega_1)^{-1}=\frac{\nu I_2+\ii t \Omega_1}{\nu^2-t^2},
    \end{align}
    we obtain
    \begin{align}
       (k_{i}^{kl})^\top (\nu_{kl}I_2-\ii t \Omega_1)^{-1}k_{j}^{kl}=\frac{\overline{C_{(k,l),i}}C_{(k,l),j}}{\nu_{kl}-t}+\frac{C_{(k,l),i}\overline{C_{(k,l),j}}}{\nu_{kl}+t}=(\mu_k-\mu_l)\left(\frac{\braket{k|X_i|l}\braket{l|X_j|k}}{\nu_{kl}-t}+\frac{\braket{l|X_i|k}\braket{k|X_j|l}}{\nu_{kl}+t}\right).
    \end{align}
    Since
    \begin{align}
        (\nu_{kl}-t)(\mu_k-\mu_l)=2((1-q)\mu_k+q\mu_l),\qquad (\nu_{kl}+t)(\mu_k-\mu_l)=2((1-q)\mu_l+q\mu_k),
    \end{align}
    we obtain
    \begin{align}
        \left(\mathcal{F}_{\Phi^{\rho},C}^{f_q}\right)_{ij}&=\sum_{(k,l)\in\mathcal{J}}\frac{(\mu_k-\mu_l)^2}{(1-q)\mu_k+q\mu_l}\braket{k|X_i|l}\braket{l|X_j|k}+\sum_{(k,l)\in\mathcal{J}}\frac{(\mu_k-\mu_l)^2}{(1-q)\mu_l+q\mu_k}\braket{l|X_i|k}\braket{k|X_j|l}\\
        &=\sum_{\substack{k,l\\(1-q)\mu_l+q \mu_k>0}}\frac{(\mu_k-\mu_l)^2}{(1-q)\mu_l+q\mu_k}\braket{l|X_i|k}\braket{k|X_j|l}=\left(\mathcal{F}_\rho^{f_q}\right)_{ij}.
    \end{align}
    Moreover, rescaling $C\mapsto\sqrt{r}C$ is equivalent to rescaling $X_i\mapsto \sqrt{r}X_i$, and hence $\mathcal{F}_{\Phi^{\rho},\sqrt{r}C}^{f_q}=r\mathcal{F}_{\Phi^{\rho},C}^{f_q}$.
\end{proof}

\begin{proof}[Proof of Lemma~\ref{lemma:local_model_convertibility}]
    We set $d\coloneqq\dim\mathcal{H}$ and $d'\coloneqq\dim\mathcal{H}'$. We first deal with the trivial cases. If $\rho'=I'/d'$, then $\rho'^{(r)}_{u,n}$ is independent of $u$, and therefore the preparation channel $A\mapsto \Tr(A)\rho'^{\otimes \floor{rn}}$ proves the claim with zero conversion error. If $\rho=I/d$, then $\mathcal{F}_{\rho}^{f_q}=0$. Since $r>0$, this implies $\mathcal{F}_{\rho'}^{f_q}=0$. Consequently, we have $[X_i',\rho']=0$ for all $i$, thus $\rho'^{(r)}_{u,n}$ is again independent of $u$, and hence the same preparation channel proves the claim. Therefore, we assume $\rho\neq I/d$ and $\rho'\neq I/d'$ below.

    Let $(\Phi^\rho,C)$ and $(\Phi^{\rho'},C')$ be the Gaussian models associated with the input and output models in Theorem~\ref{thm:RTA_sec_QLAN}, and denote their Fock spaces $\mathcal{F}_\rho$ and $\mathcal{F}_{\rho'}$. From Lemma~\ref{lem:limit_model_QFI_conservation}, the assumption of Eq.~\eqref{eq:assumption_QFI_condition} is equivalent to
    \begin{align}
        \forall q\in (0,1),\qquad \mathcal{F}^{f_q}_{\Phi^\rho,C}\geq r\mathcal{F}^{f_q}_{\Phi^{\rho'},C'}=\mathcal{F}^{f_q}_{\Phi^{\rho'},\sqrt{r}C'}.
    \end{align}
    By Theorem~\ref{thm:RTA_Gaussian_convertibility}, there exists a quantum channel $\tilde{\mathcal{E}}:\mathcal{T}_1(\mathcal{F}_\rho)\to \mathcal{T}_1(\mathcal{F}_{\rho'})$ such that
    \begin{align}
        \tilde{\mathcal{E}}(\Phi^{\rho}_{Cu})=\Phi^{\rho'}_{\sqrt{r}C'u},\qquad\forall u\in\mathbb{R}^{p}.  
    \end{align}

    Applying Theorem~\ref{thm:RTA_sec_QLAN} to $\rho_{u,n}$ and $\rho'^{(r)}_{u,n}$, there exist quantum channels $T_n:\mathcal{T}_1(\mathcal{H}^{\otimes n})\to \mathcal{T}_1(\mathcal{F}_\rho)$ and $S_n^{(r)}: \mathcal{T}_1(\mathcal{F}_{\rho'})\to \mathcal{T}_1(\mathcal{H}'^{\otimes \floor{rn}})$ such that
    \begin{align}
        \sup_{\|u\|\leq n^\epsilon}\left\|T_n(\rho_{u,n})-\Phi^{\rho}_{Cu}\right\|_1=O(n^{-1/2+\kappa}),\qquad \sup_{\|u\|\leq n^\epsilon}\left\|\rho'^{(r)}_{u,n}-S_n^{(r)}(\Phi^{\rho'}_{\sqrt{r}C'u})\right\|_1=O(n^{-1/2+\kappa}).
    \end{align}
    Therefore, for the channel $\Lambda_n\coloneqq S_n^{(r)}\circ\tilde{\mathcal{E}}\circ T_n$, by the triangle inequality and trace-norm contractivity, we obtain
    \begin{align}
        \sup_{\|u\|\leq n^\epsilon}\left\|\Lambda_n\left(\rho_{u,n}\right) -  \rho'^{(r)}_{u,n}\right\|_1\leq \sup_{\|u\|\leq n^\epsilon}\left(\left\|T_n(\rho_{u,n})-\Phi^{\rho}_{Cu}\right\|_1+\left\|\rho'^{(r)}_{u,n}-S_n^{(r)}(\Phi^{\rho'}_{\sqrt{r}C'u})\right\|_1\right)=O(n^{-1/2+\kappa}).
    \end{align}
\end{proof}

We recall the parameter-localization result used in the estimate-and-convert argument of Ref.~\cite{yamaguchi_QuantumGeometricTensorDeterminesPureState_2026}. We fix a basis $\{A_i\}_{i=1}^p$ of the Lie algebra of $G$, where $p\coloneqq \dim G$. We write $g(\theta)\coloneqq \exp(\ii\sum_{i=1}^p\theta_iA_i)$. Then the following holds.
\begin{lemma}[Lemma~D2 of Ref.~\cite{yamaguchi_QuantumGeometricTensorDeterminesPureState_2026}]\label{lem:POVM}
    Let $\rho$ be an arbitrary state on a finite-dimensional Hilbert space, let $G$ be a compact Lie group, and fix any $\eta\in(0,1/2)$. We set
    \begin{align}
        k_n\coloneqq \ceil{n^{1-\eta}}.
    \end{align}
    Then there exist constants $c_1,c_2>0$ and POVMs $M_n(\dd \hat{g})$ on $\mathcal{H}^{\otimes k_n}$ such that, with
    \begin{align}
        G_{\mathrm{succ}}^{(n)}(g)\coloneqq \left\{\hat{g}\in G\colon \exists \theta\in\mathbb{R}^{\dim G},\, \|\theta\|\leq n^{-1/2+\eta},\, \mathcal{U}_{g(\theta)}\circ \mathcal{U}_{\hat{g}}(\rho)=\mathcal{U}_g(\rho)\right\}
    \end{align}
    we have
    \begin{align}
        \inf_{g\in G}\int_{G_{\mathrm{succ}}^{(n)}(g)}\Tr\left(M_n(\dd \hat{g})\mathcal{U}_g(\rho)^{\otimes k_n}\right)\geq 1-c_1e^{-c_2 n^\eta}\label{eq:rough_estimate_succ_bound}
    \end{align}
    for all sufficiently large $n$.
\end{lemma}
\begin{proof}
    This is Lemma~D2 of Ref.~\cite{yamaguchi_QuantumGeometricTensorDeterminesPureState_2026}. The quantitative bound in Eq.~\eqref{eq:rough_estimate_succ_bound} is (D164) in Ref.~\cite{yamaguchi_QuantumGeometricTensorDeterminesPureState_2026}, which follows from the state-tomography bound in Theorem~1 of Ref.~\cite{gutaFastStateTomography2020}.
\end{proof}

\begin{proof}[Proof of Theorem~\ref{thm:nonproj_unitary_rep_direct_part}]
    By Eq.~\eqref{eq:QGTordering_QFIordering}, the assumed metric adjusted QGT inequalities imply Eq.~\eqref{eq:assumption_QFI_condition}. Thus, Lemma~\ref{lemma:local_model_convertibility} applies.
    Fix arbitrary $s\in(0,r)$ and $\kappa>0$. 
    It suffices to consider $\kappa\in (0,1/2)$ since the assertion for larger $\kappa$ follows from the result for any smaller positive exponent. Choose $\eta$ and $\epsilon$ such that
    \begin{align}
        0<\eta<\epsilon<\frac{\kappa}{6}.
    \end{align}
    
    We define
    \begin{align}
        k_n\coloneqq \ceil{n^{1-\eta}},\qquad N_n\coloneqq n-k_n,\qquad m_n\coloneqq \floor{rN_n},\qquad \ell_n\coloneqq \floor{sn}.
    \end{align}
    By Lemma~\ref{lemma:local_model_convertibility}, there are local channels $\Lambda_{N_n}$ at rate $r$ satisfying
    \begin{align}
        \sup_{\|u\|\leq N_n^\epsilon}\left\|\Lambda_{N_n}(\rho_{u,N_n})-\rho'^{(r)}_{u,N_n}\right\|_1=O(N_n^{-1/2+\kappa}).
    \end{align}
    For each $h\in G$, we define 
    \begin{align}
        \Lambda_{N_n}^{(h)}\coloneqq (\mathcal{U}_h')^{\otimes m_n}\circ \Lambda_{N_n}\circ (\mathcal{U}_{h^{-1}})^{\otimes N_n}.
    \end{align}

    Since $N_n/n\to1$ and $s<r$, we have $m_n\geq \ell_n$ for all sufficiently large $n$. Thus, let
    \begin{align}
        \mathcal{D}_n:\mathcal{T}_1(\mathcal{H}'^{\otimes m_n})\to \mathcal{T}_1(\mathcal{H}'^{\otimes \ell_n})
    \end{align}
    be the channel discarding $m_n-\ell_n$ output systems. We then define a channel $\tilde{\mathcal{E}}_n$ as follows: measure the first $k_n$ input systems with the POVM $M_n(\dd \hat{g})$ from Lemma~\ref{lem:POVM}, and conditional on the outcome $\hat{g}$, apply $\mathcal{D}_n\circ \Lambda_{N_n}^{(\hat{g})}$ to the remaining $N_n$ systems. Concretely,
    \begin{align}
        \tilde{\mathcal{E}}_n(\mathcal{U}_g(\rho)^{\otimes n})=\int_{G}\Tr \left(M_n(\dd \hat{g})\mathcal{U}_g(\rho)^{\otimes k_n}\right)\mathcal{D}_n\circ \Lambda_{N_n}^{(\hat{g})}(\mathcal{U}_g(\rho)^{\otimes N_n}).
    \end{align}

    Suppose that $\hat{g}\in G_{\mathrm{succ}}^{(n)}(g)$. Then there exists $\theta\in\mathbb{R}^p$ such that
    \begin{align}
        \|\theta\|\leq n^{-1/2+\eta},\qquad \mathcal{U}_{g(\theta)}\circ \mathcal{U}_{\hat{g}}(\rho)=\mathcal{U}_g(\rho).\label{eq:succ_set}
    \end{align}
    Let $V(\hat{g})$ be the real invertible matrix determined by
    \begin{align}
        \Ad_{\hat{g}^{-1}}(A_i)=\sum_{j=1}^p\left(V(\hat{g})\right)_{ji}A_j
    \end{align}
    For $v_n\coloneqq\sqrt{N_n}V(\hat{g})\theta $, we have $g(v_n/\sqrt{N_n})=\hat{g}^{-1}g(\theta)\hat{g}$. 
    Since the map $h\mapsto V(h)$ is continuous and $G$ is compact, 
    \begin{align}
        C_V\coloneqq \sup_{h\in G}\|V(h)\|<\infty.
    \end{align}
    Therefore,
    \begin{align}
         \|v_n\|\leq C_V\sqrt{N_n}n^{-1/2+\eta}\leq N_n^\epsilon,
    \end{align}
    where we used $\eta<\epsilon$ and $N_n/n\to1$. Moreover, Eq.~\eqref{eq:succ_set} implies $g^{-1}g(\theta)\hat{g}\in\Sym_G(\rho)$. By assumption $\Sym_G(\rho)\subset \Sym_G(\rho')$, we obtain $g^{-1}g(\theta)\hat{g}\in\Sym_G(\rho')$, and hence $\mathcal{U}'_{g(\theta)}\circ \mathcal{U}'_{\hat{g}}(\rho')=\mathcal{U}_g'(\rho')$. Combining it with $g(v_n/\sqrt{N_n})=\hat{g}^{-1}g(\theta)\hat{g}$, we have
    \begin{align}
        [\mathcal{U}_{\hat{g}^{-1}}\circ \mathcal{U}_g](\rho)=\mathcal{U}_{g(v_n/\sqrt{N_n})}(\rho),\qquad [\mathcal{U}'_{\hat{g}^{-1}}\circ \mathcal{U}'_g](\rho')=\mathcal{U}'_{g(v_n/\sqrt{N_n})}(\rho').
    \end{align}
    
    Therefore, by trace-norm contractivity of $\mathcal{D}_n$ and unitary invariance, uniformly for $g\in G$ and $\hat{g}\in G_{\mathrm{succ}}^{(n)}(g)$, we obtain
    \begin{align}
        \left\|\mathcal{D}_n\circ \Lambda_{N_n}^{(\hat{g})}(\mathcal{U}_g(\rho)^{\otimes N_n})-\mathcal{U}_g'(\rho')^{\otimes \ell_n}\right\|_1
        &\leq \left\| \Lambda_{N_n}^{(\hat{g})}(\mathcal{U}_g(\rho)^{\otimes N_n})-\mathcal{U}_g'(\rho')^{\otimes m_n}\right\|_1\\
        &=\left\| \Lambda_{N_n}(\rho_{v_n,N_n})-\rho'^{(r)}_{v_n,N_n}\right\|_1\\
        &\leq \sup_{\|u\|\leq N_n^\epsilon}\left\|\Lambda_{N_n}(\rho_{u,N_n})-\rho'^{(r)}_{u,N_n}\right\|_1=O(N_n^{-1/2+\kappa}).
    \end{align}

    On the failure event, the trace norm is at most $2$. Therefore, by convexity of the trace norm and Lemma~\ref{lem:POVM}, we obtain
    \begin{align}
        \sup_{g\in G}\left\|\tilde{\mathcal{E}}_n(\mathcal{U}_g(\rho)^{\otimes n})-\mathcal{U}_g'(\rho')^{\otimes \ell_n}\right\|_1\leq O(N_n^{-1/2+\kappa})+2c_1e^{-c_2 n^\eta}=O(n^{-1/2+\kappa}),
    \end{align}
    where we used $N_n/n\to1$ in the last equality. 

    Let $\mu_G$ be the normalized Haar measure on $G$, and define
    \begin{align}
        \mathcal{E}_n\coloneqq \int_G\dd \mu_G(h)(\mathcal{U}'_{h^{-1}})^{\otimes \ell_n}\circ \tilde{\mathcal{E}}_n\circ (\mathcal{U}_h)^{\otimes n}.
    \end{align}
    The right invariance of the Haar measure implies that $\mathcal{E}_n$ is $G$-covariant. Moreover, by the convexity and the unitary invariance of the trace norm, 
    \begin{align}
        \left\|\mathcal{E}_n(\rho^{\otimes n})-\rho'^{\otimes \ell_n}\right\|_1\leq \int_G\dd \mu_G(h)\left\|\tilde{\mathcal{E}}_n(\mathcal{U}_h(\rho)^{\otimes n})-\mathcal{U}_h'(\rho')^{\otimes \ell_n}\right\|_1=O(n^{-1/2+\kappa}).
    \end{align}
    Thus, every $s\in(0,r)$ is achievable, implying that $R(\rho\to\rho')\geq r$. 
\end{proof}

\section{Extension to projective unitary representations}\label{sec:extension_proj_rep}

We extend the conversion rate formula for nonprojective unitary representations to projective unitary representations using the arguments in Refs.~\cite{shitara_IidStateConvertibilityResourceTheory_2025,yamaguchi_QuantumGeometricTensorDeterminesPureState_2026}. Let $U$ and $U'$ be continuous projective unitary representations of $G$ on $d$- and $d'$-dimensional Hilbert spaces $\mathcal{H}$ and $\mathcal{H}'$, respectively. We define
\begin{align}
    \tilde{U}(g)\coloneqq \frac{U(g)^{\otimes d}}{\det U(g)}.
\end{align}
Then, $\tilde{U}$ is a nonprojective unitary representation of $G$~\cite[Lemma~S4]{shitara_IidStateConvertibilityResourceTheory_2025}. The continuity of $\tilde{U}$ also implies its smoothness~\cite[Corollary 3.50]{hallLieGroupsLie2015}. 

The asymptotic conversion rate $R(\rho\to\rho')$ depends on the representations $U$ and $U'$ of the input and output systems; in this section, we write it as $R_{U, U'}(\rho\to\rho')$. For the output system, we also define
\begin{align}
    \tilde{U}'(g)\coloneqq \frac{U'(g)^{\otimes d'}}{\det U'(g)}.
\end{align}
Then, Lemma~S5 of Ref.~\cite{shitara_IidStateConvertibilityResourceTheory_2025} implies
\begin{align}
    R_{U,U'}(\rho\to\rho')=\frac{d'}{d}R_{\tilde{U},\tilde{U}'}(\rho^{\otimes d}\to \rho'^{\otimes d'}).\label{eq:rate_formula_transformation}
\end{align}
Since the representations appearing on the right-hand side are nonprojective and smooth, one can apply the conversion rate formula established in the previous two sections, thereby obtaining the formula to calculate $R_{U,U'}(\rho\to\rho')$ for any continuous projective unitary representations $U,U'$. In what follows, we derive its explicit expression written in terms of metric adjusted QGTs.

For continuous projective representations $U$ and $U'$, we define their metric adjusted QGT matrices, for $q\in[0,1]$, by
\begin{align}
    \mathcal{Q}_{\rho,U}^{q}&\coloneqq \frac{1}{d}\mathcal{Q}^{q}_{\rho^{\otimes d},\tilde{U}},\\
    \mathcal{Q}_{\rho',U'}^{q}&\coloneqq \frac{1}{d'}\mathcal{Q}^{q}_{\rho'^{\otimes d'},\tilde{U}'},
\end{align}
where the right-hand sides are the metric adjusted QGTs of the differentiable unitary model generated by $\tilde{U}$ at $\rho^{\otimes d}$ and $\tilde{U}'$ at $\rho'^{\otimes d'}$, respectively.

We remark that if $U$ is differentiable, the metric adjusted QGT defined in this way agrees with the standard generator-based definition. Indeed, after choosing a phase convention such that $U(e)=I$, define
\begin{align}
    X_i&\coloneqq -\ii \partial_{i}U(g(\theta))|_{\theta=0},\qquad 
    \tilde{X}_i\coloneqq -\ii\partial_{i}\tilde{U}(g(\theta))|_{\theta=0}=\sum_{a=1}^dI^{\otimes (a-1)}\otimes X_i\otimes I^{\otimes (d-a)}-\Tr(X_i)I.
\end{align}
The term proportional to the identity does not contribute to the tangent of the unitary orbit. Thus, by additivity of metric adjusted QGT,
\begin{align}
    \mathcal{Q}^{q}_{\rho^{\otimes d},\tilde{U}}=d \mathcal{Q}_{\rho,U}^{q},
\end{align}
where the right-hand side is the standard definition of the metric adjusted QGT. Similarly, we obtain $\mathcal{Q}^{q}_{\rho'^{\otimes d'},\tilde{U}'}=d' \mathcal{Q}_{\rho',U'}^{q}$.

Moreover, since the scalar determinant factor cancels under conjugation, for $\tilde{\mathcal{U}}_g(\cdot)\coloneqq \tilde{U}(g)(\cdot)\tilde{U}(g)^\dag$, we have
\begin{align}
    \tilde{\mathcal{U}}_g(\rho^{\otimes d})=(\mathcal{U}_g(\rho))^{\otimes d}.
\end{align}
Consequently, $\Sym_{G,U}(\rho)=\Sym_{G,\tilde{U}}(\rho^{\otimes d})$, where we have made explicit the dependence of the symmetry subgroup on both the state and the representation. Therefore, by Eq.~\eqref{eq:rate_formula_transformation}, we obtain
\begin{align}
    R_{U,U'}(\rho\to\rho')&=\frac{d'}{d}R_{\tilde{U},\tilde{U}'}(\rho^{\otimes d}\to \rho'^{\otimes d'})\\
    &=
    \begin{cases}
        \frac{d'}{d}\sup\left\{r\geq 0\colon\forall q\in[1/2,1],\, \mathcal{Q}^{q}_{\rho^{\otimes d},\tilde{U}}\geq r\, \mathcal{Q}^{q}_{\rho'^{\otimes d'},\tilde{U}'} \right\}&\qquad (\text{if }\Sym_{G,\tilde{U}}(\rho^{\otimes d})\subset \Sym_{G,\tilde{U}'}(\rho'^{\otimes d'}))\\
        0&\qquad (\text{otherwise})
    \end{cases}\\
    &=
    \begin{cases}
        \frac{d'}{d}\sup\left\{r\geq 0\colon \forall q\in[1/2,1] ,\, d\mathcal{Q}^{q}_{\rho,U}\geq rd'\, \mathcal{Q}^{q}_{\rho',U'} \right\}&\qquad (\text{if }\Sym_{G,U}(\rho)\subset \Sym_{G,U'}(\rho'))\\
        0&\qquad (\text{otherwise})
    \end{cases}\\
    &=
    \begin{cases}
        \sup\left\{r\geq 0\colon \forall q\in[1/2,1],\,\mathcal{Q}^{q}_{\rho,U}\geq r\, \mathcal{Q}^{q}_{\rho',U'} \right\}&\qquad (\text{if }\Sym_{G,U}(\rho)\subset \Sym_{G,U'}(\rho'))\\
        0&\qquad (\text{otherwise})
    \end{cases},
\end{align}
which establishes the conversion rate formula.

Finally, as an extension of Theorem~\ref{thm:nonproj_unitary_rep_direct_part}, we prove that one can achieve the same conversion error rate as for nonprojective representation. 

\begin{theorem}[Extension to continuous projective representations]
    Let $U$ and $U'$ be continuous projective unitary representations of a compact Lie group $G$ on finite-dimensional Hilbert spaces $\mathcal{H}$ and $\mathcal{H}'$, respectively. Let $\rho$ and $\rho'$ be arbitrary states on $\mathcal{H}$ and $\mathcal{H}'$. Suppose that $\Sym_{G,U}(\rho)\subset \Sym_{G,U'}(\rho')$. If, for some $r>0$, 
    \begin{align}
        \forall q\in [1/2,1],\qquad \mathcal{Q}_{\rho,U}^{q}\geq r \,\mathcal{Q}_{\rho',U'}^{q},
    \end{align}
    then, for every $s\in (0,r)$ and every $\kappa>0$, there exist $G$-covariant channels $\mathcal{E}_n$ such that
    \begin{align}
        \left\|\mathcal{E}_n(\rho^{\otimes n})-\rho'^{\otimes \floor{sn}}\right\|_1=O(n^{-1/2+\kappa}).\label{eq:direct_part_proj_error_rate}
    \end{align}
    Consequently, 
    \begin{align}
        R_{U, U'}(\rho\to \rho')\geq \sup\left\{r\geq 0\colon \forall q\in [1/2,1],\, \mathcal{Q}_{\rho,U}^{q}\geq r \,\mathcal{Q}_{\rho',U'}^{q}\right\}\label{eq:direct_part_proj}
    \end{align}
\end{theorem}
\begin{proof}
    Let $d\coloneqq \dim\mathcal{H}$ and $d'\coloneqq \dim\mathcal{H}'$, and define $\tilde{U}$ and $\tilde{U}'$ as above. 
    The assumption for the metric adjusted QGT implies
    \begin{align}
        \mathcal{Q}^{q}_{\rho^{\otimes d},\tilde{U}}\geq \frac{rd}{d'}\mathcal{Q}^{q}_{\rho'^{\otimes d'},\tilde{U}'}.
    \end{align}
    Moreover, the inclusion relation between the symmetry subgroup implies
    \begin{align}
        \Sym_{G,\tilde{U}}(\rho^{\otimes d})\subset \Sym_{G,\tilde{U}'}(\rho'^{\otimes d'}).
    \end{align}
    Since $\tilde{U}$ and $\tilde{U}'$ are differentiable nonprojective unitary representations, Theorem~\ref{thm:nonproj_unitary_rep_direct_part} applies. 

    Fix $s\in(0,r)$ and choose $s<t<r$. Define $\tilde{t}\coloneqq \frac{d}{d'}t$. Since $\tilde{t}<\frac{d}{d'}r$, Theorem~\ref{thm:nonproj_unitary_rep_direct_part} implies that there are $G$-covariant channels $\mathcal{E}_N'$ that asymptotically convert $N$ copies of $\rho^{\otimes d}$ into $\floor{\tilde{t}N}$ copies of $\rho'^{\otimes d'}$ with trace-norm error $O(N^{-1/2+\kappa})$. 

    For an original input of $n$ copies, we let $N_n\coloneqq \floor{n/d}$. We first discard $n-dN_n$ input systems and apply $\mathcal{E}_{N_n}'$. The number of output copies of $\rho'$ is
    \begin{align}
        M_n\coloneqq d' \floor{\tilde{t}N_n}=d'\floor*{\frac{d}{d'}tN_n}.
    \end{align}
    Since $N_n/n\to 1/d$, we have $M_n/n\to t$. Since $t>s$, for all sufficiently large $n$, we have $M_n\geq\floor{sn}$ and hence we can discard $M_n-\floor{sn}$ copies of $\rho'$. The channel $\mathcal{E}_n$ resulting from this process is $G$-covariant, and satisfies
    \begin{align}
        \|\mathcal{E}_n(\rho^{\otimes n})-\rho'^{\otimes \floor{sn}}\|_1=O(N_n^{-1/2+\kappa})=O(n^{-1/2+\kappa}),
    \end{align}
    which establishes Eq.~\eqref{eq:direct_part_proj_error_rate}. Since $s\in(0,r)$ was arbitrary, Eq.~\eqref{eq:direct_part_proj} follows.
\end{proof}

\clearpage
\section{Detailed derivations of other results in the main text}

\subsection{Proof that a finite number of $q$ is insufficient (Proposition~\ref{prop:finite_q_is_insufficient})}\label{sec:finite_number_q_insufficient}
To prove Proposition~\ref{prop:finite_q_is_insufficient}, we use $\mathcal{Q}_\rho^q/\mathcal{Q}_{\rho'}^q=\mathcal{F}_{\rho}^{f_q}/\mathcal{F}_{\rho'}^{f_q}$ for $q\in(0,1)$. 
For each $q_0\in[1/2,1)$, we construct a pair of full-rank input and output states whose QFI ratio has a unique minimum at $q=q_0$.

Let $\mathcal{H}$ be a single-qubit Hilbert space and fix an orthonormal basis $\{\ket{0},\ket{1}\}$. We define the Hamiltonian as $ H\coloneqq \ket{1}\bra{1}$, and the unitary representation of $U(1)$ by
\begin{align}
    U(e^{\ii\theta})\coloneqq e^{-\ii H\theta},\qquad \theta\in[0,2\pi).
\end{align}

For $ \ket{\pm}\coloneqq\left(\ket{0}\pm\ket{1}\right)/\sqrt{2}$, we define
\begin{align}
    \tau_s&\coloneqq\frac{1+s}{2}\ket{+}\bra{+}+\frac{1-s}{2}\ket{-}\bra{-}=\frac{1}{2}
    \begin{pmatrix}
        1&s\\
        s&1
    \end{pmatrix},
\end{align}
where $s\in(0,1)$ is a real parameter. Under the $U(1)$ action,
\begin{align}
    U(e^{\ii\theta})\tau_s U(e^{\ii\theta})^\dag=\frac{1}{2}
    \begin{pmatrix}
        1&se^{\ii\theta}\\
        se^{-\ii\theta}&1
    \end{pmatrix}.
\end{align}
Since $s>0$, invariance requires $e^{\ii\theta}=1$, and hence
\begin{align}
    \Sym_{U(1)}(\tau_s)=\{e\}.
\end{align}
Moreover, since $\left|\braket{+|H|-}\right|^2=1/4$, the $f_q$-QFI is given by
\begin{align}
    \mathcal{F}^{f_q}_{\tau_s}&=\frac{s^2/4}{(1-q)\frac{1+s}{2}+q\frac{1-s}{2}}+\frac{s^2/4}{(1-q)\frac{1-s}{2}+q\frac{1+s}{2}}=\frac{s^2}{1-s^2u},
\end{align}
where $u\coloneqq (1-2q)^2$. 

We next introduce another qubit $\mathcal{H}_{\mathrm{flag}}$ with vanishing Hamiltonian, which serves as a classical flag system. 
Fix $a,b,c\in(0,1)$ satisfying
\begin{align}
    a^2<c^2<b^2.
\end{align}
For $w\in(0,1)$, define the input state on
$\mathcal{H}\otimes\mathcal{H}_{\mathrm{flag}}$ by
\begin{align}
    \rho_w
    \coloneqq
    w\,\tau_a\otimes\ket{0}\bra{0}
    +(1-w)\,\tau_b\otimes\ket{1}\bra{1},
\end{align}
and the output state on $\mathcal{H}$ by
\begin{align}
    \rho'\coloneqq\tau_c.
\end{align}
Since the flag system is invariant under $U(1)$, we have
\begin{align}
    \Sym_{U(1)}(\rho_w)=\Sym_{U(1)}(\rho')=\{e\}.
\end{align}
Thus, there is no obstruction in converting $\rho_w$ into $\rho'$ from the symmetry-subgroup condition.

The $f_q$-QFIs of $\rho_w$ and $\rho'$ are given by
\begin{align}
    \mathcal{F}^{f_q}_{\rho_w}&=w\frac{a^2}{1-a^2u}+(1-w)\frac{b^2}{1-b^2u},\qquad \mathcal{F}^{f_q}_{\rho'}=\frac{c^2}{1-c^2u}.
\end{align}
To calculate the conversion rate from $\rho_w$ to $\rho'$, we introduce
\begin{align}
    r_w(u)
    \coloneqq\frac{\mathcal{F}^{f_q}_{\rho_w}}{\mathcal{F}^{f_q}_{\rho'}}=
    \frac{1-c^2u}{c^2}\left(w\frac{a^2}{1-a^2u}+(1-w)\frac{b^2}{1-b^2u}\right).
\end{align}
Since the symmetry-subgroup condition is satisfied, the conversion rate is given by
\begin{align}
    R(\rho_w\to\rho')=\inf_{u\in[0,1)} r_w(u).\label{eq:inf_rho_w}
\end{align}

We now show that, by appropriately choosing $w$, the infimum is uniquely attained at an arbitrarily chosen point $u_0\coloneqq(1-2q_0)^2$. 
Differentiating $r_w(u)$ with respect to $u$, we obtain
\begin{align}
    r'_w(u)=\frac{1}{c^2}\left(w\frac{a^2(a^2-c^2)}{(1-a^2u)^2}+(1-w)\frac{b^2(b^2-c^2)}{(1-b^2u)^2}\right).
\end{align}
Since $a^2<c^2<b^2$, the condition $r'_w(u_0)=0$ is satisfied by
\begin{align}
    w=w_0
    \coloneqq
    \frac{\dfrac{b^2(b^2-c^2)}{(1-b^2u_0)^2}}{\dfrac{a^2(c^2-a^2)}{(1-a^2u_0)^2}+\dfrac{b^2(b^2-c^2)}{(1-b^2u_0)^2}}
    \in(0,1).
\end{align}

To prove that $u_0$ is the unique global minimizer, define
\begin{align}
    C&\coloneqq\frac{w_0a^2(c^2-a^2)}{(1-a^2u_0)^2}=\frac{(1-w_0)b^2(b^2-c^2)}{(1-b^2u_0)^2}>0.
\end{align}
Using this relation, the derivative can be rewritten as
\begin{align}
    r'_{w_0}(u)=\frac{C}{c^2}\left(\left(\frac{1-b^2u_0}{1-b^2u}\right)^2-\left(\frac{1-a^2u_0}{1-a^2u}\right)^2\right).
\end{align}
Since
\begin{align}
    \frac{1-b^2u_0}{1-b^2u}-\frac{1-a^2u_0}{1-a^2u}=\frac{(b^2-a^2)(u-u_0)}{(1-a^2u)(1-b^2u)},
\end{align}
we obtain
\begin{align}
    r'_{w_0}(u)
    \begin{cases}
        <0,&\qquad 0\leq u<u_0,\\
        =0,&\qquad u=u_0,\\
        >0,&\qquad u_0<u<1.
    \end{cases}
\end{align}
Consequently, $u_0$ is the unique global minimizer of $r_{w_0}(u)$ on $[0,1)$, and
\begin{align}
    R(\rho_{w_0}\to\rho')=r_{w_0}(u_0).\label{eq:rate_r_w0_u0}
\end{align}
Since $u=(1-2q)^2$, the point $u=u_0$ corresponds to $q=q_0$ uniquely on the interval $q\in[1/2,1)$. 

As a concrete example, we take 
\begin{align}
    a^2=\frac{1}{4},\qquad c^2=\frac{1}{2},\qquad b^2=\frac{3}{4}.
\end{align}
Then,
\begin{align}
    r_{w}(u)=(2-u)\left(\frac{w}{4-u}+\frac{3(1-w)}{4-3u}\right),\qquad w_0(u_0)=\frac{3(4-u_0)^2}{(4-3u_0)^2+3(4-u_0)^2}.
\end{align}
Figure~\ref{fig:arbitrary_u0_bottleneck} shows $r_{w_0(u_0)}(u)$ for several values $u_0=0.1,0.3,0.5,0.7,0.9$. For each choice of $u_0$, the corresponding curve attains the infimum in Eq.~\eqref{eq:inf_rho_w} uniquely at $u=u_0$. Thus, the rate-determining point can be placed at an arbitrary interior point of the interval.
\begin{figure}
    \centering
    \includegraphics[width=0.5\linewidth]{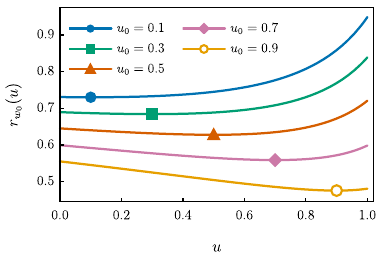}
    \caption{Tunable interior bottleneck of the metric adjusted QGT ratio (equal to the $f_q$-QFI ratio). The curves show $r_{w_0}(u)$ for several values of $u_0$. The markers indicate the unique global minima at $u=u_0$, demonstrating that the bottleneck can be moved continuously across the interior of the interval.}
    \label{fig:arbitrary_u0_bottleneck}
\end{figure}

As a consequence of the above observation, we obtain the following:
\begin{proposition}[No fixed finite $q$-sampling universally determines the conversion rate; Formal statement of Proposition~\ref{prop:finite_q_is_insufficient}]
    For $q\in[0,1]$, we define
    \begin{align}
        s_q\coloneqq \sup\left\{r\geq 0\colon \mathcal{Q}_{\rho}^{q}\geq r \mathcal{Q}_{\rho'}^{q}\right\}.
    \end{align}
    Let $S\subset [1/2,1]$ be any nonempty finite set fixed independently of the states. Then there exist full-rank input and output states $\rho,\rho'$ with $\Sym_{U(1)}(\rho)=\Sym_{U(1)}(\rho')=\{e\}$, and a point $ q_0\in (1/2,1)\setminus S$ such that
    \begin{align}
        R(\rho\to\rho')=s_{q_0}<\min_{q\in S}s_q.
    \end{align}
    Moreover, $q_0$ is the unique minimizer of $s_q$ on the interval $[1/2,1]$.  
\end{proposition}
\begin{proof}
    Choose $q_0\in(1/2,1)\setminus S$ and use the full-rank states $\rho_{w_0},\rho'$ constructed above. For $q\in[1/2,1)$, we have
    \begin{align}
        s_q=\frac{\mathcal{Q}^q_{\rho_{w_0}}}{\mathcal{Q}_{\rho'}^q}=r_{w_0}((1-2q)^2),
    \end{align}
    which has a unique minimum at $q_0$. At $q=1$, both states are full-rank, so $\mathcal{Q}_{\rho_{w_0}}^1=\mathcal{Q}_{\rho'}^1=0$ and hence $s_1=+\infty$. Thus, the endpoint ratio condition imposes no additional constraint, and every sampled bound is strictly larger than $s_{q_0}$. 
\end{proof}

\subsection{Semidefinite programming formulation of the conversion rate}
To express the conversion rate as a semidefinite program, we use the following theorem instead of Theorem~\ref{thm:RTA_Gaussian_convertibility}:
\begin{theorem}[Conditions~(ii) and~(iii) of Theorem~\ref{thm:convertibility_gaussian_shift} in \protect\hyperlink{Part2}{Part~II}]\label{thm:RTA_Gaussian_convertibility_SDP}
    Let $\Phi$ and $\Phi'$ be $m$- and $m'$-mode Gaussian states with covariance matrices $\sigma$ and $\sigma'$, and symplectic forms $\Omega$ and $\Omega'$, respectively. Using two complex matrices $C\in\mathbb{C}^{m\times p}$ and $C'\in\mathbb{C}^{m'\times p}$, define the two $p$-parameter Gaussian shift models $\mathcal{G}(\Phi,C)\coloneqq \{\Phi_{Cu}\}_{u\in\mathbb{R}^p}$ and $\mathcal{G}(\Phi',C')\coloneqq \{\Phi'_{C'u}\}_{u\in\mathbb{R}^p}$. The following conditions are equivalent:
    \begin{enumerate}[(i)]\setcounter{enumi}{1}
        %\item %The Gaussian shift model $\mathcal{G}(\Phi,C)$ is convertible to $\mathcal{G}(\Phi',C')$ without error. That is, 
        %There exists a quantum channel $\mathcal{E}$, independent of $u\in\mathbb{R}^p$, such that $\mathcal{E}(\Phi_{Cu})=\Phi'_{C'u}$ for all $u\in\mathbb{R}^p$.
        \item For all $q\in(0,1)$, $\mathcal{F}^{f_q}_{\Phi,C}\geq \mathcal{F}^{f_q}_{\Phi',C'}$, where the two matrices denote the $f_q$-QFI matrices for $\mathcal{G}(\Phi,C)$ and $\mathcal{G}(\Phi',C')$.
        \item There exists a real matrix $X\in \mathbb{R}^{2m'\times 2m}$ such that 
        \begin{align}
            XK=K',\qquad \sigma'+\ii\Omega'\geq X(\sigma +\ii \Omega)X^\top,
        \end{align}
        where
        \begin{align}
            K\coloneqq \sqrt{2}
        \begin{pmatrix}
            \Re(C_{11})&\Re(C_{12})&\cdots&\Re(C_{1p})\\
            \Im(C_{11})&\Im(C_{12})&\cdots&\Im(C_{1p})\\
            \Re(C_{21})&\Re(C_{22})&\cdots&\Re(C_{2p})\\
            \Im(C_{21})&\Im(C_{22})&\cdots&\Im(C_{2p})\\
            \vdots&\vdots & \vdots & \vdots\\
            \Re(C_{m1})&\Re(C_{m2})&\cdots&\Re(C_{mp})\\
            \Im(C_{m1})&\Im(C_{m2})&\cdots&\Im(C_{mp})\\
        \end{pmatrix},\qquad  K'\coloneqq \sqrt{2}
        \begin{pmatrix}
            \Re(C'_{11})&\Re(C'_{12})&\cdots&\Re(C'_{1p})\\
            \Im(C'_{11})&\Im(C'_{12})&\cdots&\Im(C'_{1p})\\
            \Re(C'_{21})&\Re(C'_{22})&\cdots&\Re(C'_{2p})\\
            \Im(C'_{21})&\Im(C'_{22})&\cdots&\Im(C'_{2p})\\
            \vdots&\vdots & \vdots & \vdots\\
            \Re(C'_{m'1})&\Re(C'_{m'2})&\cdots&\Re(C'_{m'p})\\
            \Im(C'_{m'1})&\Im(C'_{m'2})&\cdots&\Im(C'_{m'p})\\
        \end{pmatrix}.
        \end{align}
    \end{enumerate}
\end{theorem}

Using the above theorem, we derive a semidefinite programming formulation of the conversion rate. To this end, we introduce notation. For a quantum state $\rho$ on a finite-dimensional Hilbert space $\mathcal{H}$ of dimension $d$, let $\rho=\sum_{i=1}^d\mu_i\ket{i}\bra{i}$ be the eigenvalue decomposition. We define
\begin{align}
     \mathcal{J}_\rho&\coloneqq \{(k,l)\colon \mu_k>\mu_l\},\qquad m_\rho\coloneqq \max\{1,|\mathcal{J}_\rho|\}.
\end{align}
If $\mathcal{J}_\rho\neq \emptyset$, we define $C\in\mathbb{C}^{m_\rho\times p}$ and $K\in\mathbb{R}^{2m_\rho\times p}$ by
\begin{align}
    C_{(k,l),i}&\coloneqq \ii \sqrt{\mu_k-\mu_l}\braket{l|X_i|k},\qquad  K\coloneqq \sqrt{2}
        \begin{pmatrix}
            \Re(C_{11})&\Re(C_{12})&\cdots&\Re(C_{1p})\\
            \Im(C_{11})&\Im(C_{12})&\cdots&\Im(C_{1p})\\
            \Re(C_{21})&\Re(C_{22})&\cdots&\Re(C_{2p})\\
            \Im(C_{21})&\Im(C_{22})&\cdots&\Im(C_{2p})\\
            \vdots&\vdots & \vdots & \vdots\\
            \Re(C_{m_\rho 1})&\Re(C_{m_\rho 2})&\cdots&\Re(C_{m_\rho p})\\
            \Im(C_{m_\rho 1})&\Im(C_{m_\rho 2})&\cdots&\Im(C_{m_\rho p})\\
        \end{pmatrix},\\
    \sigma&\coloneqq \bigoplus_{(k,l)\in\mathcal{J}_\rho}\frac{\mu_k+\mu_l}{\mu_k-\mu_l}I_2,\qquad \Omega\coloneqq \bigoplus_{(k,l)\in\mathcal{J}_\rho}
    \begin{pmatrix}
        0&1\\
        -1&0
    \end{pmatrix},\qquad B\coloneqq \sigma+\ii\Omega,
\end{align}
If $\mathcal{J}_\rho=\emptyset$, we set
\begin{align}
    C=0_{1\times p},\qquad K=0_{2\times p},\qquad \sigma=I_2,\qquad \Omega=
    \begin{pmatrix}
        0&1\\
        -1&0
    \end{pmatrix},\qquad B\coloneqq \sigma+\ii\Omega.
\end{align}

For a primed state $\rho'$ on another finite-dimensional Hilbert space $\mathcal{H}'$, we define $\mathcal{J}_{\rho'}$, $m_{\rho'}$, $C'$, $K'$, $\sigma'$, $\Omega'$ and $B'$ in the same way. By construction, $B\geq 0$ and $B'\geq 0$. 

\begin{proposition}[Semidefinite programming]
    Let $U$ and $U'$ be continuous nonprojective unitary representations of a compact Lie group $G$ on finite-dimensional Hilbert spaces $\mathcal{H}$ and $\mathcal{H}'$, respectively. Set $p\coloneqq \dim G$. Using local coordinates $g(\theta)$ around the identity satisfying $g(0)=e$, we define $X_i\coloneqq -\ii\partial_{\theta_i}U(g(\theta))|_{\theta=0}$ and $X_i'\coloneqq -\ii\partial_{\theta_i}U'(g(\theta))|_{\theta=0}$.  Let $\rho$ and $\rho'$ be quantum states on $\mathcal{H}$ and $\mathcal{H}'$, respectively, satisfying $\Sym_G(\rho)\subset \Sym_G(\rho')$. Then
    \begin{align}
        R(\rho\to\rho')=r_{\mathrm{SDP}},
    \end{align}
    where
    \begin{align}
        r_{\mathrm{SDP}}&\coloneqq \sup_{r,L} \,r\\
        \mathrm{subject \,\, to}&\, \, r\geq 0,\qquad L\in\mathbb{R}^{2m_{\rho'}\times 2m_\rho},\\
        &\, \, LK=rK'\\
        &\, \,\mathcal{M}(r,L)\coloneqq 
        \begin{pmatrix}
            rB'&LB\\
            BL^\top & B
        \end{pmatrix}\geq 0.
    \end{align}
\end{proposition}
\begin{proof}
    Fix $r>0$. By Eq.~\eqref{eq:QGTordering_QFIordering}, Lemma~\ref{lem:limit_model_QFI_conservation} and Theorem~\ref{thm:RTA_Gaussian_convertibility_SDP}, applied to the output Gaussian model with $C'$ replaced by $\sqrt{r}C'$, we have
    \begin{align}
        \mathcal{Q}_\rho^{q}\geq r\mathcal{Q}_{\rho'}^{q},\qquad \forall q\in[1/2,1]\qquad 
        \Longleftrightarrow\qquad \exists X\in \mathbb{R}^{2m_{\rho'}\times 2m_\rho} \quad \text{such that }\quad 
            XK=\sqrt{r}K',\qquad B'\geq XBX^\top.
    \end{align}
    Note that if $\mathcal{J}_\rho=\emptyset$, the dummy Gaussian model introduced above has $K=0$ and hence zero $f_q$-QFI, in agreement with $\mathcal{F}^{f_q}_\rho=0$. The same observation applies to $\rho'$ when $\mathcal{J}_{\rho'}=\emptyset$.

    Making the invertible change of variable $L\coloneqq \sqrt{r}X$, we obtain
    \begin{align}
        \mathcal{Q}_\rho^{q}\geq r\mathcal{Q}_{\rho'}^{q},\qquad \forall q\in[1/2,1]\qquad 
        \Longleftrightarrow\qquad \exists L\in \mathbb{R}^{2m_{\rho'}\times 2m_\rho} \quad \text{such that }\quad 
            LK=rK',\qquad rB'-LBL^\top\geq 0.
    \end{align}
    For any real $L$, the matrix
    \begin{align}
        T_L\coloneqq 
        \begin{pmatrix}
            I_{2m_{\rho'}}&-L\\
            0&I_{2m_{\rho}}
        \end{pmatrix}
    \end{align}
    is invertible and satisfies
    \begin{align}
        T_L\mathcal{M}(r,L)T_L^\top=
        \begin{pmatrix}
            rB'-LBL^\top&0\\
            0& B
        \end{pmatrix}.
    \end{align}
    Since $B\geq 0$ and $T_L$ is invertible, we have
    \begin{align}
        \mathcal{M}(r,L)\geq 0\Longleftrightarrow rB'-LBL^\top\geq 0.
    \end{align}
    Consequently, for every $r>0$,
    \begin{align}
        \mathcal{Q}_\rho^{q}\geq r\mathcal{Q}_{\rho'}^{q},\qquad \forall q\in[1/2,1]\qquad 
        \Longleftrightarrow\qquad \exists L\in \mathbb{R}^{2m_{\rho'}\times 2m_\rho} \quad \text{such that }\quad 
            LK=rK',\qquad \mathcal{M}(r,L)\geq 0.
    \end{align}
    
    At $r=0$, the metric adjusted QGT condition holds automatically, while the semidefinite-program constraints are feasible with $L=0$.

    By the assumed inclusion $\Sym_G(\rho)\subset\Sym_G(\rho')$ and the conversion-rate formula,
    \begin{align}
        R(\rho\to\rho')&=\sup\{r\geq 0\colon \forall q\in [1/2,1],\,  \mathcal{Q}_\rho^{q}\geq r\mathcal{Q}_{\rho'}^{q}\}\\
        &=\sup\{r\geq 0\colon  \exists L\in \mathbb{R}^{2m_{\rho'}\times 2m_\rho} \, \text{such that }\,
            LK=rK',\, \mathcal{M}(r,L)\geq 0\}=r_{\mathrm{SDP}}.
    \end{align}

    Finally, we verify that the above optimization problem can be written in the standard form of a real semidefinite program. Since $B=\sigma+\ii\Omega$ and $B'=\sigma'+\ii\Omega'$, we write
    \begin{align}
        \mathcal{M}(r,L)
        =\mathcal{S}(r,L)+\ii \mathcal{A}(r,L),
    \end{align}
    where
    \begin{align}
        \mathcal{S}(r,L)
        \coloneqq
        \begin{pmatrix}
            r\sigma'&L\sigma\\
            \sigma L^\top&\sigma
        \end{pmatrix},\qquad 
        \mathcal{A}(r,L)
        \coloneqq
        \begin{pmatrix}
            r\Omega'&L\Omega\\
            \Omega L^\top&\Omega
        \end{pmatrix}
    \end{align}
    satisfying
    \begin{align}
        \mathcal{S}(r,L)^\top =\mathcal{S}(r,L),\qquad \mathcal{A}(r,L)^\top =-\mathcal{A}(r,L).
    \end{align}
    For any $z=x+\ii y$ with real vectors $x$ and $y$, we have
    \begin{align}
        z^\dagger\mathcal{M}(r,L)z=
        \begin{pmatrix}
            x\\
            y
        \end{pmatrix}^{\top}
       \widehat{\mathcal{M}}(r,L)
        \begin{pmatrix}
            x\\
            y
        \end{pmatrix},\qquad 
         \widehat{\mathcal{M}}(r,L)\coloneqq 
         \begin{pmatrix}
            \mathcal{S}(r,L)&-\mathcal{A}(r,L)\\
            \mathcal{A}(r,L)&\mathcal{S}(r,L)
        \end{pmatrix}
    \end{align}
    
    Therefore,
    \begin{align}
        \mathcal{M}(r,L)\geq0 \quad\Longleftrightarrow\quad \widehat{\mathcal{M}}(r,L)\geq0.
    \end{align}
    Since $\widehat{\mathcal{M}}(r,L)$ is a real symmetric matrix affine in the real variables $(r,L)$, while the objective and the constraint $LK=rK'$ are linear, the optimization problem is a finite-dimensional semidefinite program in the standard real form~\cite{boyd_ConvexOptimization_2004a}.
\end{proof}

\begin{remark}[Continuous projective representations]
The result extends to continuous projective unitary representations $U$ and $U'$ by the lifting construction in Appendix~\ref{sec:extension_proj_rep}.
Let $d\coloneqq\dim\mathcal{H}$ and $d'\coloneqq\dim\mathcal{H}'$, and define
\begin{align}
    \tilde{U}_g\coloneqq \frac{U_g^{\otimes d}}{\det U_g},\qquad\tilde{U}'_g\coloneqq \frac{(U'_g)^{\otimes d'}}{\det U'_g}.
\end{align}
These are continuous nonprojective unitary representations, and the lifted states have the same symmetry subgroups as the original states. Under the inclusion $\Sym_{G,U}(\rho)\subseteq\Sym_{G,U'}(\rho')$, we therefore obtain
\begin{align}
     R_{U,U'}(\rho\to\rho')=\frac{d'}{d}\,\tilde{r}_{\mathrm{SDP}},
\end{align}
where $\tilde{r}_{\mathrm{SDP}}$ is the value of the above semidefinite program evaluated for the states $\rho^{\otimes d}$ and $\rho'^{\otimes d'}$ with representations $\tilde{U}$ and $\tilde{U}'$, respectively.
\end{remark}

\subsection{Derivation of Eq.~\eqref{eq:Fq_SLD_inequality}}
We prove
\begin{align}
    \Re \mathcal{Q}_\rho^q\leq \mathcal{Q}_\rho^{1/2},\qquad \forall q\in[0,1],\label{eq:QGT_q_QGT_SLD}
\end{align}
from which Eq.~\eqref{eq:Fq_SLD_inequality} follows for $U(1)$. 

For a $p$-parameter family of states $\rho_{\theta}$ with $\theta\in\mathbb{R}^p$, the associated $f_q$-QFI $\mathcal{F}_{\rho}^{f_q}$ at $\rho\coloneqq\rho_{\theta_0}$ is given by a $p\times p$ matrix
\begin{align}
    \left(\mathcal{F}_{\rho}^{f_q}\right)_{ij}\coloneqq \sum_{\substack{k,l=1;\\(1-q)\mu_l+q\mu_k>0} }^d\frac{\braket{l|\partial_i\rho|k}\braket{k|\partial_j\rho|l}}{(1-q)\mu_l+q\mu_k}
\end{align}
for $i,j=1,\ldots,p$, where $\partial_i\rho\coloneqq \partial_{\theta_i}\rho_{\theta}|_{\theta=\theta_0}$ and $\rho=\sum_{i=1}^d \mu_i\ket{i}\bra{i}$ denotes the eigenvalue decomposition of $\rho$. For a \textit{real} vector $u\in\mathbb{R}^p$, we have
\begin{align}
    u^\top \mathcal{F}_{\rho}^{f_q} u&= \sum_{\substack{k,l=1;\\(1-q)\mu_l+q\mu_k>0} }^d\frac{\braket{l|\partial_u\rho|k}\braket{k|\partial_u\rho|l}}{(1-q)\mu_l+q\mu_k},
\end{align}
where $\partial_u\rho \coloneqq \sum_{i=1}^pu_i\partial_i \rho$. Since $q\in(0,1)$, we have
\begin{align}
    (1-q)\mu_l+q\mu_k>0\Longleftrightarrow\mu_k+\mu_l>0\Longleftrightarrow (1-q)\mu_k+q\mu_l>0.
\end{align}
Moreover, since $\partial_u\rho$ is Hermitian, $\braket{l|\partial_u\rho|k}\braket{k|\partial_u\rho|l}=|\braket{k|\partial_u\rho|l}|^2$. Therefore, 
\begin{align}
    u^\top \mathcal{F}_{\rho}^{f_q} u
    &=\frac{1}{2} \sum_{\substack{k,l=1;\\(1-q)\mu_l+q\mu_k>0} }^d\left(\frac{1}{(1-q)\mu_l+q\mu_k}+\frac{1}{(1-q)\mu_k+q\mu_l}\right)|\braket{k|\partial_u\rho|l}|^2.
\end{align}
Since
\begin{align}
    \frac{1}{(1-q)\mu_l+q\mu_k}+\frac{1}{(1-q)\mu_k+q\mu_l}=\frac{\mu_k+\mu_l}{\mu_k\mu_l+q(1-q)(\mu_k-\mu_l)^2}
\end{align}
and
\begin{align}
    \mu_k\mu_l+q(1-q)(\mu_k-\mu_l)^2-q(1-q)(\mu_k+\mu_l)^2=\mu_k\mu_l-4q(1-q)\mu_k\mu_l=(2q-1)^2\mu_k\mu_l\geq 0,
\end{align}
we get
\begin{align}
    \frac{1}{(1-q)\mu_l+q\mu_k}+\frac{1}{(1-q)\mu_k+q\mu_l}\leq \frac{1}{q(1-q)(\mu_k+\mu_l)}.
\end{align}
Consequently,
\begin{align}
    4q(1-q)u^\top \mathcal{F}_{\rho}^{f_q} u\leq  \sum_{\substack{k,l=1;\\(1-q)\mu_l+q\mu_k>0} }^d\frac{2}{\mu_k+\mu_l}|\braket{k|\partial_u\rho|l}|^2=u^\top  \mathcal{F}_{\rho}^{f_{1/2}}u,
\end{align}
Since $\mathcal{F}_{\rho}^{f_q}$ is a Hermitian matrix, $u^\top \mathcal{F}_{\rho}^{f_q}u=u^\top \Re\mathcal{F}_{\rho}^{f_q} u$. Therefore, $4q(1-q) \Re \mathcal{F}_{\rho}^{f_q}\leq \Re \mathcal{F}_{\rho}^{f_{1/2}}$. 
Since $\Re \mathcal{F}_{\rho}^{f_{1/2}}=\mathcal{F}_{\rho}^{f_{1/2}}$ for the SLD QFI, we obtain
\begin{align}
     4q(1-q) \Re \mathcal{F}_{\rho}^{f_q}\leq  \mathcal{F}_{\rho}^{f_{1/2}},
\end{align}
or equivalently,
\begin{align}
    \Re \mathcal{Q}_\rho^q\leq \mathcal{Q}_\rho^{1/2}
\end{align}
for $q\in(0,1)$. Continuously extending to endpoints $q=0,1$, we obtain Eq.~\eqref{eq:QGT_q_QGT_SLD}. 

For $U(1)$, the tensors are scalars, and hence $\Re \mathcal{Q}_\rho^q=\mathcal{Q}_\rho^q$, proving Eq.~\eqref{eq:Fq_SLD_inequality}.

\clearpage
\hypertarget{Part2}{}
\begin{center}    
\textbf{{\large \underline{Part II: Convertibility between Gaussian shift models}}} 
\end{center}

This part proves a necessary and sufficient condition for convertibility between Gaussian shift models, expressed by a one-parameter family of quantum Fisher information.

\AppendixTOCTwo{Contents of Part II}

% \section{Convertibility between Gaussian shift models}

\section{Covariance matrix and symplectic form}
Let us first review the basics of quantum continuous variable systems (for more details, see, e.g., Ref.~\cite{weedbrook_GaussianQuantumInformation_2012,serafiniQuantumContinuousVariables2017}). 
Let $r\coloneqq (x_1,p_1,\ldots,x_m,p_m)^\top$ be the vector of canonical operators for an $m$-mode bosonic system, which satisfies the canonical commutation relation
\begin{align}
    [r_i,r_j]=\ii\Omega_{ij},
\end{align}
where $\Omega$ is a $2m\times 2m$ real antisymmetric matrix called the symplectic form, given by
\begin{align}
    \Omega\coloneqq \bigoplus_{i=1}^m
    \begin{pmatrix}
        0&1\\
        -1&0
    \end{pmatrix}.
\end{align}

Let $\Phi$ be a Gaussian state of a system consisting of $m\in\mathbb{Z}_{\geq 1}$ bosonic modes. A Gaussian state is fully characterized by its first and centered second moments, represented by the displacement vector $d\in\mathbb{R}^{2m}$ and the covariance matrix $\sigma\in\mathbb{R}^{2m\times 2m}$, respectively, whose elements are defined by
\begin{align}
    d_i\coloneqq \Tr(\Phi r_i),\qquad \sigma_{ij}\coloneqq \Tr(\Phi \{r_i-d_i,r_j-d_j\})=\Tr\left(\Phi((r_i-d_i)(r_j-d_j)+(r_j-d_j)(r_i-d_i))\right).
\end{align}
For a fixed Gaussian state $\Phi$, one can shift quadratures by
\begin{align}
    r\mapsto \tilde{r}\coloneqq r-d,\label{eq:shift_quadrature}
\end{align}
so that the first moments $\Tr(\Phi \tilde{r}_i)$ vanish. With the above convention, the ordered second moments are
\begin{align}
    \Tr(\Phi\tilde{r}_i\tilde{r}_j)=\frac{1}{2}(\sigma_{ij}+\ii\Omega_{ij}),\qquad \Tr(\Phi\tilde{r}_j\tilde{r}_i)=\frac{1}{2}(\sigma_{ij}-\ii\Omega_{ij}).
\end{align}

A real symmetric matrix $\sigma$ is a valid covariance matrix of a Gaussian state if and only if it satisfies the uncertainty condition~\cite{simon_QuantumnoiseMatrixMultimodesystemsinvariance_1994, weedbrook_GaussianQuantumInformation_2012,serafiniQuantumContinuousVariables2017}:
\begin{align}
    \sigma +\ii\Omega\geq 0.\label{eq:uncertainty_plus}
\end{align}
We call such a real symmetric matrix physical. 
Taking the complex conjugate, it also holds
\begin{align}
    \sigma -\ii\Omega\geq 0.\label{eq:uncertainty_minus}
\end{align}
The uncertainty condition also implies $\sigma>0$. Indeed, if a real vector satisfied $v^\top \sigma v=0$, then $v^\dag (\sigma+\ii \Omega )v=0$. Then, positivity of $\sigma+\ii \Omega$ implies $(\sigma+\ii\Omega)v=0$, whose imaginary part gives $\Omega v=0$. This leads to $v=0$ because the symplectic form is non-singular,  

A quantum channel that maps Gaussian states to Gaussian states is called a Gaussian channel~\cite{giedke_CharacterizationGaussianOperationsdistillationGaussian_2002}. A Gaussian channel $\mathcal{E}_{\mathrm{G}}$ from $m$ input modes to $m'$ output modes is completely characterized by $X\in\mathbb{R}^{2m'\times 2m}$, $Y=Y^\top\in\mathbb{R}^{2m'\times 2m'}$, and $l\in\mathbb{R}^{2m'}$, satisfying
\begin{align}
    Y+\ii (\Omega' - X \Omega X^\top)\geq 0,
\end{align}
where $\Omega$ and $\Omega'$ denote $2m\times 2m$ and $2m'\times2m'$ symplectic forms, respectively~\cite{holevo_EvaluatingCapacitiesBosonicGaussianchannels_2001,weedbrook_GaussianQuantumInformation_2012}. 
Under this channel, the displacement vector and the covariance matrix transform as
\begin{align}
    d\mapsto Xd+l,\qquad \sigma\mapsto X\sigma X^\top+Y.\label{eq:dynamics_gaussian_channel}
\end{align}
In this case, we write $\mathcal{E}_{\mathrm{G}}=(X,Y,l)$ for short. 

\section{Gaussian shift model and quantum Fisher information matrices}
A Gaussian shift model is a parametric family of Gaussian states generated by displacement operators. Let $\Phi$ be an $m$-mode Gaussian state that serves as a reference. By shifting the quadratures as in Eq.~\eqref{eq:shift_quadrature}, we can assume that the expectation values of canonical variables $r$ vanish for $\Phi$. 
A displacement operator is defined for a complex amplitude $z=(z_1,\ldots,z_m)\in\mathbb{C}^m$ by
\begin{align}
    D(z)\coloneqq \exp\left(\sum_{j=1}^m\left(z_ja_j^\dag-\overline{z_j}a_j\right)\right)=\bigotimes_{j=1}^m\left(\exp\left(z_ja_j^\dag-\overline{z_j}a_j\right)\right),
\end{align}
where creation and annihilation operators are related to canonical operators by
\begin{align}
    a_j^\dag \coloneqq \frac{1}{\sqrt{2}}(x_j-\ii p_j),\qquad  a_j\coloneqq \frac{1}{\sqrt{2}}(x_j+\ii p_j).
\end{align}
When a state $\Phi$ evolves into $D(z)\Phi D(z)^\dag$, the first moments shift by
\begin{align}
    k\coloneqq \sqrt{2}
    \begin{pmatrix}
        \Re z_1&
        \Im z_1&
        \Re z_2&
        \Im z_2&
        \cdots&
        \Re z_m&
        \Im z_m
    \end{pmatrix}^\top,
    % k\coloneqq \sqrt{2}
    % \begin{pmatrix}
    %     \Re z_1\\
    %     \Im z_1\\
    %     \Re z_2\\
    %     \Im z_2\\
    %     \vdots\\
    %     \Re z_m\\
    %     \Im z_m\\
    % \end{pmatrix},
    \label{eq:displacement_complex_amplitude}
\end{align}
while the covariance matrix remains unchanged. Note that the displacement operator can also be written as
\begin{align}
    D(z)=\exp\left(-\ii k^\top \Omega r\right).\label{eq:displacement_real_shift}
\end{align}

For a given $m$-mode Gaussian state $\Phi$, we denote the displaced state by
\begin{align}
    \Phi_z\coloneqq D(z)\Phi D(z)^\dag,\qquad z\in\mathbb{C}^m.
\end{align}
For an integer $p\geq 1$, and a matrix $C\in\mathbb{C}^{m\times p}$, we denote
\begin{align}
    \mathcal{G}(\Phi,C)\coloneqq \{ \Phi_{Cu}\}_{u\in \mathbb{R}^p}.
\end{align}
Here, we remark that $\Phi$ and $C$ are fixed parameters specifying the model, while $u\in \mathbb{R}^p$ is regarded as an unknown parameter. 
The matrix $C$ allows for a real-linear dependence of the complex displacement amplitude on the parameter $u$. 
For example, if $p=2m$ and 
\begin{align}
    C=C_{\mathrm{bj}}^{(m)}\coloneqq 
    \begin{pmatrix}
        1&\ii&0&0&\cdots &0&0\\
        0&0&1&\ii&\cdots&0&0\\
        0&0&0&0&\cdots&1&\ii
    \end{pmatrix},\label{eq:definition_bijective_C}
\end{align}
then the map $u\mapsto Cu$ is a bijection from $\mathbb{R}^{2m}$ to $\mathbb{C}^m$. Under this parameterization, the shift is expressed as $k=Ku$, where $K\in\mathbb{R}^{2m\times p}$ is defined by
\begin{align}
    K\coloneqq \sqrt{2}
    \begin{pmatrix}
        \Re(C_{11})&\Re(C_{12})&\cdots&\Re(C_{1p})\\
        \Im(C_{11})&\Im(C_{12})&\cdots&\Im(C_{1p})\\
        \Re(C_{21})&\Re(C_{22})&\cdots&\Re(C_{2p})\\
        \Im(C_{21})&\Im(C_{22})&\cdots&\Im(C_{2p})\\
        \vdots&\vdots & \vdots & \vdots\\
        \Re(C_{m1})&\Re(C_{m2})&\cdots&\Re(C_{mp})\\
        \Im(C_{m1})&\Im(C_{m2})&\cdots&\Im(C_{mp})\\
    \end{pmatrix}.\label{eq:definition_K_from_C}
\end{align}

Let us derive the quantum Fisher information matrix. Fix $q\in(0,1)$ and set $t\coloneqq 2q-1\in (-1,1)$. For $u\in\mathbb{R}^p$, we write $\rho_u\coloneqq \Phi_{Cu}$ and $X_i(u)\coloneqq \partial_{u_i}\rho_u$. We seek operators $\{L_i(u)\}_{i=1}^p$ satisfying
\begin{align}
    X_i(u)=q \rho_u L_i(u) +(1-q)L_i(u) \rho_u.\label{eq:X_iqphi}
\end{align}
For a trace-class operator $T$, we define its Fourier--Weyl transform by~\cite{holevoProbabilisticStatisticalAspects2011,ferraro_GaussianStatesQuantumInformation_2005} 
\begin{align}
    \chi_T(\xi)\coloneqq \Tr(T W(\xi)),\qquad W(\xi)\coloneqq e^{\ii \xi^\top r},\qquad \xi\in\mathbb{R}^{2m}.
\end{align}
For density operators this is the standard characteristic function.
The Fourier--Weyl transform is injective on the trace class~\cite[Corollary~5.3.5]{holevoProbabilisticStatisticalAspects2011}. Hence, if two trace-class operators have the same Fourier--Weyl transform for all $\xi\in\mathbb{R}^{2m}$, then they are equal.

We show that both $\rho_u O$ and $O\rho_u$ are trace-class for any operator $O$ given by a linear combination of quadratures. Since any Gaussian state has finite first and second moments, 
\begin{align}
    \|\rho_u^{1/2}O\|_\HS^2=\Tr (\rho_u OO^\dag)<\infty,\qquad\|O\rho_u^{1/2}\|_\HS^2 =\Tr(\rho_u O^\dag O)<\infty.
\end{align}
Therefore, from H\"older inequality, we have
\begin{align}
    \|\rho_uO\|_1\leq  \|\rho_u^{1/2}\|_\HS\|\rho_u^{1/2}O\|_\HS<\infty,\qquad \|O\rho_u\|_1\leq \|\rho_u^{1/2}\|_\HS\|O\rho_u^{1/2}\|_\HS<\infty,
\end{align}
where we used $\|\rho_u^{1/2}\|_\HS^2=\Tr(\rho_u)=1$. Therefore, both $\rho_u O$ and $O\rho_u$ are trace-class. Since
\begin{align}
    X_i(u)=-\ii [(K^\top \Omega r)_i,\rho_u],
\end{align}
the operator $X_i(u)$ is trace-class. Using a linear ansatz 
\begin{align}
    L_i(u)=\ell^{(i)\top}(r-Ku),\qquad \ell^{(i)}\in\mathbb{C}^{2m},
\end{align}
the operator $q \rho_u L_i(u) +(1-q)L_i(u) \rho_u$ is trace-class. Therefore, it suffices to find $\ell^{(i)}$ satisfying
\begin{align}
    \chi_{X_i(u)}(\xi)=\chi_{q \rho_u L_i(u) +(1-q)L_i(u) \rho_u}(\xi),\qquad  \xi\in\mathbb{R}^{2m}.
\end{align}

For the Gaussian state $\Phi_{Cu}$, we have
\begin{align}
    \chi_{\rho_u}(\xi)=\exp\left(\ii\xi^\top K u-\frac{1}{4}\xi^\top \sigma \xi\right).
\end{align}
Differentiating with respect to $u_i$, we get
\begin{align}
   \chi_{X_i(u)}(\xi)=\Tr(\partial_{u_i}\rho_u e^{\ii \xi^\top r})= \partial_{u_i} \chi_{\rho_u}(\xi)=\ii \sum_j\xi_jK_{ji}\chi_{\rho_u}(\xi).\label{eq:char_diff}
\end{align}

For $W(\xi)\coloneqq e^{\ii\xi^\top r}$, the canonical commutation relation implies $W(\xi)W(\eta)=e^{-\frac{\ii}{2}\xi^\top \Omega\eta}W(\xi+\eta)$.
For $\eta=he_j$, where $e_j$ is the $j$th unit basis vector, we have
\begin{align}
    W(\xi+he_j)=e^{-\frac{\ii h}{2}( \Omega\xi)_j}W(\xi)e^{\ii hr_j},\qquad W(\xi+he_j)=e^{\frac{\ii h}{2}( \Omega\xi)_j}e^{\ii hr_j}W(\xi).
\end{align}
Taking derivative at $h=0$, we obtain
\begin{align}
    \partial_{\xi_j}W(\xi)=-\frac{\ii}{2}(\Omega\xi)_jW(\xi)+\ii W(\xi)r_j,\qquad \partial_{\xi_j}W(\xi)=\frac{\ii}{2}(\Omega\xi)_jW(\xi)+\ii r_jW(\xi).
\end{align}
Consequently, using the linear ansatz $L_i(u)=\ell^{(i)\top}(r-Ku)$, we have
\begin{align}
    &\Tr\left((q\rho_uL_i(u)+(1-q)L_i(u)\rho_u)W(\xi)\right)\\
    &=q\Tr(\rho_uL_i(u)W(\xi))+(1-q)\Tr (\rho_uW(\xi)L_i(u))\\
    &=\sum_j\ell_j^{(i)}\left(q\Tr(\rho_u(r-Ku)_jW(\xi))+(1-q)\Tr(\rho_uW(\xi)(r-Ku)_j)\right)
\end{align}
Now, since
\begin{align}
    \Tr(\rho_u(r-Ku)_jW(\xi))&=\Tr\left(\rho_u\left(-\ii\partial_{\xi_j}W(\xi)-\frac{1}{2}(\Omega\xi)_jW(\xi)-(Ku)_jW(\xi)\right)\right)\\
    &=-\ii\partial_{\xi_j}\chi_{\rho_u}(\xi)-\frac{1}{2}(\Omega\xi)_j\chi_{\rho_u}(\xi)-(Ku)_j\chi_{\rho_u}(\xi)\\
    &=\frac{\ii}{2}\left((\sigma+\ii\Omega)\xi\right)_j\chi_{\rho_u}(\xi),
\end{align}
and similarly
\begin{align}
    \Tr(\rho_uW(\xi)(r-Ku)_j)=\frac{\ii}{2}\left((\sigma-\ii\Omega)\xi\right)_j\chi_{\rho_u}(\xi),
\end{align}
we get
\begin{align}
    \Tr\left((q\rho_uL_i(u)+(1-q)L_i(u)\rho_u)W(\xi)\right)&=\frac{\ii}{2}\ell^{(i)\top}(\sigma+\ii t\Omega )\xi\chi_{\rho_u}(\xi)=\frac{\ii}{2}\sum_j\xi_j\left((\sigma-\ii t\Omega )\ell^{(i)}\right)_j\chi_{\rho_u}(\xi)\label{eq:chi_qq}
\end{align}
The matrix $\sigma-\ii t\Omega$ is positive definite $|t|<1$. Indeed, for $t\in[0,1)$, 
\begin{align}
    \sigma-\ii t \Omega=(1-t)\sigma +t(\sigma-\ii\Omega )\overset{\text{Eq.~\eqref{eq:uncertainty_minus}}}{\geq} (1-t)\sigma>0,
\end{align}
and for $t\in(-1,0)$, 
\begin{align}
    \sigma-\ii t \Omega=(1+t)\sigma+(-t)(\sigma+\ii\Omega)\overset{\text{Eq.~\eqref{eq:uncertainty_plus}}}{\geq}(1+t)\sigma>0.
\end{align}
Therefore, $\sigma-\ii t\Omega$ is invertible.
Comparing this with Eqs.~\eqref{eq:char_diff} and~\eqref{eq:chi_qq}, we get
\begin{align}
    \ell^{(i)}=2(\sigma-\ii t \Omega)^{-1}K_i,\qquad K_i\coloneqq 
    \begin{pmatrix}
        K_{1,i}\\
        K_{2,i}\\
        \vdots\\
        K_{2m,i}
    \end{pmatrix}\in\mathbb{R}^{2m}.
\end{align}
With this choice of $\ell^{(i)}$, Eqs.~\eqref{eq:char_diff}
and~\eqref{eq:chi_qq} show that the two trace-class operators on the two sides of Eq.~\eqref{eq:X_iqphi} have identical Fourier--Weyl transforms for all $\xi\in\mathbb{R}^{2m}$. By the injectivity of the Fourier--Weyl transform on the trace class, they are equal. Hence the above $L_i(u)$ indeed satisfies Eq.~\eqref{eq:X_iqphi}.

Since $X_i(u)=D(Cu)X_i(0)D(Cu)^\dagger$ and $L_i(u)=D(Cu)L_i(0)D(Cu)^\dagger$, the Fisher information matrix is independent of $u$. Therefore, writing $\rho_0=\Phi$ and $L_i=L_i(0)$ and evaluating it at $u=0$, we obtain 
\begin{align}
    \left(\mathcal{F}^{f_q}_{\Phi,C}\right)_{ij}&\coloneqq \Tr(X_i(u)L_j(u))=q \Tr( \Phi  L_i L_j)+(1-q)\Tr (L_i\Phi L_j)\label{eq:QFI_Fock}\\
    &=\sum_{s,s'=1}^{2m}\ell^{(i)}_{s}\ell^{(j)}_{s'}(q\Tr( \Phi  r_s r_{s'})+(1-q)\Tr( \Phi  r_{s'} r_s))\\
    &=2\left(K^\top \left(\sigma-\ii t\Omega\right)^{-1}K\right)_{ij},
\end{align}
i.e.,
\begin{align}
    \mathcal{F}^{f_q}_{\Phi,C}=2K^\top \left(\sigma-\ii t\Omega\right)^{-1}K,\qquad t\coloneqq (2q-1).\label{eq:QFI_formula_Fock}
\end{align}

\section{Necessary and sufficient condition for the convertibility between Gaussian shift models}
In this section, we prove the following theorem. 
\begin{theorem}[Convertibility between Gaussian shift model]\label{thm:convertibility_gaussian_shift}
    Let $\Phi$ and $\Phi'$ be $m$- and $m'$-mode Gaussian states with covariance matrices $\sigma$ and $\sigma'$, and symplectic forms $\Omega$ and $\Omega'$, respectively. Using two complex matrices $C\in\mathbb{C}^{m\times p}$ and $C'\in\mathbb{C}^{m'\times p}$, define the two $p$-parameter Gaussian shift models $\mathcal{G}(\Phi,C)\coloneqq \{\Phi_{Cu}\}_{u\in\mathbb{R}^p}$ and $\mathcal{G}(\Phi',C')\coloneqq \{\Phi'_{C'u}\}_{u\in\mathbb{R}^p}$. The following conditions are equivalent:
    \begin{enumerate}[(i)]
        \item The Gaussian shift model $\mathcal{G}(\Phi,C)$ is convertible to $\mathcal{G}(\Phi',C')$ without error. That is, there exists a quantum channel $\mathcal{E}$, independent of $u\in\mathbb{R}^p$, such that $\mathcal{E}(\Phi_{Cu})=\Phi'_{C'u}$ for all $u\in\mathbb{R}^p$.
        \item For all $q\in(0,1)$, 
            \begin{align}
                \mathcal{F}^{f_q}_{\Phi,C}\geq \mathcal{F}^{f_q}_{\Phi',C'}.
            \end{align}
            Equivalently, $\forall t\in(-1,1)$, 
            \begin{align}
                K^\top \left(\sigma-\ii t\Omega\right)^{-1}K\geq K^{\prime\top} \left(\sigma'-\ii t\Omega'\right)^{-1}K',
            \end{align}
            where $K\in\mathbb{R}^{2m\times p}$ and $K'\in\mathbb{R}^{2m'\times p}$ are defined from $C$ and $C'$ via Eq.~\eqref{eq:definition_K_from_C}. 
        \item There exists a real matrix $X\in \mathbb{R}^{2m'\times 2m}$ such that 
        \begin{align}
            XK=K',\qquad \sigma'+\ii\Omega'\geq X(\sigma +\ii \Omega)X^\top.
        \end{align}
        \item There exists a Gaussian channel $\mathcal{E}_{\mathrm{G}}$, independent of $u\in\mathbb{R}^p$, such that $\mathcal{E}_{\mathrm{G}}(\Phi_{Cu})=\Phi'_{C'u}$ for all $u\in\mathbb{R}^p$.
    \end{enumerate}
\end{theorem}

\begin{proof}~\par
    (i)$\implies$(ii): This is the monotonicity of the quantum Fisher information. This part is nontrivial because the monotonicity is often proved only for finite-dimensional systems. The detailed proof is provided in Section~\ref{sec:monotonicity_QFI_Gaussian}.

    (ii)$\implies$(iii): This is the most intricate part of the proof, which is provided in Section~\ref{sec:realization_Gaussian}.

    (iii)$\implies$(iv): We use centered quadratures, so that  $\Tr(\Phi r)=0$ and $\Tr(\Phi' r')=0$. Define $Y\coloneqq\sigma'- X\sigma X^\top\in \mathbb{R}^{2m'\times 2m'}$, which is real and symmetric. This matrix satisfies
    \begin{align}
        Y+\ii (\Omega' - X \Omega X^\top)=(\sigma'+\ii\Omega')-X(\sigma +\ii \Omega)X^\top \geq 0.
    \end{align}
    Thus, $X\in\mathbb{R}^{2m'\times 2m}$ and $Y\in \mathbb{R}^{2m'\times 2m'}$ define a Gaussian channel, which we denote by $\mathcal{E}_{\mathrm{G}}$. By Eq.~\eqref{eq:dynamics_gaussian_channel}, for each $u\in\mathbb{R}^p$, the output $\mathcal{E}_{\mathrm{G}}(\Phi_{Cu})$ is a Gaussian state, which has first moment $XKu=K'u$ and covariance matrix $X\sigma X^\top +Y=\sigma'$. Therefore, $\mathcal{E}_{\mathrm{G}}(\Phi_{Cu})=\Phi'_{C'u}$. 

    (iv)$\implies$(i): This is immediate because any Gaussian channel is a quantum channel. 
\end{proof}

Before proceeding to subsections that provide the detailed proof, we show its extension from states to channels as an immediate corollary of Theorem~\ref{thm:convertibility_gaussian_shift}.

\begin{corollary}[Gaussian post-processing]
    Let $\Lambda_{\mathrm{G}}$ be an $m$-mode input $s$-mode output Gaussian channel and $\Lambda_{\mathrm{G}}'$ be an $m$-mode input $s'$-mode output Gaussian channel. The following conditions are equivalent:
    \begin{enumerate}[(I)]
        \item There exists a quantum channel $\mathcal{E}$ such that $\Lambda_{\mathrm{G}}'=\mathcal{E}\circ \Lambda_{\mathrm{G}}$. 
        % \item For all $q\in(0,1)$, 
        %     \begin{align}
        %         \mathfrak{F}^{f_q}_{\Lambda_{\mathrm{G}}}\geq \mathfrak{F}^{f_q}_{\Lambda_{\mathrm{G}}'}.
        %     \end{align}
        %     Equivalently, $\forall t\in(-1,1)$, 
        %     \begin{align}
        %         X^\top ((X X^\top+Y)+\ii t\Omega)^{-1}X\geq X^{\prime\top} ((X' X^{\prime \top}+Y')+\ii t\Omega')^{-1}X'.
        %     \end{align}
        % \item There exists a real matrix $L\in \mathbb{R}^{2m'\times 2m}$ such that 
        % \begin{align}
        %     LK=K',\qquad Y'-LYL^\top +\ii\Omega'-\ii L\Omega L^\top\geq 0.
        % \end{align}
        \item There exists a Gaussian channel $\mathcal{E}_{\mathrm{G}}$ such that $\Lambda_{\mathrm{G}}'=\mathcal{E}_{\mathrm{G}}\circ \Lambda_{\mathrm{G}}$.
    \end{enumerate}
\end{corollary}
\begin{proof}
    The implication (II)$\implies$(I) is immediate since any Gaussian channel is a quantum channel. We prove (I)$\implies$(II). Write the two Gaussian channels as $\Lambda_{\mathrm{G}}=(X,Y,l)$ and $\Lambda_{\mathrm{G}}'=(X',Y',l')$. 

    Let $\Phi_0$ be the $m$-mode vacuum state, whose covariance matrix is $I_{2m}$. Let $C_0\coloneqq C_{\mathrm{bj}}^{(m)}$ be the matrix defined in Eq.~\eqref{eq:definition_bijective_C} and consider the full coherent-state shift model
    \begin{align}
        \mathcal{G}_0\coloneqq \{(\Phi_0)_{C_0u}\colon u\in\mathbb{R}^{2m}\}.
    \end{align}
    By Eq.~\eqref{eq:definition_K_from_C}, its real displacement matrix is $K_0=\sqrt{2}I_{2m}$, which is invertible. 

    Applying $\Lambda_{\mathrm{G}}$ and $\Lambda'_{\mathrm{G}}$, we obtain two Gaussian shift models
    \begin{align}
        \mathcal{G}_1\coloneqq \{\Lambda_{\mathrm{G}}((\Phi_0)_{C_0u})\colon u\in\mathbb{R}^{2m}\},\qquad  \mathcal{G}_2\coloneqq \{\Lambda'_{\mathrm{G}}((\Phi_0)_{C_0u})\colon u\in\mathbb{R}^{2m}\},
    \end{align}
    and their covariance matrices are given by 
    \begin{align}
        \Sigma\coloneqq XX^\top +Y,\qquad \Sigma'\coloneqq X'X'^\top +Y'.
    \end{align}
    Their first moments at $u=0$ are $l$ and $l'$, respectively. 

    From condition (I), $\mathcal{E} (\Lambda_{\mathrm{G}}((\Phi_0)_{C_0u}))=\Lambda'_{\mathrm{G}}((\Phi_0)_{C_0u})$ for all $u\in \mathbb{R}^{2m}$. By using (i)$\implies$(iii) in Theorem~\ref{thm:convertibility_gaussian_shift}, there exists a real matrix $L\in\mathbb{R}^{2s'\times 2s}$ such that
    \begin{align}
        LXK_0=X'K_0,\qquad \Sigma'+\ii\Omega_{s'}\geq L(\Sigma+\ii\Omega_{s})L^\top.
    \end{align}
    Since $K_0$ is invertible, we have $LX=X'$. Define 
    \begin{align}
        Z\coloneqq Y'-LYL^\top,\qquad a\coloneqq l'-Ll.
    \end{align}
    Then, $Z$ is real and symmetric, and satisfies
    \begin{align}
        Z+\ii (\Omega_{s'}-L\Omega_sL^\top)=(\Sigma'+\ii\Omega_{s'})-L(\Sigma+\ii \Omega_s)L^\top\geq 0.
    \end{align}
    Thus, $\mathcal{E}_{\mathrm{G}}\coloneqq (L,Z,a)$ defines a Gaussian channel. Moreover, the composition of the Gaussian channels yields 
    \begin{align}
        \mathcal{E}_{\mathrm{G}}\circ \Lambda_{\mathrm{G}}=(LX,LYL^\top +Z,Ll+a)=(X',Y',l')=\Lambda'_{\mathrm{G}}.
    \end{align}
\end{proof}

\subsection{Monotonicity of quantum Fisher information for Gaussian shift model (Proof of (i)$\implies$(ii))}\label{sec:monotonicity_QFI_Gaussian}
We fix $q\in(0,1)$ and set $t\coloneqq 2q-1$. 
Let $\mathcal{F}$ be the $m$-mode Fock space. Let $\mathcal{B}(\mathcal{F})$ and $\mathcal{S}_2(\mathcal{F})$ be the set of bounded operators and the set of Hilbert--Schmidt operators on $\mathcal{F}$. The logarithmic derivative
\begin{align}
    L_i =\sum_{j=1}^{2m}r_j\ell^{(i)}_j=2\sum_{j=1}^{2m}r_j\left(\left(\sigma-\ii t\Omega\right)^{-1}K\right)_{ji}
\end{align}
is in general unbounded. However, since
\begin{align}
    \Tr(\Phi L_i^\dag L_i)=\frac{1}{2}\ell^{(i)\dag} (\sigma +\ii \Omega)\ell^{(i)}<\infty,\qquad \Tr(\Phi L_i L_i^\dag)=\frac{1}{2}\ell^{(i)\dag} (\sigma -\ii \Omega)\ell^{(i)}<\infty,
\end{align}
the operators $L_i\Phi^{1/2},\Phi^{1/2}L_i$ are Hilbert--Schmidt.

For $z\in\mathbb{C}^p$, we define
\begin{align}
    X_z\coloneqq \sum_{i=1}^pz_i\partial_{u_i}\Phi_{Cu}|_{u=0},\qquad L_z\coloneqq \sum_{i=1}^pz_iL_i,
\end{align}
which satisfies
\begin{align}
    X_z=q\Phi L_z+(1-q)L_z\Phi. 
\end{align}

For $B\in\mathcal{B}(\mathcal{F})$, we define
\begin{align}
    T_{\Phi}(B)&\coloneqq \left(\sqrt{q}\Phi^{1/2}B,\sqrt{1-q} B\Phi^{1/2}\right)\in\mathcal{S}_2(\mathcal{F})\oplus \mathcal{S}_2(\mathcal{F}).
\end{align}
For any $B\in\mathcal{B}(\mathcal{F})$, 
\begin{align}
    \Tr(B^\dag X_z)=q\Tr(B^\dag \Phi L_z)+(1-q)\Tr(B^\dag L_z \Phi)=\braket{T_{\Phi}(B),T_{\Phi}(L_z)}_{\mathcal{S}_2\oplus\mathcal{S}_2},
\end{align}
where $\braket{\cdot,\cdot}_{\mathcal{S}_2\oplus\mathcal{S}_2}$ denotes the Hilbert--Schmidt inner product on $\mathcal{S}_2(\mathcal{F})\oplus\mathcal{S}_2(\mathcal{F})$. Defining
\begin{align}
     \|B\|_{\Phi}^2&\coloneqq  q \Tr(B^\dag \Phi  B)+(1-q)\Tr(B^\dag B\Phi )=\braket{T_{\Phi}(B),T_{\Phi}(B)}_{\mathcal{S}_2\oplus\mathcal{S}_2},
\end{align}
we obtain
\begin{align}
    2\Re\Tr(B^\dag X_z)- \|B\|_{\Phi }^2= \|L_z\|_{\Phi }^2-\|B-L_z\|_{\Phi }^2\leq \|L_z\|_{\Phi }^2 \overset{\text{Eq.~\eqref{eq:QFI_Fock}}}{=}z^\dag \mathcal{F}^{f_q}_{\Phi,C} z\label{eq:variation_Fock}
\end{align}
for any $B\in\mathcal{B}(\mathcal{F})$. 

Let us now approximate $L_z$ by bounded operators. We decompose
\begin{align}
    L_z=R+\ii S,\qquad R\coloneqq \frac{L_z+L_z^\dag}{2},\qquad S\coloneqq \frac{L_z-L_z^\dag}{2\ii}.
\end{align}
Since $R$ and $S$ are linear combinations of quadratures with finite real coefficients, we have
\begin{align}
    \Tr(\Phi R^2)<\infty,\qquad \Tr(\Phi S^2)<\infty.
\end{align}
Let $E_R$ be projection-valued measure such that $R=\int_{\mathbb{R}}\lambda \dd E_R(\lambda)$. We define the truncated operator 
\begin{align}
    R_N\coloneqq \int_{|\lambda|\leq N}\lambda \dd E_R(\lambda)
\end{align}
for a positive integer $N>0$. Since $R$ and $R_N$ are self-adjoint, we have
\begin{align}
    \|(R-R_N)\|_{\Phi}^2&=q\Tr((R-R_N)^\dag \Phi (R-R_N))+(1-q)\Tr((R-R_N)^\dag (R-R_N) \Phi)\\
    &=\Tr(\Phi(R-R_N)^2)=\int_{|\lambda|>N}\lambda^2 \dd\mu_{\Phi,R}(\lambda),
\end{align}
where $\dd\mu_{\Phi,R}(\lambda)\coloneqq \Tr(\Phi\dd E_R(\lambda))$. Defining $h_N(\lambda)\coloneqq \lambda^2\mathbbm{1}_{\{|\lambda|>N\}}$, for each fixed $\lambda$, 
\begin{align}
    h_N(\lambda)\to0\qquad (N\to\infty).
\end{align}
Since $0\leq h_N(\lambda)\leq \lambda^2$, dominated convergence theorem implies
\begin{align}
    \lim_{N\to\infty} \|(R-R_N)\|_{\Phi}^2=\lim_{N\to\infty}\int_{\mathbb{R}}h_N(\lambda)\dd\mu_{\Phi,R}(\lambda)=\int_{\mathbb{R}}\lim_{N\to\infty}h_N(\lambda)\dd\mu_{\Phi,R}(\lambda)=0.
\end{align}
Similarly, for $S_N\coloneqq \int_{|\lambda|<N}\lambda \dd E_S(\lambda)$, we have $ \lim_{N\to\infty} \|(S-S_N)\|_{\Phi}^2=0$. Therefore, a bounded operator 
\begin{align}
    B_N\coloneqq R_N+\ii S_N
\end{align}
satisfies
\begin{align}
    \|B_N-L_z\|_{\Phi }\leq \|(R-R_N)\|_{\Phi}+\|(S-S_N)\|_{\Phi}\to 0\qquad (N\to\infty).
\end{align}
Thus, Eq.~\eqref{eq:variation_Fock} implies
\begin{align}
    \lim_{N\to\infty}\left( 2\Re\Tr(B_N^\dag X_z)- \|B_N\|_{\Phi }^2\right)=\|L_z\|^2_\Phi
\end{align}
and hence
\begin{align}
    z^\dag \mathcal{F}^{f_q}_{\Phi,C} z= \sup_{B\in\mathcal{B}(\mathcal{F})}\left\{2\Re\Tr(B^\dag X_z)- \|B\|_{\Phi }^2\right\}.
\end{align}

For the dual map $\Lambda\coloneqq \mathcal{E}^*$, the Schwarz inequality for unital CP maps (see e.g., Ref.~\cite{choi_SchwarzInequalityPositivelinearmaps_1974_math} and Proposition~3.3 in Ref.~\cite{paulsen_CompletelyBoundedMapsOperatorAlgebras_2003}) gives
\begin{align}
    \Lambda(Y)^\dag \Lambda(Y)\leq \Lambda (Y^\dag Y),\qquad \Lambda(Y) \Lambda(Y)^\dag\leq \Lambda (YY^\dag).
\end{align}
Therefore, we obtain
\begin{align}
    \|\Lambda(B)\|_{\Phi }^2=q \Tr(\Phi  \Lambda(B) \Lambda(B)^\dag )+(1-q)\Tr(\Phi \Lambda(B)^\dag \Lambda(B))\leq q \Tr(\Phi \Lambda( B B^\dag) )+(1-q)\Tr(\Phi \Lambda(B^\dag B))=\|B\|_{\mathcal{E}(\Phi) }^2.
\end{align}

Let $\mathcal{F}'$ be the output Fock space. Let $L_i'$ be the logarithmic derivative for $\Phi'_{C'u}$, and define $X_z'\coloneqq\sum_{i=1}^pz_i\partial_{u_i}\Phi'_{C'u}|_{u=0} $ and $L_z'\coloneqq \sum_{i=1}^pz_iL_i'$. For a parameter-independent channel $\mathcal{E}$ such that $\mathcal{E}(\Phi_{Cu})=\Phi'_{C'u}$ for all $u\in\mathbb{R}^p$, the linearity of $\mathcal{E}$ gives
\begin{align}
    X_z'\coloneqq\sum_{i=1}^pz_i\partial_{u_i}\Phi'_{C'u}|_{u=0}=\sum_{i=1}^pz_i\partial_{u_i}\mathcal{E}(\Phi_{Cu})|_{u=0}=\sum_{i=1}^pz_i\mathcal{E}(\partial_{u_i}\Phi_{Cu}|_{u=0})=\mathcal{E}(X_z).
\end{align}
Therefore,
\begin{align}
    2\Re\Tr(B^\dag X_z')-\|B\|_{\mathcal{E}(\Phi) }^2&=2 \Re\Tr(B^\dag \mathcal{E}(X_z))-\|B\|_{\mathcal{E}(\Phi) }^2\leq 2\Re\Tr(\Lambda(B)^\dag X_z)-\|\Lambda(B)\|_{\Phi }^2\leq z^\dag \mathcal{F}^{f_q}_{\Phi,C} z.
\end{align}
Taking the supremum over $B\in\mathcal{B}(\mathcal{F}')$, we obtain
\begin{align}
     z^\dag \mathcal{F}^{f_q}_{\Phi,C} z\geq z^\dag \mathcal{F}^{f_q}_{\Phi',C'} z,
\end{align}
i.e.,
\begin{align}
    \forall q\in(0,1),\qquad \mathcal{F}^{f_q}_{\Phi,C}\geq  \mathcal{F}^{f_q}_{\Phi',C'} ,
\end{align}
or equivalently,
\begin{align}
    \forall t \in (-1,1),\qquad K^\top(\sigma-\ii t \Omega)^{-1}K\geq  K^{\prime\top}(\sigma'-\ii t \Omega')^{-1}K'
\end{align}
by Eq.~\eqref{eq:QFI_formula_Fock}.

\subsection{A realization theorem for Gaussian convertibility (Proof of (ii)$\implies$(iii))}\label{sec:realization_Gaussian}

\begin{lemma}\label{lem:ordering_inversion}
    Let $M$ and $N$ be positive-definite Hermitian matrices of sizes $m\times m$ and $n\times n$, respectively. For $Y\in \mathbb{C}^{m\times n}$, the following two are equivalent: (i) $M\geq YN Y^\dag$, (ii) $N^{-1}\geq Y^\dag M^{-1}Y$.
\end{lemma}
\begin{proof}
    Define $Z\coloneqq M^{-1/2}Y N^{1/2}$. The conditions (i) and (ii) are equivalent to $I\geq ZZ^\dag $ and $I\geq Z^\dag Z$, respectively. Both of these conditions are equivalent to the condition that all the singular values of $Z$ are less than or equal to one. 
\end{proof}

\begin{lemma}[Douglas' Lemma]\label{lem:Douglas}
    Let $\mathcal{H},\mathcal{K}_1,\mathcal{K}_2$ be finite-dimensional Hilbert spaces. For linear operators $D_1:\mathcal{H}\to \mathcal{K}_1$ and $D_2:\mathcal{H}\to\mathcal{K}_2$, the following two are equivalent: 
    \begin{enumerate}[(i)]
        \item There exists a contraction $K:\mathcal{K}_1\to\mathcal{K}_2$ such that $D_2=KD_1$. 
        \item $D_1^\dag D_1\geq D_2^\dag D_2$.
    \end{enumerate}
\end{lemma}
\begin{proof}
    (i)$\implies$(ii): $D_2^\dag D_2=D_1^\dag K^\dag K D_1\leq D_1^\dag D_1$.

    (ii)$\implies$(i): Since $\|D_1x\|=0$ implies $\|D_2 x\|=0$, $\ker D_1\subset \ker D_2$. 
    On $\Ran D_1$, define $K_0$ by $K_0(D_1 x)\coloneqq D_2 x$. This is well-defined since if $D_1x=D_1y$, then $x-y\in \ker D_1\subset \ker D_2$, and hence $D_2(x-y)=0$. This map $K_0$ is contraction on $\Ran D_1$ since $\|K_0(D_1x)\|^2=\|D_2x\|^2\leq \|D_1x\|^2$. We then extend $K_0$ on $\Ran D_1$ to $\mathcal{K}_1$ by $K|_{\Ran D_1}\coloneqq K_0$ and $K|_{(\Ran D_1)^\perp}\coloneqq 0$. 
\end{proof}

\begin{lemma}\label{lem:single_inversion}
    Let $H_i>0$ be an $m_i\times m_i$ Hermitian matrix, and $C_i\in\mathbb{C}^{m_i\times p}$ for $i=1,2$. Then the following two are equivalent:
    \begin{enumerate}[(i)]
        \item $\exists Y\in \mathbb{C}^{m_2\times m_1}$ such that $YC_1=C_2$ and $H_2\geq Y H_1Y^\dag$.
        \item $C_1^\dag H_1^{-1} C_1\geq C_2^\dag H_2^{-1} C_2$. 
    \end{enumerate}
\end{lemma}
\begin{proof}
    (i)$\implies$(ii): From Lemma~\ref{lem:ordering_inversion}, $H_2\geq Y H_1 Y^\dag$ implies $H_1^{-1}\geq Y^\dag H_2^{-1} Y$. Therefore, $C_1^\dag H_1^{-1} C_1\geq C_1^\dag ( Y^\dag H_2^{-1} Y)C_1=C_2^\dag H_2^{-1}C_2$. 

    (ii)$\implies$(i): Defining $D_i\coloneqq H_i^{-1/2}C_i$ for $i=1,2$, we have $D_1^\dag D_1\geq D_2^\dag D_2$. By Lemma~\ref{lem:Douglas}, there exists a contraction $K$ such that $D_2=KD_1$. Thus, for $Y\coloneqq H_2^{1/2}KH_1^{-1/2}$, we have
    \begin{align}
        YC_1=H_2^{1/2}KH_1^{-1/2}C_1=H_2^{1/2}KD_1=H_2^{1/2}D_2=C_2
    \end{align}
    and 
    \begin{align}
        YH_1Y^\dag=H_2^{1/2}KK^\dag H_2^{1/2}\leq H_2^{1/2}H_2^{1/2}=H_2.
    \end{align}
\end{proof}

As an extension, we prove an interpolated version of this lemma.
\begin{theorem}\label{thm:existence_complex_matrix}
    For $i=1,2$, let $P_i,Q_i\in\mathbb{C}^{m_i\times m_i}$ be positive-definite Hermitian matrices, and $C_i\in\mathbb{C}^{m_i\times p}$ be a matrix. For
    \begin{align}
        H_i(\lambda)\coloneqq \lambda P_i+(1-\lambda)Q_i,\qquad \lambda\in[0, 1],
    \end{align}
    the following two are equivalent:
    \begin{enumerate}[(i)]
        \item $\exists Y\in\mathbb{C}^{m_2\times m_1}$ such that $YC_1=C_2$, $P_2\geq Y P_1 Y^\dag$ and $Q_2\geq Y Q_1Y^\dag$. 
        \item For all $\lambda\in[0,1]$, $C_1^\dag H_1(\lambda)^{-1} C_1\geq C_2^\dag H_2(\lambda)^{-1}C_2$.
    \end{enumerate}
\end{theorem}
\begin{proof}
    (i)$\implies$(ii): From $P_2\geq Y P_1 Y^\dag$ and $Q_2\geq Y Q_1Y^\dag$, we have $H_2(\lambda)\geq YH_1(\lambda)Y^\dag$ for any $\lambda\in[0,1]$. Then, the condition (ii) follows from Lemma~\ref{lem:single_inversion}. 

    (ii)$\implies$(i): 
    
    Let $\mathcal{H}$ be a real vector space consisting of all $m_2\times m_2$ Hermitian matrices. As a subset of $\mathcal{H}\times \mathcal{H}$, define
    \begin{align}
        \mathcal{K}\coloneqq \{(A,B)\in\mathcal{H}\times \mathcal{H}\colon\exists Y \text{ such that }YC_1=C_2,\quad A\geq YP_1 Y^\dag,\quad B \geq YQ_1Y^\dag\}\subset \mathcal{H}\times \mathcal{H}.
    \end{align}
    Then, the condition (i) is equivalent to $(P_2,Q_2)\in\mathcal{K}$.
    
    Let us now show that $\mathcal{K}$ is convex and closed: 
    \begin{itemize}
        \item 
            To prove convexity, let $(A_j,B_j)\in\mathcal{K}$, $j=1,2$, and choose corresponding witnesses $Y_j$ such that $Y_jC_1=C_2$ and $A_j\geq Y_jP_1 Y_j^\dag,\,B_j \geq Y_jQ_1Y_j^\dag$. Then, for $Y_\theta\coloneqq \theta Y_1+(1-\theta)Y_2$, we have $Y_\theta C_1=C_2$ and 
        \begin{align}
            \theta A_1+(1-\theta)A_2-Y_\theta P_1Y_\theta^\dag&\geq \theta Y_1 P_1 Y_1^\dag+(1-\theta)Y_2 P_1 Y_2^\dag-Y_\theta P_1Y_\theta^\dag=\theta(1-\theta)(Y_1-Y_2)P_1(Y_1-Y_2)^\dag\geq 0,\label{eq:convexity_sandwitch}\\
            \theta B_1+(1-\theta)B_2-Y_\theta Q_1Y_\theta^\dag&\geq \theta Y_1 Q_1 Y_1^\dag+(1-\theta)Y_2 Q_1 Y_2^\dag-Y_\theta Q_1Y_\theta^\dag=\theta(1-\theta)(Y_1-Y_2)Q_1(Y_1-Y_2)^\dag\geq 0,
        \end{align}
        implying that $\theta(A_1,B_1)+(1-\theta)(A_2,B_2)\in \mathcal{K}$, i.e., $\mathcal{K}$ is convex. 
        \item 
            To prove the closedness, let $(A_n,B_n)$ be a sequence such that $(A_n,B_n)\to (A,B)$ in a matrix norm. For each $n$, let $Y_n$ be a matrix such that $Y_nC_1=C_2$ and $A_n\geq Y_nP_1 Y_n^\dag,\,B_n \geq Y_nQ_1Y_n^\dag$. Since $P_1$ is positive-definite, one can take $p>0$ such that $P_1\geq pI$. Then, $A_n\geq Y_nP_1Y_n^\dag \geq pY_nY_n^\dag$. Since $(A_n)$ is bounded, it follows that $(Y_n)$ is bounded as well. Therefore, since the space of matrices is finite-dimensional, $(Y_n)$ has a convergent subsequence $(Y_{n_k})$ converging to some matrix $Y$. Along the same subsequence, $A_{n_k}\to A$ and $B_{n_k}\to B$. Taking $k\to\infty$, we obtain
        \begin{align}
            YC_1=C_2,\qquad A\geq YP_1 Y^\dag,\qquad B \geq YQ_1Y^\dag
        \end{align}
        Therefore, $(A,B)\in\mathcal{K}$, implying that $\mathcal{K}$ is closed.
    \end{itemize}
    Moreover, $\mathcal{K}\neq \emptyset$. Indeed, from condition~(ii) at $\lambda=0$, by Lemma~\ref{lem:single_inversion}, there exists a complex matrix $Y_0$ such that $Y_0C_1=C_2$ and $Q_2\geq Y_0Q_1Y_0^\dag$. In particular, the first condition implies $(Y_0P_1Y_0^\dag,Y_0Q_1Y_0^\dag)\in\mathcal{K}$.

    Now, we prove $(P_2,Q_2)\in\mathcal{K}$ by contradiction. Assume that $(P_2,Q_2)\notin\mathcal{K}$. 
    We introduce an inner product $\braket{\cdot,\cdot}$ on real vector space $\mathcal{H}\times \mathcal{H}$ by
    \begin{align}
        \braket{(R,S),(A,B)}\coloneqq \Tr(RA)+\Tr(SB).
    \end{align}
    Since $\mathcal{K}$ is a nonempty closed convex subset of the finite-dimensional real vector space $\mathcal{H}\times\mathcal{H}$ and $(P_2,Q_2)\notin\mathcal{K}$, from the strict separation theorem (see e.g., Corollary 4.1.3 in Ref.~\cite{hiriart-urrutyFundamentalsConvexAnalysis2001}), the assumption that $(P_2,Q_2)\notin\mathcal{K}$ implies that there exist Hermitian matrices $R,S$ such that
    \begin{align}
        \braket{(R,S),(P_2,Q_2)}<\inf_{(A,B)\in\mathcal{K}}\{\braket{(R,S),(A,B)}\},
    \end{align}
    i.e., 
    \begin{align}
        \Tr(RP_2)+\Tr(SQ_2)<\inf_{(A,B)\in\mathcal{K}}\{\Tr(RA)+\Tr(SB)\}.\label{eq:separation}
    \end{align}
    
    These matrices $R$ and $S$ are positive semidefinite. Let us first prove $R\geq 0$ by contradiction. Suppose that $R\not\geq  0$. Then there exists a vector $v$ such that $v^\dag R v<0$. Defining $U\coloneqq vv^\dag\geq 0$, for any $(A,B)\in \mathcal{K}$, $(A+t U,B)\in\mathcal{K}$ for any $t\in\mathbb{R}_{\geq0}$ since $A+tU\geq A$. However, since $\Tr(RU)=v^\dag R v<0$, we have
    \begin{align}
        \lim_{t\to\infty}\braket{(R,S),(A+tU,B)}=\lim_{t\to\infty}\left(\Tr(RA)+\Tr(SB)+t\Tr(RU)\right)=-\infty,
    \end{align}
    which contradicts Eq.~\eqref{eq:separation}. Thus, $R\geq 0$. Repeating the same argument, we also obtain $S\geq 0$. 

    Now, we show
    \begin{align}
        \inf_{(A,B)\in\mathcal{K}}\{\Tr(RA)+\Tr(SB)\}=\inf_{YC_1=C_2}\{\Tr (RYP_1Y^\dag)+\Tr(S Y Q_1 Y^\dag)\}. 
    \end{align}
    Indeed, if $A,B,Y$ satisfy
    \begin{align}
        A\geq YP_1 Y^\dag,\qquad B\geq Y Q_1 Y^\dag,
    \end{align}
    then 
    \begin{align}
        \Tr(RA)+\Tr(SB)\geq \Tr(RYP_1 Y^\dag)+\Tr(SY Q_1 Y^\dag)
    \end{align}
    Therefore, taking infimum over $Y$ satisfying $YC_1=C_2$, we get
    \begin{align}
          \inf_{(A,B)\in\mathcal{K}}\{\Tr(RA)+\Tr(SB)\}\geq \inf_{YC_1=C_2}\{\Tr (RYP_1Y^\dag)+\Tr(S Y Q_1 Y^\dag)\}.
    \end{align}
    Conversely, for any $Y$ satisfying $YC_1=C_2$, the pair $(YP_1Y^\dag,YQ_1Y^\dag)$ belongs to $\mathcal{K}$, which gives the reverse inequality. Therefore, Eq.~\eqref{eq:separation} is equivalent to
    \begin{align}
        \Tr(RP_2)+\Tr(SQ_2)<\inf_{YC_1=C_2}\{\Tr (RYP_1Y^\dag)+\Tr(S Y Q_1 Y^\dag)\}.\label{eq:separation_Y}
    \end{align}

    We now relate the right-hand side of Eq.~\eqref{eq:separation_Y} to condition~(ii). By Lemma~\ref{lem:single_inversion}, condition (ii) implies that, for each $\lambda\in[0,1]$, there exists $Y_\lambda\in\mathbb{C}^{m_2\times m_1}$ such that
    \begin{align}
        Y_\lambda C_1&=C_2,\label{eq:YCC_lambda}\\
        H_2(\lambda)&\geq Y_\lambda H_1(\lambda)Y_\lambda^\dag.\label{eq:YHY_lambda}
    \end{align}    

    Let us first consider the case where $R>0$ and $S>0$. Define
    \begin{align}
        W\coloneqq R+S>0,\qquad E\coloneqq W^{-1/2}R W^{-1/2}.
    \end{align}
    Since $R<R+S$, we have $E< W^{-1/2}(R+S)W^{-1/2}=I$. Let
    \begin{align}
        E=\sum_j \lambda_j \ket{e_j}\bra{e_j},\qquad 0<\lambda_j<1
    \end{align}
    be the eigenvalue decomposition of $E$ with an orthonormal basis $\{\ket{e_j}\}_{j=1}^{m_2}$. Define
    \begin{align}
        \ket{w_j}\coloneqq W^{1/2}\ket{e_j},\qquad \ket{\tilde{w}_j}\coloneqq W^{-1/2}\ket{e_j}.
    \end{align}
    Then,
    \begin{align}
         R&=W^{1/2}\left(W^{-1/2}RW^{-1/2}\right)W^{1/2}=\sum_j\lambda_j\ket{w_j}\bra{w_j},\\
         S&=W-R=W^{1/2}\left(\sum_j  \ket{e_j}\bra{e_j}\right)W^{1/2}-R=\sum_j(1-\lambda_j)\ket{w_j}\bra{w_j},
    \end{align}
    and
    \begin{align}
        \braket{w_j|\tilde{w}_k}=\delta_{jk},\qquad \sum_j\ket{\tilde{w}_j}\bra{w_j}=I,
    \end{align}
    where $\delta_{jk}$ denotes the Kronecker delta. 
    For each $j$, choose $Y_{\lambda_j}$ satisfying
    \begin{align}
        Y_{\lambda_j}C_1=C_2,\qquad H_2(\lambda_j)\geq Y_{\lambda_j}H_1(\lambda_j)Y_{\lambda_j}^\dag,
    \end{align}
    and define
    \begin{align}
        Y\coloneqq \sum_j\ket{\tilde{w}_j}\bra{w_j}Y_{\lambda_j},
    \end{align}
    Then,
    \begin{align}
        YC_1=\sum_j\ket{\tilde{w}_j}\bra{w_j}Y_{\lambda_j}C_1=\sum_j\ket{\tilde{w}_j}\bra{w_j}C_2=C_2.
    \end{align}
    Moreover, since $\bra{w_j}Y=\bra{w_j}Y_{\lambda_j}$ and $Y^\dag\ket{w_j}=Y^\dag_{\lambda_j}\ket{w_j}$, we obtain
    \begin{align}
        &\Tr (RYP_1Y^\dag)+\Tr(S Y Q_1 Y^\dag)\\
        &=\sum_j\lambda_j\braket{w_j|YP_1Y^\dag|w_j} + \sum_j(1-\lambda_j)\braket{w_j|Y Q_1 Y^\dag|w_j}=\sum_j\braket{w_j|Y(\lambda_jP_1+(1-\lambda_j)Q_1)Y^\dag|w_j}\\
        &=\sum_j\braket{w_j|Y_{\lambda_j}H_1(\lambda_j)Y_{\lambda_j}^\dag|w_j}\overset{\text{Eq.~\eqref{eq:YHY_lambda}}}{\leq }\sum_j\braket{w_j|H_2(\lambda_j)|w_j}=\Tr(R P_2)+\Tr(SQ_2).
    \end{align}
    Thus, 
    \begin{align}
        \inf_{YC_1=C_2}\{\Tr (RYP_1Y^\dag)+\Tr(S Y Q_1 Y^\dag)\}\leq \Tr(R P_2)+\Tr(SQ_2),
    \end{align}
    which contradicts Eq.~\eqref{eq:separation_Y}, and hence $(P_2,Q_2)\in\mathcal{K}$. 
    
    When $R\geq0$ and $S\geq 0$, defining $R_{\epsilon}\coloneqq R+\epsilon I>0$ and $S_{\epsilon}\coloneqq S+\epsilon I>0$ for $\epsilon>0$, the above argument yields
    \begin{align}
        \inf_{YC_1=C_2}\{\Tr (R_\epsilon YP_1Y^\dag)+\Tr(S_\epsilon Y Q_1 Y^\dag)\}\leq \Tr(R_\epsilon P_2)+\Tr(S_\epsilon Q_2).
    \end{align}
    Since
    \begin{align}
        \Tr (RYP_1Y^\dag)+\Tr(S Y Q_1 Y^\dag)\leq \Tr (R_\epsilon YP_1Y^\dag)+\Tr(S_\epsilon Y Q_1 Y^\dag),
    \end{align}
    we get
    \begin{align}
        \inf_{YC_1=C_2}\{\Tr (RYP_1Y^\dag)+\Tr(S Y Q_1 Y^\dag)\}\leq \Tr(R_\epsilon P_2)+\Tr(S_\epsilon Q_2).
    \end{align}
    Taking the limit $\epsilon \to 0$, we obtain
    \begin{align}
        \inf_{YC_1=C_2}\{\Tr (RYP_1Y^\dag)+\Tr(S Y Q_1 Y^\dag)\}\leq \Tr(RP_2)+\Tr(S Q_2),
    \end{align}
    which again contradicts Eq.~\eqref{eq:separation_Y}, and hence $(P_2,Q_2)\in\mathcal{K}$. 
\end{proof}

\begin{lemma}\label{lem:inner_t}
    Let $A,B\in\mathbb{R}^{m\times m}$ such that $A^\top =A$ and $B^{\top}=-B$. If $A+\ii B\geq0$, then $A+\ii t B\geq 0$ for all $t\in[-1,1]$. 
\end{lemma}
\begin{proof}
    Let $z=x+\ii y$ with $x,y\in\mathbb{R}^m$. By $z^\dag (A+\ii B)z\geq 0$ and $\bar{z}^\dag (A+\ii B)\bar{z}\geq 0$, we have $x^\top A x+y^\top A y\geq 2 |x^\top B y|$. 
    Then, when $|t|\leq 1$, we have $x^\top A x+y^\top A y\geq 2 |x^\top tB y|$, implying that $A+\ii t B\geq 0$. 
\end{proof}

\begin{theorem}
    For $i=1,2$, let $V_i,\Omega_i\in\mathbb{R}^{m_i\times m_i}$ such that
    \begin{align}
        V_i=V_i^\top>0,\qquad \Omega_i^\top =-\Omega_i,\qquad V_i+\ii\Omega_i\geq 0.
    \end{align}
    Also, let $A_i\in\mathbb{R}^{ m_i\times p}$. Then, the following two are equivalent:
    \begin{enumerate}[(i)]
        \item $\exists X\in\mathbb{R}^{m_2\times m_1}$ such that $XA_1=A_2$ and $V_2+\ii \Omega_2\geq X(V_1+\ii \Omega_1)X^\top$.
        \item For all $t\in(-1,1)$, $A_1^\top(V_1+\ii t\Omega_1)^{-1}A_1\geq A_2^\top (V_2+\ii t \Omega_2)^{-1}A_2$. 
    \end{enumerate}
\end{theorem}
\begin{proof}
    Note that for $t\in[0,1)$, we have
    \begin{align}
        V_i+\ii t \Omega_i=(1-t)V_i+t(V_i+\ii\Omega_i)\geq (1-t)V_i>0. 
    \end{align}
    From the transpose of this inequality, we also have $V_i-\ii t \Omega_i>0$. Therefore, $V_i+\ii t \Omega_i$ is invertible for any $|t|<1$. 

    (i)$\implies$(ii): Applying Lemma~\ref{lem:inner_t} to $A\coloneqq V_2-XV_1X^\top$ and $B\coloneqq \Omega_2-X\Omega_1 X^\top$, we have
    \begin{align}
        V_2+\ii t \Omega_2\geq X(V_1+\ii t \Omega_1)X^\top
    \end{align}
    for all $t\in[-1,1]$. For $|t|<1$, since $V_i+\ii t \Omega_i$ is invertible, Lemma~\ref{lem:ordering_inversion} implies
    \begin{align}
        (V_1+\ii t \Omega_1)^{-1}\geq X^\top (V_2+\ii t \Omega_2)^{-1}X.
    \end{align}
    Therefore, $A_1^\top (V_1+\ii t \Omega_1)^{-1}A_1\geq A_1^\top X^\top (V_2+\ii t \Omega_2)^{-1}XA_1=A_2^\top (V_2+\ii t \Omega_2)^{-1}A_2$. 

    (ii)$\implies$(i): By condition~(ii), 
    \begin{align}
        A_1^\top(V_1+\ii t\Omega_1)^{-1}A_1\geq A_2^\top (V_2+\ii t \Omega_2)^{-1}A_2,\qquad t\in(-1,1).\label{eq:H_inv_pre}
    \end{align}
    For a fixed $s\in(0,1)$, we define
    \begin{align}
        P_i\coloneqq V_i+\ii s \Omega_i,\qquad Q_i\coloneqq V_i-\ii s\Omega_i.
    \end{align}
    Then, $P_i$ and $Q_i$ are positive definite, and $ H_i(\lambda)\coloneqq \lambda P_i+(1-\lambda)Q_i$ is given by
    \begin{align}
        H_i(\lambda)=V_i+\ii (2\lambda-1)s\Omega_i.
    \end{align}
    Therefore, from Eq.~\eqref{eq:H_inv_pre}, 
    \begin{align}
        A_1^\top H_1(\lambda)^{-1}A_1\geq A_2^\top H_2(\lambda)^{-1}A_2,\qquad \lambda \in [0,1].
    \end{align}
    
    Therefore, by Theorem~\ref{thm:existence_complex_matrix}, there exists a complex matrix $Y_s$ such that
    \begin{align}
        Y_sA_1&=A_2,\quad V_2+  \ii s \Omega_2\geq Y_s(V_1+ \ii s \Omega_1)Y_s^\dag,\qquad V_2-  \ii s \Omega_2\geq Y_s(V_1- \ii s \Omega_1)Y_s^\dag.\label{eq:Y_s_ineq}
    \end{align}
    Taking the complex conjugate, we also get
    \begin{align}
         \overline{Y_s}A_1&=A_2,\quad V_2-  \ii s \Omega_2\geq \overline{Y_s}(V_1- \ii s \Omega_1)\overline{Y_s}^\dag,\qquad V_2+  \ii s \Omega_2\geq \overline{Y_s}(V_1+ \ii s \Omega_1)\overline{Y_s}^\dag. \label{eq:bar_Y_s_ineq}
    \end{align}
    Defining a real matrix $X_s$ by
    \begin{align}
        X_s\coloneqq \frac{Y_s+\overline{Y_s}}{2},
    \end{align}
    it satisfies $X_sA_1=A_2$. Moreover, the convexity of the map $Y\mapsto Y H_i Y^\dag$, which is shown by the same argument as in Eq.~\eqref{eq:convexity_sandwitch}, we obtain
    \begin{align}
        V_2+\ii s \Omega_2\geq X_s(V_1+\ii s \Omega_1)X_s^\top. 
    \end{align}
    Evaluating this inequality for real vectors, we have
    \begin{align}
        V_2\geq X_sV_1X_s^\top
    \end{align}
    and therefore, $X_s$ is uniformly bounded since
    \begin{align}
        \|X_s\|^2\leq \frac{\|V_2\|}{\lambda_{\min} (V_1)}.
    \end{align}
    Choose a sequence $(s_n)_{n=1}^{\infty}\subset(0,1)$ such that $s_n\nearrow 1$. Since the family $\{X_s:0<s<1\}$ is uniformly bounded and the space of matrices is finite-dimensional, the sequence $(X_{s_n})$ admits a convergent subsequence. We denote this subsequence by $(X_{s_{n_k}})$ and write
    \begin{align}
        X_{s_{n_k}}\longrightarrow X.
    \end{align}
    Since each $X_{s_{n_k}}$ is real and the space of real matrices is closed, $X$ is also real. Moreover,
    \begin{align}
        X_{s_{n_k}}A_1=A_2,\qquad 
        V_2+\ii s_{n_k}\Omega_2\geq
        X_{s_{n_k}}
        (V_1+\ii s_{n_k}\Omega_1)
        X_{s_{n_k}}^\top.
    \end{align}
    Taking $k\to\infty$ and using continuity, we obtain
    \begin{align}
        XA_1=A_2,\qquad V_2+\ii\Omega_2\geq X(V_1+\ii\Omega_1)X^\top.
    \end{align}
    Thus, condition~(i) follows.
\end{proof}

We now apply the preceding theorem with
\begin{align}
    V_1=\sigma,\quad \Omega_1=\Omega,\quad A_1=K,\qquad V_2=\sigma',\quad \Omega_2=\Omega',\quad A_2=K'.
\end{align}
Since condition~(ii) of Theorem~\ref{thm:convertibility_gaussian_shift} holds for all $t\in(-1,1)$, replacing $t$ by $-t$ puts it in the form required by the preceding theorem. We therefore obtain a real matrix $X$ such that
\begin{align}
    XK=K',\qquad \sigma'+\ii\Omega'\geq X(\sigma+\ii\Omega)X^\top,
\end{align}
which is precisely condition~(iii).

\clearpage
\hypertarget{Part3}{}
\begin{center}
\textbf{{\large \underline{Part III: Quantum local asymptotic normality for finite-dimensional unitary-orbit models}}}
\end{center}
\AppendixTOCThree{Contents of Part III}

For vector spaces $V,V'$, we denote by $\Hom(V,V')$ the set of all linear maps from $V$ to $V'$. In particular, we also write $\End(V)\coloneqq \Hom(V,V)$. 

The notation in Part III is independent of that used elsewhere in the manuscript. For example, $\mathcal{F}$ denotes the Fock space rather than the quantum Fisher information.

\section{Setup, main theorem, and proof strategy}

\subsection{Unitary-orbit model on finite-dimensional system}\label{subsec:unitary_orbit_finite_dim}
Fix a density matrix $\rho_0$ on a finite-dimensional Hilbert space $\mathcal{H}$, which we call a reference state. Let 
\begin{align}
    \rho_0=\sum_{a\in\mathcal{A}_+}\zeta_aP_a,\quad \zeta_a>0
\end{align}
be the spectral decomposition of $\rho_0$, where the labels $a\in\mathcal{A}_+$ denote the distinct positive eigenvalues, $P_a$ is the projector onto the eigenspace $V_a$ of dimension $d_a\coloneqq \dim V_a\geq 1$. If $\rho_0$ is not full-rank, we denote by $V_0\coloneqq \ker \rho_0$ with $\zeta_0=0$. 
Introducing
\begin{align}
    \mathcal{A}\coloneqq
    \begin{cases}
        \mathcal{A}_+,\quad &\ker \rho_0=\{0\},\\
        \mathcal{A}_+\cup\{0\},\quad &\ker \rho_0\neq\{0\},
    \end{cases}
\end{align}
the Hilbert space $\mathcal{H}$ is orthogonally decomposed as $\mathcal{H}=\bigoplus_{a\in\mathcal{A}}V_a$.

We call $H\coloneqq \prod_{a\in\mathcal{A}}U(V_a)$, embedded block-diagonally in $U(\mathcal{H})$, the stabilizer of the state $\rho_0$, where $U(V)$ is the set of all unitary operators on a Hilbert space $V$. The state $\rho_0$ is invariant under the conjugate action of any stabilizer element $h\in H$, namely, $h\rho_0 h^{-1}=\rho_0$. A state with degenerate eigenvalues has non-trivial stabilizers. 

The Lie algebra associated with $H$ is $\mathfrak{h}\coloneqq \bigoplus_{a\in\mathcal{A}}\mathfrak{u}(V_a)$, where $\mathfrak{u}(V)$ denotes the set of all anti-Hermitian operators on $V$. We denote the complexification of $\mathfrak{h}$ by
\begin{align}
    \mathfrak{l}\coloneqq \mathfrak{h}\oplus \ii \mathfrak{h}=\bigoplus_{a\in\mathcal{A}}\End (V_a), \label{eq:l_alg}
\end{align}
where $\End(V)$ is the set of all linear operators on $V$. 

An arbitrary infinitesimal unitary perturbation splits into the stabilizer part $\mathfrak{h}$, and a horizontal part, which moves the state between distinct spectral eigenspaces. Precisely, we define 
\begin{align}
     \mathfrak{u}_{\hor}\coloneqq \{G\in \mathfrak{u}(\mathcal{H})\mid P_aG P_a=0 \, \forall a \in \mathcal{A}\}.
\end{align}
so that $\mathfrak{u}(\mathcal{H})=\mathfrak{h}\oplus \mathfrak{u}_{\hor}$ as an orthogonal direct sum of real vector spaces. 
Introducing the spectral pinching map $\mathcal{D}(\cdot )\coloneqq \sum_{a\in \mathcal{A}}P_a( \cdot) P_a$, any anti-Hermitian operator $K\in \mathfrak{u}(\mathcal{H})$ can be uniquely decomposed as
\begin{align}
    K=K^{\ver}+K^{\hor},\quad K^{\ver}\coloneqq \mathcal{D}(K)\in \mathfrak{h},\quad K^{\hor}\coloneqq K- \mathcal{D}(K)\in \mathfrak{u}_{\hor}. \label{eq:hor_ver_decomposition}
\end{align}

As a starting point of our study on the asymptotic behavior of statistical models, let us first prove the following proposition, which shows that at the QLAN scale, the asymptotic behavior of a unitary-orbit statistical model depends only on the horizontal part. 
\begin{proposition}[Horizontal reduction; extension of Lemma~D6. in Ref.~\cite{yamaguchi_QuantumGeometricTensorDeterminesPureState_2026}]\label{prop:horizontal_reduction}
    Let $V$ be a finite-dimensional normed real vector space and let $K:V\to \mathfrak{u}(\mathcal{H})$ be real-linear. Define $K^{\ver}\coloneqq \mathcal{D}\circ K$, $K^{\hor}\coloneqq K- K^{\ver}$ and  
    \begin{align}
        \rho_{u,n}^K\coloneqq e^{n^{-1/2}K(u)}\rho_0 e^{-n^{-1/2}K(u)},\quad \rho_{u,n}^{K^\hor}\coloneqq e^{n^{-1/2}K^{\hor}(u)}\rho_0 e^{-n^{-1/2}K^{\hor}(u)}.
    \end{align}
    Then there exists a finite constant $C_K$, independent of $n$ and $u$, such that
    \begin{align}
        \left\|\left(\rho_{u,n}^K\right)^{\otimes n}-\left(\rho_{u,n}^{K^\hor}\right)^{\otimes n}\right\|_1\leq C_K\frac{\|u\|^2}{\sqrt{n}}
    \end{align}
    for all $n$ and $u$. Consequently, for any $\epsilon<1/4$, we have
    \begin{align}
       \sup_{\|u\|\leq n^\epsilon}\left\|\left(\rho_{u,n}^K\right)^{\otimes n}-\left(\rho_{u,n}^{K^\hor}\right)^{\otimes n}\right\|_1=O(n^{-1/2+2\epsilon})
    \end{align}
    as $n\to\infty$. 
\end{proposition}

\begin{proof}
    Fix $u$ and write $A_n\coloneqq n^{-1/2} K^\hor(u)$ and $B_n\coloneqq n^{-1/2}K^{\ver}(u)$. Since $[B_n,\rho_0]=0$, we have
    \begin{align}
        e^{A_n}e^{B_n}\rho_0 e^{-B_n}e^{-A_n}=e^{A_n}\rho_0e^{-A_n}= \rho_{u,n}^{K^\hor}.
    \end{align}
    
    For anti-Hermitian operators $A,B$, defining $F_1(s)\coloneqq e^{(1-s)(A+B)}e^{sA}e^{sB}$, we get
    \begin{align}
        e^{A+B}-e^Ae^B=F_1(0)-F_1(1)=-\int_{0}^1 F_1'(s)\dd s =\int_0^1 e^{(1-s)(A+B)}[B,e^{sA}]e^{sB}\dd s.
    \end{align}
    Similarly, for $F_2(t)\coloneqq e^{tA}Be^{(s-t)A}$, we have
    \begin{align}
        [B,e^{sA}]=F_2(0)-F_2(s)=-\int_0^s F_2'(t)\dd t=-\int_0^s e^{tA}[A,B] e^{(s-t)A}\dd t.
    \end{align}
    Therefore, since $e^{tA}$ and $e^{tB}$ are unitary, we obtain
    \begin{align}
        \left\|e^{A+B}-e^Ae^B\right\|\leq \int_0^1s\|[A,B]\|\dd s=\frac{1}{2}\|[A,B]\|\leq \|A\|\|B\|.\label{eq:BCH}
    \end{align}

    Let $\ket{\Psi_0}$ be a purification of $\rho_0$. Then, 
    \begin{align}
       \ket{\Phi^K}\coloneqq \left(e^{A_n+B_n}\otimes I\right)\ket{\Psi_0},\quad  \ket{\Phi^{K^\hor}}\coloneqq \left(e^{A_n}e^{B_n}\otimes I\right)\ket{\Psi_0}
    \end{align}
    are purifications of $\rho_{u,n}^K$ and $\rho_{u,n}^{K^\hor}$, respectively. The distance between these vectors is bounded as
    \begin{align}
        \|\ket{\Phi^K}-\ket{\Phi^{K^\hor}}\|\leq 
        \|e^{A_n+B_n}\otimes I-e^{A_n}e^{B_n}\otimes I\|=\|e^{A_n+B_n}-e^{A_n}e^{B_n}\|\leq \|A_n\|\|B_n\|\leq \frac{C_K}{2}\frac{\|u\|^2}{n},
    \end{align}
    where we have used Eq.~\eqref{eq:BCH} and defined $C_K\coloneqq 2\|K^\ver\|\|K^\hor\|<\infty$. Since the trace distance is contractive under partial trace, we get
    \begin{align}
        \left\|\left(\rho_{u,n}^K\right)^{\otimes n}-\left(\rho_{u,n}^{K^\hor}\right)^{\otimes n}\right\|_1&\leq \left\|(\ket{\Phi^{K}}\bra{\Phi^{K}})^{\otimes n}-(\ket{\Phi^{K^\hor}}\bra{\Phi^{K^\hor}})^{\otimes n}\right\|_1\\
        &\leq 2\sqrt{1-|\braket{\Phi^{K}|\Phi^{K^\hor}}|^{2n}},
    \end{align}
    where in the second line, we have used $\|\ket{u}\bra{u}-\ket{v}\bra{v}\|_1= 2\sqrt{1-|\braket{u,v}|^2}$ for unit vectors $u,v$.
    
    Since, for $x\in[0,1]$, 
    \begin{align}
        1-x^n=(1-x)(1+x+x^2+\cdots+x^{n-1})\leq n(1-x),
    \end{align}
    we get
    \begin{align}
        \sqrt{1-|\braket{\Phi^{K}|\Phi^{K^\hor}}|^{2n}}\leq \sqrt{n}\sqrt{1-|\braket{\Phi^{K}|\Phi^{K^\hor}}|^{2}}\leq \sqrt{n}\|\ket{\Phi^{K}}-\ket{\Phi^{K^\hor}}\|.
    \end{align}
    Therefore,
    \begin{align}
        \left\|\left(\rho_{u,n}^K\right)^{\otimes n}-\left(\rho_{u,n}^{K^\hor}\right)^{\otimes n}\right\|_1\leq C_K\frac{\|u\|^2}{\sqrt{n}}.
    \end{align}
    Consequently, for $\|u\|\leq n^\epsilon$, 
    \begin{align}
        \left\|\left(\rho_{u,n}^K\right)^{\otimes n}-\left(\rho_{u,n}^{K^\hor}\right)^{\otimes n}\right\|_1\leq C_Kn^{-1/2+2\epsilon},
    \end{align}
    which converges to zero as $n\to\infty$ if $\epsilon<1/4$. 
\end{proof}

Intuitively, this proposition implies that, among the generators of unitary transformation, only the ``jump'' operators that map a vector in an eigenspace to a vector in a different eigenspace contribute to the asymptotic behavior of i.i.d. statistical models. In the rest of this paper, therefore, we only consider a unitary orbit generated by the horizontal part. We remark that if $\rho_0=\frac{1}{d}I_{\mathcal{H}}$, $\rho_0$ is invariant under any unitary evolution. Consequently, the asymptotic behavior of its i.i.d. statistical model is trivial, as is also reflected by $\mathfrak{u}_\hor =\{0\}$. Thus, we only consider $\rho_0\neq \frac{1}{d}I_{\mathcal{H}}$ below. We shall repeatedly use the strictly positive minimum spectral gap
\begin{align}
    \delta_*\coloneqq \min_{a,b\in \mathcal{A};\zeta_a>\zeta_b}(\zeta_a-\zeta_b).
\end{align}

To characterize the unitary orbit generated by the horizontal part, let us introduce several notations. For each ordered pair $(a,b)$ with $\zeta_a>\zeta_b$, we define a vector space $W_{ab}\coloneqq\Hom(V_b,V_a)$ equipped with the Hilbert--Schmidt inner product $\braket{X,Y}_{W_{ab}}\coloneqq \Tr_{V_b}(X^\dag Y)$. We denote its conjugate Hilbert space by
\begin{align}
    \mathcal{K}_{ab}\coloneqq \overline{W_{ab}}=\overline{\Hom(V_b,V_a)},
\end{align}
where its elements are written $\bar{X}$ for $X\in W_{ab}$, scalar multiplication is $\alpha\cdot \bar{X}\coloneqq \overline{\bar{\alpha}X}$, and its inner product is
\begin{align}
    \braket{\bar{X},\bar{Y}}_{\mathcal{K}_{ab}}\coloneqq \braket{Y,X}_{W_{ab}}=\Tr_{V_b}(Y^\dag X).\label{eq:inner_prod_K_ab}
\end{align}
For $x=\bar{X}\in\mathcal{K}_{ab}$, we write $\|x\|\coloneqq \|X\|_{\HS}$. As we shall explain soon, the use of the conjugate Hilbert space $\mathcal{K}_{ab}$ is essential to our argument.

The horizontal tangent space is a complex Hilbert space
\begin{align}
    \mathcal{K}_{\hor}\coloneqq \bigoplus_{\zeta_a>\zeta_b\geq 0}\mathcal{K}_{ab}\label{eq:horizontal_one_particle_space}
\end{align}
of dimension $\sum_{\zeta_a>\zeta_b\geq 0}d_ad_b$. Since the Fock space describing the limit model is introduced as a direct sum of symmetric subspaces of the tensor powers of $\mathcal{K}_{\hor}$ (see Section~\ref{sec:Gaussian_limit_model}), we also call it the horizontal one-particle space. 

For $\bar{X}\in\mathcal{K}_{\hor}$, we define
\begin{align}
    \ell_{ab}(\bar{X})\coloneqq X^\dag\in \Hom(V_a,V_b),\quad r_{ab}(\bar{X})\coloneqq X\in \Hom(V_b,V_a)\label{definition:ell_r}
\end{align}
and extend it to maps $\ell,r:\mathcal{K}_\hor\to\mathfrak{gl}(\mathcal{H})$ by direct sum. That is, $\ell_{ab}$ maps the higher-eigenvalue block $V_a$ to the lower-eigenvalue block $V_b$, whereas $r_{ab}$ maps $V_b$ to $V_a$. 

Since $\ell_{ab}$ maps the higher-eigenvalue block $V_a$ to the lower-eigenvalue block $V_b$, $\ell$ is a lowering root vector. Similarly, $r$ is a raising root vector that maps $V_b$ to $V_a$. 

Since $\ell(\alpha\cdot\bar{X})=\ell(\overline{\bar{\alpha}X})=(\bar{\alpha}X)^\dag =\alpha\ell(\bar{X})$, and $r(\alpha\cdot\bar{X})=\bar{\alpha}X=\bar{\alpha}r(\bar{X})$, the map $\ell$ is linear and $r$ is antilinear on $\mathcal{K}_\hor$. This explains our use of the conjugate space $\mathcal{K}_{\mathrm{hor}}$: As we shall show in Theorem~\ref{thm:QLAN}, $\ell$ corresponds to the creation operator that depends linearly on the complex amplitude, whereas $r$ corresponds to the annihilation operator that depends antilinearly on it.

The stabilizer $h=\bigoplus_{a\in\mathcal{A}}h_a\in H$ acts on $X\in W_{ab}$ by $h\cdot X\coloneqq h_a Xh_{b}^{-1}$, and hence on $\mathcal{K}_{ab}$ as $h\cdot \bar{X}\coloneqq \overline{h_a Xh_{b}^{-1}}$. Therefore, 
\begin{align}
    \ell_{ab}(h\cdot \bar{X})&=(h_a Xh_{b}^{-1})^\dag =h_bX^\dag h_a^{-1}=\Ad(h)\ell_{ab}(\bar{X}),\label{eq:block_h_action_ell}\\
    r_{ab}(h\cdot \bar{X})&=h_a Xh_{b}^{-1}=\Ad(h)r_{ab}(\bar{X}).\label{eq:block_h_action_r}
\end{align}

Since $h_a,h_b$ are unitary, the action of $H$ preserves the inner product in Eq.~\eqref{eq:inner_prod_K_ab}, meaning that this action is unitary on $\mathcal{K}_{ab}$ (and hence on $\mathcal{K}_\hor$). 

In this study, we analyze the unitary statistical model:
\begin{definition}[Generator and finite model]
    For $Z=\bigoplus_{\zeta_a>\zeta_b}Z_{ab}\in \mathcal{K}_\hor$, we define anti-Hermitian operator
    \begin{align}
        K(Z)\coloneqq \sum_{\zeta_a>\zeta_b}\frac{\ell_{ab}(Z_{ab})-r_{ab}(Z_{ab})}{\sqrt{\zeta_a-\zeta_b}}\in\mathfrak{u}(\mathcal{H}).\label{eq:definition_K_X}
    \end{align}
    Then, the statistical model for a finite-dimensional system is defined by
    \begin{align}
        \rho_{Z,n}\coloneqq \left(U_n(Z)\rho_0U_n(Z)^\dag\right)^{\otimes n},\qquad U_n(Z)\coloneqq \exp(n^{-1/2}K(Z)).
    \end{align}
\end{definition}

Let us now confirm that this model properly describes the unitary orbits generated by the horizontal part $\mathfrak{u}_\hor$. 
\begin{lemma}\label{lem:real_lin_K}
    The generator $K(Z)$ is anti-Hermitian. The map $K:\mathcal{K}_\hor \to \mathfrak{u}(\mathcal{H})$ is real-linear and isomorphism onto $\mathfrak{u}_\hor$. 
\end{lemma}
\begin{proof}
    For $Z=\sum_{\zeta_a>\zeta_b}\bar{X}_{ab}\in\mathcal{K}_\hor$, 
    \begin{align}
        K(Z)=\sum_{\zeta_a>\zeta_b}\frac{\ell_{ab}(\bar{X}_{ab})-r_{ab}(\bar{X}_{ab})}{\sqrt{\zeta_a-\zeta_b}}=\sum_{\zeta_a>\zeta_b}\frac{X_{ab}^\dag -X_{ab}}{\sqrt{\zeta_a-\zeta_b}}
    \end{align}
    is anti-Hermitian. 
    
    For $t\in\mathbb{R}$, $\ell_{ab}(t\bar{X}_{ab})-r(t\bar{X}_{ab})=t\ell_{ab}(\bar{X}_{ab})-\bar{t}r(\bar{X}_{ab})=t(\ell_{ab}(\bar{X}_{ab})-r(\bar{X}_{ab}))$. Therefore, the map $K$ is real-linear. 

    For any $a\in\mathcal{A}$, $P_aK(Z)P_a=0$. Moreover, the $(b,a)$ block of $K(Z)$ is $X_{ab}^\dag/\sqrt{\zeta_a-\zeta_b}$; hence $K(Z)=0$ implies $X_{ab}=0$ for any ordered pair, so $K$ is injective. For any $M\in\mathfrak{u}_\hor$, define $X_{ab}\coloneqq \sqrt{\zeta_a-\zeta_b}P_aM^\dag P_b$. Then for $Z_{ab}\coloneqq \bar{X}_{ab}\in\mathcal{K}_{ab}$, we have 
    \begin{align}
        P_bK(Z)P_a=\frac{X_{ab}^\dag}{\sqrt{\zeta_a-\zeta_b}}=P_bM P_a.
    \end{align}
    for any $(a,b)$ such that $\zeta_a>\zeta_b$. From the anti-Hermiticity, we also get $P_aK(Z)P_b=P_aM P_b$. Since the $P_aMP_a=0=P_aK(Z)P_a$ for any $a\in\mathcal{A}$, we get $K(Z)=M$. Thus, $K$ is also surjective. 
\end{proof}

Although the use of $\mathcal{K}_\hor$ is convenient for formal arguments, the connection to concrete models may have been obscured. Thus, we here explain a more explicit expression by using the eigenbasis of $\rho_0$. 

We fix an orthonormal basis $\{\ket{e_k}\}_{k=1}^d$ of $\mathcal{H}$ by using the eigenvalue decomposition of $\rho_0$:
\begin{align}
    \rho_0=\sum_{k=1}^d\mu_k\ket{e_k}\bra{e_k},\label{eq:eigenvalue_decomposition_state}
\end{align}
where the eigenvalues are ordered in decreasing order $\mu_1\geq \mu_2\geq\cdots\mu_R>0=\mu_{R+1}=\cdots\mu_d$ with $R\coloneqq\rank \rho_0$. Note that although there is no preferred choice for eigenvectors in a degenerate eigenspace, one can freely fix the basis. 

We introduce $I_a\coloneqq \{k\colon \mu_k=\zeta_a\}$ so that $V_a$ is spanned by $\{\ket{e_k}\colon k\in I_a\}$ and $|I_a|=d_a$. We denote the matrix unit with respect to this basis by $E_{kl}\coloneqq \ket{e_k}\bra{e_l}$. Then, $\rho_0=\sum_{k=1}^d \mu_k E_{kk}$ is diagonal, and a strictly upper-triangular matrix unit $E_{kl}$ ($k<l$) maps an eigenvalue direction to a weakly-higher-eigenvalue direction. In the rest of this paper, all highest-weight notions below refer to this fixed order basis. 

The set $\{\overline{E_{ij}}\colon i\in I_a,\, j\in I_b\}$ forms a basis of $\mathcal{K}_{ab}$. Therefore, $Z\in\mathcal{K}_\hor$ can be expanded as
\begin{align}
    Z=\sum_{(a,b);\,\zeta_a>\zeta_b}\sum_{k\in I_a}\sum_{l\in I_b}z_{kl}\overline{E_{kl}}\label{eq:Z_expansion}
\end{align}
with $z_{kl}\in\mathbb{C}$. Then
\begin{align}
    K(Z)=\sum_{(a,b);\,\zeta_a>\zeta_b}\sum_{k\in I_a}\sum_{l\in I_b}\frac{z_{kl}\ket{e_l}\bra{e_k} - \overline{z_{kl}}\ket{e_k}\bra{e_l}}{\sqrt{\zeta_a-\zeta_b}}=\sum_{(a,b);\,\zeta_a>\zeta_b}\sum_{k\in I_a}\sum_{l\in I_b}\frac{z_{kl}\ket{e_l}\bra{e_k} - \overline{z_{kl}}\ket{e_k}\bra{e_l}}{\sqrt{\mu_k-\mu_l}},\label{eq:expansion_basis_KZ}
\end{align}
which provides an explicit parameterization. 

We here explain the connection to the parameterization in Ref.~\cite{lahiryMinimaxEstimationLowrank2024} by Lahiry and Nussbaum, where QLAN is established under the assumption that positive eigenvalues are nondegenerate. In their study, the generator of the unitary is parameterized by $\mathfrak{z}_{ij}^{(\mathrm{LH})}\in\mathbb{C}$ as
\begin{align}
    \ii \sum_{1\leq k \leq R}\sum_{k<l\leq d}\frac{\Re \left(\mathfrak{z}_{kl}^{(\mathrm{LH})}\right)T_{k,l}+\Im\left(\mathfrak{z}_{kl}^{(\mathrm{LH})}\right)T_{l,k}}{\sqrt{\mu_k-\mu_l}},
\end{align}
where $T_{k,l}\coloneqq \ii E_{kl}-\ii E_{lk}$ and $T_{l,k}\coloneqq E_{lk}+E_{kl}$ for $k<l$. Since
\begin{align}
    \ii\left(\Re \left(\mathfrak{z}_{kl}^{(\mathrm{LH})}\right)T_{k,l}+\Im\left(\mathfrak{z}_{kl}^{(\mathrm{LH})}\right)T_{l,k}\right)&=\ii\frac{\mathfrak{z}_{kl}^{(\mathrm{LH})}+\overline{\mathfrak{z}_{kl}^{(\mathrm{LH})}}}{2}(\ii E_{kl}-\ii E_{lk})+\ii \frac{\mathfrak{z}_{kl}^{(\mathrm{LH})}-\overline{\mathfrak{z}_{kl}}^{(\mathrm{LH})}}{2\ii}(E_{lk}+E_{kl})\\
    &=\mathfrak{z}_{kl}^{(\mathrm{LH})}\ket{e_l}\bra{e_k}-\overline{\mathfrak{z}_{kl}^{(\mathrm{LH})}}\ket{e_k}\bra{e_l},
\end{align}
our parameter $z_{kl}$ equals the parameter $\mathfrak{z}^{(\mathrm{LH})}_{kl}$ in Ref.~\cite{lahiryMinimaxEstimationLowrank2024} when positive eigenvalues are nondegenerate.

% \clearpage
\subsection{Gaussian limit model and the main theorem}\label{sec:Gaussian_limit_model}
For a one-particle space $\mathcal{K}_\hor$, we introduce the corresponding bosonic Fock space by
\begin{align}
    \mathcal{F}_\hor\coloneqq \Gamma_{\mathrm{s}}(\mathcal{K}_\hor)\coloneqq \bigoplus_{k=0}^\infty \Sym^k\mathcal{K}_\hor,
\end{align}
where $\Sym^k\mathcal{K}_\hor$ denotes the symmetric subspace of $\mathcal{K}_{\hor}^{\otimes k}$, i.e., the subspace consisting of symmetric vectors. For readers unfamiliar with symmetric subspaces and the construction of bosonic Fock spaces, we refer to Appendix~\ref{app:Fock_space_review}. We denote the dimension of $\mathcal{K}_\hor$ by $D\coloneqq \dim \mathcal{K}_\hor$.

We introduce the Gaussian limit model on $\mathcal{F}_\hor$. To explain the notation, let us first consider a single-mode system, whose creation and annihilation operators are denoted by $a^\dag$ and $a$. The single-mode thermal state at inverse temperature $\beta$ is given by
\begin{align}
    \phi_\beta\coloneqq \frac{e^{-\beta a^\dag a}}{\Tr (e^{-\beta a^\dag a})}.
\end{align}
Introducing a number basis $\{\ket{m}\}_{m=0}^\infty$, it can also be expressed as
\begin{align}
    \phi_\beta=(1-e^{-\beta})\sum_{m=0}^\infty e^{-\beta m}\ket{m}\bra{m}.\label{eq:thermal_mode_number_basis}
\end{align}
At $\beta=\infty$, i.e., at zero temperature, it is equal to the vacuum state: $\phi_\infty=\ket{0}\bra{0}$. For a coherent amplitude $z\in\mathbb{C}$, the displacement operator $D(z)$ is defined by
\begin{align}
    D(z)\coloneqq \exp\left(za^\dag - \bar{z} a\right).\label{eq:displacement_single_mode} 
\end{align}

Now let us consider a multi-mode Fock space $\mathcal{F}_\hor$. Let $a_{kl}^\dag$ and $a_{kl}$ be single-particle creation and annihilation operators for each $(k,l)$. Let $\beta_{ab}\coloneqq - \ln (\zeta_b/\zeta_a)$ with $\zeta_a>\zeta_b$. For the kernel block, $\zeta_b=0$, we set $\beta_{a0}=-\ln 0\coloneqq \infty$. We define a $D$-mode Gaussian mode
\begin{align}
    \Phi_0\coloneqq \bigotimes_{\zeta_a>\zeta_b\geq 0}\phi_{\beta_{ab}}^{\otimes d_ad_b}=\bigotimes_{\substack{(a,b);\,\zeta_a>\zeta_b\\k\in I_a,\,l\in I_b}}\phi_{\beta_{ab}}^{(kl)},
\end{align}
where $\phi_{\beta_{ab}}^{(kl)}$ denotes a single-mode Gaussian state, i.e., 
\begin{align}
    \phi_{\beta_{ab}}^{(kl)}\coloneqq \frac{e^{-\beta_{ab} a_{kl}^\dag a_{kl}}}{\Tr (e^{-\beta_{ab} a_{kl}^\dag a_{kl}})}\propto \sum_{m=0}^\infty \left(\frac{\zeta_b}{\zeta_a}\right)^m\ket{m}\bra{m}.\label{eq:thermal_state_mode_kl}
\end{align}

For the expansion in Eq.~\eqref{eq:Z_expansion}, $Z=\sum_{(a,b);\,\zeta_a>\zeta_b}\sum_{k\in I_a}\sum_{l\in I_b}z_{kl}\overline{E_{kl}}\in\mathcal{K}_\hor$, we define creation and annihilation operators by
\begin{align}
    a^\dag(Z)\coloneqq \sum_{(a,b);\,\zeta_a>\zeta_b}\sum_{k\in I_a}\sum_{l\in I_b}z_{kl}a^\dag_{kl},\qquad a(Z)\coloneqq \sum_{(a,b);\,\zeta_a>\zeta_b}\sum_{k\in I_a}\sum_{l\in I_b}\overline{z_{kl}}a_{kl}.
\end{align}
The canonical commutation relations $[a_{kl},a_{k'l'}^\dag]=\delta_{kk'}\delta_{ll'}I$ imply that for $X=\sum_{(a,b);\,\zeta_a>\zeta_b}\sum_{k\in I_a}\sum_{l\in I_b}x_{kl}\overline{E_{kl}}\in\mathcal{K}_\hor$ and $Y=\sum_{(a,b);\,\zeta_a>\zeta_b}\sum_{k\in I_a}\sum_{l\in I_b}y_{kl}\overline{E_{kl}}\in\mathcal{K}_\hor$, we have
\begin{align}
    [a(X),a^\dag(Y)]=\sum_{(a,b);\,\zeta_a>\zeta_b}\sum_{k\in I_a}\sum_{l\in I_b}\overline{x_{kl}}y_{kl}I=\braket{X,Y}_{\mathcal{K}_\hor}I.\label{eq:CCR_Fock}
\end{align}
The multi-mode displacement operator is given by
\begin{align}
    D(Z)\coloneqq \exp\left(a^\dag(Z)-a(Z)\right)= \exp\left(\sum_{(a,b);\,\zeta_a>\zeta_b}\sum_{k\in I_a}\sum_{l\in I_b}\left(z_{kl}a^\dag_{kl}-\overline{z_{kl}}a_{kl}\right)\right)
\end{align}
The linearity and anti-linearity conventions for $a(Z)$ and $a^\dag(Z)$ are chosen so that $D(Z)$ is consistent with the standard notation in Eq.~\eqref{eq:displacement_single_mode}. Accordingly, $\ell_{ab}$ and $r_{ab}$ in Eq.~\eqref{definition:ell_r} are defined with the corresponding linearity and anti-linearity properties, using the conjugate Hilbert space $\mathcal{K}_{ab}$.

\begin{definition}[Limit model]
    For $\Phi_0=\bigotimes_{\zeta_a>\zeta_b\geq 0}\phi_{\beta_{ab}}^{\otimes d_ad_b}$, we define $\Phi_Z\coloneqq D(Z)\Phi_0D(Z)^\dag$. 
\end{definition}
We remark that 
\begin{align}
    \Phi_0=\bigotimes_{\zeta_a>\zeta_b>0}\phi_{\beta_{ab}}^{\otimes d_ad_b}\otimes \bigotimes_{\zeta_a>0,\,\zeta_b=0}(\ket{0}\bra{0})^{\otimes d_ad_b}.
\end{align}
Thus, in the limit model, each positive-positive-eigenvalue pair corresponds to a mixed Gaussian state, while each positive-zero-eigenvalue pair corresponds to a pure Gaussian state.

We now state the main theorem of Part III, which concerns QLAN for unitary-orbit models of arbitrary states on a finite-dimensional system. We write $\mathcal{T}_1(\mathcal{H})$ for the trace-class operators on a Hilbert space $\mathcal{H}$.

\begin{theorem}[Unitary-orbit QLAN]\label{thm:QLAN}
    Fix $\beta\in [0,\frac{1}{12})$. Take any $\kappa$ such that $6\beta<\kappa<\frac{1}{2}$. Then there exist quantum channels
    \begin{align}
        T_{n}:\mathcal{T}_1(\mathcal{H}^{\otimes n})\to \mathcal{T}_1(\mathcal{F}_{\hor}),\qquad S_{n}: \mathcal{T}_1(\mathcal{F}_{\hor})\to\mathcal{T}_1(\mathcal{H}^{\otimes n})
    \end{align}
    such that
    \begin{align}
        \sup_{\|Z\|\leq n^\beta}\left\|T_{n}(\rho_{Z,n})-\Phi_Z\right\|_1&=O(n^{-1/2+\kappa}),\qquad
        \sup_{\|Z\|\leq n^\beta}\left\|\rho_{Z,n}-S_{n}(\Phi_Z)\right\|_1=O(n^{-1/2+\kappa}).
    \end{align}
\end{theorem}
As a corollary of this theorem, we also obtain a rate-scaled QLAN.
\begin{corollary}[Rate-scaled unitary-orbit QLAN]\label{cor:QLAN_with_rate}
    Fix a rate $r>0$, and set $m_n\coloneqq \floor{rn}$ and $r_n\coloneqq m_n/n$. For $Z\in\mathcal{K}_\hor$, define
    \begin{align}
        \rho_{Z,n}^{(r)}\coloneqq \left(U_{m_n}(\sqrt{r_n}Z)\rho_0 U_{m_n}(\sqrt{r_n}Z)^\dag\right)^{\otimes m_n}.
    \end{align}
    Fix $\beta\in[0,1/12)$, and take any $\kappa$ such that $6\beta<\kappa<1/2$. Then there exist quantum channels
    \begin{align}
        T_{n}^{(r)}:\mathcal{T}_1(\mathcal{H}^{\otimes m_n})\to \mathcal{T}_1(\mathcal{F}_{\hor}),\qquad S_{n}^{(r)}: \mathcal{T}_1(\mathcal{F}_{\hor})\to\mathcal{T}_1(\mathcal{H}^{\otimes m_n})
    \end{align}
    such that
    \begin{align}
        \sup_{\|Z\|\leq n^\beta}\left\|T_{n}^{(r)}(\rho_{Z,n}^{(r)})-\Phi_{\sqrt{r}Z}\right\|_1&=O(n^{-1/2+\kappa}),\qquad
        \sup_{\|Z\|\leq n^\beta}\left\|\rho_{Z,n}^{(r)}-S_{n}^{(r)}(\Phi_{\sqrt{r}Z})\right\|_1=O(n^{-1/2+\kappa}).
    \end{align}
\end{corollary}
The proof of Theorem~\ref{thm:QLAN} and Corollary~\ref{cor:QLAN_with_rate} is provided in Sections~\ref{sec:proof_of_QLAN_main_theorem} and~\ref{sec:ratescaled_QLAN}, respectively.

%\clearpage
\subsection{Schur-Weyl Reduction and Typical Blocks}
In preparation for explaining the proof strategy of Theorem~\ref{thm:QLAN}, we review the Schur--Weyl decomposition, which decomposes an $n$-fold i.i.d. Hilbert space into a direct sum of Hilbert spaces labeled by Young diagrams. Following the pioneering works on QLAN~\cite{kahnLocalAsymptoticNormality2009,lahiryMinimaxEstimationLowrank2024}, we also introduce a set of typical Young diagrams. As we shall see, the probability of observing atypical diagrams is superpolynomially small in $n$, making it sufficient to restrict attention to typical diagrams in the construction of the conversion channels.

\subsubsection{Block decomposition}
We say that $\lambda=(\lambda_1,\ldots,\lambda_d)\in\mathbb{Z}_{\geq 0}^d$ is a partition of $n$ with at most $d$ parts if $\lambda_1\geq \cdots\geq \lambda_d\geq 0$ and $\sum_{i=1}^d\lambda_i=n$. In this case, we write $\lambda\vdash n$. We denote by $\ell(\lambda)$ the number of nonzero parts in $\lambda$. For notational convenience, we set $\lambda_i\coloneqq 0$ for all $i>d$. 

Let $\mathfrak{S}_n$ denote the permutation group on $\{1,\ldots,n\}$, which acts unitarily on $\mathcal{H}^{\otimes n}$ by permuting the tensor factors.
For each $X\in\End (\mathcal{H})$, define an operator $\pi_n(X)\coloneqq X^{\otimes n}$ on $\mathcal{H}^{\otimes n}$. The permutation action commutes with any $\pi_n(X)$. This is the basic mechanism behind the following Schur--Weyl duality (Theorem~4.1 in Ref.~\cite{kahnLocalAsymptoticNormality2009}; Appendix~A in Ref.~\cite{lahiryMinimaxEstimationLowrank2024}; see e.g., Corollary 4.2.11 and Section~9.1.1 in Ref.~\cite{goodmanSymmetryRepresentationsInvariants2009} for more details). 

\begin{fact}[Schur--Weyl duality]
    There exist an irreducible $\GL(\mathcal{H})$-module $H_\lambda$ of highest weight $\lambda$ and an irreducible $\mathfrak{S}_n$-module $K_\lambda$ such that
    \begin{align}
        \mathcal{H}^{\otimes n}\cong \bigoplus_{\lambda\vdash n,\,\ell(\lambda)\leq d}H_\lambda\otimes K_\lambda,
    \end{align}
    as a representation of $\GL(\mathcal{H})\times\mathfrak{S}_n$. 
\end{fact}
Note that the above identification can be chosen to be unitary with suitable inner products on $H_\lambda$ and $K_\lambda$.

Under the Schur--Weyl decomposition, $X^{\otimes n}$ commutes with $\mathfrak{S}_n$-action. Hence, by Schur's lemma, for $X\in\End(\mathcal H)$, there exist unique operators $\pi_\lambda(X)\in\End(H_\lambda)$ such that
\begin{align}
    X^{\otimes n}&=\bigoplus_{\substack{\lambda\vdash n,\, \ell(\lambda)\leq d}}\pi_\lambda(X)\otimes I_{K_\lambda}.\label{eq:tensor_lambda_decomposition}
\end{align}
The relations $(XY)^{\otimes n}=X^{\otimes n}Y^{\otimes n}$ and $(X^\dag)^{\otimes n}=(X^{\otimes n})^\dag$ implies
\begin{align}
    \pi_\lambda(XY)=\pi_\lambda(X)\pi_\lambda(Y),\quad \pi_{\lambda}(X^\dag)=\pi_\lambda(X)^\dag.
\end{align}

Define the derived representation by
\begin{align}
    \dd \pi_\lambda(X)\coloneqq \left(\frac{\dd}{\dd t}\pi_\lambda (e^{tX})\right)\biggl|_{t=0},\quad X\in \End(\mathcal{H}).
\end{align}
Then, since $\pi_n(e^{tX})=(e^{tX})^{\otimes n}$, 
\begin{align}
    \sum_{i=1}^nI^{\otimes (i-1)}\otimes X\otimes I^{\otimes (n-i)}\biggl|_{H_\lambda\otimes K_\lambda}=\dd \pi_\lambda(X)\otimes I_{K_\lambda}. \label{eq:collective_action_der}
\end{align}
From $\pi_\lambda(T^\dag)=\pi_\lambda(T)^\dag$, 
\begin{align}
    \dd\pi_\lambda(X^\dag)=\dd\pi_\lambda(X)^\dag.
\end{align}
As a standard property of the differential of a representation (see e.g., Section~1 in Ref.~\cite{goodmanSymmetryRepresentationsInvariants2009}), we also have
\begin{align}
    \dd\pi_\lambda([X,Y])=[\dd\pi_\lambda(X),\dd\pi_\lambda(Y)],\quad \pi_\lambda (e^X)=e^{\dd\pi_\lambda(X)}.
\end{align}

From Eq.~\eqref{eq:tensor_lambda_decomposition}, we have
\begin{align}
    \rho_{Z,n}=\bigoplus_{\lambda\vdash n,\,\ell(\lambda)\leq d}p_{\lambda,n}\left(\rho_{\lambda,Z,n}\otimes \frac{I_{K_\lambda}}{\dim K_\lambda}\right),\label{eq:rho_decomposition}
\end{align}
where
\begin{align}
    \pi_\lambda(\rho_0)&\geq0,\quad Z_\lambda\coloneqq \Tr_{H_\lambda}\pi_\lambda(\rho_0),\quad \rho_{\lambda,0}\coloneqq Z^{-1}_\lambda \pi_\lambda(\rho_0)\qquad (Z^{-1}_\lambda>0),\\
    U_{\lambda,n}(Z)&\coloneqq \pi_\lambda(U_n(Z)),\quad \rho_{\lambda,Z,n}\coloneqq U_{\lambda,n}(Z) \rho_{\lambda,0} U_{\lambda,n}(Z)^\dag,\quad p_{\lambda,n}\coloneqq Z_\lambda\dim K_{\lambda}. 
\end{align}
Note that for $\lambda$ such that $\pi_\lambda(\rho_0)=0$, we can formally take $\rho_{\lambda,0}$ to be any fixed state on $H_\lambda$. The choice does not contribute to Eq.~\eqref{eq:rho_decomposition}. Moreover, such $\lambda$ is not a typical diagram in the sense of Definition~\ref{def:typical_set} below, and therefore does not affect the discussion.

Graphically, a partition $\lambda\in \mathbb{Z}^d_{\geq 0}$ is represented by a Young diagram with at most $d$ rows, where the $i$-th row contains $\lambda_i$ boxes (see e.g., Ref.~\cite{fulton_YoungTableauxApplicationsRepresentationTheory_1996} for more details). A way of putting a positive integer in each box is called a filling of the diagram. A semistandard Young tableau is a filling that is weakly increasing across each row and strictly increasing down each column. 

A polynomial $s_\lambda$ defined by
\begin{align}
    s_\lambda(x_1,\cdots,x_d)\coloneqq \sum_{T}\prod_{i=1}^d(x_i)^{m_i(T)}\label{eq:Schur_poly}
\end{align}
is called the Schur polynomial, where the summation runs over semistandard Young tableaux $T$ of shape $\lambda$ with entries in $\{1,\cdots,d\}$ and $m_i(T)$ is the number of entries equal to $i$. From Theorem 6.3 in Ref.~\cite{fulton_RepresentationTheory_2004}, for any $g\in\GL(\mathcal{H})$, the trace of $\pi_\lambda(g)$ is given by the Schur polynomial $s_\lambda$ on the eigenvalues $x_1,\cdots x_d$ of $g$ as
\begin{align}
    \Tr_{H_\lambda} \pi_\lambda(g)=s_\lambda (x_1,\cdots x_d).
\end{align}
Since both sides are polynomial in the matrix entries, the character identity extends by continuity to arbitrary $g\in \End(\mathcal{H})$, including noninvertible $g$. 
Therefore, 
\begin{align}
    Z_\lambda=\Tr_{H_\lambda}\pi_\lambda(\rho_0)=s_\lambda(\mu_1,\cdots,\mu_d),\quad p_{\lambda,n}=\dim K_{\lambda}s_\lambda(\mu_1,\cdots,\mu_d).\label{eq:p_lambda_schur}
\end{align}

For a state with rank $R\coloneqq \rank \rho_0$, we obtain the following:
\begin{lemma}\label{lem:rank_poly}
    $p_{\lambda,n}>0$ holds if and only if $\lambda_{R+1}=0$. 
\end{lemma}
\begin{proof}
    For a semistandard Young tableau $T$, entries strictly increase down columns, implying that an entry in row $j$ is at least $j$. Suppose $\lambda_{R+1}>0$. Then, for any semistandard Young tableau $T$ of shape $\lambda$, $m_{i}(T)>0$ for some $i>R$. Therefore, we get
    \begin{align}
        s_\lambda(\mu_1,\cdots,\mu_d)=\sum_{T}\prod_{i=1}^d(\mu_i)^{m_i(T)}=0,
    \end{align}
    and hence $p_{\lambda,n}=0$. 

    Conversely, if $\lambda_{R+1}=0$, then the semistandard Young tableau filling every box in row $i$ with $i$ contributes to the positive monomial $\prod_{i\leq R} \mu_i^{\lambda_i}>0$. Thus, $p_{\lambda,n}>0$. 
\end{proof}
Therefore, only a shape $\lambda$ with $\lambda_{R+1}=0$ contributes to the decomposition in Eq.~\eqref{eq:rho_decomposition}.

\subsubsection{Typical diagrams and concentration}
We introduce a typical set of Young diagrams.
\begin{definition}[Typical set]\label{def:typical_set}
    For a fixed $\alpha\in(1/2,1)$, define
    \begin{align}
        \Lambda_n\coloneqq \{\lambda\vdash n\colon \lambda_{R+1}=0,\quad |\lambda_i-n\mu_i|\leq n^\alpha\text{ for }1\leq i \leq R\}.\label{eq:typical_set_definition}
    \end{align}
\end{definition}
The aim of this subsection is to prove that the sum of $p_{\lambda,n}$ over atypical diagrams is superpolynomially small in $n$. Although this fact has already been established in Refs.~\cite{kahnLocalAsymptoticNormality2009,lahiryMinimaxEstimationLowrank2024}, we provide its proof for completeness.

Let us first show the following bound for the Schur character:

\begin{lemma}\label{lem:Schur_char_bound}
    For any $\lambda\vdash n$ with $\lambda_{R+1}=0$, 
    \begin{align}
        s_\lambda(\mu_1,\ldots,\mu_R)\leq (n+1)^{R(R-1)/2}\prod_{i=1}^R\mu_i^{\lambda_i}.
    \end{align}
\end{lemma}
\begin{proof}
    Let $T$ be a semistandard tableau of shape $\lambda=(\lambda_1,\ldots,\lambda_R)$ with entries in $\{1,\ldots,R\}$. Since the entries strictly increase down column, an entry in row $j$ is at least $j$. Thus, entries from $\{1,\ldots,k\}$ occurs only in the first $k$ rows, which contains $\sum_{i=1}^k\lambda_i$ boxes, implying that
    \begin{align}
        \sum_{i=1}^km_i(T)\leq \sum_{i=1}^k\lambda_i.
    \end{align}

    Define $w_i\coloneqq \lambda_i-m_i(T)$ and $W_k\coloneqq \sum_{i=1}^k w_i\geq 0$. Since $W_R=0$, we have
    \begin{align}
        \sum_{i=1}^Rw_i\ln(\mu_i)=\sum_{i=1}^{R}(W_{i}-W_{i-1})\ln(\mu_i)=\sum_{k=1}^{R-1}W_k(\ln(\mu_k)-\ln(\mu_{k+1}))\geq 0,
    \end{align}
    where we have used $\mu_{k}\geq \mu_{k+1}$ in the last inequality. Exponentiating this relation, we get $\prod_{i=1}^R\mu_i^{\lambda_i-m_i(T)}\geq 1$, or equivalently,
    \begin{align}
        \prod_{i=1}^R\mu_i^{\lambda_i}\geq \prod_{i=1}^R\mu_i^{m_i(T)}.
    \end{align}
    Therefore, from Eq.~\eqref{eq:Schur_poly}, we obtain
    \begin{align}
        s_\lambda(\mu_1,\ldots,\mu_R) \leq\sum_T  \prod_{i=1}^R\mu_i^{\lambda_i}.\label{eq:Schur_poly_estimate1}
    \end{align}

    Let us now count the number of semistandard Young tableaux of shape $\lambda$. For a semistandard Young tableau $T$, let $\lambda^{(k)}(T)$ denote the shape formed by the boxes with entries $\leq k$. Each $\lambda^{(k)}(T)$ is again a Young-diagram shape  and satisfies $\lambda^{(1)}(T)\subset\lambda^{(2)}(T)\subset \cdots\subset \lambda^{(R)}(T)=\lambda$. Importantly, the chain $\{\lambda^{(k)}\}_{k=0}^R$ uniquely determines $T$, and therefore, the number of semistandard Young tableaux is equal to the number of such chains. Each Young tableau $\lambda^{(k)}$ is specified by its row length $\lambda_i^{(k)}$. Since entries in a column are strictly increasing, $\lambda^{(k)}$ has at most $k$ rows. Since $\lambda^{(R)}=\lambda$, a chain is specified by numbers $\{\lambda_i^{(k)}\}_{1\leq i\leq k,1\leq k\leq R-1}$, each of which takes values in $\{0,\ldots,n\}$. Thus, the number of semistandard Young tableaux is at most
    \begin{align}
        (n+1)^{1+2+\cdots+(R-1)}=(n+1)^{R(R-1)/2}.
    \end{align}
    From Eq.~\eqref{eq:Schur_poly_estimate1}, we obtain
    \begin{align}
         s_\lambda(\mu_1,\ldots,\mu_R) \leq (n+1)^{R(R-1)/2} \prod_{i=1}^R\mu_i^{\lambda_i}.
    \end{align}
\end{proof}

From Eq.~(4.11) in Ref.~\cite{fulton_RepresentationTheory_2004}, we have
\begin{align}
    \dim K_\lambda =\frac{n!}{l_1!\cdots l_R!}\prod_{1\leq i <j\leq R}(l_i-l_j),
\end{align}
where we have defined $l_i\coloneqq \lambda_i+R-i$ for $1\leq i \leq R$. Since
\begin{align}
    l_i-l_j=\lambda_i-\lambda_j+j-i\leq \lambda_i+j-i,
\end{align}
we have
\begin{align}
    \prod_{j=i+1}^R( l_i-l_j)\leq \prod_{j=i+1}^R( \lambda_i+j-i)=\frac{(\lambda_i+R-i)!}{\lambda_i!}=\frac{l_i!}{\lambda_i!}.
\end{align}
Therefore,
\begin{align}
    \dim K_\lambda= \frac{n!}{l_1!\cdots l_R!}\prod_{1\leq i <j\leq R}(l_i-l_j)\leq \frac{n!}{\lambda_1!\cdots\lambda_R!}
\end{align}
From Eq.~\eqref{eq:p_lambda_schur} and Lemma~\ref{lem:Schur_char_bound}, we obtain
\begin{align}
    p_{\lambda,n}\leq (n+1)^{R(R-1)/2}\frac{n!}{\lambda_1!\cdots\lambda_R!} \prod_{i=1}^R\mu_i^{\lambda_i}=(n+1)^{R(R-1)/2}\mathbb{P}\left(X^{(n)}=\lambda\right),
\end{align}
where $X^{(n)}=(X_1^{(n)},\ldots,X_R^{(n)})\sim \mathrm{Multinomial}(n;\mu_1,\ldots,\mu_R)$. By using this bound, we obtain the following:

\begin{proposition}[Typical concentration]\label{prop:typical_concentration}
    For any $\alpha\in(1/2,1)$, there are constants $C_1,C_2,C_3>0$ such that
    \begin{align}
        \varepsilon_{n}^{\atyp}\coloneqq \sum_{\lambda\notin\Lambda_n}p_{\lambda,n}\leq C_1n^{C_2}e^{-C_3n^{2\alpha-1}}.
    \end{align}
\end{proposition}
\begin{proof}
    By Lemma~\ref{lem:rank_poly}, only diagrams with $\lambda_{R+1}=0$ contribute. For such a diagram, $\lambda\notin \Lambda_n$ implies $\max_{1\leq i\leq R}|\lambda_i-n\mu_i|> n^\alpha$. Therefore, 
    \begin{align}
        \varepsilon_{n}^{\atyp}&\leq (n+1)^{R(R-1)/2}\sum_{\substack{\lambda\vdash n\\ \max_i |\lambda_i-n\mu_i|> n^\alpha}}\mathbb{P}\left(X^{(n)}=\lambda\right)\\
        &\leq (n+1)^{R(R-1)/2}\mathbb{P}\left(\max_{1\leq i\leq R}|X_i^{(n)}-n\mu_i|> n^\alpha\right)\\
        &\leq (n+1)^{R(R-1)/2}\sum_{i=1}^R\mathbb{P}\left(|X_i^{(n)}-n\mu_i|> n^\alpha\right).
    \end{align}
    From Hoeffding's inequality, 
    \begin{align}
        \mathbb{P}\left(|X_i^{(n)}-n\mu_i|> n^\alpha\right)\leq 2e^{-2n^{2\alpha-1}}.
    \end{align}
    Therefore,
    \begin{align}
        \varepsilon_{n}^{\atyp}\leq 2R (n+1)^{R(R-1)/2} e^{-2n^{2\alpha-1}}.
    \end{align}
\end{proof}
This proposition implies that the block corresponding to an atypical Young diagram $\lambda\notin\Lambda_n$ is not essential in asymptotic conversion. Therefore, in the rest of the manuscript, our main focus is to construct channels for $\lambda\in \Lambda_n$. 

Before closing this subsection, we introduce the block gap:
\begin{definition}[Block gap]
    Define the block average of the length by
    \begin{align}
        \bar{\lambda}_a\coloneqq \frac{1}{d_a}\sum_{i\in I_a}\lambda_i,\quad a\in\mathcal{A}_+
    \end{align}
    and set $\bar{\lambda}_0\coloneqq 0$ for the kernel block. For $\zeta_a>\zeta_b$, we define the block gap
    \begin{align}
        s_{ab}(\lambda)\coloneqq \bar{\lambda}_a-\bar{\lambda}_b.
    \end{align}
\end{definition}
The following lemma provides a typical size of block gap:
\begin{lemma}\label{lem:s_ab_lower_upper_bound}
    Uniformly over $\lambda\in \Lambda_n$ and for all sufficiently large $n$, 
    \begin{align}
        \frac{\delta_*}{2}n\leq s_{ab}(\lambda)\leq 2n.\label{eq:s_ab_lower_upper_bound}
    \end{align}
    In particular, for any ordered pairs $(a,b)$ and $(c,d)$ satisfying $\zeta_a>\zeta_b$ and $\zeta_c>\zeta_d$, $s_{ab}(\lambda)/s_{cd}(\lambda)\leq C$ and $n/s_{ab}(\lambda)\leq C'$ with constants $C\coloneqq 4/\delta_*$ and $C'\coloneqq 2/\delta_*$ for all sufficiently large $n$. 
\end{lemma}
\begin{proof}
    For $a\in\mathcal{A}_+$, typicality implies $|\lambda_i-n\zeta_a|\leq n^\alpha$, and hence
    \begin{align}
        |\bar{\lambda}_a-n\zeta_a|\leq \frac{1}{d_a}\sum_{i\in I_a}|\lambda_i-n\zeta_a|\leq n^\alpha.
    \end{align}
    For the kernel block, $\bar{\lambda}_0=0=n\zeta_0$. Therefore,
    \begin{align}
        |s_{ab}(\lambda)-n(\zeta_a-\zeta_b)|\leq 2n^\alpha.\label{eq:s_ab_upper_bound_2nalpha}
    \end{align}
    Since $\alpha<1$, for all sufficiently large $n$, $2n^\alpha\leq \min \{\frac{\delta_*}{2}n,n\}$. Using $\delta_*\leq \zeta_a-\zeta_b\leq 1$, we then obtain
    \begin{align}
        s_{ab}(\lambda)&\geq n(\zeta_a-\zeta_b)-2n^\alpha\geq \frac{\delta_*}{2}n,\\
        s_{ab}(\lambda)&\leq  n(\zeta_a-\zeta_b)+2n^\alpha\leq 2n,
    \end{align}
    which proves Eq.~\eqref{eq:s_ab_lower_upper_bound}. The two ratio bounds follow immediately from these lower and upper bounds.
\end{proof}

%\clearpage
\subsection{Proof strategy}
As shown in Eq.~\eqref{eq:rho_decomposition}, the Schur--Weyl decomposition separates the $n$-copy i.i.d. unitary-orbit model into a classical label $\lambda$ and a quantum state $\rho_{\lambda,Z,n}$ on $H_\lambda$. The multiplicity space $K_\lambda$ is maximally mixed and carries no dependence on $Z$. Intuitively, the row lengths of $\lambda$ encode the macroscopic occupation profile associated with the eigenbasis of $\rho_0$. Proposition~\ref{prop:typical_concentration} shows that the total probability of atypical blocks is superpolynomially small in $n$. It is therefore enough to establish the approximation uniformly on typical blocks and to use fixed output states on atypical blocks. Section~\ref{sec:corresp_fin_Fock} establishes the correspondence between each typical block and the Fock space model, while Section~\ref{sec:forward_reverse_channels} assembles the blockwise constructions into conversion channels. Since the detailed estimates are involved, we first summarize the logic of the proof.

\vspace{5mm}

\underline{Section~\ref{sec:ground_component}:\emph{The ground component and low excitations.}} Inside each typical block $H_\lambda$, we first identify the analog of the zero-excitation sector. In earlier QLAN results without degeneracy in the positive eigenvalues~\cite{kahnLocalAsymptoticNormality2009,lahiryMinimaxEstimationLowrank2024}, the highest-weight vector $v_\lambda$ itself plays this role. With degenerate eigenspaces, however, there is no preferred choice of basis within each eigenspace, and hence no canonical choice of a single highest-weight vector. We therefore fix one eigenbasis, take the corresponding highest-weight vector $v_\lambda$, and then take the linear span of its orbit under the stabilizer. This gives the subspace $E_\lambda$ (Eq.~\eqref{eq:E_lambda_highest_weight_vector_and_stabilizer}), which captures the freedom associated with basis choices within the degenerate eigenspaces. 
 Lemma~\ref{lem:ground_component} shows that inter-block raising operators annihilate $E_\lambda$, while the repeated inter-block lowering operators generate $H_\lambda$ from $E_\lambda$. This naturally leads to the filtration $F_kH_\lambda$, which is spanned by states obtained from $E_\lambda$ by at most $k$ inter-block lowering operations.

The reason for introducing this filtration is that the finite-size estimates become more accurate in the low-excitation regime $k\ll n$. We therefore first construct the finite-to-bosonic correspondence in these low-excitation sectors and later show, after identifying the relevant states, that the weight outside the chosen cutoff is negligible. On $F_kH_\lambda$, the occupations of the eigenspace blocks remain close to their values in $E_\lambda$. This allows us to separate the action of a block-diagonal operator into a dominant scalar part and a controlled non-scalar remainder (Theorem~\ref{thm:l_estimate}). This estimate becomes important in the next section: the commutator of the finite-dimensional ladder operators contains such a block-diagonal term, whose scalar part gives the leading scalar term in the canonical commutation relation (CCR), while the remaining part becomes a finite-size correction.

\vspace{5mm}
\underline{Section~\ref{sec:approximate_commutation_relation}: \emph{Approximate Canonical Commutation Relations.}} 
Local unitary perturbations generate transitions between different eigenspaces, which, after an appropriate rescaling, play the roles of bosonic creation and annihilation operators. The goal of this section is to identify this rescaling and to control the finite-size corrections to the canonical commutation relations. Recall that a transition from an eigenspace with a higher eigenvalue to one with a lower eigenvalue corresponds to a creation operator, while the reverse transition corresponds to an annihilation operator.

The key observation is that the commutator of a transition and its reverse measures an occupation-number difference. For example, for $k\in I_a$ and $l\in I_b$ with $\zeta_a>\zeta_b$,
\begin{align}
    [E_{kl},E_{lk}]=E_{kk}-E_{ll}.
\end{align}
On a Schur block, the diagonal operators $\dd\pi_\lambda(E_{kk})$ and $\dd\pi_\lambda(E_{ll})$ measure occupations. In the presence of degeneracies, the natural basis-independent scalar contribution is governed by the difference between the
average occupations of the two blocks, 
\begin{align}
    s_{ab}(\lambda)\coloneqq \bar{\lambda}_a-\bar{\lambda}_b.
\end{align}
For a typical block, $s_{ab}(\lambda)$ is of order $n$ (Lemma~\ref{lem:s_ab_lower_upper_bound}). We therefore normalize both ladder operators by $s_{ab}(\lambda)^{-1/2}$, so that the leading scalar part of their mixed commutator becomes of order one (see Definitions~\ref{def:rescaling_blocks} and \ref{def:creation_and_annihilation_finite}). With this normalization, Eq.~\eqref{eq:comm_cre_ann_finite} shows
\begin{align}
    [\ann(X),\cre(Y)]=\braket{X,Y}_{\mathcal{K}_\hor}I+\text{finite-$n$ corrections,}
\end{align}
which is the finite-dimensional analogue of the bosonic canonical commutation relation (Eq.~\eqref{eq:CCR_Fock}). The results of this section show that these correction terms are uniformly small in the low-excitation regime used in the subsequent construction.

\vspace{5mm}
\underline{Section~\ref{sec:Fock_embedding}: \emph{Construction of the Fock-Space Embedding.}} Having identified approximate bosonic ladder operators, we use them to build multiparticle excitations. We first introduce a map for an ordered $k$-particle tensor by
\begin{align}
    x_1\otimes \cdots \otimes x_k\otimes \xi\mapsto \frac{1}{\sqrt{k!}}\cre(x_1)\cdots  \cre(x_k)\xi,\qquad x_i\in \mathcal{K}_\hor,\,\xi\in E_\lambda,
\end{align}
and then restrict this map to the symmetric tensors, as in the usual construction of bosonic Fock space. (For a brief
review of the symmetric-tensor construction of Fock space, see Appendix~\ref{app:Fock_space_review} for a brief review.) We perform this construction up to a cutoff $L_n$ that grows slowly with $n$ and denote the resulting map by $\symword_{\lambda,\leq L_n}$. 

The remaining question is whether this map preserves inner products. Because the finite-dimensional ladder operators satisfy the bosonic commutation relations only approximately, this is not automatic. Propositions~\ref{prop:intw_cre} and~\ref{prop:intw_ann} show that the excitation map approximately intertwines the finite-dimensional ladder operators with the Fock creation and annihilation operators. These relations imply that the excitation map approximately preserves the Fock-space inner products. Equivalently, the Gram operator
\begin{align}
    \gram_{\lambda,L_n}\coloneqq \symword_{\lambda,\leq L_n}^\dag \symword_{\lambda,\leq L_n}
\end{align}
is close to the identity on the chosen low-excitation cutoff (Theorem~\ref{thm:qi}). 
A polar correction then turns this approximate embedding into the exact isometry
\begin{align}
    W_\lambda\coloneqq \symword_{\lambda,\leq L_n}\gram_{\lambda,L_n}^{-1/2}\colon \mathcal{F}_{\leq L_n}\otimes E_\lambda \to H_\lambda, 
\end{align}
which remains close to the original excitation map $\symword_{\lambda,\leq L_n}$.

\vspace{5mm}
\underline{Section~\ref{sec:Gaussian_description}: \emph{Gaussian Description of Typical Schur Blocks.}} We now use the isometry $W_\lambda$ to establish the finite-Fock correspondence for both the state and its local dynamics.  First, commuting a creation-type inter-block transition through the block reference state $\rho_{\lambda,0}$ multiplies it by $\zeta_b/\zeta_a$. Repeated excitations therefore acquire geometric weights, exactly as in a thermal bosonic mode (see Eq.~\eqref{eq:thermal_state_mode_kl}). Consequently, Proposition~\ref{prop:b_lambda_n} shows that, uniformly over typical blocks, $\rho_{\lambda,0}$ is asymptotically represented under $W_\lambda$ by the Gaussian reference state $\Phi_0$, up to a negligible cutoff error.

We next compare the local dynamics. For typical blocks, $s_{ab}(\lambda)/n$ approaches the spectral gap $\zeta_a-\zeta_b$. Hence, the finite-dimensional ladder operators are normalized so that they asymptotically match the normalization inherited from the original local generator. Consequently, under the embedding, $n^{-1/2}\dd\pi_\lambda(K(Z))$ is approximated by the Fock displacement generator $a^\dag (Z)-a(Z)$. Proposition~\ref{prop:single_covariance} integrates this generator comparison along coherent-state trajectories and shows that local finite-dimensional unitaries correspond to coherent displacements with a uniformly vanishing error.

Finally, to extend this relation from pure coherent states to mixed displaced Gaussian states, we use the Glauber--Sudarshan $P$-representation, which expresses a mixed Gaussian state as a Gaussian mixture of pure coherent states. The corresponding successive local perturbations of the finite-dimensional system can be combined into an effective local perturbation, with negligible error. Integrating the coherent-state approximation over this Gaussian mixture then yields the uniform approximation of $\rho_{\lambda,Z,n}$ by the displaced Gaussian state $\Phi_Z$. This is also where the earlier low-excitation restriction is justified for the actual states: the probability weight outside the chosen cutoff is negligible.

\vspace{5mm}
\underline{Section~\ref{sec:forward_reverse_channels}: \emph{Forward and Reverse Channels, and QLAN.}} The final step is to assemble the typical-block isometries into global forward and reverse quantum channels. On typical blocks, the preceding approximation is used directly; on atypical blocks, fixed output states are sufficient because their total probability is superpolynomially small. Thus, no further approximation inside a typical block is required. Combining the typical-block approximation with the concentration and cutoff estimates gives the error rate stated in Theorem~\ref{thm:QLAN}.

\vspace{7mm}
Before proceeding to the detailed proof, we fix the following terminology concerning constants used in the uniform estimates. In what follows, we call a positive constant $C>0$ a \emph{structural constant} if it is independent of $n$, $\lambda$, and $Z$. Note that structural constants may depend on the fixed data of the finite-dimensional model, such as the Hilbert space $\mathcal{H}$ and the reference state $\rho_0$.

\clearpage
\section{Correspondence between Finite-dimensional Space and Fock Space on Typical Blocks}\label{sec:corresp_fin_Fock}

\subsection{The ground component and low excitations}\label{sec:ground_component}
From now on, for notational convenience, we abbreviate $T v\coloneqq \dd\pi_\lambda(T)v$ for $T\in \mathfrak{g}\coloneqq \End(\mathcal{H})$ and $v\in H_\lambda$ when no confusion can arise. 

\subsubsection{Triangular decomposition, ground component, and $k$-excitation filtration}
Define
\begin{align}
    \mathfrak{n}_+\coloneqq \bigoplus_{\zeta_a>\zeta_b}\Hom(V_b,V_a),\quad \mathfrak{n}_-\coloneqq \bigoplus_{\zeta_a>\zeta_b}\Hom(V_a,V_b).
\end{align}
Then, $\mathfrak{g}\coloneqq \End(\mathcal{H})=\mathfrak{n}_-\oplus \mathfrak{l}\oplus\mathfrak{n}_+$ as a vector space where $\mathfrak{l}\coloneqq \bigoplus_{a\in \mathcal{A}}\End(V_a)$ (which was defined in Eq.~\eqref{eq:l_alg}). The spaces $\mathfrak{n}_+$ and $\mathfrak{n}_-$ are Lie subalgebras formed by elements that strictly raise and lower, respectively, the block eigenvalues, and they satisfy
\begin{align}
    [\mathfrak{l},\mathfrak{n}_{\pm}]\subset \mathfrak{n}_{\pm}.
\end{align}

Because of the convention for the basis fixed around Eq.~\eqref{eq:eigenvalue_decomposition_state}, $\mathfrak{n}_+$ (resp. $\mathfrak{n}_-$) is contained in the strictly upper-triangular (resp. lower-triangular) subalgebra, while $\mathfrak{l}$ contains strictly upper- and lower-triangular parts within a block. 

We here prove two basic properties of Lie subalgebras:
\begin{lemma}\par~\label{lem:ordering_subalg}
    \begin{enumerate}[(i)]
        \item Let $\mathfrak{a}\subset \mathfrak{g}$ be a Lie subalgebra with ordered basis $(Y_1,\ldots,Y_A)$. Then, any operator $\dd\pi_\lambda(X_1)\cdots \dd\pi_\lambda(X_m)$ with $X_i\in\mathfrak{a}$ is a linear combination of ordered operator monomials $\dd \pi_\lambda(Y_1)^{n_1}\cdots \dd \pi_\lambda(Y_A)^{n_A}$ of total degree at most $m$, i.e., $n_1+\cdots+n_A\leq m$. 
        \item Let $\mathfrak{b}\subset \mathfrak{g}$ be a Lie subalgebra with a decomposition $\mathfrak{b}=\mathfrak{b}_1\oplus\mathfrak{b}_2\oplus\mathfrak{b}_3$. Then, any vector $\dd\pi_\lambda(L_1)\cdots \dd\pi_\lambda(L_m)v$ with $L_i\in\mathfrak{b}$ and $v\in H_\lambda$ is a linear combination of vectors of the sorted form
        \begin{align}
            \dd\pi_\lambda (X_1)\cdots \dd\pi_\lambda (X_p)\dd\pi_\lambda (Y_1)\cdots \dd\pi_\lambda (Y_q)\dd\pi_\lambda (Z_1)\cdots \dd\pi_\lambda (Z_r)v\label{eq:ordering_triangular}
        \end{align}
        with $X_i\in \mathfrak{b}_1$, $Y_i\in\mathfrak{b}_2$, $Z_i\in \mathfrak{b}_3$, with $p+q+r\leq m$. 
    \end{enumerate}
\end{lemma}
\begin{proof}
(i)
Since $\dd\pi_\lambda$ is a Lie-algebra representation, adjacent out-of-order factors can be reordered using
\begin{align}
\dd\pi_\lambda(Y_s)\dd\pi_\lambda(Y_t)=\dd\pi_\lambda(Y_t)\dd\pi_\lambda(Y_s)+\dd\pi_\lambda([Y_s,Y_t]).
\end{align}
Because $\mathfrak a$ is a Lie subalgebra, the commutator $[Y_s,Y_t]$ is again a linear combination of the $\{Y_j\}_{j=1}^A$. Repeating this procedure expresses every operator word of length $m$ as a linear combination of ordered operator monomials of degree at most $m$.

(ii) 
Choose ordered bases of $\mathfrak{b}_1$, $\mathfrak{b}_2$, and $\mathfrak{b}_3$, and concatenate them in this order to obtain an ordered basis of $\mathfrak{b}$ as a vector space. Applying part~(i) for $\mathfrak{b}$ expresses every operator word of length $m$ as a linear combination of ordered monomials of total degree at most $m$, obtaining Eq.~\eqref{eq:ordering_triangular}. 
\end{proof}

Now, we introduce two important spaces: 
\begin{definition}
    Let $v_\lambda \in H_\lambda$ be the highest-weight vector, normalized as $\|v_\lambda\|=1$. We define
    \begin{align}
        E_\lambda\coloneqq U(\mathfrak{l})v_\lambda=\Span\{\dd\pi_\lambda(L_1)\cdots \dd\pi_\lambda(L_m)v_\lambda \colon m\geq 0,\, L_i\in\mathfrak{l}\}\subset H_\lambda,
    \end{align}
    which we call ground component. In addition, we define the $k$-excitation filtration
    \begin{align}
        F_k H_\lambda\coloneqq U_{\leq k}(\mathfrak{n}_-)E_\lambda\coloneqq \Span\{\dd \pi_\lambda(X_1)\cdots \dd \pi_\lambda(X_m)\xi\colon 0\leq m \leq k,\, X_i\in\mathfrak{n}_-,\, \xi\in E_\lambda\}.
    \end{align}
\end{definition}

To explain the intuition behind the vector space $E_\lambda$, we note that we have fixed an ordered orthonormal basis adapted to the eigenspace decomposition $\mathcal H=\bigoplus_{a\in\mathcal{A}}V_a$. This choice singles out a
highest-weight vector $v_\lambda$. When some eigenspaces are degenerate, however, the choice of basis within each block $V_a$ is not canonical. Any other adapted orthonormal basis is obtained by the action of an element of the
stabilizer $H= \prod_{a\in\mathcal{A}}U(V_a)$. Under such a change of basis, $v_\lambda$ is transformed into
$\pi_\lambda(h)v_\lambda$ with $h\in H$, which is the highest-weight vector relative to the transformed basis. Since the complexification of the Lie algebra of $H$ is $\mathfrak{l}$, we have
\begin{align}
    E_\lambda=U(\mathfrak{l})v_\lambda=\Span_{\mathbb{C}}\{\pi_\lambda(h)v_\lambda\colon h\in H\}.\label{eq:E_lambda_highest_weight_vector_and_stabilizer}
\end{align} 
Indeed, $W\coloneqq \Span_{\mathbb{C}}\{\pi_\lambda(h)v_\lambda\colon h\in H\}$ is $H$-invariant and therefore invariant under $\mathfrak{h}$ and its complexification $\mathfrak{l}$. As $v_\lambda \in W$, this gives $U(\mathfrak{l})v_\lambda\subset W$. Conversely, $U(\mathfrak{l})v_\lambda$ is invariant under $\dd\pi_\lambda(\mathfrak{l})$. Since any $h\in H$ has a block-diagonal anti-Hermitian logarithm $A\in\mathfrak{h}\subset \mathfrak{l}$ with $h=e^A$, the identity $\pi_\lambda(h)=e^{\dd\pi_\lambda(A)}$ shows that $U(\mathfrak{l})v_\lambda$ is $H$-invariant. As it contains $v_\lambda$, we obtain $U(\mathfrak{l})v_\lambda\supset W$, and hence $U(\mathfrak{l})v_\lambda= W$. Thus, $E_\lambda$ collects the degrees of freedom associated with the non-uniqueness of the basis inside the degenerate eigenspaces. 

The terminology ``ground component'' reflects the algebraic role stated in the following lemma:
\begin{lemma} (i) $\mathfrak{n}_+E_\lambda=0$. (ii) $H_\lambda=U(\mathfrak{n}_-)E_\lambda = \bigcup_{k\geq 0}F_kH_\lambda$.\label{lem:ground_component}
\end{lemma}
In our construction of the correspondence between the finite-dimensional Schur module and the Fock space, the block-raising and block-lowering operators, namely the operators associated with $\mathfrak n_+$ and $\mathfrak n_-$, play the roles of annihilation and creation operators, respectively. Accordingly, properties~(i) and~(ii) parallel the familiar facts that the Fock vacuum is annihilated by all annihilation operators and that the Fock space is generated from the vacuum by repeated applications of creation operators.

\begin{proof}[Proof of Lemma~\ref{lem:ground_component}]
    (i)
    Any element in $E_\lambda$ is a linear combination of $\dd \pi_\lambda(L_1)\cdots \dd \pi_\lambda(L_m)v_\lambda$ with $L_i\in\mathfrak{l}$. For $X\in \mathfrak{n}_+$, iterating $XL=LX+[X,L]$ and using $[X,L_i]\in\mathfrak{n}_+$, we find that $\dd \pi_\lambda (X)\dd \pi_\lambda(L_1)\cdots \dd \pi_\lambda(L_m)v_\lambda$ is a linear combination of terms of the form $\dd\pi_\lambda(L_{1}')\cdots\dd\pi_\lambda(L_{k}')\dd\pi_\lambda(Y')v_\lambda$, where $L_j'\in\mathfrak{l}$, $Y'\in\mathfrak{n}_+$ and $k\leq m$. Since $\dd\pi_\lambda(Y')v_\lambda=0$ for any $Y'\in\mathfrak{n}_+$, we get
    \begin{align}
        \dd \pi_\lambda (X)\dd \pi_\lambda(L_1)\cdots \dd \pi_\lambda(L_m)v_\lambda=0,
    \end{align}
    and hence $\mathfrak{n}_+E_\lambda=0$. 

    (ii) We first prove $H_\lambda=U(\mathfrak{g})v_\lambda$ (see, e.g., Proposition~14.13(ii) in Ref.~\cite{fulton_RepresentationTheory_2004}). Let $W\coloneqq U(\mathfrak{g})v_\lambda$, which is not empty since $v_\lambda\neq0$. Since $X W\subset W$ for any $X\in\mathfrak{g}$, $W$ is an invariant subspace of $H_\lambda$. Since $H_\lambda$ is irreducible, we get $W=H_\lambda$. 

    From the property~(ii) of Lemma~\ref{lem:ordering_subalg}, $U(\mathfrak{g})v_\lambda$ is spanned by the sorted vectors of the form 
    \begin{align}
        \dd\pi_\lambda (X_1)\cdots \dd\pi_\lambda (X_p)\dd\pi_\lambda (Y_1)\cdots \dd\pi_\lambda (Y_q)\dd\pi_\lambda (Z_1)\cdots \dd\pi_\lambda (Z_r)v_\lambda,
    \end{align}
    where $X_i\in\mathfrak{n}_-$, $Y_i\in\mathfrak{l}$, and $Z_i\in\mathfrak{n}_+$, i.e., $U(\mathfrak{g})v_\lambda=U(\mathfrak{n}_-)U(\mathfrak{l})U(\mathfrak{n}_+)v_\lambda$. Since $\mathfrak{n}_+E_\lambda=0$ and hence $U(\mathfrak{n}_+)v_\lambda=\mathbb{C}v_\lambda$, we get $U(\mathfrak{g})v_\lambda=U(\mathfrak{n}_-)U(\mathfrak{l})v_\lambda=U(\mathfrak{n}_-)E_\lambda$. 
\end{proof}

The following lemma shows the stability of $k$-excitation filtration under $\mathfrak{l}$, which we shall use later. 
\begin{lemma}\label{lem:stability_filtration}
    For each $k\geq 0$, $\mathfrak{l}F_kH_\lambda\subset F_kH_\lambda$. Moreover, for any $g\in\prod_{a\in\mathcal{A}}\GL(V_a)$, $\pi_\lambda(g)F_kH_\lambda=F_kH_\lambda$. 
\end{lemma}

\begin{proof}
    Since $E_\lambda=U(\mathfrak{l})v_\lambda$, $\dd\pi_\lambda (L)E_\lambda \subset E_\lambda$ for any $L\in\mathfrak{l}$. The space $F_kH_\lambda$ is spanned by vectors of the form
    \begin{align}
        w=\dd\pi_\lambda(X_1)\cdots \dd\pi_\lambda(X_m)\xi,
    \end{align}
    where $0\leq m\leq k$, $X_i\in\mathfrak{n}_-$ and $\xi\in E_\lambda$. For $L\in\mathfrak{l}$, consider $\dd\pi_\lambda(L)w$. Moving $\dd\pi_\lambda(L)$ successively to the right, and by using $\dd\pi_\lambda(L)\xi\in E_\lambda$ and $[\mathfrak{l},\mathfrak{n}_-]\subset \mathfrak{n}_-$, we obtain a linear combination of vectors containing at most $m \leq k$ factors from $\mathfrak{n}_-$, and hence $\dd\pi_\lambda(L)w\in F_kH_\lambda$. Thus, $\mathfrak{l}F_kH_\lambda\subset F_kH_\lambda$. 

    Since $F_kH_\lambda$ is invariant under $\dd\pi_\lambda(\mathfrak{l})$, it is also invariant under $e^{\dd\pi_\lambda(L)}=\pi_\lambda(\exp(L))$ for any $L\in\mathfrak{l}$. Let $g\in\prod_{a\in\mathcal{A}}\GL(V_a)$. Its polar decomposition is given by $g=up$, where $u$ is unitary and $p=(g^\dag g)^{1/2}$ is a positive definite operator. Since $g$ is block diagonal, both $u$ and $p$ are also block diagonal. Due to the spectral theorem, there exist anti-Hermitian operator $A$ and Hermitian operator $B$ such that $u=e^A$ and $p=e^B$, and moreover, $A,B\in\bigoplus_{a}\End(V_a)=\mathfrak{l}$. Therefore, $\pi_\lambda(g)=\pi_\lambda(e^A)\pi_\lambda(e^B)$ leaves $F_kH_\lambda$ invariant, i.e., $\pi_\lambda(g)F_kH_\lambda\subset F_kH_\lambda$. Applying the same argument to $g^{-1}$ yields the inverse inclusion. Thus, $\pi_\lambda(g)F_kH_\lambda= F_kH_\lambda$.
\end{proof}

\subsubsection{Weight bookkeeping and blockwise spread}
The main aim of this subsection is to evaluate the spread of the weights of $F_kH_\lambda$. 

Let $\mathfrak{t}\coloneqq \Span\{E_{ii}\colon 1\leq i \leq d\}\subset \mathfrak{l}$. Since the operators $\dd\pi_\lambda(E_{ii})$ are commuting and self-adjoint, $H_\lambda$ decomposes orthogonally into their joint eigenspaces
\begin{align}
    (H_\lambda)_\omega\coloneqq \{v\colon \forall i,\, \dd\pi_\lambda(E_{ii})v=\omega_iv\},\quad \omega=(\omega_1,\ldots,\omega_d)\in\mathbb{R}^d,
\end{align}
which are called weight spaces. We denote the set of weights by
\begin{align}
    \Wt(H_\lambda)\coloneqq \{\omega\colon (H_\lambda)_\omega\neq \{0\}\}.
\end{align}
On $H_\lambda$, the operator $\dd\pi_\lambda(E_{ii})=\sum_jI^{\otimes(j-1)}\otimes E_{ii}\otimes I^{\otimes (n-j)}$ counts the occurrence of the basis index $i$ among the tensor slots. Therefore, $\omega\in\mathbb{Z}_{\geq 0}^d$ and 
\begin{align}
    \omega \in\Wt(H_\lambda)\implies \omega_i\in \{0,1,\ldots,n\},\quad \sum_{i=1}^d\omega_i=n.
\end{align}

If $O\in\mathfrak{g}$ satisfies $[E_{ii},O]=\theta_iO$ for all $i$, we call it an $\mathfrak{t}$-weight element of weight $\theta=(\theta_1,\ldots,\theta_d)$. In this case, $\dd\pi_\lambda(O)$ maps $(H_\lambda)_\omega$ to $(H_\lambda)_{\omega+\theta}$. Indeed, for any $v\in (H_\lambda)_\omega$, 
\begin{align}
    \dd\pi_\lambda(E_{ii})\dd\pi_\lambda(O)v=\left(\dd\pi_\lambda(O)\dd\pi_\lambda(E_{ii})+[\dd\pi_\lambda(E_{ii}),\dd\pi_\lambda(O)]\right)v=(\omega_i+\theta_i)\dd\pi_\lambda(O)v,
\end{align}
implying that $\dd\pi_\lambda(O)v\in (H_\lambda)_{\omega+\theta}$. In particular, the matrix units are weight elements. Indeed, since $[E_{ii},E_{pq}]=(\delta_{ip}-\delta_{iq})E_{pq}$, $E_{pq}$ has weight $e_p-e_q$, where $\{e_i\}_{i=1}^d$ denotes the standard basis vector of $\mathbb{Z}^d$. Consequently, we obtain the following:

\begin{lemma}[Weights in the filtration]\par~\label{lem:wt_filt}
    \begin{enumerate}[(i)]
        \item $E_\lambda$ is spanned by $\mathfrak{t}$-weight vectors. Any weight $\omega'$ of $E_\lambda$ has the form
        \begin{align}
            \omega'=\lambda+\vartheta_1+\cdots+\vartheta_l
        \end{align}
        with $\vartheta_j\in \{e_p-e_q\colon p,q\in I_a\text{ for some }a\}$ and a nonnegative integer $l$. In particular, for any $a\in\mathcal{A}$, weights $\omega'$ of $E_\lambda$ satisfy $\sum_{i\in I_a}\omega_i'=\sum_{i\in I_a}\lambda_i\eqqcolon|\lambda^{(a)}|$. 
        \item $F_kH_\lambda$ is spanned by $\mathfrak{t}$-weight vectors. Any weight $\omega$ of $F_kH_\lambda$ has the form
        \begin{align}
            \omega=\omega'+\theta_1+\cdots+\theta_m,\quad 0\leq m\leq k,
         \end{align}
        with $\omega'\in\Wt(E_\lambda)$ and $\theta_j\in\{e_p-e_q\colon p\in I_b,\,q\in I_a,\,\zeta_a>\zeta_b\}$. Consequently, 
        \begin{align}
            \left|\sum_{i\in I_a}\omega_i-|\lambda^{(a)}|\right|\leq k.
        \end{align}
    \end{enumerate}
\end{lemma}
\begin{proof}
    (i)
    The algebra $\mathfrak{l}$ is spanned by the matrix units $E_{pq}$ for which $p$ and $q$ belong to the same block. Therefore, $E_\lambda$ is spanned by the vectors of the form $\dd\pi_\lambda(E_{p_1q_1})\cdots \dd\pi_\lambda(E_{p_lq_l})v_\lambda$ with some $l\geq 0$ and, for each $j$, $p_j$ and $q_j$ belong to the same block. Any nonzero vector of this form has weight $\lambda+\vartheta_1+\cdots+\vartheta_l$ with $\vartheta_j =e_{p_j}-e_{q_j}$. Note that each $\vartheta_j$ transfers one unit of weight between two indices in the same block. Therefore, 
    \begin{align}
        \sum_{i\in I_a}\omega_i'=\sum_{i\in I_a}\lambda_i.
    \end{align}

    (ii) 
    Since $E_\lambda$ is invariant under commuting self-adjoint operators $\dd \pi_\lambda(E_{11}),\ldots,\dd \pi_\lambda(E_{dd})$, it is spanned by weight vectors. By definition, $F_kH_\lambda=U_{\leq k}(\mathfrak{n}_-)E_\lambda$. The space $\mathfrak{n}_-$ is spanned by the matrix units $E_{pq}$ satisfying $p\in I_b$, $q\in I_a$ and $\zeta_a>\zeta_b$. Consequently, $F_kH_\lambda$ is spanned by vectors of the form $\dd\pi_\lambda(E_{p_1q_1})\cdots \dd\pi_\lambda(E_{p_mq_m})\xi$ with $0\leq m\leq k$ and $\xi\in E_\lambda$. Note that $E_\lambda$ is spanned by weight vectors. If $\xi$ is a weight vector with weight $\omega'\in \Wt(E_\lambda)$, then, $\dd\pi_\lambda(E_{p_1q_1})\cdots \dd\pi_\lambda(E_{p_mq_m})\xi$ has weights $\omega'+\theta_1+\cdots+\theta_m$ with $\theta_i\coloneqq e_{p_i}-e_{q_i}$. For a fixed block $I_a$, each $\theta_j$ changes the total weight in the block by either $-1,0$, or $1$. Therefore,
    \begin{align}
        \left|\sum_{i\in I_a}(\theta_j)_i\right|\leq 1
    \end{align}
    and hence
    \begin{align}
        \left|\sum_{i\in I_a}\omega_i-|\lambda^{(a)}|\right|=\left|\sum_{j=1}^m\sum_{i\in I_a}(\theta_j)_i\right|\leq m \leq k.
    \end{align}
\end{proof}

Let us now investigate the spread of $E_\lambda$. We prove several lemmas.

\begin{lemma}\label{lem:maj_spread}
    Let $\mu=(\mu_1,\ldots,\mu_m)\in\mathbb{R}^m$ with $\mu_1\geq \cdots\geq \mu_m$. Define $\bar{\mu}\coloneqq \frac{1}{m}\sum_i\mu_i$ and $\Delta\coloneqq \max_i|\mu_i-\bar{\mu}|$. If $y\in\mathbb{R}^m$ satisfies $y_1\geq \cdots\geq y_m$, $\sum_{i=1}^my_i=\sum_{i=1}^m\mu_i$, and $\sum_{i=1}^sy_i\leq \sum_{i=1}^s\mu_i$ for $1\leq s<m$, then $\max_i|y_i-\bar{\mu}|\leq \Delta$. 
\end{lemma}
\begin{proof}
    If $m=1$, $y_1=\mu_1=\bar{\mu}$, and hence the claim is trivial. Assume $m\geq 2$. The partial-sum inequality gives $y_1\leq \mu_1\leq \bar{\mu}+\Delta$. Using the total sums $\sum_{i=1}^my_i=\sum_{i=1}^m\mu_i$ together with the partial-sum inequality for $s=m-1$ gives $y_m\geq \mu_m\geq \bar{\mu}-\Delta$. Since $y_1\geq y_i\geq y_m$ for any $i$, we obtain $|y_i-\bar{\mu}|\leq \Delta$.
\end{proof}

\begin{lemma}\label{lem:weight_maj}
    Fix a block $a\in \mathcal{A}$ and list $I_a=\{i_1<i_2<\cdots<i_{d_a}\}$. Then any $\omega\in \Wt(E_\lambda)$ satisfies  
    \begin{align}
        \sum_{r=1}^k\omega_{i_r}\leq \sum_{r=1}^k\lambda_{i_r}\quad (1\leq k \leq d_a),\qquad \sum_{r=1}^{d_a}\omega_{i_r}=\sum_{r=1}^{d_a}\lambda_{i_r}=|\lambda^{(a)}|.
    \end{align}
\end{lemma}
\begin{proof}
    Since $\mathfrak{l}=\End(V_a)\oplus(\bigoplus_{b\neq a}\End(V_b))$ and $[\End(V_a),\End(V_b)]=0$ for $a\neq b$, $E_\lambda$ is spanned by vectors $uw$, where $u\in U(\End(V_a))$ and $w\coloneqq w'v_\lambda$ with $w'\in U(\bigoplus_{b\neq a}\End(V_b))$. 

    Let $p,q\in I_a$ such that $p<q$. Since $E_{pq}$ is a raising operator within the block $a$, $[E_{pq},w']=0$. Thus, $E_{pq}w=E_{pq}w'v_\lambda=w'E_{pq}v_\lambda=0$. Similarly, since $[E_{ii},w']=0$ for $i\in I_a$, we obtain $ E_{ii}w=E_{ii}w'v_\lambda=w'E_{ii}v_\lambda=\lambda_iw$. Therefore, $w$ is a highest-weight vector for $\mathfrak{gl}(V_a)$ with block-$a$ weight $\lambda^{(a)}=(\lambda_{i_1},\ldots,\lambda_{i_{d_a}})$. 

    Define $M_w\coloneqq U(\mathfrak{gl}(V_a))w$. For a triangular decomposition of $\mathfrak{gl}(V_a)=\mathfrak{n}^{(a)}_-\oplus \mathfrak{t}^{(a)}\oplus\mathfrak{n}_+^{(a)}$, by using (ii) in Lemma~\ref{lem:ordering_subalg}, we have $U(\mathfrak{gl}(V_a))=U(\mathfrak{n}^{(a)}_-)U(\mathfrak{t}^{(a)})U(\mathfrak{n}_+^{(a)})$. Since $\mathfrak{n}_+^{(a)}w=0$ and $\mathfrak{t}^{(a)}$ acts as a scalar on $w$, we get $M_w=U(\mathfrak{n}^{(a)}_-) w$. Therefore,
    \begin{align}
        E_\lambda=\sum_{w\in U(\bigoplus_{b\neq a}\End(V_b))v_\lambda}M_w.
    \end{align}

    For a matrix unit $E_{pq}$ of $\mathfrak{n}_-^{(a)}$, we let $p=i_t$ and $q=i_{t'}$ with $t>t'$. Define a partial sum $S_k(\omega)\coloneqq \sum_{r=1}^k\omega_{i_r}$. Under $E_{pq}$, the change in the partial sum is given by
    \begin{align}
        S_k(\omega+e_p-e_q)-S_k(\omega)= 
        \begin{cases}
            0\quad &\text{ if }k< t' \lor t\leq k\\
            -1\quad &\text{ if }t'\leq k <t
        \end{cases},
    \end{align}
    and, in particular, $S_k(\omega+e_p-e_q)\leq S_k(\omega)$ and $S_{d_a}(\omega+e_p-e_q)=S_{d_a}(\omega)$. Thus, for any $\omega\in \Wt(E_\lambda)$, 
    \begin{align}
        \sum_{r=1}^k\omega_{i_r}=S_k(\omega)\leq S_k(\lambda)=\sum_{r=1}^k\lambda_{i_r}\quad (1\leq k \leq d_a),\qquad \sum_{r=1}^{d_a}\omega_{i_r}=S_{d_a}(\omega)=S_{d_a}(\lambda)=\sum_{r=1}^{d_a}\lambda_{i_r}=|\lambda^{(a)}| .
    \end{align}
\end{proof}

\begin{lemma}\label{lem:permutation_weight}
    Let $\sigma$ be a permutation of $\{1,\ldots,d\}$ that permutes $I_a$ and fixes all other indices. For a weight $\omega\in\Wt(E_\lambda)$, we define $\sigma\cdot\omega$ by $(\sigma\cdot\omega)_i\coloneqq \omega_{\sigma^{-1}(i)}$. Then, $\sigma\cdot\omega\in\Wt(E_\lambda)$. 
\end{lemma}
\begin{proof}
    Defining the permutation matrix  $P_\sigma\in \GL(V_a)\subset \prod_{b}\GL(V_b)$ such that $P_\sigma e_j=e_{\sigma(j)}$, $P_\sigma^{-1} E_{ii}P_\sigma=E_{\sigma^{-1}(i)\sigma^{-1}(i)}$. By using the relation $\pi_\lambda(g)\dd\pi_\lambda(T)\pi_\lambda(g)^{-1}=\dd\pi_\lambda(gTg^{-1})$, we get
    \begin{align}
        \dd \pi_\lambda(E_{ii})\pi_\lambda(P_\sigma)v&=\pi_\lambda(P_\sigma)\dd \pi_\lambda(P_\sigma^{-1}E_{ii}P_\sigma)v\\
        &=\pi_\lambda(P_\sigma)\dd\pi_\lambda(E_{\sigma^{-1}(i)\sigma^{-1}(i)})v.
    \end{align}
    If $v$ is a nonzero weight vector with $\omega\in\Wt(E_\lambda)$, 
    \begin{align}
        \dd\pi_\lambda(E_{ii})\pi_\lambda(P_\sigma) v=\omega _{\sigma^{-1}(i)}\pi_\lambda(P_\sigma) v.
    \end{align}
    Since $\pi_\lambda(P_\sigma)$ is invertible $\pi_\lambda(P_\sigma) v\neq 0$. Moreover, from Lemma~\ref{lem:stability_filtration}, $\pi_\lambda(P_\sigma)E_\lambda=E_\lambda$ since $P_\sigma\in\prod_b\GL(V_b)$. 
    Thus, $\sigma\cdot \omega\in\Wt(E_\lambda)$. 
\end{proof}

Combining these lemmas, we obtain the spread of $E_\lambda$-weights.
\begin{lemma}\label{lem:wt_E_lambda}
    Let $\lambda\in\Lambda_n$. For any $\omega\in \Wt(E_\lambda)$, 
    \begin{align}
        |\omega_i-\bar{\lambda}_a|\leq 2 n^\alpha
    \end{align}
    holds for any $a\in\mathcal{A}$ and $i\in I_a$. In particular, $|\omega_i-\omega_j|\leq 4n^\alpha$ for $i,j$ in a common block. 
\end{lemma}
\begin{proof}
    Fix $a\in\mathcal{A}_+$ and let $I_a=\{i_1<\cdots<i_{d_a}\}$. For $\omega\in \Wt(E_\lambda)$, its $a$-block components are $(\omega_{i_1},\cdots,\omega_{i_{d_a}})$. Let $\omega^\downarrow\in\mathbb{Z}^d$ be the element obtained by sorting the components within $a$-block in decreasing order, while leaving other components unchanged. It satisfies $\omega^\downarrow=\sigma\cdot \omega$ with a permutation $\sigma$ in block $a$. Thus, from Lemma~\ref{lem:permutation_weight}, $\omega^\downarrow\in\Wt(E_\lambda)$. We now consider two decreasing vectors $\mu\coloneqq (\lambda_{i_1},\ldots,\lambda_{i_{d_a}})$ and $y\coloneqq (\omega^\downarrow_{i_1},\ldots,\omega^\downarrow_{i_{d_a}})$. From Lemma~\ref{lem:weight_maj}, they satisfy
    \begin{align}
         \sum_{r=1}^ky_{r}\leq \sum_{r=1}^k\mu_{r}\quad (1\leq k \leq d_a),\qquad \sum_{r=1}^{d_a}y_{r}=\sum_{r=1}^{d_a}\mu_{r}.
    \end{align}
    Since $\bar{\lambda}_a= \bar{\mu}$, from Lemma~\ref{lem:maj_spread},
    \begin{align}
        \max_{i\in I_a}|\omega_{i}-\bar{\lambda}_a|=\max_{i\in I_a}|\omega^\downarrow_{i}-\bar{\mu}|\leq \Delta,
    \end{align}
    where $\Delta\coloneqq \max_{i\in I_a}|\mu_i-\bar{\mu}|$.
    
    Now, since $\lambda\in\Lambda_n$, we have $|\lambda_i-n\zeta_a|\leq n^\alpha$ for any $i\in I_a$. Thus, we get
    \begin{align}
        |\bar{\lambda}_a-n\zeta_a|\leq \frac{1}{d_a}\sum_{i\in I_a}|\lambda_i-n\zeta_a|\leq n^\alpha
    \end{align}
    and hence
    \begin{align}
        |\lambda_i-\bar{\lambda}_a|\leq |\lambda_i-n\zeta_a|+|\bar{\lambda}_a-n\zeta_a|\leq 2n^\alpha\quad (i\in I_a),
    \end{align}
    implying that $\Delta\leq 2n^\alpha$. 

    For the kernel block, if $\lambda\in\Lambda_n$, then $\lambda_i=0$ for all $i\in I_0$. Thus, $\bar{\lambda}_0=0$. Moreover, from Lemma~\ref{lem:weight_maj}, $\omega_i=0$ for all $i\in I_0$. Thus, $|\omega_i-\bar{\lambda}_0|=0$. 
\end{proof}

By using these lemmas, we obtain the following:
\begin{proposition}[Spread of $F_kH_\lambda$-weights]\label{prop:spread_FkHlambda_weights}
    Let $\lambda\in\Lambda_n$ and $k\geq 0$. Then any $\omega\in \Wt(F_kH_\lambda)$ satisfies
    \begin{align}
        |\omega_i-\bar{\lambda}_a|\leq 2n^\alpha+k
    \end{align}
    for all $a\in\mathcal{A}$ and $i\in I_a$. In particular, for $i,j$ in a common block,
    \begin{align}
        |\omega_i-\omega_j|\leq 4(n^\alpha+k).
    \end{align}
\end{proposition}
\begin{proof}
    From (ii) of Lemma~\ref{lem:wt_filt}, $\omega\in\Wt(F_kH_\lambda)$ satisfies
    \begin{align}
        \omega=\omega'+\theta_1+\cdots+\theta_m
    \end{align}
    with $\omega'\in\Wt(E_\lambda)$, $\theta_j\in\{e_p-e_q\colon p\in I_b,\,q\in I_a,\,\zeta_a>\zeta_b\}$, and $0\leq m\leq k$. Thus, 
    \begin{align}
        |\omega_i-\omega_i'|\leq \left|\sum_{r=1}^m(\theta_r)_i\right|\leq \sum_{r=1}^m \left|(\theta_r)_i\right|\leq m\leq k.
    \end{align}
    By using Lemma~\ref{lem:wt_E_lambda},
    \begin{align}
        |\omega_i-\bar{\lambda}_a|\leq |\omega_i'-\bar{\lambda}_a|+|\omega_i-\omega_i'|\leq 2n^\alpha+k.
    \end{align}
    Moreover, for $i,j\in I_a$, 
    \begin{align}
        |\omega_i-\omega_j|\leq |\omega_i-\bar{\lambda}_a|+|\omega_j-\bar{\lambda}_a|\leq 4n^\alpha+2k\leq 4(n^\alpha+k).
    \end{align}
\end{proof}

\subsubsection{Uniform bounds for block-diagonal operators}
An operator acting within an eigenspace block may have norm of order $O(n)$ on $H_\lambda$. Using the weight estimates
from the previous subsection, this subsection proves that its centered action has norm $O(n^\alpha+k)$ on the low-excitation sector $F_kH_\lambda$.

We first show that a set of (anti-)Hermitian operators induces a direct sum decomposition of a vector space into invariant subspaces.
\begin{lemma}\label{lem:orthogonal_direct_sum}
    Let $W$ be a finite-dimensional complex inner-product space. Let $\mathcal{M}$ be a set of Hermitian or anti-Hermitian operators on $W$. Then $W$ admits an orthogonal direct sum decomposition $W=\bigoplus_m^\perp W_m$, where each $W_m$ is invariant under $\mathcal{M}$ and minimal in the sense that no $W_m$ contains a nonzero proper $\mathcal{M}$-invariant subspace. 
\end{lemma}
\begin{proof}
    If $W=\{0\}$, there is nothing to prove. Suppose that $W\neq\{0\}$. Since $W$ itself is invariant under any operator in $\mathcal{M}$, the collection of nonzero common $\mathcal{M}$-invariant subspaces of $W$ is nonempty. From these $\mathcal{M}$-invariant subspaces, choose one with minimal dimension, denoted by $W_1$. By its minimality, $W_1$ contains no nonzero proper common $\mathcal{M}$-invariant subspace.
    
    Now, for any $v\in W_1^\perp$, $w\in W_1$ and $M\in\mathcal{M}$, 
    \begin{align}
        \braket{Mv,w}=\braket{v,M^\dag w}=\pm\braket{v,M w}=0
    \end{align}
    since $Mw\in W_1$. Therefore, $W_1^\perp$ is also a $\mathcal{M}$-invariant subspace. Since $\dim W_1^\perp< \dim W<\infty$, and the restriction of $M\in\mathcal{M}$ on $W_1^\perp$ is Hermitian or anti-Hermitian, repeatedly applying the same argument on $W_1^\perp$, we find an orthogonal direct sum decomposition $W=\bigoplus_m^\perp W_m$. 
\end{proof}

We also utilize a standard fact about irreducible angular-momentum multiplets. 
\begin{lemma}\label{lem:sl2_basis}
    Let $W\neq \{0\}$ be a finite-dimensional complex inner-product space, and let $J_+,J_-,J_3$ be operators on $W$ such that
    \begin{align}
        [J_3,J_\pm]=\pm J_\pm,\quad [J_+,J_-]=2 J_3,\quad J_+^\dag =J_-,\quad J_3^\dag=J_3,
    \end{align}
    Suppose that $W$ is irreducible in the sense that it contains no nonzero proper subspace invariant under all of $J_+,J_-,J_3$. Then, there exists $j\in \frac{1}{2}\mathbb{Z}_{\geq 0}$ and orthonormal basis 
    \begin{align}
        \{\ket{j,m}\colon m=-j,-j+1,\ldots,j\}
    \end{align}
    of $W$ such that
    \begin{align}
        J_3\ket{j,m}=m\ket{j,m},\quad J_{\pm}\ket{j,m}=\sqrt{(j\mp m)(j\pm m+1)}\ket{j,m\pm1}.
    \end{align}
    Consequently, $\dim W=2j+1$, and
    \begin{align}
        \|J_+\|=\|J_-\|\leq j+\frac{1}{2}.
    \end{align}
\end{lemma}
\begin{proof}
    For the construction of the orthonormal basis, see, for example, Chapter~3 of Ref.~\cite{georgiLieAlgebrasParticle2018}. The norm bound follows directly from
    \begin{align}
        \|J_+\|=\max_{-j\leq m\leq j}\sqrt{(j-m)(j+m+1)}\leq j+\frac{1}{2}
    \end{align}
    and $\|J_-\|=\|J_+\|$ since $J_-^\dag =J_+$.
\end{proof}

\begin{theorem}\label{thm:l_estimate}
    For $L\in\mathfrak{l}$, let $L^{(a)}$ denote its block components so that $L=\sum_{a\in\mathcal{A}}L^{(a)}$, and define
    \begin{align}
        \chi_\lambda(L)\coloneqq \sum_{a\in\mathcal{A}}\frac{\Tr_{V_a}L^{(a)}}{d_a}|\lambda^{(a)}|.\label{eq:chi_lambda_definition}
    \end{align}
    Then there exists a constant $C_{\mathrm{tl}}$ depending only on $d$ such that, for any $L\in\mathfrak{l}$, $\lambda\in\Lambda_n$, and $k\geq0$, 
    \begin{align}
        \|(\dd\pi_\lambda(L)-\chi_\lambda(L)I)|_{F_kH_\lambda}\|_\op\leq C_{\mathrm{tl}}\|L\|_\op(n^\alpha+k).
    \end{align}
\end{theorem}
\begin{proof}
    By Lemma~\ref{lem:wt_filt}, $F_kH_\lambda$ is spanned by weight vectors. Since the weight spaces $(H_\lambda)_\omega$ are joint eigenspaces of commuting Hermitian operators $\dd\pi_\lambda(E_{11}),\ldots,\dd\pi_\lambda(E_{dd})$, $W\coloneqq F_kH_\lambda$ admits the orthogonal weight-space decomposition
    \begin{align}
        W=\bigoplus_{\omega\in\Wt(W)}^\perp\tilde{W}_\omega,\quad \tilde{W}_\omega\coloneqq W\cap (H_\lambda)_\omega.
    \end{align}
    
    Fix $a\in\mathcal{A}$ and distinct indices $i,j\in I_a$. Let $J_\pm,J_3$ be operators on $W$ defined by
    \begin{align}
        J_+\coloneqq \dd\pi_\lambda(E_{ij})|_W,\quad J_-\coloneqq \dd\pi_\lambda(E_{ji})|_W,\quad J_3\coloneqq \frac{1}{2}\dd\pi_\lambda(E_{ii}-E_{jj})|_W.
    \end{align}
    These restrictions are well-defined since $E_{ij},E_{ji},E_{ii}-E_{jj}\in\mathfrak{l}$, and $\mathfrak{l}$ preserves $W$ by Lemma~\ref{lem:stability_filtration}. The operators satisfy
    \begin{align}
        [J_3,J_\pm]=\pm J_\pm,\quad [J_+,J_-]=2 J_3,\quad J_+^\dag =J_-,\quad J_3^\dag=J_3.
    \end{align}
    Applying Lemma~\ref{lem:orthogonal_direct_sum} to Hermitian/anti-Hermitian pair $J_3$ and $J_+ -J_-$ gives an orthogonal decomposition $W=W_1\oplus^\perp\cdots\oplus^\perp W_N$ into minimal invariant subspaces. Each $W_r$ is also invariant under $J_++J_-=[J_3,J_+-J_-]$, and hence under $J_+$ and $J_-$. Therefore, $W_r$ is irreducible under $J_+,J_-,J_3$. Applying Lemma~\ref{lem:sl2_basis} for each $W_r$, there exists $j_r\in\frac{1}{2}\mathbb{Z}_{\geq 0}$ such that
    \begin{align}
        j_r\in\spec(J_3|_{W_r}),\quad \|J_+|_{W_r}\|\leq j_r+\frac{1}{2}.
    \end{align}

    For $v\in \tilde{W}_\omega$, 
    \begin{align}
        J_3v=\frac{1}{2}\dd\pi_\lambda(E_{ii}-E_{jj})v=\frac{\omega_i-\omega_j}{2}v.
    \end{align}
    Therefore, 
    \begin{align}
        \spec(J_3|_W)\subset \left\{\frac{\omega_i-\omega_j}{2}\colon \omega\in\Wt(W)\right\}.
    \end{align}
    Since $W_r\subset W$, we obtain
    \begin{align}
        j_r\leq \max_{\omega\in\Wt(W)}\left\{\frac{1}{2}|\omega_i-\omega_j|\right\}\leq 2(n^\alpha+k),
    \end{align}
    where we have used Proposition~\ref{prop:spread_FkHlambda_weights}. Therefore,
    \begin{align}
        \|J_+|_{W_r}\|_\op\leq j_r+\frac{1}{2}\leq  \frac{5}{2}(n^\alpha+k),
    \end{align}
    where we have used $n^\alpha+k\geq 1$. Since the subspaces $W_r$ are mutually orthogonal and invariant under $J_+$, we obtain
    \begin{align}
        \|\dd\pi_\lambda(E_{ij})|_W\|_\op\leq\frac{5}{2}(n^\alpha+k).
    \end{align}

    For $L=\sum_{a\in\mathcal{A}}L^{(a)}\in\mathfrak{l}$, decompose each block as
    \begin{align}
        L^{(a)}=M^{(a)}+c_aI_{V_a},\quad c_a\coloneqq \frac{\Tr_{V_a}L^{(a)}}{d_a}
    \end{align}
    so that $\Tr_{V_a} M^{(a)}=0$. For $Z_a\coloneqq \sum_{i\in I_a}E_{ii}$, 
    \begin{align}
        \dd\pi_\lambda(L)-\chi_\lambda(L)I=\sum_{a\in\mathcal{A}}\dd\pi_\lambda(M^{(a)})+\sum_{a\in\mathcal{A}}c_a(\dd\pi_\lambda(Z_a)-|\lambda^{(a)}|I)
    \end{align}
    Furthermore, $|c_a|\leq \|L^{(a)}\|_\op\leq \|L\|_\op$ and hence
    \begin{align}
        \|M^{(a)}\|_\op&\leq \|L^{(a)}\|_\op+|c_a|\leq 2\|L\|_\op.\label{eq:tr_norm_M_a}
    \end{align}

    We first derive a bound for the central part. For $v\in \tilde{W}_\omega$, 
    \begin{align}
        \sum_{a\in\mathcal{A}}c_a(\dd\pi_\lambda(Z_a)-|\lambda^{(a)}|I)v=\sum_{a\in\mathcal{A}}c_a\left(\sum_{i\in I_a}\omega_{i}-|\lambda^{(a)}|\right)v.
    \end{align}
    Therefore, 
    \begin{align}
        \left\|\sum_{a\in\mathcal{A}}c_a(\dd\pi_\lambda(Z_a)-|\lambda^{(a)}|I)\biggl|_{W}\right\|_\op &\leq \sum_{a\in\mathcal{A}}|c_a|\max_{\omega\in\Wt(F_kH_\lambda)}\left|\sum_{i\in I_a}\omega_{i}-|\lambda^{(a)}|\right|\\
        &\leq d \|L\|_\op k\leq d\|L\|_\op (n^\alpha+k)
    \end{align}
    where we have used Lemma~\ref{lem:wt_filt} and $|\mathcal{A}|\leq d$. 

    Next, we derive a bound for the traceless part. Write
    \begin{align}
        M^{(a)}=\sum_{\substack{i,j\in I_a\\i\neq j}}m_{ij}E_{ij}+\sum_{i\in I_a}m_{ii}E_{ii}.
    \end{align}
    From Eq.~\eqref{eq:tr_norm_M_a}, we have
    \begin{align}
        |m_{ij}|\leq \|M^{(a)}\|_\op\leq 2\|L\|_\op.
    \end{align}
    Therefore,
    \begin{align}
        \left\|\sum_{a\in\mathcal{A}}\sum_{\substack{i,j\in I_a\\i\neq j}}m_{ij}\dd\pi_\lambda(E_{ij})\biggl|_{W}\right\|_\op\leq \sum_{a\in\mathcal{A}}\sum_{\substack{i,j\in I_a\\i\neq j}}|m_{ij}|\|\dd\pi_\lambda(E_{ij})|_W\|\leq 5d^2\|L\|_\op (n^\alpha+k).
    \end{align}

    For $v\in \tilde{W}_\omega$, 
    \begin{align}
        \sum_{a\in\mathcal{A}}\sum_{i\in I_a}m_{ii}\dd\pi_\lambda(E_{ii})v=\left(\sum_{a\in\mathcal{A}}\sum_{i\in I_a}m_{ii}\omega_i\right)v
    \end{align}
    Since $\Tr M^{(a)}=0$, $\sum_{i\in I_a}m_{ii}=0$. Therefore,
    \begin{align}
        \sum_{a\in\mathcal{A}}\sum_{i\in I_a}m_{ii}\omega_i=\sum_{a\in\mathcal{A}}\sum_{i\in I_a}m_{ii}(\omega_i-\bar{\lambda}_a).
    \end{align}
    Since $|\omega_i-\bar{\lambda}_a|\leq 2n^\alpha+k$ by Proposition~\ref{prop:spread_FkHlambda_weights}, we obtain
    \begin{align}
        \left\|\sum_{a\in\mathcal{A}}\sum_{i\in I_a}m_{ii}\dd\pi_\lambda(E_{ii})\biggl|_{W}\right\|_\op&\leq \sum_{a\in\mathcal{A}}\sum_{i\in I_a}|m_{ii}|\max_{\omega\in\Wt(F_kH_\lambda)}|(\omega_i-\bar{\lambda}_a)|\\
        &\leq 2d\|L\|_\op(2n^\alpha+k)\\
        &\leq 4d\|L\|_\op(n^\alpha+k).
    \end{align}

    Therefore,
    \begin{align}
        &\|(\dd\pi_\lambda(L)-\chi_\lambda(L)I)|_{F_kH_\lambda}\|_\op\nonumber\\
        &\leq \left\|\sum_{a\in\mathcal{A}}c_a(\dd\pi_\lambda(Z_a)-|\lambda^{(a)}|I)\biggl|_{W}\right\|_\op+\left\|\sum_{a\in\mathcal{A}}\sum_{\substack{i,j\in I_a\\i\neq j}}m_{ij}\dd\pi_\lambda(E_{ij})\biggl|_{W}\right\|_\op+\left\|\sum_{a\in\mathcal{A}}\sum_{i\in I_a}m_{ii}\dd\pi_\lambda(E_{ii})\biggl|_{W}\right\|_\op\\
        &\leq C_{\mathrm{tl}}\|L\|_\op(n^\alpha+k),
    \end{align}
    where $C_{\mathrm{tl}}\coloneqq 5d(d+1)$. 
\end{proof}

% \clearpage
\subsection{Approximate Canonical Commutation Relations}\label{sec:approximate_commutation_relation}
In this subsection, we introduce normalized raising and lowering operators on each typical Schur block and compare their commutation relations with the canonical commutation relations on bosonic Fock space. The normalization is determined by the block gaps $s_{ab}(\lambda)$, which are of order $n$ uniformly over typical diagrams.

The resulting operators $\cre(x)$ and $\ann(x)$ behave as approximate creation and annihilation operators. Their mixed commutator has the scalar term proportional to the identity, together with three different finite-$n$ corrections: off-diagonal remainders of order $n^{-1/2}$, a centered block-diagonal error of order $n^{\alpha-1}+k/n$ on $F_kH_\lambda$, and commutators of this error with the ladder operators of order $n^{-1}$.

Throughout this subsection, we fix $\lambda\in \Lambda_n$ and take $n$ sufficiently large so that Eq.~\eqref{eq:s_ab_lower_upper_bound} holds. In particular, every $s_{ab}(\lambda)$ is strictly positive, and therefore the square-root normalizations below are well-defined.

\subsubsection{Normalized ladder operators}

Recall that for $\bar{X}\in\mathcal{K}_{\hor}$, we defined
\begin{align}
    \ell_{ab}(\bar{X})\coloneqq X^\dag\in \Hom(V_a,V_b),\quad r_{ab}(\bar{X})\coloneqq X\in \Hom(V_b,V_a).
\end{align}

\begin{definition}\label{def:rescaling_blocks}
     We define two linear bijections $\Lmat:\mathcal{K}_{\hor}\to \mathfrak{n}_-$ and $\Umat:\overline{\mathcal{K}_{\hor}}\to\mathfrak{n}_+$ by
    \begin{align}
        \Lmat(x)&\coloneqq \sum_{\zeta_a>\zeta_b}s_{ab}^{-1/2}(\lambda)\ell_{ab}(x_{ab})\in\mathfrak{n}_-,\qquad x=\sum_{\zeta_a>\zeta_b}x_{ab}\in\mathcal{K}_{\hor}, \\
        \Umat(\bar{u})&\coloneqq  \sum_{\zeta_a>\zeta_b}s_{ab}^{-1/2}(\lambda)r_{ab}(u_{ab})\in\mathfrak{n}_+,\qquad    \bar{u}=\sum_{\zeta_a>\zeta_b}\overline{u_{ab}}\in\overline{\mathcal{K}_{\hor}}.
    \end{align}
\end{definition}

The following is a key property of these maps which we use throughout in this subsection. 
\begin{lemma}\label{lem:norm_L_U_inv}
    For any $x,u\in\mathcal{K}_{\hor}$, 
    \begin{align}
        \|\Lmat(x)\|_\HS\leq \sqrt{\frac{2}{\delta_*n}}\|x\|,\quad \|\Umat(\bar{u})\|_\HS\leq \sqrt{\frac{2}{\delta_*n}}\|u\|.
    \end{align}
    For $X\in\mathfrak{n}_-$ and $Y\in\mathfrak{n}_{+}$,
    \begin{align}
        \|(\Lmat)^{-1}(X)\|\leq \sqrt{2n}\|X\|_\HS,\quad \|(\Umat)^{-1}(Y)\|\leq \sqrt{2n}\|Y\|_\HS.
    \end{align}
\end{lemma}
\begin{proof}
    Since distinct matrix blocks are orthogonal with respect to the Hilbert--Schmidt inner product, and $\ell_{ab}$ and $r_{ab}$ preserve the Hilbert--Schmidt norm, we get
    \begin{align}
         \|\Lmat(x)\|_\HS^2=\sum_{\zeta_a>\zeta_b}\frac{\|x_{ab}\|^2}{s_{ab}}\leq \frac{2}{\delta_*n}\|x\|^2,\quad \|\Umat(\bar{u})\|_\HS^2=\sum_{\zeta_a>\zeta_b}\frac{\|u_{ab}\|^2}{s_{ab}}\leq \frac{2}{\delta_*n}\|u\|^2,
    \end{align}
    where we have used $s_{ab}\geq \frac{\delta_*n}{2}$. Similarly, for $X=\sum_{\zeta_a>\zeta_b}X_{ba}\in\mathfrak{n}_-$ and $Y=\sum_{\zeta_a>\zeta_b}Y_{ab}\in\mathfrak{n}_+$, by using $s_{ab}\leq 2n$
    \begin{align}
         \|(\Lmat)^{-1}(X)\|^2\leq \sum_{\zeta_a>\zeta_b}s_{ab}(\lambda)\|X_{ba}\|_\HS^2\leq2n\|X\|_\HS^2,\quad  \|(\Umat)^{-1}(Y)\|^2\leq \sum_{\zeta_a>\zeta_b}s_{ab}(\lambda)\|Y_{ab}\|_\HS^2\leq2n\|Y\|_\HS^2.
    \end{align}
\end{proof}

By using the maps $\Lmat$ and $\Umat$, the approximate creation and annihilation operators are defined as follows.
\begin{definition}\label{def:creation_and_annihilation_finite}
     Define
    \begin{align}
        \cre(x)\coloneqq \dd\pi_\lambda(\Lmat(x)),\quad \ann^{\lin}(\bar{u})\coloneqq \dd\pi_\lambda(\Umat(\bar{u})).
    \end{align}
    The anti-linear annihilation operator is defined by
    \begin{align}
        \ann(u)\coloneqq \ann^{\lin}(\bar{u}). 
    \end{align}
\end{definition}

\begin{lemma}\label{lem:properties_cre_ann}
    For any $x\in\mathcal{K}_{\hor}$, 
    \begin{enumerate}[(i)]
        \item $\cre(x)^\dag=\ann(x)$,
        \item $\ann(x)\xi=0$ for any $\xi\in E_\lambda$, 
        \item $\cre(x)F_kH_\lambda\subset F_{k+1}H_\lambda$ and $\ann(x)F_kH_\lambda\subset F_kH_\lambda$ for any $k\geq 0$, 
        \item For $h\in H$, $\pi_\lambda(h)\cre(x)\pi_\lambda(h)^\dag =\cre(h\cdot x)$ and $\pi_\lambda(h)\ann(x)\pi_\lambda(h)^\dag =\ann(h\cdot x)$. 
    \end{enumerate}
\end{lemma}
\begin{proof}
    (i) 
    For $x_{ab}=\bar{X}_{ab}\in\mathcal{K}_{ab}$, $\ell_{ab}(\bar{X}_{ab})^\dag=(X_{ab}^\dag)^\dag =X_{ab}=r_{ab}(\bar{X}_{ab})$, and hence $\cre(x)^\dag = \ann(x)$. 

    (ii) Since $\Umat(\bar{x})\in\mathfrak{n}_+$, the claim directly follows from Lemma~\ref{lem:ground_component}. 

    (iii) Since $\Lmat(x)\in\mathfrak{n}_-$, $\cre(x)F_kH_\lambda\subset F_{k+1}H_\lambda$ directly follows from the definition of $F_kH_\lambda$. We now prove $\mathfrak{n}_+F_kH_\lambda\subset F_kH_\lambda$ by induction on $k$. For $k=0$, the statement is equivalent to the above property (ii). Suppose that $\mathfrak{n}_+F_lH_\lambda\subset F_lH_\lambda$ for $l\geq 0$. 
    The space $F_{l+1}H_\lambda$ is spanned by $F_lH_\lambda$ and by vectors of the form $u=Xu'$, where $X\in\mathfrak n_-$ and $u'\in F_lH_\lambda$. 
    For $Y\in\mathfrak{n}_+$, $Yu=XYu'+[Y,X]u'$. 
    Since $\mathfrak{n}_+F_lH_\lambda\subset F_lH_\lambda$, we have $Yu'\in F_{l}H_\lambda$, implying that $XYu'\in F_{l+1}H_\lambda$. Since $[Y,X]\in\mathfrak{g}=\mathfrak{n}_-\oplus\mathfrak{l}\oplus\mathfrak{n}_+$, it can be decomposed as $[Y,X]=N_++L+N_-$, where $L\in\mathfrak{l}$ and $N_{\pm}\in\mathfrak{n}_{\pm}$. Since $N_-u' \in F_{l+1}H_\lambda$ by the definition of $F_{l+1}H_\lambda$, $Lu'\in F_lH_\lambda$ from Lemma~\ref{lem:stability_filtration}, and $N_+ u'\in F_lH_\lambda$ by the assumption. Since $F_lH_\lambda\subset F_{l+1}H_\lambda$, we get $[Y,X] u'\in F_{l+1}H_\lambda$. Therefore, $\mathfrak{n}_+F_{l+1}H_\lambda\subset F_{l+1}H_\lambda$. 

    (iv) For $g\in \GL(\mathcal{H})$ and $T\in\mathfrak{g}$, $\pi_\lambda(g)\dd\pi_\lambda(T)\pi_\lambda(g)^{-1}=\dd\pi_\lambda(gTg^{-1})$. For unitary $h\in H$, the operator $\pi_\lambda(h)$ is unitary, and hence $\pi_\lambda(h)^{-1}=\pi_\lambda(h)^\dag$. From Eqs.~\eqref{eq:block_h_action_ell} and \eqref{eq:block_h_action_r}, $h\ell(x)h^{-1}=\ell (h\cdot x)$ and $hr(x)h^{-1}=r(h\cdot x)$. Since $h\cdot x$ stays in the same block summand, and the normalization $s_{ab}^{-1/2}$ is unchanged for each block, we obtain $\pi_\lambda(h)\cre(x)\pi_\lambda(h)^\dag =\cre(h\cdot x)$ and $\pi_\lambda(h)\ann(x)\pi_\lambda(h)^\dag =\ann(h\cdot x)$. 
\end{proof}

\subsubsection{Commutators between creation and annihilation operators}
Since $\mathfrak{n}_-$ is a Lie subalgebra, $[\Lmat(x),\Lmat(y)]\in\mathfrak{n}_-$ for $x,y\in\mathcal{K}_\hor$. Defining a bilinear map $\Xi_1:\mathcal{K}_\hor\times \mathcal{K}_\hor\to \mathcal{K}_\hor$ by
\begin{align}
    \Xi_1(x,y)\coloneqq \sqrt{n}(\Lmat)^{-1}\left([\Lmat(x),\Lmat(y)]\right)\in \mathcal{K}_\hor,
\end{align}
we have
\begin{align}
    [\cre(x),\cre(y)]=n^{-1/2}\cre(\Xi_1(x,y)).
\end{align}
Since $\ann(x)=\cre(x)^\dag$, we also have
\begin{align}
    [\ann(x),\ann(y)]=n^{-1/2}\ann(\Xi_1(y,x)),
\end{align}
where we used $\Xi_1(y,x)=-\Xi_1(x,y)$ and the anti-linearity of $\ann$.

For $\bar{u}\in\overline{\mathcal{K}_\hor}$ and $x\in\mathcal{K}_{\hor}$, we define 
\begin{align}
    \Qmat(\bar{u},x)\coloneqq [\Umat (\bar{u}),\Lmat(x)]\in\mathfrak{g}.
\end{align}
The decomposition $\mathfrak{g}=\End(\mathcal{H})=\mathfrak{n}_-\oplus \mathfrak{l}\oplus\mathfrak{n}_+$ is orthogonal with respect to the Hilbert--Schmidt product. We denote the corresponding orthogonal projections by $P_-:\mathfrak{g}\to \mathfrak{n}_-$, $P_0:\mathfrak{g}\to \mathfrak{l}$, and $P_+:\mathfrak{g}\to \mathfrak{n}_+$. Then we write
\begin{align}
    \Qmat=\Qmat^-+\Qmat ^0 +\Qmat ^+,\qquad \Qmat ^\sigma\coloneqq P_\sigma\Qmat\quad (\sigma\in \{-,0,+\}).
\end{align}
We define bilinear maps $\Xi_2:\overline{\mathcal{K}_\hor}\times \mathcal{K}_\hor\to \mathcal{K}_\hor$ and $\Xi_3:\overline{\mathcal{K}_\hor}\times \mathcal{K}_\hor\to \overline{\mathcal{K}_\hor}$ by
\begin{align}
    \Xi_2(\bar{u},x)\coloneqq \sqrt{n}(\Lmat)^{-1}(\Qmat^-(\bar{u},x))\in\mathcal{K}_\hor,\qquad
    \Xi_3(\bar{u},x)\coloneqq \sqrt{n}(\Umat)^{-1}(\Qmat^+(\bar{u},x))\in\overline{\mathcal{K}_\hor}.
\end{align}
Introducing
\begin{align}
    \Dmat(\bar{u},x)\coloneqq \dd\pi_\lambda(\Qmat^0(\bar{u},x))-\braket{u,x}I,
\end{align}
we get
\begin{align}
    [\ann^\lin(\bar{u}),\cre(x)]&=\dd\pi_\lambda\left(\Qmat^-(\bar{u},x)+\Qmat ^0(\bar{u},x) +\Qmat ^+(\bar{u},x)\right)\\
    &=\braket{u,x}I+ \Dmat(\bar{u},x)+n^{-1/2}\cre(\Xi_2(\bar{u},x))+n^{-1/2}\ann^\lin(\Xi_3(\bar{u},x)).
\end{align}

At higher orders, not only commutators among the creation and annihilation operators but also commutators involving $\Dmat$ and these operators appear. Since $\Qmat^0(\bar{u},x)\in\mathfrak{l}$ and $[\mathfrak{l},\mathfrak{n}_{\pm}]\subset \mathfrak{n}_{\pm}$, we introduce the trilinear maps $\Xi_4:\overline{\mathcal{K}_\hor}\times \mathcal{K}_\hor\times \mathcal{K}_\hor\to \mathcal{K}_\hor$ and $\Xi_5:\overline{\mathcal{K}_\hor}\times \mathcal{K}_\hor\times \overline{\mathcal{K}_\hor}\to \overline{\mathcal{K}_\hor}$:
\begin{align}
    \Xi_4(\bar{u},x,y)\coloneqq n(\Lmat)^{-1}([\Qmat^0(\bar{u},x),\Lmat(y)])\in\mathcal{K}_\hor,\quad  \Xi_5(\bar{u},x,\bar{y})\coloneqq n(\Umat)^{-1}([\Qmat^0(\bar{u},x),\Umat(\bar{y})])\in\overline{\mathcal{K}_\hor}.
\end{align}
Then, we get
\begin{align}
    [\Dmat(\bar{u},x),\cre(y)]=n^{-1}\cre(\Xi_4(\bar{u},x,y)),\quad [\Dmat(\bar{u},x),\ann^\lin(\bar{y})]=n^{-1}\ann^\lin(\Xi_5(\bar{u},x,\bar{y})).
\end{align}

We now prove the following theorem, which summarizes the above commutation relations and includes bounds on the norm of multilinear maps $\Xi_i$.
\begin{theorem}\label{thm:commutators_multilin_maps}
    For $x,y,u,v\in\mathcal{K}_\hor$, the following identity holds on $H_\lambda$:
    \begin{align}
        [\cre(x),\cre(y)]&=n^{-1/2}\cre(\Xi_1(x,y)),\\
        [\ann(u),\ann(v)]&=n^{-1/2}\ann(\Xi_1(v,u))\\
        [\ann^\lin(\bar{u}),\cre(x)]&=\braket{u,x}I+ \Dmat(\bar{u},x)+n^{-1/2}\cre(\Xi_2(\bar{u},x))+n^{-1/2}\ann^\lin(\Xi_3(\bar{u},x)),\label{eq:comm_cre_ann_finite}\\
        [\Dmat(\bar{u},x),\cre(y)]&=n^{-1}\cre(\Xi_4(\bar{u},x,y)),\label{eq:comm_D_C}\\
        [\Dmat(\bar{u},x),\ann^\lin(\bar{y})]&=n^{-1}\ann^\lin(\Xi_5(\bar{u},x,\bar{y})).
    \end{align}
    The multilinear maps $\Xi_i$ satisfy
    \begin{align}
        \|\Xi_1(x,y)\|&\leq \frac{4\sqrt{2}}{\delta_*}\|x\|\|y\|, \quad  \|\Xi_2(\bar{u},x)\|\leq \frac{4\sqrt{2}}{\delta_*}\|u\|\|x\|,\quad  \|\Xi_3(\bar{u},x)\|\leq \frac{4\sqrt{2}}{\delta_*}\|u\|\|x\|\label{eq:bound_xi_1_xi_2_xi_3},\\
        \|\Xi_4(\bar{u},x,y)\|&\leq \frac{4}{\delta_*}\|u\|\|x\|\|y\|,\quad \|\Xi_5(\bar{u},x,\bar{y})\|\leq \frac{4}{\delta_*}\|u\|\|x\|\|y\|.
    \end{align}
    The defect term $\Dmat$ satisfies
    \begin{align}
        \|\Dmat(\bar{u},x)|_{F_kH_\lambda}\|_\op \leq C_{\mathsf{D}}\|u\|\|x\|\varepsilon_0(k),\quad \varepsilon_0(k)\coloneqq n^{\alpha-1}+\frac{k}{n}
    \end{align}
    for a structural constant $C_{\mathsf{D}}$. Moreover, $\chi_\lambda(\Qmat^0(\bar{u},x))=\braket{u,x}$ and $\Qmat^0(\bar{u},x)^\dag=\Qmat^0(\bar{x},u)$.
\end{theorem}
\begin{proof}
    Below, we repeatedly use $\|[R,S]\|_\HS\leq \|RS\|_\HS+\|SR\|_\HS\leq 2\|R\|_\HS\|S\|_\HS$. 

    By using the bound on $\Lmat$ and $\Umat$ in Lemma~\ref{lem:norm_L_U_inv}, 
    \begin{align}
        \|[\Lmat(x),\Lmat(y)]\|_\HS&\leq 2\|\Lmat(x)\|_\HS\|\Lmat(y)\|_\HS\leq \frac{4}{\delta_*n}\|x\|\|y\|,\\
        \|\Qmat^\pm(\bar{u},x)\|_\HS&\leq \|\Qmat(\bar{u},x)\|_\HS\leq 2\|\Umat (\bar{u})\|_\HS\|\Lmat(x)\|_\HS\leq  \frac{4}{\delta_*n}\|u\|\|x\|,
    \end{align}
    where we used contractivity of $P_\pm$ in the first inequality in the second line. Using the bound on $(\Lmat)^{-1}$ and $(\Umat)^{-1}$ in Lemma~\ref{lem:norm_L_U_inv}, we obtain Eq.~\eqref{eq:bound_xi_1_xi_2_xi_3}. 

    We now prove $\chi_\lambda(\Qmat^0(\bar{u},x))=\braket{u,x}$, where $\chi_\lambda$ is defined in Eq.~\eqref{eq:chi_lambda_definition}. Since the block-diagonal projection of $[\Umat(\bar{u}),\Lmat(x)]$ only has contributions from the commutators belonging to the same ordered block pair, we have
    \begin{align}
        \Qmat^0(\bar{u},x)=\sum_{\zeta_a>\zeta_b}s_{ab}^{-1}[r_{ab}(u_{ab}),\ell_{ab}(x_{ab})].\label{eq:Qzero_block}
    \end{align}
    Since for $x_{ab}=\bar{X}_{ab}$ and $u_{ab}=\bar{U}_{ab}$, 
    \begin{align}
        \chi_\lambda ([r_{ab}(u_{ab}),\ell_{ab}(x_{ab})])=\chi_\lambda([U_{ab},X_{ab}^\dag])=\Tr_{V_a}(U_{ab}X_{ab}^\dag)\frac{|\lambda^{(a)}|}{d_a}-\Tr_{V_b}(X_{ab}^\dag U_{ab})\frac{|\lambda^{(b)}|}{d_b}=\braket{u_{ab},x_{ab}}s_{ab}(\lambda),
    \end{align}
    we get $\chi_\lambda(\Qmat^0(\bar{u},x))=\braket{u,x}$. Taking adjoint of Eq.~\eqref{eq:Qzero_block}, we also have $\Qmat^0(\bar{u},x)^\dag=\Qmat^0(\bar{x},u)$. From Eq.~\eqref{eq:Qzero_block},
    \begin{align}
        \| \Qmat^0(\bar{u},x)\|_\op&\leq \sum_{\zeta_a>\zeta_b}s_{ab}^{-1}\|[r_{ab}(u_{ab}),\ell_{ab}(x_{ab})]\|_\op\\
        &\leq \frac{2}{\delta_*n}\sum_{\zeta_a>\zeta_b}\max\{\|U_{ab}X_{ab}^\dag\|_\op,\|X_{ab}^\dag U_{ab}\|_\op\}\\
        &\leq \frac{2}{\delta_*n}\sum_{\zeta_a>\zeta_b}\|U_{ab}\|_\op\|X_{ab}\|_\op&& (\|U_{ab}X_{ab}^\dag\|_\op\leq \|U_{ab}\|_\op\|X_{ab}\|_\op,\,\|X_{ab}^\dag U_{ab}\|_\op\leq \|U_{ab}\|_\op\|X_{ab}\|_\op)\\
        &\leq  \frac{2}{\delta_*n}\sum_{\zeta_a>\zeta_b}\|u_{ab}\|\|x_{ab}\|&&(\|U_{ab}\|_\op\leq \|u_{ab}\|,\, \|X_{ab}\|_\op\leq \|x_{ab}\|)\\
        &\leq \frac{2}{\delta_*n} \|u\|\|x\|&& (\text{Cauchy--Schwarz inequality}).
    \end{align}
    From Theorem~\ref{thm:l_estimate}, we therefore get
    \begin{align}
         \|\Dmat(\bar{u},x)|_{F_kH_\lambda}\|&=\left\|\left(\dd\pi_\lambda(\Qmat^0(\bar{u},x))-\chi_\lambda(\Qmat^0(\bar{u},x))I\right)\big|_{F_kH_\lambda}\right\|\\
         &\leq C_{\mathrm{tl}}\| \Qmat^0(\bar{u},x)\|_\op(n^\alpha+k)\\
         &\leq C_{\mathsf{D}}\|u\|\|x\|\varepsilon_0(k)&&(C_{\mathsf{D}}\coloneqq 2C_{\mathrm{tl}}/\delta_*).
    \end{align}

    Finally, we bound $\Xi_4$ and $\Xi_5$. Write $y_{ab}=\bar{Y_{ab}}$. Because $\Qmat^0(\bar{u},x)$ is block diagonal, the normalization $s_{ab}^{-1/2}$ is canceled when $(\Lmat)^{-1}$ is applied. Thus,
    \begin{align}
        \|\Xi_4(\bar{u},x,y)\|&=n\left(\sum_{\zeta_a>\zeta_b}\|[\Qmat^0(\bar{u},x),Y_{ab}^\dag]\|_\HS^2\right)^{1/2}\\
        &\leq 2n\|\Qmat^0(\bar{u},x)\|_\op\left(\sum_{\zeta_a>\zeta_b}\|Y_{ab}\|_\HS^2\right)^{1/2}\\
        &\leq  2n\frac{2}{\delta_*n}\|u\|\|x\|\|y\|=\frac{4}{\delta_*}\|u\|\|x\|\|y\|.
    \end{align}
    The same argument, with $\Umat$ in place of $\Lmat$, gives the stated bound for $\Xi_5$. 
    
\end{proof}

In particular, for $C_{\Xi}\coloneqq 8\sqrt{2}/\delta_*$, we obtain the following bounds
\begin{align}
        &\|\Xi_1(x,y)\|\leq C_\Xi\|x\|\|y\|,\quad \|\Xi_2(\bar{u},x)\|+\|\Xi_3(\bar{u},x)\|\leq C_\Xi \|u\|\|x\|,\\
        &\|\Xi_4(\bar{u},x,y)\|\leq C_{\Xi}\|u\|\|x\|\|y\|,\quad \|\Xi_5(\bar{u},x,\bar{y})\|\leq C_{\Xi}\|u\|\|x\|\|y\|
\end{align}
which we shall use in the next section.

% \clearpage
\subsection{Construction of the Fock-Space Embedding}\label{sec:Fock_embedding}
Throughout this subsection, $\lambda\in\Lambda_n$ is fixed. We denote by $N$ a positive integer describing a cutoff, which will later be chosen as $N=L_n\coloneqq \floor{n^\eta}$ with $0<\eta<1/6$. 

We write $\mathcal{F}_{\leq N}\coloneqq \bigoplus_{k=0}^N\Sym^k\mathcal{K}_\hor$ for the truncated Fock space. We fix a block-adapted orthonormal basis $\{e_j\}_{j=1}^D$ with $D\coloneqq \dim \mathcal{K}_\hor$ such that for all $j$, $e_j\in\mathcal{K}_{a_jb_j}$ for some pair $(a_j,b_j)$ with $\zeta_{a_j}>\zeta_{b_j}$. We denote by $a^\dag (x)$ and $a(x)$ the bosonic creation and annihilation operators on $\Gamma_{\mathrm{s}}(\mathcal{K}_\hor)$. For the fixed basis $\{e_j\}_{j=1}^D$, we introduce shorthand notations:
\begin{align}
    a_j^\dag \coloneqq a^\dag (e_j),\quad a_j\coloneqq a(e_j),\quad \crej\coloneqq \cre(e_j),\quad \annj\coloneqq \ann(e_j).
\end{align}

\subsubsection{Multiparticle excitations generated by ladder operators}

\begin{definition}[Multiparticle excitation map]
    For $k\geq 0$, we define a multiparticle excitation map $\word_{\lambda,k}:\mathcal{K}_\hor^{\otimes k}\otimes E_\lambda\to H_\lambda$ by linear extension of
    \begin{align}
        \word_{\lambda,k}(x_1\otimes\cdots\otimes x_k\otimes \xi)\coloneqq \frac{1}{\sqrt{k!}}\cre(x_1)\cdots \cre(x_k)\xi,\qquad \word_{\lambda,0}(\xi)\coloneqq\xi.
    \end{align}
    We further define $\symword_{\lambda,k}\coloneqq \word_{\lambda,k}\big|_{\Sym^k\mathcal{K}_\hor\otimes E_\lambda}$, 
    and $\symword_{\lambda, \leq N}:\mathcal{F}_{\leq N}\otimes E_\lambda\to H_\lambda$ by
    \begin{align}
        \symword_{\lambda,\leq N}\left(\bigoplus_{k=0}^Nu_k\right)\coloneqq \sum_{k=0}^N\symword_{\lambda,k}u_k,
    \end{align}
    and the Gram operator $\gram_{\lambda,N}:\mathcal{F}_{\leq N}\otimes E_\lambda\to \mathcal{F}_{\leq N}\otimes E_\lambda$ by $\gram_{\lambda,N}\coloneqq  \symword_{\lambda,\leq N}^\dag  \symword_{\lambda,\leq N}$.

\end{definition}
Because the finite-dimensional ladder operators do not commute exactly, the excitation map is first defined on the full tensor product and then restricted to the symmetric sector.

We first prove the $H$-covariance of these maps.
\begin{lemma}\label{lemma:h_equivalence_symword}
    For $h\in H$, let $u_h:\mathcal{K}_\hor\to\mathcal{K}_\hor$, $u_hx\coloneqq h\cdot x$ and define
    \begin{align}
        U_{\lambda,h}^{(\leq N)}\coloneqq \bigoplus_{k=0}^N\left(u_h^{\otimes k}|_{\Sym^k\mathcal{K}_\hor}\otimes \pi_\lambda(h)|_{E_\lambda}\right)
    \end{align}
    on $\mathcal{F}_{\leq N}\otimes E_\lambda$. Then, for any $h\in H$, 
    \begin{align}
        \pi_\lambda(h)\word_{\lambda,k}=\word_{\lambda,k}(u_h^{\otimes k}\otimes \pi_\lambda(h)|_{E_\lambda})\quad (0\leq k\leq N),\quad \pi_\lambda(h)\symword_{\lambda,\leq N}=\symword_{\lambda,\leq N}U_{\lambda,h}^{(\leq N)},\quad [\gram_{\lambda,N},U_{\lambda,h}^{(\leq N)}]=0.
    \end{align}
\end{lemma}
\begin{proof}
    For a decomposable tensor, the covariance of $\cre$ gives
    \begin{align}
        \pi_\lambda(h)\word_{\lambda,k}(x_1\otimes\cdots\otimes x_k\otimes \xi)=\frac{1}{\sqrt{k!}}\cre(h\cdot x_1)\cdots \cre(h\cdot x_k)\pi_\lambda(h)\xi=\word_{\lambda,k}(h\cdot x_1\otimes \cdots \otimes h\cdot x_k\otimes \pi_\lambda(h)\xi).
    \end{align}
    Thus, by linearity, we obtain $\pi_\lambda(h)\word_{\lambda,k}=\word_{\lambda,k}(u_h^{\otimes k}\otimes \pi_\lambda(h)|_{E_\lambda})$ for $0\leq k\leq N$. Since $u_h^{\otimes k}$ commutes with the permutation action on tensor factors in $\mathcal{K}_\hor^{\otimes k}$, the restriction of this identity gives $\pi_\lambda(h)\symword_{\lambda,k}=\symword_{\lambda,k}(u_h^{\otimes k}|_{\Sym^k\mathcal{K}_\hor}\otimes \pi_\lambda(h)|_{E_\lambda})$ and hence $\pi_\lambda(h)\symword_{\lambda,\leq N}=\symword_{\lambda,\leq N}U_{\lambda,h}^{(\leq N)}$. Thus, 
    \begin{align}
        U_{\lambda,h}^{(\leq N)\dag }\symword_{\lambda,\leq N}^\dag \symword_{\lambda,\leq N}U_{\lambda,h}^{(\leq N)} =\symword_{\lambda,\leq N}^\dag \pi_\lambda(h)^\dag \pi_\lambda(h)\symword_{\lambda,\leq N} =\symword_{\lambda,\leq N}^\dag \symword_{\lambda,\leq N},
    \end{align}
    or equivalently, $[\gram_{\lambda,N},U_{\lambda,h}^{(\leq N)}]=0$. 
\end{proof}

\subsubsection{Approximate intertwining for creation operators}
%\subsubsection{Approximate exchange symmetry of the excitations}
From this subsubsection, we will evaluate errors originating by deviation of commutators between $\cre$ and $\ann$ from the canonical commutation relation. We first establish two facts on tensor norms that will be repeatedly used.  

\begin{lemma}[Insertion into tensor slots]\label{lem:insertion_tensor_slot}
    Let $T:U_1\otimes\cdots\otimes U_r\to V$ be a bounded linear map between finite-dimensional Hilbert spaces. For arbitrary nonzero Hilbert spaces $\mathcal{H}_L,\mathcal{H}_R$, the map $I_{\mathcal{H}_L}\otimes T\otimes I_{\mathcal{H}_R}$ has operator norm $\|T\|$. The same conclusion also holds if tensor slots are permuted before or after applying $T$.  
\end{lemma}
\begin{proof}
    Let $\{f_\mu\}$ be an orthonormal basis of $\mathcal{H}_L\otimes \mathcal{H}_R$. After a unitary permutation of the slots, write a vector as $v=\sum_\mu f_\mu \otimes u_\mu$ with $u_\mu\in U_1\otimes\cdots\otimes U_r$. Then, 
    \begin{align}
        \|(I_{\mathcal{H}_L}\otimes T\otimes I_{\mathcal{H}_R})v\|^2=\sum_\mu \|Tu_\mu\|^2\leq \|T\|^2\sum_\mu\|u_\mu\|^2=\|T\|^2\|v\|^2.\label{eq:definition_symword_cutoff}
    \end{align}
    For a unit vector $u\in U_1\otimes\cdots\otimes U_r$ such that $\|Tu\|=\|T\|$, $v=f_L\otimes u\otimes f_R$ with unit vectors $f_L,f_R$ attains the equality. Therefore, $\|I_{\mathcal{H}_L}\otimes T\otimes I_{\mathcal{H}_R}\|=\|T\|$. The norm is invariant under slot permutations since they are unitary.
\end{proof}

\begin{lemma}[Tensorization of multilinear estimates]\label{lem:tensorization_multilin}
    Let $U_1,\ldots,U_r$ be finite-dimensional Hilbert spaces and let $E,H$ be Hilbert spaces. Suppose that $T:U_1\times \cdots\times U_r\times E\to H$ is multilinear and satisfies
    \begin{align}
        \|T(u_1,\ldots,u_r,\xi)\|\leq C_T\left(\prod_{i=1}^r\|u_i\|\right)\|\xi\|
    \end{align}
    with a constant $C_T$. Then its linear extension $\widetilde{T}:U_1\otimes \cdots\otimes U_r\otimes E\to H$ satisfies $ \left\|\widetilde{T}\right\|\leq C_T\sqrt{\prod_{i=1}^r \dim U_i}$. 
\end{lemma}
\begin{proof}
    Take an orthonormal basis $\{e_j^{(i)}\}_j$ of $U_i$. Then, $f_\mu\coloneqq e_{\mu_1}^{(1)}\otimes \cdots\otimes e_{\mu_r}^{(r)}$ with $\mu=(\mu_1,\ldots,\mu_r)$ forms an orthonormal basis of $U_1\otimes \cdots \otimes U_r$. Note that the number of such labels is $ \prod_{i=1}^r\dim U_i$. Since $\widetilde{T}(f_\mu\otimes \xi_\mu)=T( e_{\mu_1}^{(1)},\ldots, e_{\mu_r}^{(r)},\xi_\mu)$,
    \begin{align}
        \left\|\widetilde{T}(f_\mu\otimes \xi_\mu)\right\|=\left\|T( e_{\mu_1}^{(1)},\ldots, e_{\mu_r}^{(r)},\xi_\mu)\right\|\leq C_T\|\xi_\mu\|,
    \end{align}
    where we have used $\|e_{\mu_i}^{(i)}\|=1$. Thus, for $v=\sum_\mu f_\mu\otimes \xi_\mu$, 
    \begin{align}
        \left\|\widetilde{T}v\right\|\leq \sum_\mu \left\|\widetilde{T}(f_\mu\otimes \xi_\mu)\right\|\leq C_T\sum_\mu\|\xi_\mu\|\leq C_T\sqrt{\prod_{i=1}^r \dim U_i}\|v\|,
    \end{align}
    where we used the Cauchy--Schwarz inequality in the last inequality. Thus, $ \left\|\widetilde{T}\right\|\leq C_T\sqrt{\prod_{i=1}^r \dim U_i}$. 
\end{proof}

Let $D\coloneqq \dim\mathcal{K}_\hor$. The bilinear maps $\Xi_1,\Xi_2,\Xi_3$ induce tensor-slot insertions of norm at most $DC_\Xi$, and the trilinear maps $\Xi_4$, $\Xi_5$ induce tensor-slot insertions of norm at most $D^{3/2}C_\Xi$. For later use, define a single structural constant that bounds all insertion norms by
\begin{align}
    \cins\coloneqq \max\{DC_\Xi, D^{3/2}C_\Xi,\sqrt{D},D\}\label{eq:definition_C_ins}.
\end{align}

By using these lemmas, we now show that the multiparticle excitations are approximately symmetric under permutations of the creation operators and derive the corresponding approximate intertwining of creation operators.

For a given fixed cutoff $N\geq 1$, we define
\begin{align}
     \Delta_N\coloneqq \left\|\symword_{\lambda,\leq N}^\dag \symword_{\lambda,\leq N}-I\right\|.
\end{align}
Also, for $m\geq 0$, we define $q_m\coloneqq \max_{0\leq k \leq m}\|\word_{\lambda,k}\|$. 
These quantities are finite since the spaces involved are finite-dimensional. We first prove the following:

\begin{lemma}\label{lem:norm_adjacent}
    Let $U_{\tau_i}$ be the unitary on $\mathcal{K}_{\hor}^{\otimes k}$ that exchanges $i$-th and $(i+1)$-st tensor slots.
    For $1\leq i\leq k-1$ and $1\leq k \leq N$,
    \begin{align}
        \|\word_{\lambda,k}-\word_{\lambda,k}(U_{\tau_i}\otimes I_{E_\lambda})\|\leq \cins n^{-1/2}k^{-1/2}q_{k-1}. 
    \end{align}
\end{lemma}

\begin{proof}
    On a decomposable tensor, 
    \begin{align}
        \left(\word_{\lambda,k}-\word_{\lambda,k}(U_{\tau_i}\otimes I_{E_\lambda})\right)(x_1\otimes \cdots \otimes x_k\otimes \xi)&=\frac{1}{\sqrt{k!}}\cre(x_1)\cdots[\cre(x_i),\cre(x_{i+1})]\cdots \cre(x_k)\xi\\
        &=\frac{n^{-1/2}}{\sqrt{k!}}\cre(x_1)\cdots \cre(x_{i-1})\cre(\Xi_1(x_i,x_{i+1}))\cre(x_{i+2})\cdots \cre(x_k)\xi\\
        &=\frac{n^{-1/2}}{\sqrt{k}}\word_{\lambda,k-1}\circ f_i(x_1\otimes \cdots \otimes x_k\otimes \xi),
    \end{align}
    where $f_i:\mathcal{K}_\hor^{\otimes k}\otimes E_\lambda \to \mathcal{K}_\hor^{\otimes (k-1)}\otimes E_\lambda$ is the map replacing the adjacent pair $x_i\otimes x_{i+1}$ by $\Xi_1(x_i\otimes x_{i+1})$, i.e., 
    \begin{align}
        f_i(x_1\otimes \cdots \otimes x_k\otimes \xi)\coloneqq x_1\otimes \cdots \otimes x_{i-1}\otimes \Xi_1(x_i,x_{i+1})\otimes x_{i+2}\otimes \cdots \otimes x_k\otimes\xi.
    \end{align}
    Since $\Xi_1:\mathcal{K}_\hor\times \mathcal{K}_\hor\to \mathcal{K}_\hor$ is a bilinear map, there is a unique linear extension $\tilde{\Xi}_1:\mathcal{K}_\hor\otimes \mathcal{K}_\hor\to \mathcal{K}_\hor$, satisfying $\Xi_1(x,y)=\tilde{\Xi}_1(x\otimes y)$. Since $f_i=I_{\mathcal{K}_\hor^{\otimes i-1}}\otimes\tilde{\Xi}_1\otimes I_{\mathcal{K}_\hor^{\otimes k-i-1}\otimes E_\lambda}$, we get
    \begin{align}
        \|f_i\|\overset{\text{Lemma}~\ref{lem:insertion_tensor_slot}}{=}\|\tilde{\Xi}_1\|\overset{\text{Lemma}~\ref{lem:tensorization_multilin}}{\leq} DC_\Xi \overset{\text{Eq.}~\eqref{eq:definition_C_ins}}{\leq }\cins.
    \end{align}
    Therefore, we obtain $\| \left(\word_{\lambda,k}-\word_{\lambda,k}(U_{\tau_i}\otimes I_{E_\lambda})\right)\|\leq n^{-1/2}k^{-1/2}\|\word_{\lambda,k-1}\|\|f_i\|\leq \cins n^{-1/2}k^{-1/2}q_{k-1}$. 
\end{proof}

As an immediate consequence, we obtain the following bound regarding symmetrization:
\begin{lemma}\label{lem:sym_word_bound}
    Let $P_\sym^{(k)}\coloneqq \frac{1}{k!}\sum_{\sigma\in S_k}U_\sigma$. For $1\leq k \leq N$, 
    \begin{align}
        \|\word_{\lambda,k}-\word_{\lambda,k}(P_{\sym}^{(k)}\otimes I_{E_\lambda})\|\leq  \frac{\cins}{2}n^{-1/2}k^{3/2}q_{k-1}. 
    \end{align}
\end{lemma}
\begin{proof}
    Let $\tau_i\in S_k$ be the adjacent exchange between $i$th and $(i+1)$-st slot. 
    Any permutation $\sigma\in S_k$ is a product of at most $k(k-1)/2$ adjacent transpositions, and therefore, $U_\sigma=U_{\tau_{i_1}}\cdots U_{\tau_{i_l}}$ with $l\leq k(k-1)/2$. Defining $V_m\coloneqq U_{\tau_{i_m}}\cdots U_{\tau_{i_l}}$ for $m=1,\ldots,l$ with $V_{l+1}\coloneqq I$, we have $V_m=U_{\tau_{i_m}}V_{m+1}$ and hence
    \begin{align}
        \word_{\lambda,k}-\word_{\lambda,k}(U_\sigma\otimes I_{E_\lambda})&=\word_{\lambda,k}(V_{l+1}\otimes I_{E_\lambda})-\word_{\lambda,k}(V_{1}\otimes I_{E_\lambda})\\
        &=\sum_{m=1}^l(\word_{\lambda,k}(V_{m+1}\otimes I_{E_\lambda})-\word_{\lambda,k}(V_{m}\otimes I_{E_\lambda}))\\
        &=\sum_{m=1}^l(\word_{\lambda,k}-\word_{\lambda,k}(U_{\tau_{i_m}}\otimes I_{E_\lambda}))V_{m+1}\otimes I_{E_\lambda}.
    \end{align}
    Since $V_{m+1}\otimes I_{E_\lambda}$ is unitary, $\|(\word_{\lambda,k}-\word_{\lambda,k}(U_{\tau_{i_m}}\otimes I_{E_\lambda}))V_{m+1}\otimes I_{E_\lambda}\|=\|\word_{\lambda,k}-\word_{\lambda,k}(U_{\tau_{i_m}}\otimes I_{E_\lambda})\|$ and hence
    \begin{align}
         \|\word_{\lambda,k}-\word_{\lambda,k}(U_\sigma\otimes I_{E_\lambda})\|
         &\leq \sum_{m=1}^l\|\word_{\lambda,k}-\word_{\lambda,k}(U_{\tau_{i_m}}\otimes I_{E_\lambda})\|\\
         &\leq \frac{k(k-1)}{2}\cins n^{-1/2}k^{-1/2}q_{k-1}\leq \frac{\cins}{2}n^{-1/2}k^{3/2}q_{k-1}.
    \end{align}
    Therefore,
    \begin{align}
        \|\word_{\lambda,k}-\word_{\lambda,k}(P_{\sym}^{(k)}\otimes I_{E_\lambda})\|\leq \frac{1}{k!}\sum_{\sigma\in S_k}\|\word_{\lambda,k}-\word_{\lambda,k}(U_\sigma\otimes I_{E_\lambda})\|\leq \frac{\cins}{2}n^{-1/2}k^{3/2}q_{k-1}.
    \end{align}
\end{proof}

By using this lemma, we now prove approximate intertwining for creation operators.

\begin{proposition}\label{prop:intw_cre}
    For $1\leq k \leq N$ and $1\leq j \leq D$, there is an operator $\mathsf{E}_{j,k}^+:\Sym^{k-1}\mathcal{K}_\hor\otimes E_\lambda\to H_\lambda$ such that
    \begin{align}
        \crej \symword_{\lambda,k-1}=\symword_{\lambda,k}a_j^\dag+\mathsf{E}_{j,k}^+,\qquad \|\mathsf{E}_{j,k}^+\|\leq  C_+n^{-1/2}k^{2}q_{N},
    \end{align}
    for some structural constant $C_+$.
\end{proposition}
\begin{proof}
    For $u\in\Sym^{k-1}\mathcal{K}_\hor\otimes E_\lambda$, the Fock creation operator satisfies $a_j^\dag u=\sqrt{k}(P_\sym^{(k)}\otimes I_{E_\lambda})(e_j\otimes u)$. Since $a_j^\dag u\in \Sym^{k}\otimes E_\lambda$, we have 
    \begin{align}
        \symword_{\lambda,k}a_j^\dag u= \word_{\lambda,k}a_j^\dag u=\sqrt{k}\word_{\lambda,k}(P_\sym^{(k)}\otimes I_{E_\lambda})(e_j\otimes u).
    \end{align}
    On the other hand, for a decomposable vector $v= x_1\otimes \cdots\otimes x_{k-1}\otimes \xi\in\mathcal{K}_\hor^{\otimes (k-1)}\otimes E_\lambda$, 
    \begin{align}
        \word_{\lambda,k}(e_j\otimes v)=\frac{1}{\sqrt{k!}}\cre(e_j)\cre(x_1)\cdots \cre(x_{k-1})\xi=\frac{1}{\sqrt{k}}\crej\word_{\lambda,k-1}(v).
    \end{align}
    Since decomposable vectors spans $\mathcal{K}_\hor^{\otimes (k-1)}\otimes E_\lambda$, $\crej\word_{\lambda,k-1}=\sqrt{k}\word_{\lambda,k}L_j$, where $L_jw\coloneqq e_j\otimes w$ for any $w\in\mathcal{K}_\hor^{\otimes (k-1)}\otimes E_\lambda$. Since $\symword_{\lambda,k-1}u=\word_{\lambda,k-1}u$ for $u\in\Sym^{k-1}\mathcal{K}_{\hor}\otimes E_\lambda$, 
    \begin{align}
        \crej\symword_{\lambda,k-1}u=\sqrt{k}\word_{\lambda,k}(e_j\otimes u).
    \end{align}

    Therefore,
    \begin{align}
        \mathsf{E}_{j,k}^+u=\crej\symword_{\lambda,k-1}u- \symword_{\lambda,k}a_j^\dag u=\sqrt{k}\word_{\lambda,k}(I_{\mathcal{K}_\hor^{\otimes k}\otimes E_\lambda}-P_\sym^{(k)}\otimes I_{E_\lambda})(e_j\otimes u).
    \end{align}
    From Lemma~\ref{lem:sym_word_bound}, $\|\word_{\lambda,k}(I_{\mathcal{K}_\hor^{\otimes k}\otimes E_\lambda}-P_\sym^{(k)}\otimes I_{E_\lambda})(e_j\otimes u)\|\leq   \frac{\cins}{2}n^{-1/2}k^{3/2}q_{k-1}\|u\|$ and therefore,
    \begin{align}
        \|\mathsf{E}_{j,k}^+u\|\leq \frac{\cins}{2}n^{-1/2}k^{2}q_{k-1}\|u\|\leq \frac{\cins}{2}n^{-1/2}k^{2}q_{N}\|u\|.
    \end{align}
\end{proof}

If $\|\mathsf{E}_{j,k}^+\|\approx0$, then this proposition shows the approximate intertwining relation $\crej \symword_{\lambda,k-1}\approx\symword_{\lambda,k}a_j^\dag$. The following proposition provides an upper bound on $q_N$ that appears in the upper bound on $\|\mathsf{E}_{j,k}^+\|$. 

\begin{proposition}\label{prop:cins_bound}
    If $ \cins n^{-1/2}N^{3/2}\leq 1$, then $q_N\leq 2\sqrt{1+\Delta_N}$.  
\end{proposition}
\begin{proof}
    From Lemma~\ref{lem:sym_word_bound}, for $1\leq k\leq N$, 
    \begin{align}
        \|\word_{\lambda,k}\|\leq \|\word_{\lambda,k}(P_{\sym}^{(k)}\otimes I_{E_\lambda})\|+\frac{\cins}{2}n^{-1/2}k^{3/2}q_{k-1}\leq \|\symword_{\lambda,k}\|+\frac{\cins}{2}n^{-1/2}N^{3/2}q_{N}.
    \end{align}
    Let $Q_k$ be the projector onto the $k$-particle sector. Since $Q_k(\symword_{\lambda,\leq N}^\dag \symword_{\lambda,\leq N} -I)Q_k=\symword_{\lambda,k}^\dag \symword_{\lambda,k}-I_k$, and the projector $Q_k$ does not increase the operator norm, we obtain
    \begin{align}
        \|\symword_{\lambda,k}^\dag \symword_{\lambda,k}-I_k\|\leq \|\symword_{\lambda,\leq N}^\dag \symword_{\lambda,\leq N} -I\|=\Delta_N.
    \end{align}
    Therefore,
    \begin{align}
        \|\symword_{\lambda,k}\|^2=\|\symword_{\lambda,k}^\dag \symword_{\lambda,k}\|\leq \|I_k\|+\|\symword_{\lambda,k}^\dag \symword_{\lambda,k}-I_k\|=1+\Delta_N,
    \end{align}
    which implies
    \begin{align}
         \|\word_{\lambda,k}\|\leq \sqrt{1+\Delta_N}+\frac{\cins}{2}n^{-1/2}N^{3/2}q_{N}.
    \end{align}
    Moreover, for $k=0$, since $\word_{\lambda,0}=I_{E_\lambda}$, $\|\word_{\lambda,0}\|=1\leq \sqrt{1+\Delta_N}$. Thus, taking the maximum over $0\leq k \leq N$, we obtain
    \begin{align}
        q_N\leq  \sqrt{1+\Delta_N}+\frac{\cins}{2}n^{-1/2}N^{3/2}q_{N}.
    \end{align}
    From the assumption, $\frac{\cins}{2}n^{-1/2}N^{3/2}\leq \frac{1}{2}$. Therefore, $q_N\leq  \sqrt{1+\Delta_N}+\frac{1}{2}q_{N}$, which implies $q_N\leq 2\sqrt{1+\Delta_N}$. 
\end{proof}

\subsubsection{Approximate intertwining for annihilation operators}
This subsubsection derives approximate intertwining for annihilation operators. We first introduce several notations. From Theorem~\ref{thm:commutators_multilin_maps}, 
\begin{align}
    \|\Dmat(\bar{u},x)|_{F_kH_\lambda}\|\leq C_{\mathsf{D}}\|u\|\|x\|\varepsilon_0(k),\quad \varepsilon_0(k)\coloneqq n^{\alpha-1}+\frac{k}{n}
\end{align}
for some constant $C_{\mathsf{D}}>0$. Introducing $\tilde{\varepsilon}_0(k)\coloneqq C_{\mathsf{D}}\varepsilon_0(k)$, 
\begin{align}
    \|\Dmat(\bar{u},x)|_{F_mH_\lambda}\|\leq \|u\|\|x\|\tilde{\varepsilon}_0(k)\quad \text{for } 0\leq m \leq k.\label{eq:dmax_norm_fm}
\end{align}

Throughout this subsubsection, we impose the standing smallness assumptions
\begin{align}
    \max_{1\leq k\leq N}\{\tilde{\varepsilon}_0(k)+k(n^{-1}+n^{-1/2})\}\leq C_{\mathrm{sc}},\quad \cins n^{-1/2}\sqrt{N}\leq \frac{1}{2}.\label{eq:standing_assumption}
\end{align}
Here, $C_{\mathrm{sc}}$ is a fixed structural constant. Since $\alpha<1$, these conditions hold for $N=\floor{n^\eta}$ with $\eta<1/2$ and all sufficiently large $n$. 

For $1\leq k \leq N$, we define the unsymmetrized Fock contraction $a_j^{[k]}:\mathcal{K}_\hor^{\otimes k} \otimes E_\lambda\to \mathcal{K}_\hor^{\otimes (k-1)}\otimes E_\lambda$ by
\begin{align}
    a_j^{[k]}(x_1\otimes\cdots\otimes x_k\otimes \xi)\coloneqq \frac{1}{\sqrt{k}}\sum_{i=1}^k\braket{e_j,x_i}x_1\otimes \cdots \widehat{x_i}\cdots \otimes x_k\otimes\xi,
\end{align}
whose restriction on $\Sym^k\mathcal{K}_\hor\otimes E_\lambda$ is $a_j\otimes I_{E_\lambda}$. 

For fixed $j$ and $1\leq k \leq N$, we introduce the three different branch maps on decomposable tensors by
\begin{align}
    \mathsf{B}^{\mathrm{cre}}_{j,k}(x_1\otimes \cdots\otimes x_k\otimes \xi)&\coloneqq \frac{1}{\sqrt{k!}}\sum_{i=1}^k\cre(x_1)\cdots \cre(x_{i-1})\cre(\Xi_2(\bar{e}_j,x_i))\cre(x_{i+1})\cdots \cre(x_k)\xi,\\
    \mathsf{B}^{\mathrm{diag}}_{j,k}(x_1\otimes \cdots\otimes x_k\otimes \xi)&\coloneqq \frac{1}{\sqrt{k!}}\sum_{i=1}^k\cre(x_1)\cdots \cre(x_{i-1})\Dmat(\bar{e}_j,x_i)\cre(x_{i+1})\cdots \cre(x_k)\xi,\\
    \mathsf{B}^{\mathrm{ann}}_{j,k}(x_1\otimes \cdots\otimes x_k\otimes \xi)&\coloneqq \frac{1}{\sqrt{k!}}\sum_{i=1}^k\cre(x_1)\cdots \cre(x_{i-1})\annlin_\lambda(\Xi_3(\bar{e}_j,x_i))\cre(x_{i+1})\cdots \cre(x_k)\xi.
\end{align}
Each map extends uniquely and linearly to $\mathcal{K}_\hor^{\otimes k}\otimes E_\lambda$. 

For a decomposable tensor $v\coloneqq x_1\otimes \cdots\otimes x_k\otimes \xi$, 
\begin{align}
    &\annj\word_{\lambda,k}v\\
    &=\frac{1}{\sqrt{k!}}\sum_{i=1}^k\cre(x_1)\cdots \cre(x_{i-1})[\annj,\cre(x_i)]\cre (x_{i+1})\cdots \cre(x_k)\xi\\
    &=\frac{1}{\sqrt{k!}}\sum_{i=1}^k\cre(x_1)\cdots \cre(x_{i-1})\left(\braket{e_j,x_i}I+\Dmat(\bar{e}_j,x_i)+n^{-1/2}\cre(\Xi_2(\bar{e}_j,x_i))+n^{-1/2}\annlin_\lambda(\Xi_3(\bar{e}_j,x_i))\right)\cre (x_{i+1})\cdots \cre(x_k)\xi\\
    &=\word_{\lambda,k-1}a_j^{[k]}v+\mathsf{B}^{\mathrm{diag}}_{j,k}v+n^{-1/2}\mathsf{B}^{\mathrm{cre}}_{j,k}v+n^{-1/2}\mathsf{B}^{\mathrm{ann}}_{j,k}v.\label{eq:aw_wa}
\end{align}
Below we derive bounds on the norms of the maps $\mathsf{B}^{\mathrm{cre}}_{j,k}, \mathsf{B}^{\mathrm{diag}}_{j,k},\mathsf{B}^{\mathrm{ann}}_{j,k}$. 

\begin{lemma}\label{lem:Bcre}
    For $1\leq k \leq N$, $\|\mathsf{B}^{\mathrm{cre}}_{j,k}\|\leq \cins k q_N$.
\end{lemma}

\begin{proof}
    Define $L_j:\mathcal{K}_\hor\to \mathcal{K}_\hor$ by $L_jx\coloneqq \Xi_2(\bar{e}_j,x)$. Since $e_j$ is a unit vector, $\|L_jx\|\leq \| \Xi_2(\bar{e}_j,x)\|\leq C_\Xi\|x\|$, and hence $\|L_j\|\leq C_\Xi$. Defining $f_{i,j}:\mathcal{K}_\hor^{\otimes k}\otimes E_\lambda\to \mathcal{K}_\hor^{\otimes k}\otimes E_\lambda$ by $f_{i,j}\coloneqq I_{\mathcal{K}_\hor^{\otimes (i-1)}}\otimes L_j\otimes I_{\mathcal{K}_\hor^{\otimes (k-i)}\otimes E_\lambda}$, from Lemma~\ref{lem:insertion_tensor_slot}, we get $ \|f_{i,j}\|=\|L_j\|$, and hence $ \|f_{i,j}\|\leq C_\Xi$. 
    Since $\mathsf{B}^{\mathrm{cre}}_{j,k}=\sum_{i=1}^k\word_{\lambda,k}\circ f_{i,j}$, we get
    \begin{align}
        \|\mathsf{B}^{\mathrm{cre}}_{j,k}\|\leq \sum_{i=1}^k\|\word_{\lambda,k}\|\| f_{i,j}\|\leq  \sum_{i=1}^kq_N C_\Xi= C_\Xi kq_N. 
    \end{align}
    Since $ C_\Xi\leq \cins$, we get $\|\mathsf{B}^{\mathrm{cre}}_{j,k}\|\leq \cins k q_N$.
\end{proof}

\begin{lemma}\label{lem:Bdiag}
    For $1\leq k \leq N$, $\|\mathsf{B}^{\mathrm{diag}}_{j,k}\|\leq \cins (\tilde{\varepsilon}_0(k)\sqrt{k}+n^{-1}k^{3/2})q_N\leq \cins (\tilde{\varepsilon}_0(k)k+n^{-1}k^{2})q_N$. 
\end{lemma}
\begin{proof}
    By repeatedly using Eq.~\eqref{eq:comm_D_C},
    \begin{align}
        &\cre(x_1)\cdots \cre(x_{i-1})\Dmat(\bar{e}_j,x_i)\cre(x_{i+1})\cdots \cre(x_k)\xi\\
        &=\cre(x_1)\cdots \widehat{\cre(x_{i})}\cdots \cre(x_k)\Dmat(\bar{e}_j,x_i)\xi+n^{-1}\sum_{l=i+1}^k\cre(x_1)\cdots \widehat{\cre(x_{i})}\cdots\cre(x_{l-1}) \cre(\Xi_4(\bar{e}_j,x_i,x_l))\cre(x_{l+1})\cdots \cre (x_k)\xi.
    \end{align}
    
   We first consider the first contribution. From Eq.~\eqref{eq:dmax_norm_fm}, $\|\Dmat(\bar{e}_j,x_i)\xi\|\leq \|x_i\|\|\xi\|\tilde{\varepsilon}_0(k)$. Since $\Dmat(\bar{e}_j,x_i)$ is block-diagonal, which preserves each filtration space $F_kH_\lambda$ including $E_\lambda=F_0H_\lambda$, we have $\Dmat(\bar{e}_j,x_i)\xi\in E_\lambda$. Defining a bilinear map $T:\mathcal{K}_\hor\times E_\lambda\to E_\lambda$ by $T(x,\xi)\coloneqq \Dmat(\bar{e}_j,x)\xi$, we have $\|T(x,\xi)\|\leq \tilde{\varepsilon}_0(k)\|x\|\|\xi\|$. Let $\tilde{T}:\mathcal{K}_\hor\otimes E_\lambda\to E_\lambda$ be the unique linear extension of $T$. From Lemma~\ref{lem:tensorization_multilin}, $\|\tilde{T}\|\leq \sqrt{D} \tilde{\varepsilon}_0(k)$. Defining a unitary permutation $P_i:\mathcal{K}_\hor^{\otimes k}\otimes E_\lambda\to \mathcal{K}_\hor^{\otimes (k-1)}\otimes (\mathcal{K}_\hor\otimes E_\lambda)$ by $ P_i(x_1\otimes\cdots \otimes x_k\otimes \xi)\coloneqq x_1\otimes \cdots \widehat{x_i}\cdots\otimes x_k\otimes x_i\otimes \xi$, we get
   \begin{align}
       \frac{1}{\sqrt{k!}}\cre(x_1)\cdots \widehat{\cre(x_{i})}\cdots \cre(x_k)\Dmat(\bar{e}_j,x_i)\xi=k^{-1/2}\left(\word_{\lambda,k-1}\circ(I_{\mathcal{K}_\hor^{\otimes (k-1)}}\otimes \tilde{T})\circ P_i\right)(x_1\otimes \cdots\otimes x_k\otimes\xi).
   \end{align}
   By using Lemma~\ref{lem:insertion_tensor_slot}, $\|I_{\mathcal{K}_\hor^{\otimes (k-1)}}\otimes \tilde{T}\|=\|\tilde{T}\|$, and hence
   \begin{align}
       \left\|\word_{\lambda,k-1}\circ(I_{\mathcal{K}_\hor^{\otimes (k-1)}}\otimes \tilde{T})\circ P_i\right\|\leq q_N\sqrt{D} \tilde{\varepsilon}_0(k)\overset{\text{Eq.~\eqref{eq:definition_C_ins}}}{\leq} \cins q_N\tilde{\varepsilon}_0(k)
   \end{align}

    Next, we consider the commutator terms. Define a bilinear map $T':\mathcal{K}_\hor\times \mathcal{K}_\hor\to \mathcal{K}_\hor$ by $T'(x_i,x_l)\coloneqq \Xi_4(\bar{e}_j,x_i,x_l)$. Since $\|\Xi_4(\bar{e}_j,x_i,x_l)\|\leq C_\Xi\|x_i\|\|x_l\|$, its linear extension $\tilde{T'}:\mathcal{K}_\hor\otimes \mathcal{K}_\hor\to \mathcal{K}_\hor$ satisfies $\|\tilde{T'}\|\leq DC_\Xi$ by Lemma~\ref{lem:tensorization_multilin}. Defining a unitary permutation $P_{i,l}:\mathcal{K}_\hor^{\otimes k}\otimes E_\lambda \to \mathcal{K}_\hor^{\otimes k}\otimes E_\lambda$ by
    \begin{align}
        P_{i,l}(x_1\otimes \cdots\otimes x_k\otimes \xi)\coloneqq x_1\otimes \cdots \widehat{x_i}\cdots \otimes x_{l-1}\otimes (x_i\otimes x_l)\otimes x_{l+1}\otimes \cdots \otimes x_k\otimes \xi,
    \end{align}
    we have
    \begin{align}
        &\frac{1}{\sqrt{k!}}\cre(x_1)\cdots \widehat{\cre(x_{i})}\cdots\cre(x_{l-1}) \cre(\Xi_4(\bar{e}_j,x_i,x_l))\cre(x_{l+1})\cdots \cre (x_k)\xi\\
        &=k^{-1/2}\left(\word_{\lambda,k-1}\circ( I_{\mathcal{K}_\hor^{\otimes (l-2)}}\otimes \tilde{T'}\otimes I_{\mathcal{K}_\hor^{\otimes (k-l)}\otimes E_\lambda})\circ P_{i,l} \right)(x_1\otimes \cdots\otimes x_k\otimes \xi).
    \end{align}
    Again, by using Lemma~\ref{lem:insertion_tensor_slot}, $\|I_{\mathcal{K}_\hor^{\otimes (l-2)}}\otimes \tilde{T'}\otimes I_{\mathcal{K}_\hor^{\otimes (k-l)}\otimes E_\lambda}\|=\| \tilde{T'}\| $, and hence
    \begin{align}
        \left\|\word_{\lambda,k-1}\circ( I_{\mathcal{K}_\hor^{\otimes (l-2)}}\otimes \tilde{T'}\otimes I_{\mathcal{K}_\hor^{\otimes (k-l)}\otimes E_\lambda})\circ P_{i,l} \right\|\leq q_NDC_\Xi\overset{\text{Eq.~\eqref{eq:definition_C_ins}}}{\leq}\cins q_N.
    \end{align}
    
    To summarize, we obtain
    \begin{align}
        \mathsf{B}^{\mathrm{diag}}_{j,k}=\sum_{i=1}^kk^{-1/2}\left(\word_{\lambda,k-1}\circ(I_{\mathcal{K}_\hor^{\otimes (k-1)}}\otimes \tilde{T})\circ P_i\right)+n^{-1}\sum_{1\leq i < l \leq k}k^{-1/2}\left(\word_{\lambda,k-1}\circ( I_{\mathcal{K}_\hor^{\otimes (l-2)}}\otimes \tilde{T'}\otimes I_{\mathcal{K}_\hor^{\otimes (k-l)}\otimes E_\lambda})\circ P_{i,l} \right).
    \end{align}
    Therefore, using $\binom{k}{2}\leq k^2/2\leq k^2$, we get
    \begin{align}
         \|\mathsf{B}^{\mathrm{diag}}_{j,k}\|&\leq k k^{-1/2}q_N\sqrt{D} \tilde{\varepsilon}_0(k)+\binom{k}{2}n^{-1}k^{-1/2}q_NDC_\Xi\leq \cins\left(\sqrt{k} \tilde{\varepsilon}_0(k)+n^{-1}k^{3/2}\right)q_N
    \end{align}
\end{proof}

\begin{lemma}\label{lem:Bann}
    Assume that
    \begin{align}
        \max_{1\leq k \leq N}\{\tilde{\varepsilon}_0(k)+kn^{-1}+k^{1/2}n^{-1/2}\}\leq C_{\mathrm{sc}},\quad n^{-1/2}\sqrt{N}\leq \frac{1}{2\cins}.\label{eq:assumptions_Bann}
    \end{align}
    Then
    \begin{align}
        \|\mathsf{B}^{\mathrm{ann}}_{j,k}\|\leq 2\cins^2(1+C_{\mathrm{sc}})kq_N.
    \end{align}
\end{lemma}
\begin{proof}
    For $0\leq r\leq s\leq N$, we define $\mathsf{M}_{r,s}:\overline{\mathcal{K}_\hor}\otimes \mathcal{K}_\hor^{\otimes s}\otimes E_\lambda \to H_\lambda$ on decomposable tensors by
    \begin{align}
        \mathsf{M}_{r,s}(\bar{u}\otimes x_1\otimes \cdots \otimes x_s\otimes \xi)\coloneqq \frac{1}{\sqrt{s!}}\cre(x_1)\cdots\cre(x_r)\annlin_\lambda(\bar{u})\cre (x_{r+1})\cdots \cre(x_s)\xi,
    \end{align}
    and set $m_s\coloneqq \max_{0\leq r\leq s}\|\mathsf{M}_{r,s}\|$. 

    We first prove 
    \begin{align}
        m_s\leq C_{\mathsf{M}} \sqrt{s} q_N\qquad C_{\mathsf{M}} \coloneqq 2\cins(1+C_{\mathrm{sc}})\label{eq:m_s_bound_ind}
    \end{align}
    for $0\leq s \leq N$ by induction, where $C_{\mathsf{M}}$ is a structural constant. For $s=0$, $m_s=0$ since $\mathsf{M}_{0,0}(\bar{u}\otimes \xi)=\annlin_\lambda(\bar{u})\xi=0$, and hence Eq.~\eqref{eq:m_s_bound_ind} trivially holds. Assume that Eq.~\eqref{eq:m_s_bound_ind} holds at $s-1$ with $s\geq 1$. For $0\leq r\leq s$, we have
    \begin{align}
        &\cre(x_1)\cdots\cre(x_r)\annlin_\lambda(\bar{u})\cre (x_{r+1})\cdots \cre(x_s)\xi\\
        &=\sum_{i=r+1}^s\cre(x_1)\cdots\cre(x_{i-1})[\annlin_\lambda(\bar{u}),\cre (x_{i})]\cre(x_{i+1})\cdots \cre(x_s)\xi\\
        &=\sum_{i=r+1}^s\cre(x_1)\cdots\cre(x_{i-1})\left(\braket{u,x_i}I+\Dmat(\bar{u},x_i)+n^{-1/2}\cre(\Xi_2(\bar{u},x_i))+n^{-1/2}\annlin_\lambda(\Xi_3(\bar{u},x_i))\right)\cre(x_{i+1})\cdots \cre(x_s)\xi.
    \end{align}
    We derive bound for each contributions.

    \textit{Scalar branch:} Define $f^{\mathrm{sc}}_i:\overline{\mathcal{K}}_\hor\otimes \mathcal{K}_\hor^{\otimes s}\otimes E_\lambda\to \mathcal{K}_\hor^{\otimes (s-1)}\otimes E_\lambda$ by $f^{\mathrm{sc}}_i(\bar{u}\otimes x_1\otimes\cdots\otimes x_s\otimes \xi)\coloneqq \braket{u,x_i} x_1\otimes\cdots\otimes\widehat{x_i}\cdots\otimes x_s\otimes \xi$. Then, the scalar branch is expressed as
    \begin{align}
        \frac{1}{\sqrt{s!}}\sum_{i=r+1}^s\cre(x_1)\cdots\cre(x_{i-1})\left(\braket{u,x_i}I\right)\cre(x_{i+1})\cdots \cre(x_s)\xi=s^{-1/2}\sum_{i=r+1}^s\word_{\lambda,s-1}\circ f^{\mathrm{sc}}_i(\bar{u}\otimes x_1\otimes \cdots \otimes x_s\otimes \xi).
    \end{align}
    For an orthonormal basis $\{e_j\}_{j=1}^D$ of $\mathcal{K}_\hor$, any vector $v\in\overline{\mathcal{K}}_\hor\otimes \mathcal{K}_\hor$ can be expanded as $v=\sum_{jk}c_{jk}\bar{e}_j\otimes e_k$. Since the linear pairing map $p:\overline{\mathcal{K}_{\hor}}\otimes \mathcal{K}_{\hor}\to \mathbb{C}$, defined by $p(\bar{u}\otimes x)=\braket{u,x}$ satisfies
    \begin{align}
        |p(v)|=\left|\sum_{jk}c_{jk}\braket{e_j,e_k}\right|=\left|\sum_{j}c_{jj}\right|\leq \sqrt{D}\left(\sum_j|c_{jj}|^2\right)^{1/2}\leq \sqrt{D}\|v\|,
    \end{align}
    we find $\|f^{\mathrm{sc}}_i\|\leq \sqrt{D}$. Therefore, 
    \begin{align}
       \left \|s^{-1/2}\sum_{i=r+1}^s\word_{\lambda,s-1}\circ f^{\mathrm{sc}}_i\right\|\leq \sqrt{D}s^{1/2}q_N\leq \cins s^{1/2}q_N.
    \end{align}

    \textit{Block-diagonal branch:} The block-diagonal branch can be expanded as
    \begin{align}
        &\frac{1}{\sqrt{s!}}\sum_{i=r+1}^s\cre(x_1)\cdots\cre(x_{i-1})\Dmat(\bar{u},x_i)\cre(x_{i+1})\cdots \cre(x_s)\xi\\
        &=\frac{1}{\sqrt{s!}}\sum_{i=r+1}^s\cre(x_1)\cdots\widehat{\cre(x_{i})}\cdots \cre(x_s)\Dmat(\bar{u},x_i)\xi\label{eq:bd_branch_term_part}\\
        &\quad +n^{-1}\frac{1}{\sqrt{s!}}\sum_{r+1\leq i <j\leq s}\cre(x_1)\cdots\widehat{\cre(x_{i})}\cdots \cre(x_{j-1})\cre(\Xi_4(\bar{u},x_i,x_{j}))\cre(x_{j+1})\cdots \cre(x_s)\xi.\label{eq:bd_branch_comm_part}
    \end{align}
    We evaluate each part separately. 
    \begin{itemize}
        \item Terminal part (Eq.~\eqref{eq:bd_branch_term_part}): Define a trilinear map $f_i^{(\text{diag,term})}:\overline{\mathcal{K}_\hor}\times \mathcal{K}_\hor\times E_\lambda\to E_\lambda$ by $f_i^{(\text{diag,term})}(\bar{u},x,\xi)\coloneqq \Dmat(\bar{u},x)\xi$ and let $ \tilde{f}_i^{(\text{diag,term})}:\overline{\mathcal{K}_\hor}\otimes \mathcal{K}_\hor\otimes E_\lambda\to E_\lambda$ be its linear extension.
        Using a permutation unitary $U_i(\bar{u}\otimes x_1\otimes \cdots\otimes x_s\otimes\xi)\coloneqq x_1\otimes \cdots\widehat{x_i}\cdots\otimes x_s\otimes \bar{u}\otimes x_i\otimes \xi$, the terminal part is expressed as
        \begin{align}
            &\frac{1}{\sqrt{s!}}\sum_{i=r+1}^s\cre(x_1)\cdots\widehat{\cre(x_{i})}\cdots \cre(x_s)\Dmat(\bar{u},x_i)\xi\\
            &=s^{-1/2}\left(\sum_{i=r+1}^s\word_{\lambda,s-1}\circ \left(I_{\mathcal{K}_{\hor}^{\otimes (s-1)}}\otimes \tilde{f}_i^{(\text{diag,term})}\right)\circ U_i\right)(\bar{u}\otimes x_1\otimes \cdots\otimes x_s\otimes\xi).
        \end{align}
        Since $\|\Dmat(\bar{u},x)\xi\|\leq \tilde{\varepsilon}_0(s)\|u\|\|x\|\|\xi\|$, its linear extension $ \tilde{f}_i^{(\text{diag,term})}$ satisfies $\|\tilde{f}_i^{(\text{diag,term})}\|\leq D\tilde{\varepsilon}_0(s)\leq \cins \tilde{\varepsilon}_0(s)$ by Lemma~\ref{lem:tensorization_multilin}. Therefore,
        \begin{align}
            \left\|s^{-1/2}\sum_{i=r+1}^s\word_{\lambda,s-1}\circ \left(I_{\mathcal{K}_{\hor}^{\otimes (s-1)}}\otimes \tilde{f}_i^{(\text{diag,term})}\right)\circ U_i\right\|\leq \cins s^{1/2}q_N\tilde{\varepsilon}_0(s).
        \end{align}

        \item Commutator part (Eq.~\eqref{eq:bd_branch_comm_part}): Using a permutation unitary $U_{i,l}(\bar{u}\otimes x_1\otimes \cdots\otimes x_s\otimes\xi)\coloneqq x_1\otimes \cdots\widehat{x_i}\cdots\otimes x_{l-1}\otimes (\bar{u}\otimes x_i\otimes x_l)\otimes x_{l+1}\otimes \cdots\otimes x_s\otimes \xi$, we have
        \begin{align}
            &n^{-1}\frac{1}{\sqrt{s!}}\sum_{r+1\leq i <l\leq s}\cre(x_1)\cdots\widehat{\cre(x_{i})}\cdots \cre(x_{l-1})\cre(\Xi_4(\bar{u},x_i,x_{l}))\cre(x_{l+1})\cdots \cre(x_s)\xi\\
            &=n^{-1}s^{-1/2}\left(\sum_{r+1\leq i <l\leq s}\word_{\lambda, s-1}\circ (I_{\mathcal{K}_\hor^{\otimes (l-2)}}\otimes \tilde{\Xi}_4\otimes I_{\mathcal{K}_\hor^{\otimes (s-l)}\otimes E_\lambda})\circ U_{i,l}\right)(\bar{u}\otimes x_1\otimes \cdots\otimes x_s\otimes\xi),
        \end{align}
        where $ \tilde{\Xi}_4 $ denotes the linear extension of $\Xi_4$. Since $\|\tilde{\Xi}_4\|\leq D^{3/2}C_\Xi\leq \cins$, we have
        \begin{align}
            \left\|n^{-1}s^{-1/2}\left(\sum_{r+1\leq i <l\leq s}\word_{\lambda, s-1}\circ(I_{\mathcal{K}_\hor^{\otimes (l-2)}}\otimes \tilde{\Xi}_4\otimes I_{\mathcal{K}_\hor^{\otimes (s-l)}\otimes E_\lambda})\circ U_{i,l}\right)\right\|\leq \cins n^{-1}s^{3/2}q_N.
        \end{align}
    \end{itemize}
    Thus, the norm of block-diagonal branch is bounded by $\cins(\tilde{\varepsilon}_0(s)\sqrt{s}+n^{-1}s^{3/2})q_N$. 

    \textit{Creation-operator branch:} Defining a permutation unitary $U_i'(\bar{u}\otimes x_1\otimes \cdots x_s\otimes \xi)\coloneqq x_1\otimes\cdots \otimes x_{i-1}\otimes(\bar{u}\otimes x_i)\otimes x_{i+1}\otimes \cdots \otimes x_s\otimes \xi $, 
    \begin{align}
        &n^{-1/2}\frac{1}{\sqrt{s!}}\sum_{i=r+1}^s\cre(x_1)\cdots\cre(x_{i-1})\cre(\Xi_2(\bar{u},x_i))\cre(x_{i+1})\cdots \cre(x_s)\xi\\
        &=n^{-1/2}\sum_{i=r+1}^s\left(\word_{\lambda,s}\circ \left(I_{\mathcal{K}_\hor^{\otimes (i-1)}}\otimes \tilde{\Xi}_2\otimes I_{\mathcal{K}_\hor^{\otimes (s-i)}\otimes E_\lambda}\right)\circ U'_i\right)(\bar{u}\otimes x_1\otimes \cdots x_s\otimes \xi),
    \end{align}
    where $\tilde{\Xi}_2$ denotes the linear extension of $\Xi_2$. Since $\|\tilde{\Xi}_2\|\leq DC_{\Xi}\leq \cins$, we get
    \begin{align}
        n^{-1/2}\left\|\sum_{i=r+1}^s\word_{\lambda,s}\circ \left(I_{\mathcal{K}_\hor^{\otimes (i-1)}}\otimes \tilde{\Xi}_2\otimes I_{\mathcal{K}_\hor^{\otimes (s-i)}\otimes E_\lambda}\right)\circ U'_i\right\|\leq \cins n^{-1/2}sq_N.
    \end{align}

    \textit{Annihilation-operator branch:} Defining a permutation unitary $U''_i(\bar{u}\otimes x_1\otimes\cdots\otimes x_s\otimes\xi)\coloneqq \bar{u}\otimes x_i\otimes x_1\otimes\cdots \widehat{x_i}\cdots\otimes x_s\otimes \xi$, 
    \begin{align}
        &n^{-1/2}\frac{1}{\sqrt{s!}}\sum_{i=r+1}^s\cre(x_1)\cdots\cre(x_{i-1})\annlin_\lambda(\Xi_3(\bar{u},x_i))\cre(x_{i+1})\cdots \cre(x_s)\xi\\
        &=n^{-1/2}s^{-1/2}\left(\sum_{i=r+1}^s\mathsf{M}_{i-1,s-1}\circ \left(\tilde{\Xi}_3\otimes I_{\mathcal{K}_\hor^{\otimes s-1}\otimes E_\lambda}\right)\circ U''_i\right)(\bar{u}\otimes x_1\otimes\cdots\otimes x_s\otimes\xi),
    \end{align}
    where $\tilde{\Xi}_3$ is a linear extension of $\Xi_3$. Since $\|\tilde{\Xi}_3\|\leq DC_\Xi\leq \cins$, 
    \begin{align}
        n^{-1/2}s^{-1/2}\left\|\sum_{i=r+1}^s\mathsf{M}_{i-1,s-1}\circ \left(\tilde{\Xi}_3\otimes I_{\mathcal{K}_\hor^{\otimes s-1}\otimes E_\lambda}\right)\circ U''_i\right\|\leq \cins n^{-1/2}\sqrt{s}m_{s-1}.
    \end{align}

    Summarizing these four contributions, by using $s\leq N$, 
    \begin{align}
        m_s&\leq \cins (\sqrt{s}+\tilde{\varepsilon}_0(s)\sqrt{s}+n^{-1}s^{3/2}+n^{-1/2}s)q_N+\cins n^{-1/2}\sqrt{s}m_{s-1}\\
        &\leq  \cins(1 +C_{\mathrm{sc}})\sqrt{s}q_N+\cins n^{-1/2}\sqrt{N}m_{s-1}\\
        &\leq \frac{1}{2}C_{\mathsf{M}}\sqrt{s}q_N+\frac{1}{2}m_{s-1}.
    \end{align}
    From the assumption of the mathematical induction, $m_{s-1}\leq C_{\mathsf{M}}\sqrt{s-1}q_N\leq C_{\mathsf{M}}\sqrt{s}q_N$. Thus,
    \begin{align}
        m_s\leq   \frac{1}{2}C_{\mathsf{M}}\sqrt{s}q_N+\frac{1}{2}C_{\mathsf{M}}\sqrt{s}q_N=C_{\mathsf{M}}\sqrt{s}q_N.
    \end{align}

    Now, by using this inequality, we evaluate the norm of $\|\mathsf{B}^{\mathrm{ann}}_{j,k}\|$. Define $g_{i,j}:\mathcal{K}_\hor^{\otimes k}\otimes E_\lambda\to \overline{\mathcal{K}_\hor}\otimes \mathcal{K}_\hor^{\otimes (k-1)}\otimes E_\lambda$ by
    \begin{align}
        g_{i,j}(x_1\otimes\cdots\otimes x_k\otimes \xi)\coloneqq \Xi_3(\bar{e}_j,x_i)\otimes x_1\otimes \cdots \widehat{x_i}\cdots \otimes x_k\otimes \xi.
    \end{align}
    Since $\|\Xi_3(\bar{e}_j,x_i)\|\leq C_\Xi\|x_i\|\leq \cins \|x_i\|$, $\|g_{i,j}\|\leq \cins$. Since $\mathsf{B}^{\mathrm{ann}}_{j,k}=k^{-1/2}\sum_{i=1}^k\mathsf{M}_{i-1,k-1}\circ g_{i,j}$, we get
    \begin{align}
        \|\mathsf{B}^{\mathrm{ann}}_{j,k}\|\leq k^{1/2}m_{k-1}\cins\leq C_{\mathsf{M}}\cins kq_N=2\cins^2(1+C_{\mathrm{sc}})kq_N.
    \end{align}
\end{proof}

We finally obtain the following.
\begin{proposition}\label{prop:intw_ann}
    Under the assumption in Eq.~\eqref{eq:standing_assumption}, for $1\leq k \leq N$ and $1\leq j \leq D$, there exists an operator $\mathsf{E}_{j,k}^{-}:\Sym^k\mathcal{K}_\hor\otimes E_\lambda\to H_\lambda$ such that
    \begin{align}
        \annj\symword_{\lambda,k}=\symword_{\lambda,k-1}a_j+\mathsf{E}_{j,k}^{-},\quad \|\mathsf{E}_{j,k}^{-}\|\leq C_-\left(\sqrt{k} \tilde{\varepsilon}_0(k)+n^{-1}k^{3/2}+n^{-1/2}k\right)q_N\label{eq:ann_intw}
    \end{align}
    for some structural constant $C_-$.
\end{proposition}
\begin{proof}
    Since $\sqrt{k}\leq k$ for $k\geq 1$, the standing assumptions in Eq.~\eqref{eq:standing_assumption} imply the assumptions of Lemma~\ref{lem:Bann} in Eq.~\eqref{eq:assumptions_Bann}.
    
    From Eq.~\eqref{eq:aw_wa},
    \begin{align}
    \annj\word_{\lambda,k}-\word_{\lambda,k-1}a_j^{[k]}=\mathsf{B}^{\mathrm{diag}}_{j,k}+n^{-1/2}\mathsf{B}^{\mathrm{cre}}_{j,k}+n^{-1/2}\mathsf{B}^{\mathrm{ann}}_{j,k}.
    \end{align}
    From Lemmas~\ref{lem:Bcre}, \ref{lem:Bdiag} and \ref{lem:Bann}, 
    \begin{align}
        \left\|\annj\word_{\lambda,k}-\word_{\lambda,k-1}a_j^{[k]}\right\|
        &\leq \left\|\mathsf{B}^{\mathrm{diag}}_{j,k}\right\|+\left\|n^{-1/2}\mathsf{B}^{\mathrm{cre}}_{j,k}\right\|+\left\|n^{-1/2}\mathsf{B}^{\mathrm{ann}}_{j,k}\right\|\\
        &\leq \cins\left(\sqrt{k} \tilde{\varepsilon}_0(k)+n^{-1}k^{3/2}\right)q_N+n^{-1/2} \cins k q_N+ 2\cins^2(1+C_{\mathrm{sc}})n^{-1/2}kq_N\\
        &\leq C_-\left(\sqrt{k} \tilde{\varepsilon}_0(k)+n^{-1}k^{3/2}+n^{-1/2}k\right)q_N,
    \end{align}
    where we have introduced a constant $C_-\coloneqq 2\cins^2(1+C_{\mathrm{sc}})+\cins$. Since
    \begin{align}
        \annj\word_{\lambda,k}\biggl|_{\Sym^k\mathcal{K}_\hor\otimes E_\lambda}=\annj\symword_{\lambda,k},\quad \word_{\lambda,k-1}a_j^{[k]}\biggl|_{\Sym^k\mathcal{K}_\hor\otimes E_\lambda}=\symword_{\lambda,k-1}a_j,
    \end{align}
    we obtain Eq.~\eqref{eq:ann_intw} with 
    \begin{align}
        \mathsf{E}_{j,k}^{-}\coloneqq \left(\mathsf{B}^{\mathrm{diag}}_{j,k}+n^{-1/2}\mathsf{B}^{\mathrm{cre}}_{j,k}+n^{-1/2}\mathsf{B}^{\mathrm{ann}}_{j,k}\right)\biggl|_{\Sym^k\mathcal{K}_\hor\otimes E_\lambda}.
    \end{align}
\end{proof}

\subsubsection{Near-orthogonality of the excitation sectors}
Throughout this subsubsection, the standing assumptions in Eq.~\eqref{eq:standing_assumption} remain in force. 
We here show that excitation vectors of the same degree have almost the Fock-space inner product, while excitation sectors of different degrees are nearly orthogonal. The argument uses only the approximate intertwining relations for the creation and annihilation operators proved in the preceding two subsubsections, i.e., Propositions~\ref{prop:intw_cre} and \ref{prop:intw_ann}, which prove
\begin{align}
        \crej \symword_{\lambda,k-1}&=\symword_{\lambda,k}a_j^\dag+\mathsf{E}_{j,k}^+,\qquad \|\mathsf{E}_{j,k}^+\|\leq  C_+n^{-1/2}k^{2}q_{N},\label{eq:intw_cre}\\
        \annj\symword_{\lambda,k}&=\symword_{\lambda,k-1}a_j+\mathsf{E}_{j,k}^{-},\quad \|\mathsf{E}_{j,k}^{-}\|\leq C_-\left(\sqrt{k} \tilde{\varepsilon}_0(k)+n^{-1}k^{3/2}+n^{-1/2}k\right)q_N.\label{eq:intw_ann}
\end{align}

For notational simplicity, we introduce
\begin{align}
    \varepsilon(k)\coloneqq n^{-1/2}k^{2}+\sqrt{k} \tilde{\varepsilon}_0(k)+n^{-1}k^{3/2}+n^{-1/2}k,\qquad 1\leq k\leq N
\end{align}
and $C\coloneqq C_++C_-$ so that
\begin{align}
    \|\mathsf{E}_{j,k}^{\pm}\|\leq C\varepsilon(k)q_N.\label{eq:bound_E_pm}
\end{align}
We also define
\begin{align}
    \mathscr{E}_N\coloneqq \sum_{k=1}^N\varepsilon(k), \quad \varepsilon_N^*\coloneqq \max_{1\leq k\leq N}\varepsilon(k).
\end{align}

In order to evaluate $\Delta_N =\|\symword_{\lambda,\leq N}^\dag \symword_{\lambda,\leq N}-I\|$, we define
\begin{align}
    \Gamma_{r,s}\coloneqq \symword_{\lambda,r}^\dag \symword_{\lambda,s}-\delta_{r,s}I_{\Sym^r\mathcal{K}_\hor\otimes E_\lambda}
\end{align}
so that the $(r,s)$-block of $\symword_{\lambda,\leq N}^\dag \symword_{\lambda,\leq N}-I$ is $\Gamma_{r,s}$. We separately derive an upper bound for $\Gamma_{r,s}$ in the three contributions, namely, (i) diagonal block: $r=s$, (ii) vacuum sector: $r=0$ or $s=0$, (iii) off-diagonal block: $1\leq r,s\leq N$ and $r\neq s$. 

\textit{Case (i): Diagonal block.}
\begin{lemma}\label{lem:diagonal_block}
    For $\Delta_k^{\mathrm{diag}}\coloneqq \|\Gamma_{k,k}\|$, $\Delta_k^{\mathrm{diag}}\leq \Delta_{k-1}^{\mathrm{diag}}+C_{\mathrm{diag}}\varepsilon(k)q_N^2$ with a structural constant $C_{\mathrm{diag}}$. Consequently,
    \begin{align}
        \max_{0\leq k\leq N}\Delta_k^{\mathrm{diag}}\leq C_{\mathrm{diag}} \mathscr{E}_N q_N^2.
    \end{align}
\end{lemma}

\begin{proof}
    Since $\Gamma_{k,k}=\symword_{\lambda,k}^\dag \symword_{\lambda,k}-I$ is self-adjoint,
    \begin{align}
        \Delta_k^{\mathrm{diag}}=\sup_{u\in\Sym^k\mathcal{K}_\hor\otimes E_\lambda,\,\|u\|=1}\left|\|\symword_{\lambda,k}u\|^2-1\right|.
    \end{align}

    Fix any unit vector $u\in\Sym^k\mathcal{K}_\hor\otimes E_\lambda$. Since
    \begin{align}
        \sum_{j=1}^Da_j^\dag a_j=kI\quad \text{on }\Sym^k\mathcal{K}_\hor\otimes E_\lambda, \label{eq:num_op_identity}
    \end{align}
    we have
    \begin{align}
        k\|\symword_{\lambda,k}u\|^2&=k\braket{\symword_{\lambda,k}u,\symword_{\lambda,k}u}=\sum_{j=1}^D\braket{\symword_{\lambda,k}u,\symword_{\lambda,k}a_j^\dag a_j u}\\
        &=\sum_{j=1}^D\braket{\symword_{\lambda,k}u,(\crej\symword_{\lambda,k-1}-\mathsf{E}_{j,k}^{+})u_j}      &&(\text{Eq.}~\eqref{eq:intw_cre},\,\,u_j\coloneqq a_ju)\\
        &=\sum_{j=1}^D\braket{\annj\symword_{\lambda,k}u,\symword_{\lambda,k-1} u_j} -\sum_{j=1}^D\braket{\symword_{\lambda,k}u,\mathsf{E}_{j,k}^{+}u_j}   &&(\annj=\crej^\dag)\\
        &=\sum_{j=1}^D\braket{(\symword_{\lambda,k-1}a_j+\mathsf{E}_{j,k}^{-})u,\symword_{\lambda,k-1} u_j} -\sum_{j=1}^D\braket{\symword_{\lambda,k}u,\mathsf{E}_{j,k}^{+}u_j}   &&(\text{Eq.}~\eqref{eq:intw_ann})\\
        &=\sum_{j=1}^D\|\symword_{\lambda,k-1}u_j\|^2+\sum_{j=1}^D\left(\braket{\mathsf{E}_{j,k}^-u,\symword_{\lambda,k-1}u_j}-\braket{\symword_{\lambda,k}u,\mathsf{E}_{j,k}^{+}u_j}\right).\label{eq:k_sym_u}
    \end{align}
    By using $\|\symword_{\lambda,k}\|,\|\symword_{\lambda,k-1}\|\leq q_N$ and $\|\mathsf{E}_{j,k}^{\pm}\|\leq C\varepsilon(k)q_N$, we find
    \begin{align}
       \left| \sum_{j=1}^D\left(\braket{\mathsf{E}_{j,k}^-u,\symword_{\lambda,k-1}u_j}-\braket{\symword_{\lambda,k}u,\mathsf{E}_{j,k}^{+}u_j}\right)\right|\leq  2C\varepsilon(k)q_N^2\sum_{j=1}^D\|u_j\|.
    \end{align}
    Since 
    \begin{align}
        \sum_{j=1}^D\|u_j\|^2= \sum_{j=1}^D\braket{u,a_j^\dag a_j u}= k\|u\|^2=k,
    \end{align}
    by using Cauchy--Schwarz inequality, we obtain
    \begin{align}
        \left| \sum_{j=1}^D\left(\braket{\mathsf{E}_{j,k}^-u,\symword_{\lambda,k-1}u_j}-\braket{\symword_{\lambda,k}u,\mathsf{E}_{j,k}^{+}u_j}\right)\right|\leq  2C\varepsilon(k)q_N^2\sqrt{kD}.\label{eq:residual_Ju}
    \end{align}
    On the other hand, by the definition of $\Delta_{k-1}^{\mathrm{diag}}$, 
    \begin{align}
        \left|\sum_{j=1}^D\|\symword_{\lambda,k-1}u_j\|^2-k\right|=\left|\sum_{j=1}^D(\|\symword_{\lambda,k-1}u_j\|^2-\|u_j\|^2)\right|\leq \Delta_{k-1}^{\mathrm{diag}}\sum_j\|u_j\|^2=k\Delta_{k-1}^{\mathrm{diag}}.
    \end{align}
    Therefore, 
    \begin{align}
        |\|\symword_{\lambda,k}u\|^2-1|&=\frac{1}{k}|k\|\symword_{\lambda,k}u\|^2-k|\\
        &\leq \frac{1}{k}\left(\left|\sum_{j=1}^D\|\symword_{\lambda,k-1}u_j\|^2-k\right|+\left| \sum_{j=1}^D\left(\braket{\mathsf{E}_{j,k}^-u,\symword_{\lambda,k-1}u_j}-\braket{\symword_{\lambda,k}u,\mathsf{E}_{j,k}^{+}u_j}\right)\right|\right)\\
        &\leq \Delta_{k-1}^{\mathrm{diag}}+2Ck^{-1/2}\varepsilon(k)q_N^2\sqrt{D}\leq \Delta_{k-1}^{\mathrm{diag}}+C_{\mathrm{diag}}\varepsilon(k)q_N^2 \qquad\qquad (k^{-1/2}\leq 1,\, C_{\mathrm{diag}}\coloneqq 2C\sqrt{D}).
    \end{align}
    Taking the supremum over unit vector $u$, we obtain
    \begin{align}
        \Delta_{k}^{\mathrm{diag}}\leq \Delta_{k-1}^{\mathrm{diag}}+C_{\mathrm{diag}}\varepsilon(k)q_N^2 .
    \end{align}
    Since $\Delta_0^{\mathrm{diag}}=0$, iterating the inequality up $k\leq N$ gives
    \begin{align}
        \Delta_{k}^{\mathrm{diag}}\leq C_{\mathrm{diag}}\sum_{m=1}^k\varepsilon(m)q_N^2\leq   C_{\mathrm{diag}}\mathscr{E}_Nq_N^2.
    \end{align}
    Therefore, $\max_{0\leq k \leq N}\Delta_{k}^{\mathrm{diag}}\leq C_{\mathrm{diag}}\mathscr{E}_Nq_N^2$. 
\end{proof}

\textit{Case (ii): Vacuum sector.}

\begin{lemma}\label{lem:vacuum_sector}
    For $r\geq 1$, $\symword_{\lambda,0}^\dag \symword_{\lambda,r}=0$ and $\symword_{\lambda,r}^\dag \symword_{\lambda,0}=0$. 
\end{lemma}
\begin{proof}
    We prove $\Ran(\word_{\lambda,r})\perp E_\lambda$ for $r\geq 1$.
    Since $\rho_0= \sum_{a\in\mathcal{A}_+}\zeta_aI_{V_a}$, for any vector $\xi\in E_\lambda$, 
    \begin{align}
        \dd\pi_\lambda(\rho_0)\xi =\sum_{a\in\mathcal{A}_+}\zeta_a\dd\pi_\lambda(I_{V_a})\xi=\theta_\lambda\xi,\qquad\theta_\lambda\coloneqq \sum_{a\in\mathcal{A}_+}\zeta_a|\lambda^{(a)}|.
    \end{align}
    Therefore, $E_\lambda$ is contained in the eigenspace of $\dd\pi_\lambda(\rho_0)$ with eigenvalue $\theta_\lambda$.

    Now, we consider $\Ran(\word_{\lambda,r})$, which is spanned by vectors of the form of $\cre(x_1)\cdots\cre(x_r)\xi$  with $x_{i}\in\mathcal{K}_{a_ib_i}$ and $\zeta_{a_i}>\zeta_{b_i}$ for $1\leq i \leq r$. Since
    \begin{align}
        [\rho_0,\ell_{a_ib_i}(x_{i})]=-\omega_i\ell_{a_ib_i}(x_{i}),\qquad \omega_i\coloneqq \zeta_{a_i}-\zeta_{b_i},
    \end{align}
    and $\cre(x_{i})$ is equal to $\dd\pi_\lambda(\ell_{a_ib_i}(x_{i}))$ up to normalization factor, we also have $ [\dd\pi_\lambda(\rho_0),\cre(x_i)]=-\omega_i\cre(x_i)$. Therefore,
    \begin{align}
        \dd\pi_\lambda (\rho_0)\cre(x_1)\cdots\cre(x_r)\xi=\left(\theta_\lambda-\sum_{i=1}^r\omega_i\right)\cre(x_1)\cdots\cre(x_r)\xi,
    \end{align}
    implying that $\cre(x_1)\cdots\cre(x_r)\xi$ is an eigenvector of $\dd\pi_\lambda (\rho_0)$ with eigenvalue $\left(\theta_\lambda-\sum_{i=1}^r\omega_i\right)$. 
    Since $\omega_i>0$, $\theta_\lambda-\sum_{i=1}^r\omega_i<\theta_\lambda$, and therefore, $\cre(x_1)\cdots\cre(x_r)\xi$ is orthogonal to any $\xi\in E_\lambda$. Therefore, $\Ran(\word_{\lambda,r})\perp E_\lambda$. 

    Since $\word_{\lambda,0}:E_\lambda\to H_\lambda$ is an inclusion map, we obtain $\Ran(\word_{\lambda,r})\perp \Ran(\word_{\lambda,0}) $ for $r\geq 1$. Also, since $\symword_{\lambda,r}$ is a restriction of $\word_{\lambda,r}$, we have $\Ran(\symword_{\lambda,r})\subset \Ran(\word_{\lambda,r})$. Therefore, $\Ran(\symword_{\lambda,r})\perp  \Ran(\word_{\lambda,0}) $, which implies
    \begin{align}
        \symword_{\lambda,r}^\dag \symword_{\lambda,0}=0,\quad \symword_{\lambda,0}^\dag \symword_{\lambda,r}=0.
    \end{align}
\end{proof}

\textit{Case (iii): Off-diagonal block.}

We first prove the following.
\begin{lemma}\label{lemma:off_rs}
    For $1\leq r,s\leq N$ with $r\neq s$, $\|\symword_{\lambda,r}^\dag \symword_{\lambda,s}\|\leq C_{\mathrm{off}}'\frac{(\varepsilon(r)+\varepsilon(s))q_N^2\sqrt{N}}{|r-s|}$ with $C_{\mathrm{off}}'\coloneqq 2C\sqrt{D}$. 
\end{lemma}
\begin{proof}
    Let $x\in\Sym^r\mathcal{K}_\hor\otimes E_\lambda$ and $y\in\Sym^s\mathcal{K}_\hor\otimes E_\lambda$ be unit vectors, and define $x_j\coloneqq a_jx$ and $y_j\coloneqq a_jy$. From Eq.~\eqref{eq:num_op_identity} and the Cauchy--Schwarz inequality, we get
    \begin{align}
        \sum_{j=1}^D\|x_j\|^2=r,\quad \sum_{j=1}^D\|x_j\|\leq \sqrt{Dr}\leq \sqrt{DN},\qquad \sum_{j=1}^D\|y_j\|^2=s,\quad \sum_{j=1}^D\|y_j\|\leq \sqrt{Ds}\leq \sqrt{DN}. 
    \end{align}
    By using $sy=\sum_ja_j^\dag a_j y=\sum_j a_j^\dag y_j$, similarly to Eq.~\eqref{eq:k_sym_u}, we get
    \begin{align}
        s\braket{\symword_{\lambda,r}x,\symword_{\lambda,s}y}=\sum_{j=1}^D\braket{\symword_{\lambda,r-1}x_j,\symword_{\lambda,s-1}y_j}+\sum_{j=1}^D\left(\braket{\mathsf{E}_{j,r}^-x,\symword_{\lambda,s-1}y_j}-\braket{\symword_{\lambda,r}x,\mathsf{E}_{j,s}^{+}y_j}\right)
    \end{align}
    From Eq.~\eqref{eq:bound_E_pm}, similarly to Eq.~\eqref{eq:residual_Ju}, we obtain
    \begin{align}
        \left|\sum_{j=1}^D\left(\braket{\mathsf{E}_{j,r}^-x,\symword_{\lambda,s-1}y_j}-\braket{\symword_{\lambda,r}x,\mathsf{E}_{j,s}^{+}y_j}\right)\right|\leq C(\varepsilon(r)+\varepsilon(s))q_N^2\sqrt{DN}. 
    \end{align}
    Similarly, by using $rx=\sum_ja_j^\dag a_j x=\sum_j a_j^\dag x_j$, 
    \begin{align}
        r\braket{\symword_{\lambda,r}x,\symword_{\lambda,s}y}=\sum_{j=1}^D\braket{\symword_{\lambda,r-1}x_j,\symword_{\lambda,s-1}y_j}+\sum_{j=1}^D\left(\braket{\symword_{\lambda,r-1}x_j,\mathsf{E}_{j,s}^-y}-\braket{\mathsf{E}_{j,r}^{+}x_j,\symword_{\lambda,s}y}\right)
    \end{align}
    and
    \begin{align}
        \left|\sum_{j=1}^D\left(\braket{\symword_{\lambda,r-1}x_j,\mathsf{E}_{j,s}^-y}-\braket{\mathsf{E}_{j,r}^{+}x_j,\symword_{\lambda,s}y}\right)\right|\leq C(\varepsilon(r)+\varepsilon(s))q_N^2\sqrt{DN}
    \end{align}

    Therefore, $|(s-r)\braket{\symword_{\lambda,r}x,\symword_{\lambda,s}y}|\leq 2C(\varepsilon(r)+\varepsilon(s))q_N^2\sqrt{DN}$, i.e., 
    \begin{align}
        |\braket{x,\symword_{\lambda,r}^\dag \symword_{\lambda,s}y}|\leq 2C\sqrt{D}\frac{(\varepsilon(r)+\varepsilon(s))q_N^2\sqrt{N}}{|r-s|}.
    \end{align}
    Taking the supremum over unit vectors $x,y$ gives the claimed operator-norm bound.
\end{proof}

To proceed further, we use the Schur test:
\begin{lemma}\label{lem:Schur_test}
    Let $A=(A_{r,s})_{r,s=0}^N$ be operator on $\bigoplus_{r=0}^N\mathcal{H}_r$, where $A_{r,s}:\mathcal{H}_s\to \mathcal{H}_r$. If
    \begin{align}
        \sup_{r}\sum_{s=0}^N\|A_{r,s}\|\leq M,\quad \sup_s\sum_{r=0}^N\|A_{r,s}\|\leq M,\label{eq:schur_test_condition}
    \end{align}
    then $\|A\|\leq M$. 
\end{lemma}
\begin{proof}
    For $u=\bigoplus_s u_s$, we have
    \begin{align}
       \sum_r \|(Au)_r\|^2&=\sum_r\left\|\sum_sA_{r,s}u_s\right\|^2\leq \sum_r\left(\sum_s\|A_{r,s}\|\|u_s\|\right)^2\\
        &\leq \sum_r\left(\sum_s\|A_{r,s}\|\right)\left(\sum_s\|A_{r,s}\|\|u_s\|^2\right)&&(\text{Cauchy--Schwarz ineq.})\\
        &\leq M\sum_s\left(\sum_r\|A_{r,s}\|\right)\|u_s\|^2  &&(\text{Eq.~\eqref{eq:schur_test_condition}})\\
        &\leq M^2 \sum_s\|u_s\|^2=M^2\|u\|^2&&(\text{Eq.~\eqref{eq:schur_test_condition}}).
    \end{align}
    Therefore, $\|Au\|\leq M\|u\|$, implying that $\|A\|\leq M$. 
\end{proof}

By using these two lemmas, we obtain the bound on the off-diagonal contribution.
\begin{lemma}\label{lem:off_diagonal_block}
    For $\Gamma_N^{\mathrm{off}}\coloneqq \sum_{1\leq r,s\leq N,\,r\neq s}\Gamma_{r,s}$, $\|\Gamma_N^{\mathrm{off}}\|\leq C_{\mathrm{off}}q_N^2\varepsilon_N^*\sqrt{N}\ln N$ with a structural constant $C_{\mathrm{off}}$. 
\end{lemma}
\begin{proof}
    For $N=1$, $\Gamma_N^{\mathrm{off}}=0$ and hence the bound is trivial. We therefore assume $N\geq 2$. From Lemma~\ref{lemma:off_rs}, for each $1\leq r \leq N$, 
    \begin{align}
        \sum_{\substack{1\leq s\leq N\\ s\neq r}}C_{\mathrm{off}}'\frac{(\varepsilon(r)+\varepsilon(s))q_N^2\sqrt{N}}{|r-s|}\leq 2C_{\mathrm{off}}'\varepsilon_N^*q_N^2\sqrt{N}\sum_{\substack{1\leq s\leq N\\ s\neq r}}\frac{1}{|r-s|}.
    \end{align}
    Now,
    \begin{align}
        \sum_{\substack{1\leq s\leq N\\ s\neq r}}\frac{1}{|r-s|}&=\sum_{m=1}^{r-1} \frac{1}{m}+\sum_{m=1}^{N-r}\frac{1}{m}\leq 2\sum_{m=1}^{N-1}\frac{1}{m}
    \end{align}
    Since $\sum_{m=1}^n\frac{1}{m}\leq 1+\int_{1}^n\frac{\dd x}{x}=1+\ln n$, we get
    \begin{align}
         \sum_{\substack{1\leq s\leq N\\ s\neq r}}\frac{1}{|r-s|}\leq 2(1+\ln(N-1))\leq 2(1+\ln N)
    \end{align}
    For $N\geq 2$, we have $1\leq \ln N/\ln2$. Therefore,
    \begin{align}
        \sum_{\substack{1\leq s\leq N\\ s\neq r}}\frac{1}{|r-s|}\leq 2\left(\frac{1}{\ln 2}+1\right)\ln N.
    \end{align}
    Thus, from Lemma~\ref{lem:Schur_test}, 
    \begin{align}
        \|\Gamma_N^{\mathrm{off}}\|\leq C_{\mathrm{off}}q_N^2\varepsilon_N^*\sqrt{N}\ln N,\qquad  C_{\mathrm{off}}\coloneqq  4\left(\frac{1}{\ln 2}+1\right)C_{\mathrm{off}}'.
    \end{align}
\end{proof}

Combining these results, we obtain the following.
\begin{theorem}\label{thm:delta_N}
    Define $\Omega_N\coloneqq \mathscr{E}_N+\varepsilon_N^*\sqrt{N}\ln N$. Assume the standing conditions in Eq.~\eqref{eq:standing_assumption} and $\cins n^{-1/2}N^{3/2}\leq 1$. There exists a structural constant $C_0$ such that, if $C_0\Omega_N\leq \frac{1}{2}$, then
    \begin{align}
        \Delta_N\leq 2C_0\Omega_N,\quad q_N<3.
    \end{align}
\end{theorem}
\begin{proof}
    By Lemma~\ref{lem:vacuum_sector}, $\symword_{\lambda,\leq N}^\dag \symword_{\lambda,\leq N}-I=\Gamma_N^{\mathrm{diag}}+\Gamma_N^{\mathrm{off}}$, where $\Gamma_N^{\mathrm{diag}}\coloneqq \sum_{k=1}^N\Gamma_{k,k}$. Since the diagonal block acts on mutually orthogonal particle-number sectors,
    \begin{align}
       \| \Gamma_N^{\mathrm{diag}}\|=\max_{0\leq k\leq N}\Delta_k^{\mathrm{diag}}.
    \end{align}
    Thus, from Lemmas~\ref{lem:diagonal_block} and~\ref{lem:off_diagonal_block}, 
    \begin{align}
        \Delta_N\leq \frac{1}{4}C_0\Omega_N q_N^2,
    \end{align}
    where $C_0\coloneqq 4\max\{C_{\mathrm{diag}},C_{\mathrm{off}}\}$. From Proposition~\ref{prop:cins_bound}, $q_N^2\leq 4(1+\Delta_N)$. Therefore,
    \begin{align}
        \Delta_N\leq C_0\Omega_N (1+\Delta_N).
    \end{align}
    If $C_0\Omega_N\leq \frac{1}{2}$, we have $ \Delta_N\leq C_0\Omega_N +\frac{1}{2}\Delta_N$, and hence 
    \begin{align}
        \Delta_N\leq 2C_0\Omega_N \leq 1.
    \end{align}
    Moreover, $q_N\leq 2\sqrt{1+\Delta_N}\leq 2\sqrt{2}<3$. 
\end{proof}

We finally apply these results to our cutoff parameters. Let $L_n\coloneqq \floor{n^\eta}$. Then, since $\tilde{\varepsilon}_0(k)=C_{\mathsf{D}}(n^{\alpha-1}+k/n)$, for $1\leq k \leq L_n$, 
\begin{align}
    \varepsilon(k)=n^{-1/2}k^{2}+\sqrt{k} \tilde{\varepsilon}_0(k)+n^{-1}k^{3/2}+n^{-1/2}k\leq \tilde{C}(n^{\alpha-1}k+n^{-1/2}k^2),
\end{align}
where we defined $\tilde{C}\coloneqq C_{\mathsf{D}}+3$. Consequently,
\begin{align}
    \mathscr{E}_{L_n}&=\sum_{k=1}^{L_n}\varepsilon(k)\leq \tilde{C}(n^{\alpha-1}L_n^2+n^{-1/2}L_n^3),\qquad \varepsilon_{L_n}^*\leq  \tilde{C}(n^{\alpha-1}L_n+n^{-1/2}L_n^2).
\end{align}

For all sufficiently large $n$, $ \ln L_n\leq n^{\eta/2}$, and hence, $\sqrt{L_n}\ln L_n\leq n^\eta$. Thus,
\begin{align}
    \Omega_{L_n}\leq 2\tilde{C}\left(n^{\alpha-1+2\eta}+n^{-1/2+3\eta}\right).\label{eq:omega_N}
\end{align}
Defining
\begin{align}
    \delta_n\coloneqq n^{\alpha-1+2\eta}+n^{-1/2+3\eta},\label{eq:definition_delta_n}
\end{align}
we obtain the following:
\begin{theorem}\label{thm:qi}
    Assume $0<\eta<\frac{1}{6}$ and $\frac{1}{2}<\alpha<1-2\eta$. Then, there exists a structural constant $C^{\mathrm{qi}}$ such that, for all sufficiently large $n$ and all $\lambda\in\Lambda_n$, 
    \begin{align}
        \left\|\gram_{\lambda,L_n}-I\right\|\leq C^{\mathrm{qi}}\delta_n,
    \end{align}
    where $\gram_{\lambda,L_n}\coloneqq \symword_{\lambda,\leq L_n}^\dag \symword_{\lambda,\leq L_n}$, and $C^{\mathrm{qi}}$ is a structural constant. 
\end{theorem}
\begin{proof}
    Take $N=L_n=\floor{n^\eta}$. The assumptions on $\alpha$ and $\eta$ imply
    \begin{align}
        \cins n^{-1/2}L_{n}^{3/2}\to 0,\qquad \cins n^{-1/2}\sqrt{L_n}\to 0,\quad  \max_{1\leq k \leq L_n}\{\tilde{\varepsilon}_0(k)+kn^{-1}+kn^{-1/2}\}\to 0
    \end{align}
    as $n\to\infty$. Moreover, Eq.~\eqref{eq:omega_N} gives $\Omega_{L_n}\to 0$. Hence, for all sufficiently large $n$, all smallness assumptions used to derive Proposition~\ref{prop:intw_ann} and Theorem~\ref{thm:delta_N} are satisfied. 
    Applying Theorem~\ref{thm:delta_N}, we obtain
    \begin{align}
         \Delta_{L_n}\leq 2C_0\Omega_{L_n}\leq C^{\mathrm{qi}}\delta_n,
    \end{align}
    where $C^{\mathrm{qi}}\coloneqq 4C_0\tilde{C}$. 
\end{proof}

\subsubsection{Polar correction}\label{sec:Polar_correction}
Since $\delta_n\to 0$ as $n\to\infty$, Theorem~\ref{thm:qi} implies that $\gram_{\lambda,L_n}=\symword_{\lambda,\leq L_n}^\dag \symword_{\lambda,\leq L_n}\approx I$, implying that $\symword_{\lambda,\leq L_n}$ is an approximate isometry. We here construct an exact isometry and investigate its properties. 

From now on, we assume that $n$ is sufficiently large, and therefore, $ \left\|\gram_{\lambda,L_n}-I\right\|\leq \frac{1}{2}$. Since $\spec(\gram_{\lambda,L_n})\subset [\frac{1}{2},\frac{3}{2}]$, the operator $\gram_{\lambda,L_n}$ is invertible, and $\gram_{\lambda,L_n}^{\pm 1/2}$ is defined by the standard finite-dimensional spectral calculus. For notational simplicity, we denote $\symword_\lambda\coloneqq \symword_{\lambda,\leq L_n}$ and $\gram_\lambda\coloneqq \gram_{\lambda,L_n}$.
 
\begin{definition}[Exact isometry]\label{def:exact_isometry}
    Define $W_\lambda\coloneqq \symword_\lambda \gram_{\lambda}^{-1/2}:\mathcal{F}_{\leq L_n}\otimes E_\lambda \to H_\lambda$. 
\end{definition}

We show basic properties of $W_\lambda$:
\begin{lemma}\par ~\label{lemma:properties_W_lambda}
    \begin{enumerate}[(i)]
        \item $W_\lambda$ is an isometry. Consequently, $\|W_\lambda\sigma W_\lambda^\dag \|_1=\|\sigma\|_1$ for $\sigma\in\mathcal{T}_1(\mathcal{F}_{\leq L_n}\otimes E_\lambda)$. 
        \item $\|W_\lambda-\symword_{\lambda}\|\leq C^{\mathrm{qi}}\delta_n$. 
        \item $W_\lambda$ is $H$-covariant. 
        \item $W_\lambda(\ket{0}\otimes \xi)=\xi$ for any $\xi\in E_\lambda$, where $\ket{0}$ is the Fock vacuum. 
    \end{enumerate}
\end{lemma}
\begin{proof}
    (i) $W_\lambda^\dag W_\lambda =\gram_\lambda^{-1/2}\symword_\lambda^\dag \symword_\lambda \gram_\lambda^{-1/2}=I$. Since $(W_\lambda\sigma W_\lambda^\dag)^\dag W_\lambda\sigma W_\lambda^\dag=W_\lambda\sigma^\dag \sigma W_\lambda^\dag=(W_\lambda|\sigma| W_\lambda^\dag)^2$ and $W_\lambda|\sigma| W_\lambda^\dag\geq 0$, from the uniqueness of positive roots of a positive operator, we get $W_\lambda|\sigma| W_\lambda^\dag=|W_\lambda\sigma W_\lambda^\dag|$. Therefore, $\|\sigma\|_1=\Tr(|\sigma|)=\Tr(W_\lambda|\sigma|W_\lambda^\dag)=\Tr(|W_\lambda\sigma W_\lambda^\dag|)=\|W_\lambda\sigma W_\lambda^\dag \|_1$.

    (ii) Since $W_\lambda$ is an isometry and $\symword_\lambda=W_\lambda \gram_\lambda^{1/2}$,
    \begin{align}
        \|W_\lambda-\symword_\lambda\|=\|I-\gram_\lambda^{1/2}\|=\sup_{t\in\spec (\gram_\lambda)}|1-\sqrt{t}|\leq \sup_{t\in\spec (\gram_\lambda)}|t-1|=\|\gram_\lambda-I\|\leq C^{\mathrm{qi}}\delta_n,
    \end{align}
    where we have used $|1-\sqrt{t}|=\frac{|1-t|}{|1+\sqrt{t}|}\leq |1-t|$ for $t\geq 0$. 

    (iii) By Lemma~\ref{lemma:h_equivalence_symword}, for any $h\in H$, $\pi_\lambda(h)\symword_{\lambda,\leq L_n}=\symword_{\lambda,\leq L_n}U_{\lambda,h}^{(\leq L_n)}$ and $[\gram_{\lambda,L_n},U_{\lambda,h}^{(\leq L_n)}]=0$. Therefore, 
    \begin{align}
        \pi_\lambda(h)W_\lambda= \symword_{\lambda,\leq L_n}U_{\lambda,h}^{(\leq L_n)}\gram_{\lambda,L_n}^{-1/2}=\symword_{\lambda,\leq L_n}\gram_{\lambda,L_n}^{-1/2}U_{\lambda,h}^{(\leq L_n)}=W_\lambda U_{\lambda,h}^{(\leq L_n)}.
    \end{align}

    (iv) From Lemma~\ref{lem:vacuum_sector}, 
    \begin{align}
        \gram_\lambda \ket{0}\otimes \xi=\sum_{r=0}^{L_n}\symword_{\lambda,r}^\dag \symword_{\lambda,0}\ket{0}\otimes \xi=\symword_{\lambda,0}^\dag \symword_{\lambda,0}\ket{0}\otimes \xi=\ket{0}\otimes \xi.
    \end{align}
    Therefore, $\gram_\lambda^{-1/2} \ket{0}\otimes \xi= \ket{0}\otimes \xi $, and hence $W_\lambda  \ket{0}\otimes \xi=\symword_{\lambda,0} \ket{0}\otimes \xi=\xi$. 
\end{proof}

% \clearpage
\subsection{Gaussian Description of Typical Schur Blocks}\label{sec:Gaussian_description}
\subsubsection{Identification of the Block Reference State in Fock Space}
In the previous section, we constructed isometry $W_\lambda:\mathcal{F}_{\leq L_n}\otimes E_\lambda\to H_\lambda$ for a typical Young diagram $\lambda$. In this section, we identify the unnormalized block reference state $\pi_\lambda(\rho_0)$ in this Fock-space representation. 

For $x=\sum_{\zeta_a>\zeta_b}x_{ab}\in \mathcal{K}_\hor$, we define a one-particle multiplier $d_{\zeta}$  by 
\begin{align}
    d_\zeta x\coloneqq \sum_{\zeta_a>\zeta_b}\frac{\zeta_b}{\zeta_a}x_{ab}. 
\end{align}
We introduce
\begin{align}
    D_{\zeta}^{\leq L_n}\coloneqq \bigoplus_{k=0}^{L_n}d_\zeta ^{\otimes k}\biggl|_{\Sym^k\mathcal{K}_\hor},\qquad D_{\zeta}\coloneqq \bigoplus_{k=0}^{\infty}d_\zeta ^{\otimes k}\biggl|_{\Sym^k\mathcal{K}_\hor}. 
\end{align}
By using $D_\zeta ^{\leq L_n}$, the image of the unnormalized reference state $\pi_\lambda(\rho_0)$ is given by the following theorem:
\begin{theorem}\label{thm:unnormalized_states_embedding}
    For any sufficiently large $n$ and any $\lambda \in \Lambda_n$, 
    \begin{align}
        W_\lambda^\dag \pi_\lambda(\rho_0)W_\lambda = c_\lambda (D_\zeta ^{\leq L_n}\otimes I_{E_\lambda}),
    \end{align}
    where $c_\lambda\coloneqq \prod_{a\in\mathcal{A}_+}\zeta_a^{|\lambda^{(a)}|}$. 
\end{theorem}

\begin{proof}
    For $Y\in \Hom(V_a,V_b)$ with $\zeta_a>\zeta_b\geq 0$, we have $\rho_0Y=\frac{\zeta_b}{\zeta_a}Y\rho_0$. Writing $\dd\pi_n(Y)\coloneqq \sum_{i=1}^nI^{\otimes (i-1)}\otimes Y \otimes I^{\otimes (n-i)}$, this implies $\rho_0^{\otimes n}\dd\pi_n(Y)=\frac{\zeta_b}{\zeta_a}\dd\pi_n(Y)\rho_0^{\otimes n}$. Restricting both sides to $H_\lambda\otimes K_\lambda$ gives $\pi_\lambda(\rho_0)\dd\pi_\lambda(Y)=\frac{\zeta_b}{\zeta_a}\dd\pi_\lambda(Y)\pi_\lambda(\rho_0)$. Applying this relation to $Y\coloneqq s_{ab}^{-1/2}\ell_{ab}(x)$ with $x\in\mathcal{K}_{ab}$ yields
    \begin{align}
        \pi_\lambda(\rho_0)\cre(x)=\frac{\zeta_b}{\zeta_a}\cre (x)\pi_\lambda(\rho_0). 
    \end{align}

    The space $\mathcal{K}_\hor^{\otimes k}\otimes E_\lambda$ is spanned by vectors of the form $x_1\otimes\cdots\otimes x_k\otimes \xi$, where $x_i\in\mathcal{K}_{a_ib_i}$ and $\xi\in E_\lambda$. By using the above relation repeatedly, we have
    \begin{align}
        \pi_\lambda(\rho_0)\word_{\lambda,k}(x_1\otimes\cdots\otimes x_k\otimes \xi)&= \frac{1}{\sqrt{k!}}\pi_\lambda(\rho_0)\cre(x_1)\cdots\cre(x_k)\xi=\frac{1}{\sqrt{k!}}\left(\prod_{i=1}^k\frac{\zeta_{b_i}}{\zeta_{a_i}}\right)\cre(x_1)\cdots\cre(x_k)\pi_\lambda(\rho_0)\xi\\
        &=\frac{1}{\sqrt{k!}}\cre(d_\zeta x_1)\cdots\cre(d_\zeta x_k)\pi_\lambda(\rho_0)\xi.
    \end{align}
    
    Now, since $\rho_0$ commutes with any block-diagonal operator $L\in\mathfrak{l}$, $\pi_\lambda(\rho_0)$ commutes with $\dd\pi_\lambda(L)$. Since $v_\lambda$ is a weight vector of weight $\lambda$, the diagonal operator $\pi_\lambda(\rho_0)$ acts on it by the scalar $\prod_{i=1}^d \mu_i^{\lambda_i}$, i.e., $\pi_\lambda(\rho_0)v_\lambda=c_\lambda v_\lambda$ with $c_\lambda=\prod_{i=1}^R \mu_i^{\lambda_i}=\prod_{a\in\mathcal{A}_+}\zeta_a^{|\lambda^{(a)}|}$.
    Therefore, $\pi_\lambda(\rho_0)\xi=c_\lambda \xi$ for any $\xi\in E_\lambda =U(\mathfrak{l})v_\lambda$. We note that for $\lambda\in \Lambda_n$, 
    \begin{align}
        c_\lambda=\prod_{a\in\mathcal{A}_+}\zeta_a^{|\lambda^{(a)}|}>0
    \end{align}
    since the kernel-block part of $\lambda$ vanishes.
    Thus, 
      \begin{align}
        \pi_\lambda(\rho_0)\word_{\lambda,k}(x_1\otimes\cdots\otimes x_k\otimes \xi)&=\frac{c_\lambda}{\sqrt{k!}}\cre(d_\zeta x_1)\cdots\cre(d_\zeta x_k)\xi=c_\lambda \word_{\lambda,k} (d_\zeta^{\otimes k}\otimes I_{E_\lambda})(x_1\otimes\cdots\otimes x_k\otimes \xi).
    \end{align}
    Restricting to the symmetric tensor sectors, we get
    \begin{align}
        \pi_\lambda(\rho_0)\symword_{\lambda,k}=c_\lambda \symword_{\lambda,k}(d_\zeta ^{\otimes k}|_{\Sym^k\mathcal{K}_\hor}\otimes I_{E_\lambda}). 
    \end{align}
    Summing over $0\leq k \leq L_n$, we obtain
    \begin{align}
        \pi_\lambda(\rho_0)\symword_{\lambda,\leq L_n}=c_\lambda  \symword_{\lambda,\leq L_n}(D_\zeta ^{\leq L_n}\otimes I_{E_\lambda}). 
    \end{align}
    Consequently,
    \begin{align}
        \symword_{\lambda,\leq L_n}^\dag  \pi_\lambda(\rho_0)\symword_{\lambda,\leq L_n}=c_\lambda  \gram_{\lambda,L_n}(D_\zeta ^{\leq L_n}\otimes I_{E_\lambda}).
    \end{align}
    
    Since the left-hand side is self-adjoint, we have
    \begin{align}
        c_\lambda  \gram_{\lambda, L_n}(D_\zeta ^{\leq L_n}\otimes I_{E_\lambda})=\overline{c_\lambda} (D_\zeta^{\leq L_n}\otimes I_{E_\lambda})^\dag  \gram_{\lambda, L_n}^\dag.
    \end{align}
    Since $c_\lambda>0$ for $\lambda\in\Lambda_n$, and $\gram_{\lambda, L_n}$ and $D_\zeta^{\leq L_n}\otimes I_{E_\lambda}$ are self-adjoint, we obtain
    \begin{align}
        \gram_{\lambda,L_n}(D_\zeta^{\leq L_n}\otimes I_{E_\lambda})=(D_\zeta^{\leq L_n}\otimes I_{E_\lambda}) \gram_{\lambda, L_n},
    \end{align}
    i.e., $D_\zeta^{\leq L_n}\otimes I_{E_\lambda}$ commutes with $\gram_{\lambda,L_n}$. Therefore, we get
    \begin{align}
        W_{\lambda}^\dag \pi_\lambda(\rho_0) W_\lambda&= \gram_{\lambda, L_n}^{-1/2}\symword_{\lambda,\leq L_n}^\dag  \pi_\lambda(\rho_0)\symword_{\lambda,\leq L_n}\gram_{\lambda, L_n}^{-1/2}=c_\lambda (D_\zeta^{\leq L_n}\otimes I_{E_\lambda}).
    \end{align}
\end{proof}

We now derive relation between normalized states by using Theorem~\ref{thm:unnormalized_states_embedding}. 
Let $P_{\leq M}$ denote the orthogonal projection of $\Gamma_{\mathrm{s}}$ onto $\mathcal{F}_{\leq M}$ with $M\in\mathbb{Z}_{\geq 0}$. We define
\begin{align}
    Z_\lambda\coloneqq \Tr \pi_\lambda(\rho_0),\quad Z_F\coloneqq \Tr(D_\zeta),\quad \Phi_0\coloneqq Z_F^{-1}D_\zeta,\quad Z_F^{\leq M}\coloneqq \Tr(P_{\leq M}D_\zeta).
\end{align}
For the fixed block-adapted orthonormal basis $\{e_j\}_{j=1}^D$ of $\mathcal{K}_\hor$, let $q_j\coloneqq \zeta_{b_j}/\zeta_{a_j}$. Then
\begin{align}
    Z_F^{\leq M}=\sum_{\substack{m=(m_1,\ldots,m_D)\in\mathbb{Z}_{\geq 0}^D\\|m|\leq M}}\prod_{j=1}^Dq_j^{m_j},\quad |m|\coloneqq \sum_{j=1}^D m_j.
\end{align}
Since $0\leq q_j<1$ for each mode $j$, the product for $Z_F$ converges:
\begin{align}
    Z_F=\sum_{m=(m_1,\ldots,m_D)\in\mathbb{Z}_{\geq 0}^D}\prod_{j=1}^Dq_j^{m_j}=\prod_{j=1}^D\frac{1}{1-q_j},
\end{align}
and therefore $D_\zeta$ is trace class. We use the convention $0^0=1$. Define the normalization factor
\begin{align}
    b_{\lambda,n}\coloneqq \frac{c_\lambda Z_F\dim E_{\lambda}}{Z_\lambda}.
\end{align}

Also, we introduce normalized states on $E_\lambda$ and $H_\lambda$ defined by
\begin{align}
    \omega_\lambda\coloneqq \frac{I_{E_\lambda}}{\dim E_\lambda},\quad \omega_\lambda^{H_\lambda}\coloneqq \frac{P_{E_\lambda}}{\dim E_\lambda},
\end{align}
where $P_{E_\lambda}$ is the orthogonal projection of $H_\lambda$ onto $E_\lambda$. Also, we define 
\begin{align}
    P_{\lambda,M}\coloneqq W_\lambda(P_{\leq M}\otimes I_{E_\lambda})W_\lambda^\dag,\qquad 0\leq M\leq L_n.\label{eq:projector_onto_E_lambda}
\end{align}
These operators satisfy the following properties.
\begin{lemma}\label{lem:properties_omegas}\par ~ (i) $W_\lambda (\ket{0}\bra{0}\otimes \omega_\lambda)W_\lambda^\dag =\omega_\lambda^{H_\lambda}$. (ii)  $\pi_\lambda(h)\omega_\lambda^{H_\lambda}\pi_\lambda(h)^\dag=\omega_\lambda^{H_\lambda}$ for any $h\in H$. (iii) $P_{\lambda,M}$ is an orthogonal projection and $P_{\lambda,M}\leq W_\lambda W_\lambda^\dag$. 
\end{lemma}
\begin{proof}
    (i) Let $\{\xi_i\}_i$ be an orthonormal basis of $E_\lambda$. From (iv) of Lemma~\ref{lemma:properties_W_lambda}, $W_\lambda(\ket{0}\otimes \xi_i)=\xi_i$ for all $i$. Since $\omega_\lambda=\frac{1}{\dim E_\lambda}\sum_i\xi_i\xi_i^\dag$, we get $W_\lambda(\ket{0}\bra{0}\otimes\omega_\lambda)W_\lambda^\dag =\omega_\lambda^{H_\lambda}$. 

    (ii) Since $\pi_\lambda(h)$ is unitary and preserves $E_\lambda$, it commutes with the orthogonal projection $P_{E_\lambda}$. Therefore, $\pi_\lambda(h)\omega_{\lambda}^{H_\lambda}\pi_\lambda(h)^\dag =\omega_{\lambda}^{H_\lambda}\pi_\lambda(h)\pi_\lambda(h)^\dag =\omega_{\lambda}^{H_\lambda}$. 

    (iii) Since $W_\lambda$ is an isometry, $(W_\lambda (P_{\leq M}\otimes I_{E_\lambda})W_\lambda^\dag)^2=W_\lambda (P_{\leq M}\otimes I_{E_\lambda})W_\lambda^\dag$. Thus, since $W_\lambda (P_{\leq M}\otimes I_{E_\lambda})W_\lambda^\dag$ is self-adjoint, it is an orthogonal projection. Since $ (P_{\leq M}\otimes I_{E_\lambda})\leq I$, we get $P_{\lambda,M}=W_\lambda (P_{\leq M}\otimes I_{E_\lambda})W_\lambda^\dag\leq W_\lambda W_\lambda^\dag$. 
\end{proof}

\begin{proposition}\label{prop:b_lambda_n}
    There exist structural constants $C_\tail,C_\tail'>0$ such that, for all sufficiently large $n$, all typical $\lambda\in\Lambda_n$ and any $0\leq M\leq L_n$, the following hold:
    \begin{align}
        P_{\lambda, M}\rho_{\lambda,0}P_{\lambda, M}&=b_{\lambda,n}W_{\lambda}(P_{\leq M}\Phi_{0} P_{\leq M}\otimes \omega_\lambda)W_\lambda^\dag,\\
        0\leq b_{\lambda,n}-1&\leq C_\tail e^{-C_\tail'L_n},\\
        \Tr((I-P_{\lambda,M})\rho_{\lambda,0})&\leq \Tr((I-P_{\leq M})\Phi_0)\leq C_\tail e^{-C_\tail'M}\\
        \|\rho_{\lambda,0}-P_{\lambda,M}\rho_{\lambda,0}P_{\lambda,M}\|_1&\leq 2\sqrt{C_\tail}e^{-C_\tail'M/2}.
    \end{align}
\end{proposition}
Let us first prove a lemma which we use in the proof of Proposition~\ref{prop:b_lambda_n}.
\begin{lemma}\label{lem:overcount}
    Let $V$ be a finite-dimensional Hilbert space, and let $K$ be a positive operator on $V$. Let $\{e_i\}_{i=1}^m$ be a spanning family of nonzero eigenvectors of $K$, with $Ke_i=k_ie_i$. Then
    \begin{align}
        \Tr(K)\leq \sum_{i=1}^m k_i .\label{eq:trace_overcount}
    \end{align}
\end{lemma}
\begin{proof}
    Let $V_k$ be the eigenspace of $K$ with eigenvalue $k$, and let $n(k)$ be the number of vectors in $\{e_i\}_{i=1}^m$ such that $Ke_i=ke_i$. 
    Projecting the spanning family onto $V_k$ shows that its members with eigenvalue $k$ span $V_k$. Therefore,
    \begin{align}
        \dim V_k\leq n(k),
    \end{align}
    which implies
    \begin{align}
        \Tr(K)=\sum_{k\in\spec(K)}k \dim V_k \leq \sum_{k\in\spec(K)}kn(k)= \sum_{i=1}^m k_i.
    \end{align}
\end{proof}

\begin{proof}[Proof of Proposition~\ref{prop:b_lambda_n}]
    Dividing the identity in Theorem~\ref{thm:unnormalized_states_embedding} by $Z_\lambda$, 
    \begin{align}
        W_\lambda^\dag \rho_{\lambda,0}W_\lambda = \frac{c_\lambda}{Z_\lambda} (D_\zeta ^{\leq L_n}\otimes I_{E_\lambda}).
    \end{align}
    Therefore, 
    \begin{align}
        P_{\lambda, M}\rho_{\lambda,0}P_{\lambda, M}&=\frac{c_\lambda}{Z_\lambda}W_\lambda (P_{\leq M} D_\zeta ^{\leq L_n}P_{\leq M} \otimes I_{E_\lambda})W_\lambda ^\dag\\
        &=\frac{c_\lambda Z_F\dim E_\lambda}{Z_\lambda}W_\lambda \left(P_{\leq M} \frac{D_\zeta}{Z_F} P_{\leq M} \otimes \frac{I_{E_\lambda}}{\dim E_\lambda}\right)W_\lambda ^\dag\\
        &=b_{\lambda,n}W_{\lambda}(P_{\leq M}\Phi_{0} P_{\leq M}\otimes \omega_\lambda)W_\lambda^\dag.\label{eq:p_rho_p}
    \end{align}

    Since $P_{\lambda,L_n}=W_\lambda W_\lambda^\dag$ and $\pi_\lambda(\rho_0)\geq 0$, we get
    \begin{align}
        Z_\lambda=\Tr (\pi_\lambda(\rho_0))\geq \Tr (\pi_\lambda(\rho_0)P_{\lambda,L_n})=\Tr ( W_\lambda^\dag \pi_\lambda(\rho_0)W_\lambda)=c_\lambda \dim E_{\lambda} Z_F^{\leq L_n},
    \end{align}
    where we have used the identity in Theorem~\ref{thm:unnormalized_states_embedding} in the last equality. 
    To derive an upper bound on $Z_\lambda$, use the ordered block-adapted basis $Y_j=\Lmat(e_j)$ of $\mathfrak{n}_-$, so that $Y_j\in \Hom(V_{a_j},V_{b_j})$, and let $\{\xi_r\}_{r=1}^{\dim E_\lambda}$ be an orthonormal basis of $E_\lambda$. 
    For $m\coloneqq (m_1,\ldots,m_D)\in\mathbb{Z}_{\geq 0}^D$, the vectors of the form
    \begin{align}
        f_{m,r}\coloneqq \dd\pi_\lambda (Y_1)^{m_1}\cdots \dd\pi_\lambda (Y_D)^{m_D}\xi_r,
    \end{align}
    span $H_\lambda$ by Lemmas~\ref{lem:ordering_subalg} and~\ref{lem:ground_component}. By $\pi_\lambda(\rho_0)\dd\pi_\lambda(Y_j)=q_j\dd\pi_\lambda(Y_j)\pi_\lambda(\rho_0)$ and $\pi_\lambda(\rho_0)\xi=c_\lambda \xi$, any nonzero $f_{m,r}$ is an eigenvector of $\pi_\lambda(\rho_0)$ with
    \begin{align}
        \pi_\lambda(\rho_0)f_{m,r}=\kappa_{m}f_{m,r},\qquad \kappa_{m}\coloneqq c_\lambda\prod_{j=1}^Dq_j^{m_j}.
    \end{align}
    Since $H_\lambda$ is finite-dimensional, there is a finite set $\mathcal{M}\subset \mathbb{Z}_{\geq 0}^D\times \{1,\ldots,\dim E_\lambda\}$ such that the nonzero vectors $\{f_{m,r}\}_{(m,r)\in\mathcal{M}}$ span $H_\lambda$. Therefore, by Lemma~\ref{lem:overcount}, we obtain
    \begin{align}
        Z_\lambda &=\Tr \pi_\lambda(\rho_0)\leq \sum_{(m,r)\in\mathcal{M}}\kappa_m\leq \sum_{m\in\mathbb{Z}_{\geq 0}^D}\sum_{r=1}^{\dim E_\lambda}\kappa_{m}= c_\lambda\sum_{i=1}^{\dim E_\lambda}\sum_{m\in\mathbb{Z}_{\geq 0}^D}\left(\prod_{j=1}^Dq_j^{m_j}\right)=c_\lambda \dim E_\lambda Z_F.
    \end{align}
    Thus,
    \begin{align}
        1\leq b_{\lambda,n}\leq \frac{Z_F}{Z_F^{\leq L_n}}.
    \end{align}

    Let $q_*\coloneqq \max_{\zeta_a>\zeta_b}\frac{\zeta_b}{\zeta_a}<1$ and choose $C_\tail'>0$ such that $e^{C_\tail'}q_*<1$. Then, for any integer $M\geq 0$,
    \begin{align}
        Z_F-Z_F^{\leq M}&=\sum_{|m|>M}\prod_{j=1}^Dq_j^{m_j}\leq e^{-C_\tail'M}\sum_{|m|>M}\prod_{j=1}^D\left(e^{C_\tail'}q_j\right)^{m_j}\leq e^{-C_\tail'M}\sum_{m\in\mathbb{Z}_{\geq 0}^D}\prod_{j=1}^D\left(e^{C_\tail'}q_j\right)^{m_j}=C_\tail e^{-C_\tail' M},
    \end{align}
    where we have used $1\leq e^{C_\tail'(\sum_jm_j-M)}$ for any $|m|>M$ in the first inequality, and defined 
    \begin{align}
        C_\tail\coloneqq \sum_{m\in\mathbb{Z}_{\geq 0}^D}\prod_{j=1}^D\left(e^{C_\tail'}q_j\right)^{m_j}= \prod_{j=1}^D\left(\frac{1}{1-e^{C_\tail'}q_j}\right).
    \end{align}

    Since the vacuum occupation $m=0$ contributes $1$, we have $Z_{F}^{\leq L_n}\geq 1$. Therefore, 
    \begin{align}
        0\leq b_{\lambda,n}-1\leq \frac{Z_F-Z_{F}^{\leq L_n}}{Z_{F}^{\leq L_n}}\leq Z_F-Z_{F}^{\leq L_n}\leq C_\tail e^{-C_\tail' L_n}.
    \end{align}
    Similarly, 
    \begin{align}
        \Tr((I-P_{\leq M})\Phi_0)=\frac{Z_F-Z_F^{\leq M}}{Z_F}\leq C_\tail e^{-C_\tail' M}. 
    \end{align}

    From Eq.~\eqref{eq:p_rho_p} and $b_{\lambda,n}\geq 1$, we get
    \begin{align}
        \Tr((I-P_{\lambda,M})\rho_{\lambda,0})=1-b_{\lambda,n}\Tr(P_{\leq M}\Phi_{0} )\leq 1- \Tr(P_{\leq M}\Phi_0)\leq C_\tail e^{-C_\tail'M}.
    \end{align}
    Therefore, from the gentle measurement lemma~\cite{winter_CodingTheoremStrongconversequantum_1999,ogawa_MakingGoodCodesClassicalQuantumChannel_2007a}, we have
    \begin{align}
        \|\rho_{\lambda,0}-P_{\lambda,M}\rho_{\lambda,0}P_{\lambda,M}\|_1\leq 2\sqrt{\Tr(I-P_{\lambda,M})\rho_{\lambda,0}}\leq 2\sqrt{C_\tail}e^{- C_\tail'M/2}.
    \end{align}
\end{proof}

Proposition~\ref{prop:b_lambda_n} establishes the following asymptotic relation:
\begin{align}
    \rho_{\lambda,0}&\approx W_{\lambda}(P_{\leq M}\Phi_{0} P_{\leq M}\otimes \omega_\lambda)W_\lambda^\dag.
\end{align}
The following two subsubsections extend this relation to
\begin{align}
    U_{\lambda,n}(Z)\rho_{\lambda,0}U_{\lambda,n}(Z)^\dag &\approx W_{\lambda}(P_{\leq M}D(Z)\Phi_{0}D(Z)^\dag P_{\leq M}\otimes \omega_\lambda)W_\lambda^\dag
\end{align}
for $\|Z\|\leq n^\beta$.

\subsubsection{Coherent-State Dynamics under Local Unitary Transformations}
In the previous subsubsection, we identified the block reference state with the Gaussian reference state in the Fock-space representation. The purpose of this subsubsection is to show that this state correspondence is approximately covariant under local unitary transformations. 

Throughout this subsubsection, we assume that
\begin{align}
    0<\eta<\frac{1}{6},\qquad 0\leq \beta <\frac{\eta}{2},\qquad \frac{1}{2}<\alpha<1-2\eta.\label{eq:scale_eta_beta_alpha}
\end{align}
We also introduce an intermediate exponent $\bar{\beta}\coloneqq \frac{1}{2}\left(\beta+\frac{\eta}{2}\right)$ so that
\begin{align}
    \beta<\bar{\beta}<\frac{\eta}{2}.
\end{align}
We also introduce a smaller cutoff
\begin{align}
    N_n\coloneqq \floor*{\frac{L_n}{2}}
\end{align}
so that $N_n+1\leq L_n$ holds for all sufficiently large $n$. Under $N_n+1\leq L_n$, one creation operator can act on $\mathcal{F}_{\leq N_n}$ without leaving the domain of $\symword_\lambda$ and $W_\lambda$. 

We shall also use the error scale $\delta_n$ defined in Eq.~\eqref{eq:definition_delta_n}:
\begin{align}
    \delta_n\coloneqq n^{\alpha-1+2\eta}+n^{-1/2+3\eta}.
\end{align}
Under Eq.~\eqref{eq:scale_eta_beta_alpha}, $\delta_n\to0$ as $n\to\infty$. 

As in Section~\ref{sec:Polar_correction}, we use notations $\symword_\lambda\coloneqq \symword_{\lambda,\leq L_n}$ and $\gram_\lambda\coloneqq \gram_{\lambda,L_n}$ for short. 

For $X\in\mathcal{K}_{\hor}$, the corresponding coherent state is given by
\begin{align}
    \ket{X}\coloneqq e^{-\|X\|^2/2}\sum_{k=0}^\infty \frac{X^{\otimes k}}{\sqrt{k!}}\in\Gamma_{\mathrm{s}}(\mathcal{K}_\hor).
\end{align}
Differentiating the coherent state gives
\begin{align}
    \frac{\dd}{\dd t}\ket{tX}=\mathsf{B}(X)\ket{tX},\qquad \mathsf{B}
    (X)\coloneqq a^\dag (X)-a(X).
\end{align}
Hereafter, a ket vector with a capital letter always describes a coherent state. 
Let $\Pi_m$ denote the projection onto the $m$-particle sector and $P_{\leq M}$ be the projector defined by $P_{\leq M}\coloneqq \sum_{m=0}^{M}\Pi_m$ for $M\in\mathbb{Z}_{\geq 0}$. Then, $\Pi_m\ket{X}=e^{-\|X\|^2/2}\frac{X^{\otimes m}}{\sqrt{m!}}$. Therefore, 
\begin{align}
   \| \Pi_m \ket{X}\|^2=e^{-\|X\|^2}\frac{\|X\|^{2m}}{m!},
\end{align}
implying that the particle-number distribution of $\ket{X}$ is Poisson with mean $\mu=\|X\|^2$. 

We introduce a coherent state with a cutoff: 
\begin{align}
    t_X\coloneqq \|P_{\leq N_n}\ket{X}\|^2,\qquad \ket{\widehat{X}}\coloneqq t^{-1/2}_XP_{\leq N_n}\ket{X}.\label{eq:cutoff_coherent_state}
\end{align}
The following lemma provides an estimate for the tail.
\begin{lemma}\label{lem:m_particle_sector_bound}
    For all sufficiently large $n$, any $X\in\mathcal{K}_\hor$ satisfying $\|X\|\leq 3 n^{\bar{\beta}}$, and any $m\geq N_n$, 
    \begin{align}
        1-t_X\leq 2^{-N_n},\qquad \|\Pi_m\ket{X}\|\leq 2^{-m/2}\leq 2^{-N_n/2}.
    \end{align}
    In particular $t_X\geq \frac{1}{2}$. 
\end{lemma}
\begin{proof}
    In the case with mean $\mu=0$, $\|\Pi_0\ket{X}\|=1$ and $\|\Pi_m\ket{X}\|=0$ for $m\geq 1$, and hence the bound is trivial. Assume $\mu\neq 0$. Since $2\bar{\beta}<\eta$, for all sufficiently large $n$, 
    \begin{align}
        9en^{2\bar{\beta}}\leq \frac{N_n}{2}.
    \end{align}
    For $\|X\|\leq 3n^{\bar{\beta}}$, this implies $e\mu/N_n\leq 1/2$. For a random variable $S$ that follows the Poisson distribution with mean $\mu$, for $N>\mu$, the exponential Markov bound gives
    \begin{align}
        \mathbb{P}[S\geq N]\leq \left(\frac{e\mu}{N}\right)^N.\label{eq:exp_Mar}
    \end{align}
    Indeed, for any $r>0$, the Markov inequality yields
    \begin{align}
         \mathbb{P}[S\geq N]= \mathbb{P}[e^{rS}\geq e^{rN}]\leq e^{-rN}\mathbb{E}[e^{rS}]=\exp(\mu(e^r-1)-rN).
    \end{align}
    Substituting $r\coloneqq \ln(N/\mu)>0$, we get Eq.~\eqref{eq:exp_Mar}. By using Eq.~\eqref{eq:exp_Mar}, 
    \begin{align}
        1-t_X=\mathbb{P}[S> N_n]\leq \mathbb{P}[S\geq N_n] \leq \left(\frac{e\mu}{N_n}\right)^{N_n}\leq 2^{-N_n}.
    \end{align}
    In particular, $t_X\geq 1-2^{-N_n}\geq 1/2$ as long as $N_n\geq 1$. 
    Moreover, $m!\geq (m/e)^m$ gives, for $m\geq N_n$, 
    \begin{align}
        \|\Pi_m\ket{X}\|^2=e^{-\mu}\frac{\mu^m}{m!}\leq \left(\frac{e\mu}{N_n}\right)^m\leq 2^{-m}.
    \end{align}
\end{proof}

The displacement generator on Fock space is given by
\begin{align}
    \fdisgen(X) = a^\dag (X)-a(X),
\end{align}
while on the Schur block,
\begin{align}
    \sdisgen_{\lambda,n}(X)\coloneqq n^{-1/2}\dd\pi_\lambda(K(X)).
\end{align}
We prove an approximate intertwining relation between these generators. To this end, we first show the following.

\begin{lemma}\label{lem:sdisgen}
    Let $X=\sum_{j=1}^Dx_je_j\in\mathcal{K}_{\hor}$, where $\{e_j\}_{j=1}^D$ denotes the block-adapted orthonormal basis. We define
    \begin{align}
        \sdisgen_{\lambda,n}(X)\coloneqq n^{-1/2}\dd\pi_\lambda(K(X))=\sum_{j=1}^D\gamma_{j}(\lambda)(x_j\crej-\overline{x_j}\annj),\quad \gamma_{j}(\lambda)\coloneqq \sqrt{\frac{s_{a_jb_j}(\lambda)}{n(\zeta_{a_j}-\zeta_{b_j})}}. 
    \end{align}
    Then, there exists a structural constant $C_\gamma$ such that
    \begin{align}
        |\gamma_j-1|\leq C_\gamma n^{\alpha-1}
    \end{align}
    uniformly over $\lambda\in \Lambda_n$ and $j\in \{1,\ldots,D\}$. In particular, for all sufficiently large $n$, $|\gamma_j|\leq 2$. 
\end{lemma}
\begin{proof}
    From Eqs.~\eqref{eq:s_ab_lower_upper_bound} and \eqref{eq:s_ab_upper_bound_2nalpha}, 
    \begin{align}
        |\gamma_j^2-1|=\frac{|s_{a_jb_j}(\lambda)-n(\zeta_{a_j}-\zeta_{b_j})|}{n(\zeta_{a_j}-\zeta_{b_j})}\leq \frac{2n^\alpha}{\delta_*n}=\frac{2}{\delta_*}n^{\alpha-1}.
    \end{align}
    Since $\gamma_{j}>0$, $\gamma_{j}+1>1$, and hence $|\gamma_{j}-1|\leq |\gamma_{j}^2-1| \leq C_\gamma n^{\alpha-1}$ with $C_\gamma\coloneqq \frac{2}{\delta_*}$. 
\end{proof}

We now prove an approximate intertwining relation on the truncated space $\mathcal{F}_{\leq N_n}\otimes E_\lambda$. 
\begin{lemma}\label{lem:intw_gen}
    For all sufficiently large $n$, any $\lambda\in \Lambda_n$, $X\in\mathcal{K}_\hor$, and $\psi\in\mathcal{F}_{\leq N_n}\otimes E_\lambda$, 
    \begin{align}
        \|\sdisgen_{\lambda,n}(X)\symword_{\lambda} \psi-\symword_{\lambda}(\fdisgen(X)\otimes I_{E_{\lambda}})\psi\|\leq C_{\mathrm{gen}}  (n^{\alpha-1}L_n^{3/2}+n^{-1/2}L_n^{5/2})\|X\|\|\psi\|.
    \end{align}
    Consequently, if $\|X\|\leq 3n^{\bar{\beta}}$, then
    \begin{align}
         \|\sdisgen_{\lambda,n}(X)\symword_{\lambda } \psi-\symword_{\lambda}(\fdisgen(X)\otimes I_{E_{\lambda}})\psi\|\leq 3C_{\mathrm{gen}}  \delta_n\|\psi\|\label{eq:bound_intw_gen}
    \end{align}
\end{lemma}
\begin{proof}
     Expanding $X=\sum_{j=1}^Dx_je_j$ in the block-adapted orthonormal basis and using Lemma~\ref{lem:sdisgen}, on $\mathcal{F}_{\leq N_n}\otimes E_\lambda$, we have 
    \begin{align}
        &\sdisgen_{\lambda,n}(X)\symword_{\lambda} -\symword_{\lambda}(\fdisgen(X)\otimes I_{E_{\lambda}})\\
        &=\sum_{j=1}^D\gamma_jx_j\left(\crej\symword_{\lambda}-\symword_{\lambda}a_j^\dag\right)-\sum_{j=1}^D\gamma_j\overline{x_j}\left(\annj\symword_{\lambda}-\symword_{\lambda}a_j\right)+\sum_{j=1}^D(\gamma_j-1)\left(x_j\symword_{\lambda} a_j^\dag -\overline{x_j}\symword_{\lambda} a_j\right).
    \end{align}

    Set $\kappa_n\coloneqq n^{\alpha-1}L_n^{3/2}+n^{-1/2}L_n^{5/2}$. Write $\psi=\sum_{k=0}^{N_n}\psi_k$. By Proposition~\ref{prop:intw_cre} together with $q_{L_n}<3$, proven in Theorem~\ref{thm:delta_N}, we have
    \begin{align}
        \|(\crej \symword_{\lambda}-\symword_\lambda a_j^\dag)\psi\|\leq \sum_{k=1}^{N_n+1}\|(\crej \symword_{\lambda,k-1}-\symword_{\lambda,k} a_j^\dag)\psi_{k-1}\|\leq 3C_+n^{-1/2}(N_n+1)^2\sum_{k=0}^{N_n}\|\psi_k\|\leq C_+'\kappa_n\|\psi\|
    \end{align}
    uniformly in $j$, where $C_+'\coloneqq 3C_+$. Here, we used $N_n+1\leq L_n$ and $\sum_{k}\|\psi_k\|\leq \sqrt{L_n}\|\psi\|$. 
    Similarly, by Proposition~\ref{prop:intw_ann}, we have
    \begin{align}
        \|(\annj\symword_\lambda-\symword_\lambda a_j)\psi\|&\leq \sum_{k=1}^{N_n} \|(\annj\symword_{\lambda,k}-\symword_{\lambda,k-1} a_j)\psi_k\|\\
        &\leq 3C_-\left(C_{\mathsf{D}}\sqrt{L_n}\left(n^{\alpha-1}+\frac{L_n}{n}\right)+n^{-1}L_n^{3/2}+n^{-1/2}L_n\right)\sum_{k=0}^{N_n}\|\psi_k\|\leq  C_-'\kappa_n\|\psi\|,
    \end{align}
    uniformly in $j$, where $C_-'\coloneqq 3C_-(C_{\mathsf{D}}+2)$. Since $|\gamma_j|\leq 2$ for sufficiently large $n$ and $\sum_j|x_j|\leq \sqrt{D}\|X\|$, we have
    \begin{align}
        &\left\|\left(\sum_{j=1}^D\gamma_jx_j\left(\crej\symword_{\lambda}-\symword_{\lambda}a_j^\dag\right)-\sum_{j=1}^D\gamma_j\overline{x_j}\left(\annj\symword_{\lambda}-\symword_{\lambda}a_j\right)\right)\psi\right\|\leq 2 (C_+'+C_-')\sqrt{D}\kappa_n\|X\|\|\psi\|.
    \end{align}
     
     Since $\|\gram_\lambda-I\|\leq 1/2$ for all sufficiently large $n$ by Theorem~\ref{thm:qi}, we have $\|\symword_\lambda\|=\|\gram_\lambda\|^{1/2}\leq \sqrt{3/2}<2$. Since
    \begin{align}
        \|a_j^\dag \psi\|^2&=\sum_{k=0}^{N_n}\|a_j^\dag \psi_k\|^2\leq \sum_{k=0}^{N_n}(k+1)\| \psi_k\|^2\leq (N_n+1)\|\psi\|^2,\\
        \|a_j \psi\|^2&=\sum_{k=0}^{N_n}\|a_j \psi_k\|^2\leq \sum_{k=0}^{N_n}k\| \psi_k\|^2\leq N_n\|\psi\|^2
    \end{align}
    we have $\|a_j^\dag \psi\|\leq \sqrt{N_n+1}\|\psi\|$ and $\|a_j \psi\|\leq \sqrt{N_n}\|\psi\|$. Thus, Lemma~\ref{lem:sdisgen} yields
    \begin{align}
        \left\|\sum_{j=1}^D(\gamma_j-1)\left(x_j\symword_{\lambda} a_j^\dag -\overline{x_j}\symword_{\lambda} a_j\right)\psi\right\|\leq 4C_\gamma n^{\alpha-1}\sqrt{L_n}\|\psi\|\sum_{j=1}^D\|x_j\|\leq  4C_\gamma\sqrt{D} \kappa_n\|X\|\|\psi\|.
    \end{align}
    Therefore, 
    \begin{align}
         \|\sdisgen_{\lambda,n}(X)\symword_{\lambda} \psi-\symword_{\lambda}(\fdisgen(X)\otimes I_{E_{\lambda}})\psi\|\leq C_{\mathrm{gen}} (n^{\alpha-1}L_n^{3/2}+n^{-1/2}L_n^{5/2})\|X\|\|\psi\|,
    \end{align}
    where $C_{\mathrm{gen}}\coloneqq  2 (C_+'+C_-')\sqrt{D}+ 4C_\gamma\sqrt{D}$. 
    If $\|X\|\leq 3n^{\bar{\beta}}$, then $\bar{\beta}<\eta/2$ gives $n^{\bar{\beta}}\kappa_n\leq n^{\alpha-1+2\eta}+n^{-1/2+3\eta}=\delta_n$, which proves Eq.~\eqref{eq:bound_intw_gen}.
\end{proof}

In order to prove the approximate covariance, we also need to control the error caused by inserting the cutoff $P_{\leq N_n}$. Since creation and annihilation operators change the particle number by only one, this error is supported entirely on the boundary sectors $N_n$ and $N_n+1$, which the following lemma addresses:
\begin{lemma}\label{lem:bdry_error}
    For any $X\in\mathcal{K}_\hor$, 
    \begin{align}
        [\fdisgen(X),P_{\leq {N_n}}]=a^\dag (X)\Pi_{N_n}+a(X)\Pi_{N_n+1}.
    \end{align}
    Consequently, for $\|X\|\leq 3n^{\bar{\beta}}$ and $t\in [0,1]$,
    \begin{align}
        \| [\fdisgen(X),P_{\leq N_n}]\ket{tX}\|\leq 2 \sqrt{N_n+1}\|X\|2^{-N_n/2}.
    \end{align}
    In particular, for all sufficiently large $n$, $\| [\fdisgen(X),P_{\leq N_n}]\ket{tX}\|\leq \delta_n$. 
\end{lemma}
\begin{proof}
    The creation operator crosses the boundary of $P_{\leq N_n}$ only from the $N_n$-particle sector, while the annihilation operator crosses the boundary only from the $N_n+1$-particle sector. Therefore, we get
    \begin{align}
        a^\dag (X)P_{\leq N_n}-P_{\leq N_n}a^\dag (X)=a^{\dag}(X)\Pi_{N_n},\qquad P_{\leq N_n}a (X)-a (X)P_{\leq N_n}=a(X)\Pi_{N_n+1},
    \end{align}
    and hence $[\fdisgen(X),P_{\leq {N_n}}]=a^\dag (X)\Pi_{N_n}+a(X)\Pi_{N_n+1}$.

    Moreover, by using $\|a^\dag(X) |_{\Sym^k\mathcal{K}_\hor}\|\leq \sqrt{k+1}\|X\|,\,\|a(X) |_{\Sym^k\mathcal{K}_\hor}\|\leq \sqrt{k}\|X\|$ (see Eq.~\eqref{eq:cre_ann_sector_bound}), we get
    \begin{align}
        \| [\fdisgen(X),P_{\leq N_n}]\ket{tX}\|&\leq \|a^\dag (X)\Pi_{N_n}\ket{tX}\|+\|a(X)\Pi_{N_n+1}\ket{tX}\|\\
        &\leq\sqrt{N_n+1}\|X\|\left(\|\Pi_{N_n}\ket{tX}\|+\|\Pi_{N_n+1}\ket{tX}\|\right)\\
        &\leq 2\sqrt{N_n+1}\|X\| 2^{-N_n/2},
    \end{align}
    where we used Lemma~\ref{lem:m_particle_sector_bound} in the last line. Since $2^{-N_n/2}=\exp\left(-\frac{\ln 2}{2}N_n\right)$, the right-hand side is superpolynomially small in $n$. Therefore, for all sufficiently large $n$, $\| [\fdisgen(X),P_{\leq N_n}]\ket{tX}\|\leq \delta_n$. 
\end{proof}

Combining these lemmas, we now prove the approximate covariance.
\begin{proposition}[Approximate covariance on coherent trajectories]\label{prop:single_covariance}
    For all sufficiently large $n$, any $\lambda\in \Lambda_n$, any $X\in\mathcal{K}_\hor$ with $\|X\|\leq3n^{\bar{\beta}}$, and any unit vector $\xi\in E_\lambda$, 
    \begin{align}
        \|U_{\lambda,n}(X)\xi-W_\lambda (\ket{\widehat{X}}\otimes \xi)\|\leq \frac{C_{\mathrm{cv}}}{2}\delta_n
    \end{align}
    for a structural constant $C_{\mathrm{cv}}$. 
    Consequently, 
    \begin{align}
        \|U_{\lambda,n}(X)\omega_\lambda^{H_\lambda}U_{\lambda,n}(X)^\dag - W_\lambda (\ket{\widehat{{X}}}\bra{\widehat{X}}\otimes \omega_\lambda)W_\lambda^\dag\|_1\leq C_{\mathrm{cv}}\delta_n.
    \end{align}
\end{proposition}

\begin{proof}
    For $t\in[0,1]$, we define 
    \begin{align}
         F(t)\coloneqq e^{(1-t)\sdisgen_{\lambda,n}(X)}\symword_\lambda (P_{\leq N_n}\ket{tX}\otimes \xi).
    \end{align}
    Since $\symword_{\lambda,0}(\ket{0}\otimes\xi)=\xi$, 
    \begin{align}
        F(0)=e^{\sdisgen_{\lambda,n}(X)}\xi=U_{\lambda,n}(X)\xi,\qquad F(1)=\symword_\lambda (P_{\leq N_n}\ket{X}\otimes \xi).
    \end{align}
    By using $\frac{\dd}{\dd t}\ket{tX}=\fdisgen(X)\ket{tX}$, we obtain
    \begin{align}
        F'(t)&= e^{(1-t)\sdisgen_{\lambda,n}(X)}\left(-\sdisgen_{\lambda,n}(X)\symword_\lambda (P_{\leq N_n}\ket{tX}\otimes \xi)+\symword_\lambda (P_{\leq N_n}\fdisgen(X)\ket{tX}\otimes \xi)\right)\\
        &=e^{(1-t)\sdisgen_{\lambda,n}(X)}\left((\symword_\lambda(\fdisgen(X)\otimes I_{E_\lambda})-\sdisgen_{\lambda,n}(X)\symword_\lambda) (P_{\leq N_n}\ket{tX}\otimes \xi)-\symword_\lambda ([\fdisgen(X),P_{\leq N_n}]\ket{tX}\otimes \xi)\right).
    \end{align}
    By using Lemma~\ref{lem:intw_gen} with $\|(P_{\leq N_n}\ket{tX}\otimes \xi)\|\leq 1$ and Lemma~\ref{lem:bdry_error} with $\|\symword_\lambda\|<2$, we get
    \begin{align}
        \|F'(t)\|\leq 3C_{\mathrm{gen}}  \delta_n+2\delta_n,
    \end{align}
    and hence $\sup_{t\in[0,1]}\|F'(t)\|\leq (3C_{\mathrm{gen}}+2)\delta_n$. Therefore,
    \begin{align}
        \|U_{\lambda,n}(X)\xi-\symword_\lambda( P_{\leq N_n}\ket{X}\otimes \xi)\|\leq \int_0^1\|F'(t)\|\dd t\leq (3C_{\mathrm{gen}}+2)\delta_n.
    \end{align}
    
    By Lemma~\ref{lem:m_particle_sector_bound}, 
    \begin{align}
         \|P_{\leq N_n}\ket{X}-\ket{\widehat{X}}\|=1-\sqrt{t_X}\leq 1-t_X\leq 2^{-N_n},
    \end{align}
    which is superpolynomially small in $n$. Therefore, for all sufficiently large $n$, $\|P_{\leq N_n}\ket{X}-\ket{\widehat{X}}\|\leq \delta_n$. Moreover, since $\ket{\widehat{X}}\otimes \xi$ is a unit vector, by Property~(ii) in Lemma~\ref{lemma:properties_W_lambda}
    \begin{align}
        \|\symword_\lambda(\ket{\widehat{X}}\otimes \xi)-W_\lambda (\ket{\widehat{X}}\otimes \xi)\|\leq \|\symword_\lambda-W_\lambda\|\leq C^{\mathrm{qi}}\delta_n
    \end{align}
    By using $\|\symword_\lambda\|<2$, we get
    \begin{align}
        &\|U_{\lambda,n}(X)\xi-W_\lambda( \ket{\widehat{X}}\otimes \xi)\|\\
        &\leq \|U_{\lambda,n}(X)\xi-\symword_\lambda( P_{\leq N_n}\ket{X}\otimes \xi)\|+\|\symword_\lambda( P_{\leq N_n}\ket{X}\otimes \xi-\ket{\widehat{X}}\otimes \xi)\|+\|\symword_\lambda(\ket{\widehat{X}}\otimes \xi)-W_\lambda (\ket{\widehat{X}}\otimes \xi)\|\\
        &\leq \frac{C_{\mathrm{cv}}}{2}\delta_n
    \end{align}
    with $C_{\mathrm{cv}}\coloneqq 2(3C_{\mathrm{gen}}+4+C^{\mathrm{qi}})$. 

    Let $\{\xi_i\}_{i=1}^{\dim E_\lambda}$ be an orthonormal basis of $E_\lambda$. Then, 
    \begin{align}
        U_{\lambda,n}(X)\omega_\lambda^{H_\lambda}U_{\lambda,n}(X)^\dag &=\frac{1}{\dim E_\lambda}\sum_{i=1}^{\dim E_\lambda}\ket{u_i}\bra{u_i},\qquad \ket{u_i}\coloneqq U_{\lambda,n}(X)\xi_i,\\
        W_\lambda(\ket{\widehat{X}}\bra{\widehat{X}}\otimes \omega_\lambda)W_\lambda^\dag &=\frac{1}{\dim E_\lambda}\sum_{i=1}^{\dim E_\lambda}\ket{v_i}\bra{v_i},\qquad \ket{v_i}\coloneqq W_\lambda(\ket{\widehat{X}}\otimes \xi_i).
    \end{align}
    Since $U_{\lambda,n}(X)$ is unitary and $W_\lambda$ is an isometry, $\ket{u_i}$ and $\ket{v_i}$ are unit vectors. For unit vectors $\ket{u},\ket{v}$, we have
    \begin{align}
        \|\ket{u}\bra{u}-\ket{v}\bra{v}\|_1\leq 2\|u-v\|.
    \end{align}
    Therefore,
    \begin{align}
        \left\|U_{\lambda,n}(X)\omega_\lambda^{H_\lambda}U_{\lambda,n}(X)^\dag -W_\lambda(\ket{\widehat{X}}\bra{\widehat{X}}\otimes \omega_\lambda)W_\lambda^\dag \right\|_1\leq \frac{2}{\dim E_\lambda}\sum_{i=1}^{\dim E_\lambda}\|u_i-v_i\|\leq C_{\mathrm{cv}}\delta_n.
    \end{align}
\end{proof}

\subsubsection{Uniform Approximation of Displaced Block States}
The scale conditions in Eq.~\eqref{eq:scale_eta_beta_alpha} remain in force throughout this subsubsection. 
Proposition~\ref{prop:single_covariance} establishes the asymptotic correspondence between a \textit{pure} coherent state in Fock space and a unitary orbit on an $n$-fold finite-dimensional system. In order to establish our main theorem, Theorem~\ref{thm:QLAN}, we also need to establish such a correspondence for \textit{mixed} Gaussian states. To this end, we review an alternative representation of a mixed Gaussian state written in terms of an integral of pure coherent states over a Gaussian noise, which is known as Glauber–Sudarshan $P$-representation~\cite{glauber_CoherentIncoherentStatesRadiationField_1963,sudarshan_EquivalenceSemiclassicalQuantumMechanicalDescriptions_1963} (see also Appendix~C.1 in Ref.~\cite{lahiryMinimaxEstimationLowrank2024}).

We first consider a single-mode system. For a pure coherent state
\begin{align}
    \ket{z}\coloneqq D(z)\ket{0}=e^{-|z|^2/2}\sum_{m=0}^\infty\frac{z^m}{\sqrt{m!}}\ket{m},
\end{align}
we have
\begin{align}
    \frac{e^\beta-1}{\pi}\int e^{-(e^\beta-1)|z|^2}\ket{z}\bra{z}\dd^2z&= \sum_{m=0}^\infty\frac{e^\beta-1}{\pi}\int_0^\infty e^{-(e^\beta-1)r^2}e^{- r^2} \frac{r^{2m}}{m!}\ket{m}\bra{m}2\pi r\dd r\\
    &=(e^\beta-1)\sum_{m=0}^\infty \frac{1}{e^{\beta(m+1)}}\ket{m}\bra{m}=\phi_\beta.
\end{align}

We extend this relation to a multi-mode Gaussian state. Let $\mathcal{K}_{\mathrm{th}}\coloneqq \bigoplus_{\zeta_a>\zeta_b>0}\mathcal{K}_{ab}$. With respect to the fixed block-adapted basis $\{e_j\}_{j=1}^D$ of $\mathcal{K}_\hor$, let $\mathcal{J}_{\mathrm{th}}\coloneqq \{j\colon \zeta_{b_j}>0\}$ be the set of labels describing thermal states, let $\beta_j$ denote the corresponding inverse temperature. Identifying $\mathcal{K}_{\mathrm{th}}$ with $\mathbb{C}^{|\mathcal{J}_{\mathrm{th}}|}$, we define
\begin{align}
    \dd \mu(Y)\coloneqq \prod_{j\in\mathcal{J}_{\mathrm{th}}}\frac{e^{\beta_{j}}-1}{\pi}e^{-(e^{\beta_j}-1)|y_j|^2}\dd^2y_j,\qquad Y=\sum_{j\in\mathcal{J}_{\mathrm{th}}}y_je_j\in\mathcal{K}_{\mathrm{th}}.
\end{align}
Then, by using the pure coherent state $\ket{Y}\in\Gamma_{\mathrm{s}}(\mathcal{K}_{\mathrm{th}})$, we have
\begin{align}
    \int_{\mathcal{K}_{\mathrm{th}}} \ket{Y}\bra{Y} \dd \mu(Y)=\bigotimes_{\zeta_a>\zeta_b>0}\phi_{\beta_{ab}}^{\otimes d_ad_b}
\end{align}
as an identity of states on $\Gamma_{\mathrm{s}}(\mathcal{K}_{\mathrm{th}})$. 

We define the kernel-mode subspace $\mathcal{K}_{\ker}\coloneqq\bigoplus_{\zeta_a>\zeta_b=0}\mathcal{K}_{ab}$ so that
\begin{align}
    \mathcal{K}_\hor=\mathcal{K}_{\mathrm{th}}\oplus \mathcal{K}_{\ker}.
\end{align}
We define an embedding $\iota :\mathcal{K}_{\mathrm{th}}\hookrightarrow \mathcal{K}_\hor$ by $\iota (Y)\coloneqq (Y,0)$. Since $\ket{\iota Y}=\ket{Y}_{\mathcal{K}_{\mathrm{th}}}\otimes \ket{0}_{\mathcal{K}_{\ker}}$, we obtain
\begin{align}
    \Phi_0=\bigotimes_{\zeta_a>\zeta_b>0}\phi_{\beta_{ab}}^{\otimes d_ad_b}\otimes \bigotimes_{\zeta_a>0,\,\zeta_b=0}(\ket{0}\bra{0})^{\otimes d_ad_b}= \int_{\mathcal{K}_{\mathrm{th}}} \ket{\iota Y}\bra{\iota Y}\dd \mu(Y).\label{eq:displaced_state_int_expression_no_Z}
\end{align}
Consequently, the displaced state is written as
\begin{align}
    \Phi_Z=D(Z)\Phi_0D(Z)^\dag=\int_{\mathcal{K}_{\mathrm{th}}} \ket{Z+\iota Y}\bra{Z+\iota Y}\dd \mu(Y).\label{eq:displaced_state_int_expression}
\end{align}
Here we used the Weyl relation $D(Z)\ket{\iota Y}=e^{\ii\theta(Z,Y)}\ket{Z+\iota Y}$ for some phase $\theta(Z,Y)$, and hence $D(Z)\ket{\iota Y}\bra{\iota Y}D(Z)^\dag=\ket{Z+\iota Y}\bra{Z+\iota Y}$.

We now define
\begin{align}
    G_n\coloneqq \{ Y\in\mathcal{K}_{\mathrm{th}}\colon \|Y\|\leq n^{\bar{\beta}}\}\subset \mathcal{K}_{\mathrm{th}}.
\end{align}
By using the above integral expression, we prove the following bounds about the tails. 

\begin{proposition}\label{prop:Gaussian_tail_mu}
    Let $G_n^c$ be the complement set of $G_n$. Then,
    \begin{align}
        \mu(G_n^c)\leq 2^D e^{-\frac{\delta_*}{2}n^{2\bar{\beta}}}.
    \end{align}
    Moreover, uniformly for $\|Z\|\leq n^\beta$ and all sufficiently large $n$, 
    \begin{align}
        \Tr((I-P_{\leq N_n})\Phi_Z)\leq \frac{1}{4}\delta_n^2,\qquad \|\Phi_Z-P_{\leq N_n}\Phi_ZP_{\leq N_n}\|_1\leq  \delta_n.\label{eq:tail_gaussian_state_Z}
    \end{align}
\end{proposition}
\begin{proof}
    For $s>0$, by using the exponential Markov inequality, we have
    \begin{align}
        \mu(G_n^c)=\mathbb{P}(\|Y\|^2>n^{2\bar{\beta}})=\mathbb{P}(e^{s\|Y\|^2}>e^{sn^{2\bar{\beta}}})\leq e^{-sn^{2\bar{\beta}}}\mathbb{E}e^{s\|Y\|^2}.
    \end{align}
    Setting $s=\delta_*/2$, we obtain
    \begin{align}
        \mu(G_n^c)\leq e^{-\frac{\delta_*}{2}n^{2\bar{\beta}}}\mathbb{E}e^{\frac{\delta_*}{2}\|Y\|^2}.
    \end{align}
    Since $\zeta_{b_j}\leq 1$, we have $e^{\beta_j}-1=\frac{\zeta_{a_j}}{\zeta_{b_j}}-1\geq \zeta_{a_j}-\zeta_{b_j}\geq \delta_*>\frac{\delta_*}{2}$ and therefore
    \begin{align}
        \mathbb{E}e^{\frac{\delta_*}{2}\|Y\|^2}=\prod_{j\in\mathcal{J}_{\mathrm{th}}}\mathbb{E}e^{\frac{\delta_*}{2}\|y_j\|^2}=\prod_{j\in\mathcal{J}_{\mathrm{th}}}\frac{e^{\beta_j}-1}{e^{\beta_j}-1-\frac{\delta_*}{2}}\leq\prod_{j\in\mathcal{J}_{\mathrm{th}}}\frac{\delta_*}{\delta_*-\frac{\delta_*}{2}} =2^{|\mathcal{J}_{\mathrm{th}}|}\leq 2^{D}.
    \end{align}
    Thus, $ \mu(G_n^c)\leq 2^D e^{-\frac{\delta_*}{2}n^{2\bar{\beta}}}$. 

    By using Eq.~\eqref{eq:displaced_state_int_expression}, we have
    \begin{align}
        \Tr((I-P_{\leq N_n})\Phi_Z)&=\int_{\mathcal{K}_{\mathrm{th}}} \Tr((I-P_{\leq N_n})\ket{Z+\iota Y}\bra{Z+\iota Y})\dd \mu(Y)\\
        &=\int_{\mathcal{K}_{\mathrm{th}}} (1-t_{Z+\iota Y})\dd \mu(Y)=\int_{G_n} (1-t_{Z+\iota Y})\dd \mu(Y)+\int_{G_n^c} (1-t_{Z+\iota Y})\dd \mu(Y),
    \end{align}
    where $t_X\coloneqq \|P_{\leq N_n}\ket{X}\|^2$ is defined in Eq.~\eqref{eq:cutoff_coherent_state}. For $Y\in G_n$, we have $\|Z+\iota Y\|\leq \|Z\|+\|Y\|\leq n^{\beta}+n^{\bar{\beta}}\leq 2n^{\bar{\beta}}$. By Lemma~\ref{lem:m_particle_sector_bound}, we get
    \begin{align}
        \int_{G_n} (1-t_{Z+\iota Y})\dd \mu(Y)+\int_{G_n^c} (1-t_{Z+\iota Y})\dd \mu(Y)\leq 2^{-N_n} +\mu(G_n^c).
    \end{align}
    Therefore,
    \begin{align}
        \Tr((I-P_{\leq N_n})\Phi_Z)\leq 2^{-N_n}+2^D e^{-\frac{\delta_*}{2}n^{2\bar{\beta}}}.
    \end{align}
    Since the right-hand side is superpolynomially small in $n$, for all sufficiently large $n$, it holds
    \begin{align}
        \Tr((I-P_{\leq N_n})\Phi_Z)\leq \frac{1}{4}\delta_n^2.
    \end{align}
    The other inequality follows by using the gentle measurement lemma~\cite{winter_CodingTheoremStrongconversequantum_1999,ogawa_MakingGoodCodesClassicalQuantumChannel_2007a}.
\end{proof}

As Eq.~\eqref{eq:displaced_state_int_expression} shows, the limit state $\Phi_Z$ can be represented as a Gaussian mixture of coherent states obtained by the successive displacements $D(Z)D(\iota Y)$. We now establish the finite-dimensional analogue by replacing these two Weyl displacements with the corresponding horizontal unitaries.

Recall the orthogonal decomposition $\mathfrak{u}(\mathcal{H})=\mathfrak{h}\oplus\mathfrak{u}_\hor$ in Section~\ref{subsec:unitary_orbit_finite_dim}. 

\begin{lemma}\label{lem:prod_local_displacement}
    There exist structural constants $s_0$, $C_H$ such that, for all $A,B\in\mathfrak{u}_\hor$ such that $\|A\|+\|B\|\leq s_0$, there exists $M\in\mathfrak{u}_\hor$ and $h\in H$ satisfying
    \begin{align}
        e^Ae^B=e^Mh,\qquad \|M-(A+B)\|\leq C_H(\|A\|+\|B\|)^2.
    \end{align}
\end{lemma}
\begin{proof}
    Note that the above statement and the proof below are all about properties of $\mathfrak{u}(\mathcal{H})$ and its decomposition  $\mathfrak{u}(\mathcal{H})=\mathfrak{h}\oplus\mathfrak{u}_\hor$. Thus, all constants appearing in the arguments are structural. 

    For $W$ such that $\|W\|<1$, its logarithm is defined by
    \begin{align}
        \ln (I+W)\coloneqq \sum_{k\geq 1}\frac{(-1)^{k+1}}{k}W^k.
    \end{align}
    Moreover, if $\|W\|\leq 1/2$, by using $\|W^k\|\leq \|W\|^k$ and the triangle inequality, we get
    \begin{align}
        \|\ln(I+W)-W\|\leq \sum_{k=2}\frac{\|W\|^k}{k}\leq \frac{1}{2}\sum_{k=2}^\infty\|W\|^k=\frac{1}{2}\frac{\|W\|^2}{1-\|W\|}\leq \|W\|^2.\label{eq:ln_IplusW}
    \end{align}

    Let $U$ be unitary with $\|U-I\|\leq 2/5$. Let $e^{\ii \theta}$ be an eigenvalue of $U$ with normalized eigenvector $v$, where $\theta\in(-\pi,\pi]$ is its principal argument. Then, $|e^{\ii\theta}-1|=\|(U-I)v\|\leq \|U-I\|\leq 2/5$. Since $|e^{\ii\theta}-1|=2|\sin(\theta/2)|$, we have $|\theta|\leq 2\arcsin(1/5)<\pi$. Thus, the spectrum of $U$ avoids the negative real axis. By the spectral theorem, the logarithm defined by the above convergent power series coincides with the principal logarithm:
    \begin{align}
        \ln U=\sum_{k\geq 1}\frac{(-1)^{k+1}}{k}(U-I)^k=\sum_j\left(\sum_{k\geq 1}\frac{(-1)^{k+1}}{k}(e^{\ii \theta_j}-1)^k\right)P_j=\sum_j\ii\theta_jP_j,
    \end{align}
    where $U=\sum_je^{\ii\theta_j}P_j$ is the spectral decomposition, and $\theta_j\in(-\pi,\pi)$. Consequently, $(\ln U)^\dag = -\ln U$, implying that $\ln U\in\mathfrak{u}(\mathcal{H})$.

    By using $\|ST\|\leq \|S\|\|T\|$, we have
    \begin{align}
        \|e^Se^T-I\|=\left\|\sum_{\substack{j,k\geq 0\\ j+k\geq 1}}\frac{S^jT^k}{j!k!}\right\|\leq \sum_{\substack{j,k\geq 0\\ j+k\geq 1}}\frac{\|S\|^j\|T\|^k}{j!k!}=e^{\|S\|+\|T\|}-1.
    \end{align}
    Thus, for $S,T\in\mathfrak{u}(\mathcal{H})$ satisfying $\|S\|+\|T\|<1/3$, the operator $e^{S}e^T$ is unitary and $\|e^Se^T-I\|\leq e^{1/3}-1<\frac{2}{5}$. Therefore, its principal logarithm is well-defined and belongs to $\mathfrak{u}(\mathcal{H})$.

    Let $\pr_{\mathfrak{u}_\hor}$ and $\pr_{\mathfrak{h}}$ be orthogonal projections via $\mathfrak{u}=\mathfrak{u}_\hor\oplus \mathfrak{h}$. For $X\in\mathfrak{u}(\mathcal{H})$, write
    \begin{align}
        X=X_{\mathfrak{u}_\hor}+ X_{\mathfrak{h}},\qquad X_{\mathfrak{u}_\hor}\coloneqq \pr_{\mathfrak{u}_\hor} X,\qquad X_{\mathfrak{h}}\coloneqq \pr_{\mathfrak{h}}X.
    \end{align}
    Since the projections $\pr_{\mathfrak{u}_\hor}$ and $\pr_{\mathfrak{h}}$ are bounded as $\mathfrak{u}(\mathcal{H})$ is finite-dimensional, there exists a structural constant $C_{\pr}$ such that
    \begin{align}
        \|X_{\mathfrak{u}_\hor}\|+\| X_{\mathfrak{h}}\|\leq C_{\pr}\|X\|.\label{eq:cpr}
    \end{align}
    Therefore, for $X$ such that $\|X\|< \frac{1}{3C_{\pr}}$, logarithm of $e^{X_{\mathfrak{u}_\hor}}e^{X_{\mathfrak{h}}}$ is well-defined. For $X\in \mathfrak{u}(\mathcal{H})$ with $\|X\|< \frac{1}{3C_{\pr}}$, we define a function
    \begin{align}
        f(X)\coloneqq \ln (e^{X_{\mathfrak{u}_\hor} } e^{X_{\mathfrak{h}}})\in \mathfrak{u}(\mathcal{H}).
    \end{align}
    Since the exponential map and the above local logarithm are analytic, $f$ is smooth. 
    
    Now, we show that $f(X)=X+O(\|X\|^2)$. For $r<1/3$, we have
    \begin{align}
        e^r-1&=\sum_{k=1}^\infty\frac{r^k}{k!}\leq \sum_{k=1}^\infty r^k=\frac{r}{1-r}\leq \frac{3}{2}r,\\
        e^r-(1+r)&=\sum_{k=2}^\infty\frac{r^k}{k!}\leq \frac{1}{2}\sum_{k=2}^\infty r^k=\frac{r^2}{2(1-r)}\leq r^2.
    \end{align}
    Since $r\coloneqq \|X_{\mathfrak{u}_\hor}\|+\| X_{\mathfrak{h}}\|<1/3$ for $\|X\|< \frac{1}{3C_{\pr}}$, 
    \begin{align}
        \|e^{X_{\mathfrak{u}_\hor}}e^{X_{\mathfrak{h}}}-I\|&\leq\sum_{\substack{j\geq 0,\,k\geq 0\\j+k\geq 1}}\frac{\|X_{\mathfrak{u}_\hor}\|^j\|X_{\mathfrak{h}}\|^k}{j!k!}=\sum_{l=1}^\infty\frac{r^l}{l!}\leq \frac{3}{2}r\\
        \|e^{X_{\mathfrak{u}_\hor}}e^{X_{\mathfrak{h}}}-I-X\|&\leq \sum_{\substack{j\geq 0,\,k\geq 0\\j+k\geq 2}}\frac{\|X_{\mathfrak{u}_\hor}\|^j\|X_{\mathfrak{h}}\|^k}{j!k!}=\sum_{l=2}^\infty\frac{r^l}{l!}\leq r^2.
    \end{align}
    Thus, we obtain 
    \begin{align}
        \|f(X)-X\|&\leq \|\ln (e^{X_{\mathfrak{u}_\hor} } e^{X_{\mathfrak{h}}})-(e^{X_{\mathfrak{u}_\hor}}e^{X_{\mathfrak{h}}}-I)\|+\|e^{X_{\mathfrak{u}_\hor}}e^{X_{\mathfrak{h}}}-I-X\|\\
        &\leq \|e^{X_{\mathfrak{u}_\hor}}e^{X_{\mathfrak{h}}}-I\|^2+\|e^{X_{\mathfrak{u}_\hor}}e^{X_{\mathfrak{h}}}-I-X\| &&(\text{Eq.}~\eqref{eq:ln_IplusW})\\
        &\leq \frac{13}{4}r^2\leq \frac{13}{4} C_{\pr}^2\|X\|^2,\label{eq:f_X_X_bound}&&(\text{Eq.}~\eqref{eq:cpr})
    \end{align}
    i.e., $f(X)=X+O(\|X\|^2)$.

    At the origin, we have $f(0)=0$. Since $f(\varepsilon X)=\varepsilon X+O(\varepsilon^2)$, we have $\dd f(0)=I$, which is invertible. Therefore, from the inverse function theorem, there are neighborhoods of $\mathcal{U}$ and $\mathcal{V}$ of $0$ such that $f:\mathcal{U}\to\mathcal{V}$ is a diffeomorphism. We denote the inverse function by $g\coloneqq f^{-1}$. 

    Shrinking $\mathcal{U}$, and replacing $\mathcal{V}$ by $f(\mathcal{U})$, if necessary, we may assume that $\frac{13}{4} C_{\pr}^2\|X\|\leq \frac{1}{2}$ for $X\in \mathcal{U}$. We define $V\coloneqq f(X)$. Since $f(X)=V$, Eq.~\eqref{eq:f_X_X_bound} implies $V=X+R(X)$ with $\|R(X)\|\leq  \frac{13}{4} C_{\pr}^2\|X\|^2\leq \frac{1}{2}\|X\|$. Then, we get
    \begin{align}
        \|V\|\geq \|X\|-\|R(X)\|\geq \frac{1}{2}\|X\|.
    \end{align}
    Therefore,
    \begin{align}
        \|g(V)-V\|=\|X-V\|=\|R(X)\|\leq  \frac{13}{4} C_{\pr}^2\|X\|^2\leq 13 C_{\pr}^2\|V\|^2,\label{eq:gV_V}
    \end{align}
    which implies that local inverse satisfies $g(V)=V+O(\|V\|^2)$.

    We now apply these coordinates to a unitary $U=e^Ae^B$ with $A,B\in \mathfrak{u}_\hor$. When $s\coloneqq \|A\|+\|B\|\leq 1/3$, 
    \begin{align}
        \|e^Ae^B-I\|\leq \frac{3}{2}s,\qquad \|e^Ae^B-(I+A+B)\|\leq s^2.
    \end{align}
    For $V\coloneqq \ln U$, we have
    \begin{align}
        \|V-(A+B)\|\leq \frac{13}{4}s^2.
    \end{align}
    Thus, we also obtain
    \begin{align}
        \|V\|\leq \|V-(A+B)\|+\|A+B\|\leq \frac{13}{4}s^2+s\leq 3s,
    \end{align}
    where we have used $s\leq 1/3$. Since $\mathcal{V}$ is a neighborhood of $0$, there exists $\varepsilon_{\mathcal{V}}>0$ such that $\{Y\in\mathfrak{u}(\mathcal{H})\colon \|Y\|\leq \varepsilon_{\mathcal{V}}\}\subset\mathcal{V}$. Since $\|V\|\leq 3s$, choosing $s_0\coloneqq \min \{1/3,\varepsilon_{\mathcal{V}}/3\}$, we have $V\in\mathcal{V}$ whenever $s\leq  s_0$.
    Set $X\coloneqq g(V)$ and decompose $X=M+T$, where $M\coloneqq \pr_{\mathfrak{u}_\hor} X$ and $T\coloneqq \pr_{\mathfrak{h}}X$. By the definition of $f$, we have
    \begin{align}
        V=f(X)=\ln (e^Me^T).
    \end{align}
    Therefore,
    \begin{align}
        e^Ae^B=U=e^V=e^Me^T=e^Mh,
    \end{align}
    where we have defined $h\coloneqq e^T\in H$.
    
    By using Eq.~\eqref{eq:gV_V}, 
    \begin{align}
        \|V-X\|=\|f(X)-X\|\leq13 C_{\pr}^2\|V\|^2\leq 117 C_{\pr}^2s^2.
    \end{align}
    Since $\pr_{\mathfrak{u}_\hor}(A+B)=A+B$, we get
    \begin{align}
        \|M-(A+B)\|&=\|\pr_{\mathfrak{u}_\hor}(X-(A+B))\|\\
        &\leq C_{\pr}\|X-(A+B)\|\leq  C_{\pr}\left(\|X-V\|+\|V-(A+B)\|\right)\leq C_Hs^2,
    \end{align}
    where $C_H\coloneqq \frac{13}{4}C_{\pr}+ 117 C_{\pr}^3$.
\end{proof}

As a consequence, we prove the following.
\begin{proposition}[Addition of local displacement operators]\label{prop:local_displacement_addition}
    For all sufficiently large $n$, any $Z\in\mathcal{K}_\hor$ with $\|Z\|\leq n^\beta$, and any $Y\in\mathcal{K}_{\mathrm{th}}$ with $\|Y\|\leq n^{\bar{\beta}}$, there exist $X_n=X_n(Z,Y)\in\mathcal{K}_\hor$ and $h_n=h_n(Z,Y)\in H$ such that
    \begin{align}
        e^{n^{-1/2}K(Z)}e^{n^{-1/2}K(\iota Y)}&=e^{n^{-1/2}K(X_n)}h_n,\\
        \|X_n-(Z+\iota Y)\|&\leq C_{\mathrm{ld}} n^{-1/2+2\bar{\beta}},\\
        \|X_n\|&\leq 3 n^{\bar{\beta}},
    \end{align}
    where $C_{\mathrm{ld}}$ is a structural constant. Consequently,
    \begin{align}
        U_{\lambda,n}(Z)U_{\lambda,n}(\iota Y)\omega_\lambda^{H_\lambda}U_{\lambda,n}(\iota Y)^\dag U_{\lambda,n}(Z)^\dag =U_{\lambda,n}(X_n)\omega_\lambda^{H_\lambda}U_{\lambda,n}(X_n)^\dag.\label{eq:uz_uy_ux_omega}
    \end{align}
\end{proposition}
\begin{proof}
    Set $A_n\coloneqq n^{-1/2}K(Z)\in\mathfrak{u}_\hor$ and $B_n\coloneqq  n^{-1/2}K(\iota Y)\in\mathfrak{u}_\hor$. Since $K$ is bounded, there exists a structural constant $C_K$ such that
    \begin{align}
        \|A_n\|+\|B_n\|\leq C_Kn^{-1/2}(n^\beta+n^{\bar{\beta}})\leq 2C_K n^{-1/2+\bar{\beta}}\to 0.
    \end{align}
    By Lemma~\ref{lem:prod_local_displacement}, for all sufficiently large $n$, there are $M_n\in \mathfrak{u}_\hor$ and $h_n\in H$ such that
    \begin{align}
        e^{A_n}e^{B_n}=e^{M_n}h_n,\qquad \|M_n-(A_n+B_n)\|\leq 4C_K^2C_Hn^{-1+2\bar{\beta}}.\label{eq:uz_uy_ux}
    \end{align}
    As shown in Lemma~\ref{lem:real_lin_K}, the map $K:\mathcal{K}_\hor\to \mathfrak{u}_\hor$ is a real-linear isomorphism, and hence its inverse $K^{-1}$ exists. Thus, define $X_n\coloneqq K^{-1}(\sqrt{n}M_n)\in \mathcal{K}_\hor$. Since $K^{-1}$ is bounded, there exists a structural constant $C_K'$ such that
    \begin{align}
        \|X_n-(Z+\iota Y)\|\leq C_K'\sqrt{n}\|M_n-(A_n+B_n)\|\leq C_{\mathrm{ld}}n^{-1/2+2\bar{\beta}},
    \end{align}
    where $C_{\mathrm{ld}}\coloneqq 4C_K^2C_HC_K'$. For large $n$, $C_{\mathrm{ld}}n^{-1/2+2\bar{\beta}}\leq n^{\bar{\beta}}$ since $\bar{\beta}<1/2$. Thus, since $\beta<\bar{\beta}$, we get
    \begin{align}
        \|X_n\|\leq \|X_n-(Z+\iota Y)\|+\|Z\|+\|\iota Y\|\leq 3 n^{\bar{\beta}}. 
    \end{align}
    
    Applying $\pi_\lambda$ to Eq.~\eqref{eq:uz_uy_ux}, we get
    \begin{align}
        U_{\lambda,n}(Z)U_{\lambda,n}(\iota Y)=U_{\lambda,n}(X_n)\pi_\lambda(h_n).
    \end{align}
    By Lemma~\ref{lem:properties_omegas}, $\omega_\lambda^{H_\lambda}$ is invariant under $\pi_\lambda(h_n)$, thereby yielding Eq.~\eqref{eq:uz_uy_ux_omega}. 
\end{proof}

\begin{lemma}\label{lem:comparison_coherent_traj}
    There exists a structural constant $C_{\mathrm{traj}}$ such that, for all sufficiently large $n$, any $\lambda\in\Lambda_n$, any $Z\in\mathcal{K}_\hor$ with $\|Z\|\leq n^{\beta}$, and any $Y\in G_n$, we have
    \begin{align}
        \|U_{\lambda,n}(Z)W_\lambda (P_{\leq N_n}\ket{\iota Y}\bra{\iota Y}P_{\leq N_n}\otimes \omega_\lambda)W_\lambda^\dag U_{\lambda,n}(Z)^\dag -W_\lambda (P_{\leq N_n}\ket{V}\bra{V}P_{\leq N_n}\otimes \omega_\lambda)W_\lambda^\dag\|_1\leq C_{\mathrm{traj}}\delta_n,\label{eq:comparison_coherent_traj}
    \end{align}
    where $V\coloneqq Z+\iota Y$. 
\end{lemma}

\begin{proof}
    From Proposition~\ref{prop:local_displacement_addition}, there exists $X_n\in\mathcal{K}_\hor$ such that $\|X_n-(Z+\iota Y)\|\leq C_{\mathrm{ld}} n^{-1/2+2\bar{\beta}}$, $\|X_n\|\leq 3n^{\bar{\beta}}$, and 
    \begin{align}
        U_{\lambda,n}(Z)U_{\lambda,n}(\iota Y)\omega_\lambda^{H_\lambda}U_{\lambda,n}(\iota Y)^\dag U_{\lambda,n}(Z)^\dag =U_{\lambda,n}(X_n)\omega_\lambda^{H_\lambda}U_{\lambda,n}(X_n)^\dag,\label{eq:Z_Y_X_n}
    \end{align}
    Since $Y\in G_n$, $\|\iota Y\|\leq n^{\bar{\beta}}$, and hence, $\|V\|\leq n^\beta+n^{\bar{\beta}}\leq 2n^{\bar{\beta}}$. 
    
    Let $\ket{\widehat{X}_n}$, $\ket{\widehat{V}}$ and $\ket{\widehat{\iota Y}}$ be the normalized coherent states with cutoff $N_n$, defined in Eq.~\eqref{eq:cutoff_coherent_state}. 
    From the triangle inequality, we have
    \begin{align}
        &\|U_{\lambda,n}(Z)W_\lambda (P_{\leq N_n}\ket{\iota Y}\bra{\iota Y}P_{\leq N_n}\otimes \omega_\lambda)W_\lambda^\dag U_{\lambda,n}(Z)^\dag -W_\lambda (P_{\leq N_n}\ket{V}\bra{V}P_{\leq N_n}\otimes \omega_\lambda)W_\lambda^\dag\|_1\\
        &\leq \|U_{\lambda,n}(Z)W_\lambda (P_{\leq N_n}\ket{\iota Y}\bra{\iota Y}P_{\leq N_n}\otimes \omega_\lambda)W_\lambda^\dag U_{\lambda,n}(Z)^\dag -U_{\lambda,n}(Z)W_\lambda (\ket{\widehat{\iota Y}}\bra{\widehat{\iota Y}}\otimes \omega_\lambda)W_\lambda ^\dag U_{\lambda,n}(Z)^\dag\|_1\nonumber\\
        &+\|U_{\lambda,n}(Z)W_\lambda (\ket{\widehat{\iota Y}}\bra{\widehat{\iota Y}}\otimes \omega_\lambda)W_\lambda^\dag U_{\lambda,n}(Z)^\dag-U_{\lambda,n}(Z)U_{\lambda,n}(\iota Y)\omega_\lambda^{H_\lambda }U_{\lambda,n}(\iota Y)^\dag U_{\lambda,n}(Z)^\dag\|_1\nonumber\\
        &+\|U_{\lambda,n}(Z)U_{\lambda,n}(\iota Y)\omega_\lambda^{H_\lambda }U_{\lambda,n}(\iota Y)^\dag U_{\lambda,n}(Z)^\dag-U_{\lambda,n}(X_n)\omega_\lambda^{H_\lambda}U_{\lambda,n}(X_n)^\dag\|_1\nonumber\\
        &+\|U_{\lambda,n}(X_n)\omega_\lambda^{H_\lambda}U_{\lambda,n}(X_n)^\dag-W_\lambda(\ket{\widehat{X_n}}\bra{\widehat{X_n}}\otimes \omega_\lambda)W_\lambda^\dag\|_1\nonumber\\
        &+\|W_\lambda(\ket{\widehat{X_n}}\bra{\widehat{X_n}}\otimes \omega_\lambda)W_\lambda^\dag-W_\lambda(\ket{\widehat{V}}\bra{\widehat{V}}\otimes \omega_\lambda)W_\lambda^\dag\|_1\nonumber\\
        &+\|W_\lambda(\ket{\widehat{V}}\bra{\widehat{V}}\otimes \omega_\lambda)W_\lambda^\dag-W_\lambda(P_{\leq N_n}\ket{V}\bra{V}P_{\leq N_n}\otimes \omega_\lambda)W_\lambda^\dag\|_1.
    \end{align}
    By using the fact that $U_{\lambda,n}$ is unitary, $W_\lambda$ is an isometry and $\omega_\lambda$ is a normalized state, each term is evaluated as follows:

    \textit{First and sixth terms:} Since $P_{\leq N_n}\ket{\iota Y}\bra{\iota Y}P_{\leq N_n}\otimes \omega_\lambda =t_{\iota Y}(\ket{\widehat{\iota Y}}\bra{\widehat{\iota Y}}\otimes \omega_\lambda)$ and $\|\ket{\widehat{\iota Y}}\bra{\widehat{\iota Y}}\otimes \omega_\lambda\|_1=1$, the first term is equal to $1-t_{\iota Y}$, which is superpolynomially small in $n$ by Lemma~\ref{lem:m_particle_sector_bound}. Repeating the same argument for $V$ instead of $\iota Y$, we find that the sixth term is also superpolynomially small in $n$.

    \textit{Second and fourth terms:} From Proposition~\ref{prop:single_covariance}, these terms are bounded by $C_{\mathrm{cv}}\delta_n$.

    \textit{Third term:} From Eq.~\eqref{eq:Z_Y_X_n}, this term vanishes. 

    \textit{Fifth term:} Since $W_\lambda$ is isometry, and $\|\omega_\lambda\|_1=1$, we get
    \begin{align}
        &\|W_\lambda(\ket{\widehat{X_n}}\bra{\widehat{X_n}}\otimes \omega_\lambda)W_\lambda^\dag-W_\lambda(\ket{\widehat{V}}\bra{\widehat{V}}\otimes \omega_\lambda)W_\lambda^\dag\|_1\\
        &=\| \ket{\widehat{X_n}}\bra{\widehat{X_n}}-\ket{\widehat{V}}\bra{\widehat{V}}\|_1\leq \|\ket{\widehat{X_n}}\bra{\widehat{X_n}}-\ket{X_n}\bra{X_n}\|_1+\|\ket{X_n}\bra{X_n}-\ket{V}\bra{V}\|_1+\|\ket{V}\bra{V}-\ket{\widehat{V}}\bra{\widehat{V}}\|_1.
    \end{align}
    The second term is evaluated by using the inner product between pure coherent states as
    \begin{align}
        \|\ket{X_n}\bra{X_n}-\ket{V}\bra{V}\|_1=2\sqrt{1-|\braket{X_n|V}|^2}=2\sqrt{1-e^{-\|X_n-V\|^2}}\leq 2\|X_n-V\|\leq 2C_{\mathrm{ld}} n^{-1/2+2\bar{\beta}}\label{eq:coherent_state_distance}
    \end{align}
    where we used $1-e^{-x}\leq x$ in the first inequality.  Since $2\bar{\beta}<\eta$, we have $-1/2+2\bar{\beta}<-1/2+\eta<-1/2+3\eta$, and hence $\|\ket{X_n}\bra{X_n}-\ket{V}\bra{V}\|_1\leq 2C_{\mathrm{ld}} \delta_n$.
    The first and third terms are evaluated as
    \begin{align}
        \|\ket{\widehat{X_n}}\bra{\widehat{X_n}}-\ket{X_n}\bra{X_n}\|_1=2\sqrt{1-|\braket{\widehat{X_n}|X_n}|^2}&=2\sqrt{1-t_{X_n}},\\
        \|\ket{\widehat{V}}\bra{\widehat{V}}-\ket{V}\bra{V}\|_1=2\sqrt{1-|\braket{\widehat{V}|V}|^2}&=2\sqrt{1-t_{V}}.
    \end{align}
    Since $\|X_n\|\leq 3n^{\bar{\beta}}$ and $\|V\|\leq n^\beta+n^{\bar{\beta}}\leq 2n^{\bar{\beta}}$, Lemma~\ref{lem:m_particle_sector_bound} implies both terms are superpolynomially small in $n$.

    The sum of all superpolynomially small terms is at most $\delta_n$ by taking $n$ sufficiently large. 
    Summing over all the contributions, we obtain Eq.~\eqref{eq:comparison_coherent_traj} with $C_{\mathrm{traj}}\coloneqq 2C_{\mathrm{cv}}+2C_{\mathrm{ld}}+1$. 
\end{proof}

Integrating this relation over $Y$, we obtain the uniform Gaussian approximation on a typical block:
\begin{theorem}\label{thm:typical_block_isometry_approximation}
    Assume the scale conditions in Eq.~\eqref{eq:scale_eta_beta_alpha}.
    There exists a structural constant $C_\typ$ such that, for all sufficiently large $n$, any $\lambda\in\Lambda_n$, and any $Z\in\mathcal{K}_\hor$ with $\|Z\|\leq n^\beta$, we have
    \begin{align}
        \|\rho_{\lambda,Z,n}-W_\lambda(P_{\leq N_n}\Phi_ZP_{\leq N_n}\otimes \omega_\lambda)W_\lambda^\dag\|_1\leq C_\typ\delta_n.\label{eq:uniform_approximation}
    \end{align}
\end{theorem}

\begin{proof}
    From the triangle inequality, we have
    \begin{align}
        &\|\rho_{\lambda,Z,n}-W_\lambda(P_{\leq N_n}\Phi_ZP_{\leq N_n}\otimes \omega_\lambda)W_\lambda^\dag\|_1\\
        &\leq \|\rho_{\lambda,Z,n}- U_{\lambda,n}(Z)W_\lambda (P_{\leq N_n}\Phi_0P_{\leq N_n}\otimes \omega_\lambda)W_\lambda^\dag U_{\lambda,n}(Z)^\dag \|_1\\
        &+\|U_{\lambda,n}(Z)W_\lambda (P_{\leq N_n}\Phi_0P_{\leq N_n}\otimes \omega_\lambda)W_\lambda^\dag U_{\lambda,n}(Z)^\dag-W_\lambda(P_{\leq N_n}\Phi_ZP_{\leq N_n}\otimes \omega_\lambda)W_\lambda^\dag\|_1.
    \end{align}
    
    By Proposition~\ref{prop:b_lambda_n}, the first term is evaluated as
    \begin{align}
        &\|\rho_{\lambda,Z,n}- U_{\lambda,n}(Z)W_\lambda (P_{\leq N_n}\Phi_0P_{\leq N_n}\otimes \omega_\lambda)W_\lambda^\dag U_{\lambda,n}(Z)^\dag \|_1\\
        &=\|\rho_{\lambda,0}- W_\lambda (P_{\leq N_n}\Phi_0P_{\leq N_n}\otimes \omega_\lambda)W_\lambda^\dag \|_1\\
        &\leq\|\rho_{\lambda,0}-P_{\lambda,N_n}\rho_{\lambda,0}P_{\lambda,N_n}\|_1+\|P_{\lambda,N_n}\rho_{\lambda,0}P_{\lambda,N_n}-W_\lambda (P_{\leq N_n}\Phi_0P_{\leq N_n}\otimes \omega_\lambda)W_\lambda^\dag \|_1\\
        &\leq 2\sqrt{C_\tail}e^{-C_\tail' N_n/2}+C_\tail e^{-C_\tail' N_n},
    \end{align}
    which is superpolynomially small in $n$. 

    Using the coherent-state representation of $\Phi_0$ and $\Phi_Z$, with $V\coloneqq Z+\iota Y$, the second term is bounded as
    \begin{align}
        &\|U_{\lambda,n}(Z)W_\lambda (P_{\leq N_n}\Phi_0P_{\leq N_n}\otimes \omega_\lambda)W_\lambda^\dag U_{\lambda,n}(Z)^\dag-W_\lambda(P_{\leq N_n}\Phi_ZP_{\leq N_n}\otimes \omega_\lambda)W_\lambda^\dag\|_1\\
        &=\left\|\int_{\mathcal{K}_{\mathrm{th}}}\left(U_{\lambda,n}(Z)W_\lambda (P_{\leq N_n}\ket{\iota Y}\bra{\iota Y}P_{\leq N_n}\otimes \omega_\lambda)W_\lambda^\dag U_{\lambda,n}(Z)^\dag-W_\lambda(P_{\leq N_n}\ket{V}\bra{V}P_{\leq N_n}\otimes \omega_\lambda)W_\lambda^\dag\right)\dd \mu(Y)\right\|_1\\
        &\leq \int_{G_n}\left\|\left(U_{\lambda,n}(Z)W_\lambda (P_{\leq N_n}\ket{\iota Y}\bra{\iota Y}P_{\leq N_n}\otimes \omega_\lambda)W_\lambda^\dag U_{\lambda,n}(Z)^\dag-W_\lambda(P_{\leq N_n}\ket{V}\bra{V}P_{\leq N_n}\otimes \omega_\lambda)W_\lambda^\dag\right)\right\|_1\dd \mu(Y)\\
        &+\int_{G_n^c}\left(\left\|U_{\lambda,n}(Z)W_\lambda (P_{\leq N_n}\ket{\iota Y}\bra{\iota Y}P_{\leq N_n}\otimes \omega_\lambda)W_\lambda^\dag U_{\lambda,n}(Z)^\dag\right\|_1+\left\|W_\lambda(P_{\leq N_n}\ket{V}\bra{V}P_{\leq N_n}\otimes \omega_\lambda)W_\lambda^\dag\right\|_1\right)\dd \mu(Y)\\
        &\leq C_{\mathrm{traj}}\delta_n+2\mu(G_n^c),
    \end{align}
    where we have used Lemma~\ref{lem:comparison_coherent_traj} and 
    \begin{align}
        \left\|U_{\lambda,n}(Z)W_\lambda (P_{\leq N_n}\ket{\iota Y}\bra{\iota Y}P_{\leq N_n}\otimes \omega_\lambda)W_\lambda^\dag U_{\lambda,n}(Z)^\dag\right\|_1\leq 1,\quad \left\|W_\lambda(P_{\leq N_n}\ket{V}\bra{V}P_{\leq N_n}\otimes \omega_\lambda)W_\lambda^\dag\right\|_1\leq 1.
    \end{align}
    Note that $\mu(G_n^c)$ is superpolynomially small in $n$ by Proposition~\ref{prop:Gaussian_tail_mu}. 

    The sum of all superpolynomially small terms $2\sqrt{C_\tail}e^{-C_\tail'N_n/2}+C_\tail e^{-C_\tail' N_n}+2\mu(G_n^c)$ is at most $\delta_n$ by taking $n$ sufficiently large. 
    Summing over all the contributions, we obtain Eq.~\eqref{eq:uniform_approximation} with $C_\typ\coloneqq C_{\mathrm{traj}}+1$. 
\end{proof}

\clearpage
\section{Forward and Reverse Channels, and QLAN}\label{sec:forward_reverse_channels}
This section constructs the forward and reverse channels by using the isometry $W_\lambda: \mathcal{F}_{\leq L_n}\otimes E_\lambda\to H_\lambda$. 
Throughout this section, we identify an operator $O\in\mathcal{T}_1(\mathcal{F}_{\leq L_n})$
with the operator $O\oplus0$ on $\mathcal{F}_\hor=\mathcal{F}_{\leq L_n}\oplus\mathcal{F}_{>L_n}$.

\subsection{Forward channels}
The forward channels are defined as follows.
\begin{definition}
    For each $\lambda\in \Lambda_n$, we define $\mathcal{A}_{\lambda}\colon\mathcal{T}_1(H_\lambda)\to\mathcal{T}_1(\mathcal{F}_\hor)$ by
    \begin{align}
        \mathcal{A}_{\lambda}(\sigma)\coloneqq \Tr_{E_\lambda}\left(W_\lambda^\dag \sigma W_\lambda\right)+\Tr_{H_\lambda}((I_{H_\lambda}-W_\lambda W_\lambda^\dag)\sigma)\ket{0}\bra{0}.
    \end{align}
    Let $\Pi_\lambda$ denote the orthogonal projector onto $H_\lambda\otimes K_\lambda$, and define $\Pi_{\Lambda_n}\coloneqq \sum_{\lambda\in\Lambda_n}\Pi_\lambda$. We then define the forward channel $T_n:\mathcal{T}_1(\mathcal{H}^{\otimes n})\to\mathcal{T}_1(\mathcal{F}_\hor)$ by
    \begin{align}
        T_n(\sigma)\coloneqq\sum_{\lambda\in\Lambda_n}\mathcal{A}_{\lambda}\left(\Tr_{K_\lambda}(\Pi_\lambda \sigma \Pi_\lambda)\right)+\Tr((I-\Pi_{\Lambda_n})\sigma)\ket{0}\bra{0}.
    \end{align}
\end{definition}

Let us confirm that these maps are quantum channels:
\begin{itemize}
    \item The map $\mathcal{A}_\lambda$: The map $\sigma\mapsto W_\lambda^\dag \sigma W_\lambda$, and the partial trace over $E_\lambda$ are completely positive. Also, $\sigma\mapsto\Tr_{H_\lambda}((I_{H_\lambda}-W_\lambda W_\lambda^\dag)\sigma)\ket{0}\bra{0}$ is a measure-and-prepare channel, which is completely positive. The map $\mathcal{A}_\lambda$ is trace-preserving because
    \begin{align}
        \Tr(\mathcal{A}_\lambda(\sigma))=\Tr\left(W_\lambda^\dag \sigma W_\lambda\right)+\Tr_{H_\lambda}((I_{H_\lambda}-W_\lambda W_\lambda^\dag)\sigma)=\Tr \sigma,
    \end{align}
    where we used the cyclicity of the trace. 
    \item The map $T_n$: Each transformation $\sigma\mapsto \Pi_\lambda \sigma \Pi_\lambda\mapsto \Tr_{K_\lambda}(\Pi_\lambda \sigma \Pi_\lambda)\mapsto \mathcal{A}_{\lambda}\left(\Tr_{K_\lambda}(\Pi_\lambda \sigma \Pi_\lambda)\right)$ is respectively completely positive. Also, $\sigma\mapsto \Tr((I-\Pi_{\Lambda_n})\sigma)\ket{0}\bra{0}$ is again a measure-and-prepare channel, and hence completely positive. The map $T_n$ is trace-preserving because 
    \begin{align}
        \Tr(T_n(\sigma))=\sum_{\lambda\in\Lambda_n}\Tr(\Pi_\lambda \sigma)+\Tr((I-\Pi_{\Lambda_n})\sigma)=\Tr\sigma.
    \end{align}
\end{itemize}

The following proves the asymptotic error rate in conversion via the forward channels:
\begin{theorem}\label{thm:fwd}
    Suppose that the scale conditions in Eq.~\eqref{eq:scale_eta_beta_alpha} are satisfied.
    There exists a structural constant $C_\fwd>0$ such that, for all sufficiently large $n$,
    \begin{align}
        \sup_{\|Z\|\leq n^\beta}\|T_n(\rho_{Z,n})-\Phi_Z\|_1\leq C_\fwd\delta_n+2\varepsilon^{\atyp}_n.\label{eq:fwd_scale}
    \end{align}
\end{theorem}
\begin{proof}
    Let $\lambda\in \Lambda_n$. Since $\mathcal{A}_\lambda$ is a quantum channel, it is trace-norm contractive, i.e.,
    \begin{align}
        &\|\mathcal{A}_\lambda(\rho_{\lambda,Z,n})-\mathcal{A}_\lambda(W_\lambda (P_{\leq N_n}\Phi_ZP_{\leq N_n}\otimes \omega_\lambda)W_\lambda^\dag)\|_1\\
        &\leq  \|\rho_{\lambda,Z,n}-W_\lambda (P_{\leq N_n}\Phi_ZP_{\leq N_n}\otimes \omega_\lambda)W_\lambda^\dag\|_1\overset{\text{Thm.~\ref{thm:typical_block_isometry_approximation}}}{\leq}C_\typ\delta_n
    \end{align}
    for all sufficiently large $n$. Since $W_\lambda$ is an isometry, 
    \begin{align}
        \mathcal{A}_\lambda(W_\lambda (P_{\leq N_n}\Phi_ZP_{\leq N_n}\otimes \omega_\lambda)W_\lambda^\dag)=\Tr_{E_\lambda}(P_{\leq N_n}\Phi_ZP_{\leq N_n}\otimes \omega_\lambda)=P_{\leq N_n}\Phi_ZP_{\leq N_n}.
    \end{align}
    From Eq.~\eqref{eq:tail_gaussian_state_Z}, we have $\|\Phi_Z-P_{\leq N_n}\Phi_ZP_{\leq N_n}\|_1\leq  \delta_n$. Therefore, from the triangle inequality,
    \begin{align}
        \|\mathcal{A}_\lambda(\rho_{\lambda,Z,n})-\Phi_Z\|_1\leq C_\fwd \delta_n,\qquad C_\fwd\coloneqq C_\typ+1
    \end{align}
    
    By using $\sum_{\lambda\in\Lambda_n}p_{\lambda,n}+\varepsilon^{\atyp}_n=1$, we have
    \begin{align}
        \|T_n(\rho_{Z,n})-\Phi_Z\|_1&=\left\|\sum_{\lambda\in\Lambda_n}p_{\lambda,n}\left(\mathcal{A}_\lambda(\rho_{\lambda,Z,n})-\Phi_Z\right)+\varepsilon^{\atyp}_n\left(\ket{0}\bra{0}-\Phi_Z\right)\right\|_1\\
        &\leq C_\fwd\delta_n\sum_{\lambda\in\Lambda_n}p_{\lambda,n}+2\varepsilon^{\atyp}_n\leq C_\fwd\delta_n+2\varepsilon^{\atyp}_n .
    \end{align}
    where we have used $\left\|\ket{0}\bra{0}-\Phi_Z\right\|_1\leq \|\ket{0}\bra{0}\|+\|\Phi_Z\|_1\leq 2$ and $\sum_{\lambda\in\Lambda_n}p_{\lambda,n}\leq 1$.
\end{proof}

\subsection{Reverse channels}

The reverse channels are defined as follows.
\begin{definition}
    For each $\lambda\in \Lambda_n$, we define $\mathcal{B}_{\lambda}\colon\mathcal{T}_1(\mathcal{F}_\hor)\to\mathcal{T}_1(H_\lambda)$ by
    \begin{align}
        \mathcal{B}_{\lambda}(\sigma)\coloneqq W_\lambda(P_{\leq N_n}\sigma P_{\leq N_n}\otimes \omega_\lambda)W_\lambda^\dag + \Tr((I-P_{\leq N_n})\sigma)\omega_\lambda^{H_\lambda}.
    \end{align}
    Define a positive semidefinite operator $\tau_n^\atyp$ by
    \begin{align}
        \tau_n^\atyp\coloneqq \bigoplus_{\lambda\notin\Lambda_n}p_{\lambda,n}\left(\rho_{\lambda,0}\otimes \omega_\lambda^{K_\lambda}\right), \qquad \omega_\lambda^{K_\lambda}\coloneqq \frac{I_{K_\lambda}}{\dim K_\lambda},
    \end{align}
    which satisfies $\Tr  \tau_n^\atyp=\varepsilon^\atyp_n$. We then define the reverse channel $S_n:\mathcal{T}_1(\mathcal{F}_\hor)\to \mathcal{T}_1(\mathcal{H}^{\otimes n})$ by
    \begin{align}
        S_n(\sigma)\coloneqq \bigoplus_{\lambda\in\Lambda_n}p_{\lambda,n}\left(\mathcal{B}_\lambda(\sigma)\otimes \omega_\lambda^{K_\lambda}\right)+ \Tr(\sigma)\tau_n^\atyp.
    \end{align}
\end{definition}

Similar to the arguments on $\mathcal{A}_{\lambda}$, we find $\mathcal{B}_{\lambda}$ is a quantum channel. Tensoring with the normalized state $\omega_\lambda^{K_\lambda}$ is completely positive and trace-preserving. The map $\sigma\mapsto \Tr(\sigma)\tau_n^\atyp$ is completely positive since $\tau_n^{\atyp}\geq 0$. Thus, $S_n$ is completely positive. Moreover, since $\Tr  \tau_n^\atyp=\varepsilon^\atyp_n=1-\sum_{\lambda\in\Lambda_n}p_{\lambda,n}$, we have
\begin{align}
    \Tr S_n(\sigma)=\sum_{\lambda\in\Lambda_n}p_{\lambda,n}\Tr(\sigma)+\Tr(\sigma)\Tr(\tau_n^\atyp)=\Tr\sigma,
\end{align}
implying that $S_n$ is trace-preserving, and hence, it is a quantum channel.

The following proves the asymptotic error rate in conversion via the reverse channels:
\begin{theorem}\label{thm:rev}
    Suppose that the scale conditions in Eq.~\eqref{eq:scale_eta_beta_alpha} are satisfied.
    There exists a structural constant $C_\rev>0$ such that, for all sufficiently large $n$,
    \begin{align}
        \sup_{\|Z\|\leq n^\beta}\|\rho_{Z,n}-S_n(\Phi_Z)\|_1\leq C_\rev\delta_n+2\varepsilon^{\atyp}_n.\label{eq:rev_scale}
    \end{align}
\end{theorem}
\begin{proof}
    Let $\lambda\in \Lambda_n$. By the definition of $\mathcal{B}_\lambda$ and the triangle inequality, we have
    \begin{align}
        \|\mathcal{B}_\lambda(\Phi_Z)-\rho_{\lambda,Z,n}\|_1\leq \|W_\lambda(P_{\leq N_n}\Phi_Z P_{\leq N_n}\otimes \omega_\lambda)W_\lambda^\dag-\rho_{\lambda,Z,n}\|_1+\Tr((I-P_{\leq N_n})\Phi_Z)\leq C_\typ\delta_n+ \frac{1}{4}\delta_n^2,
    \end{align}
    where we have used Theorem~\ref{thm:typical_block_isometry_approximation} and Eq.~\eqref{eq:tail_gaussian_state_Z}. For all sufficiently large $n$, since $\frac{1}{4}\delta_n^2\leq \delta_n$, the right-hand side is bounded by $C_\rev\delta_n$ with $C_\rev\coloneqq C_\typ+1$. 

    By using $\sum_{\lambda\in\Lambda_n}p_{\lambda,n}+\varepsilon^{\atyp}_n=1$, we have
    \begin{align}
        \|\rho_{Z,n}-S_n(\Phi_Z)\|_1&\leq\sum_{\lambda\in\Lambda_n}p_{\lambda,n} \|\mathcal{B}_\lambda(\Phi_Z)-\rho_{\lambda,Z,n}\|_1+\left\|\tau_n^\atyp-\bigoplus_{\lambda\notin\Lambda_n}p_{\lambda,n}\left(\rho_{\lambda,Z,n}\otimes \omega_\lambda^{K_\lambda}\right)\right\|_1\\
        &\leq \sum_{\lambda\in\Lambda_n}p_{\lambda,n}C_\rev\delta_n+ 2\sum_{\lambda\notin\Lambda_n}p_{\lambda,n}\leq C_\rev\delta_n+2\varepsilon^{\atyp}_n,
    \end{align}
    where we have used 
    \begin{align}
        \left\|\tau_n^\atyp-\bigoplus_{\lambda\notin\Lambda_n}p_{\lambda,n}\left(\rho_{\lambda,Z,n}\otimes \omega_\lambda^{K_\lambda}\right)\right\|_1\leq \left\|\tau_n^\atyp\right\|_1+\left\|\bigoplus_{\lambda\notin\Lambda_n}p_{\lambda,n}\left(\rho_{\lambda,Z,n}\otimes \omega_\lambda^{K_\lambda}\right)\right\|_1=2\sum_{\lambda\notin\Lambda_n}p_{\lambda,n}=2\varepsilon^{\atyp}_n.
    \end{align}
\end{proof}

\subsection{Proof of the main theorem (Theorem~\ref{thm:QLAN})}\label{sec:proof_of_QLAN_main_theorem}

\begin{theorem}[Restatement of Theorem~\ref{thm:QLAN}]
    Fix $\beta\in [0,\frac{1}{12})$. Take any $\kappa$ such that $6\beta<\kappa<\frac{1}{2}$. Then there exist quantum channels
    \begin{align}
        T_{n}:\mathcal{T}_1(\mathcal{H}^{\otimes n})\to \mathcal{T}_1(\mathcal{F}_{\hor}),\qquad S_{n}: \mathcal{T}_1(\mathcal{F}_{\hor})\to\mathcal{T}_1(\mathcal{H}^{\otimes n})
    \end{align}
    such that
    \begin{align}
        \sup_{\|Z\|\leq n^\beta}\left\|T_{n}(\rho_{Z,n})-\Phi_Z\right\|_1&=O(n^{-1/2+\kappa}),\qquad
        \sup_{\|Z\|\leq n^\beta}\left\|\rho_{Z,n}-S_{n}(\Phi_Z)\right\|_1=O(n^{-1/2+\kappa}).
    \end{align}
\end{theorem}

\begin{proof}
    For the fixed $\kappa$, one may choose $\eta$ such that
    \begin{align}
        2\beta <\eta<\frac{1}{3}\kappa.
    \end{align}
    Since $1-2\eta>1-1/3>1/2$ and $\kappa-2\eta>\eta>0$, one can choose $\alpha$ such that
    \begin{align}
        \frac{1}{2}<\alpha<\min\left\{1-2\eta,\,\,\frac{1}{2}+\kappa-2\eta\right\}.
    \end{align}
    Thus, $\beta,\alpha,\eta$ satisfy the scale conditions in Eq.~\eqref{eq:scale_eta_beta_alpha}. Therefore, Theorems~\ref{thm:fwd} and~\ref{thm:rev} imply that for all sufficiently large $n$, there are quantum channels $T_n$ and $S_n$ satisfying Eqs.~\eqref{eq:fwd_scale} and ~\eqref{eq:rev_scale}, respectively. Note that for the finitely many remaining values of $n$, one can define $T_n$ and $S_n$ arbitrarily; this clearly does not affect the asymptotic error bounds.
    
    Since
    \begin{align}
        \alpha-1+2\eta<\frac{1}{2}+\kappa-2\eta-1+2\eta=-\frac{1}{2}+\kappa,\qquad -\frac{1}{2}+3\eta<-\frac{1}{2}+\kappa,
    \end{align}
    we get
    \begin{align}
        \delta_n= n^{\alpha-1+2\eta}+n^{-1/2+3\eta}=O(n^{-1/2+\kappa}).
    \end{align}
    Therefore, since $\varepsilon^{\atyp}_n$ is superpolynomially small in $n$ by Proposition~\ref{prop:typical_concentration}, we obtain
    \begin{align}
        \sup_{\|Z\|\leq n^\beta}\left\|T_{n}(\rho_{Z,n})-\Phi_Z\right\|_1&=O(n^{-1/2+\kappa}),\qquad
        \sup_{\|Z\|\leq n^\beta}\left\|\rho_{Z,n}-S_{n}(\Phi_Z)\right\|_1=O(n^{-1/2+\kappa}).
    \end{align}
\end{proof}

\subsection{Proof of Rate-scaled QLAN (Corollary~\ref{cor:QLAN_with_rate})}\label{sec:ratescaled_QLAN}

We first prove an upper bound on the trace distance between two Gaussian states, which extends the bound in Eq.~\eqref{eq:coherent_state_distance} to mixed Gaussian states. 
\begin{lemma}[Trace-norm continuity of Gaussian shifts]\label{lem:gaussian-shift-continuity}
    For any $X,Y\in\mathcal{K}_\hor$, 
    \begin{align}
        \|\Phi_X-\Phi_Y\|_1\leq 2\sqrt{1-e^{-\|X-Y\|^2}}\leq 2\|X-Y\|.
    \end{align}
\end{lemma}
\begin{proof}
    By Eq.~\eqref{eq:displaced_state_int_expression}, we have
    \begin{align}
        \Phi_X-\Phi_Y=\int_{\mathcal{K}_{\mathrm{th}}}\left(\ket{X+\iota Z}\bra{X+\iota Z}-\ket{Y+\iota Z}\bra{Y+\iota Z}\right)\dd \mu(Z).
    \end{align}
    Therefore, by the triangle inequality, we obtain
    \begin{align}
        \|\Phi_X-\Phi_Y\|_1&\leq \int_{\mathcal{K}_{\mathrm{th}}}\|\left(\ket{X+\iota Z}\bra{X+\iota Z}-\ket{Y+\iota Z}\bra{Y+\iota Z}\right)\|_1\dd \mu(Z)\\
        &=\int_{\mathcal{K}_{\mathrm{th}}}2\sqrt{1-e^{-\|(X+\iota Z)-(Y+\iota Z)\|^2}}\dd \mu(Z)=2\sqrt{1-e^{-\|X-Y\|^2}}\\
        &\leq 2\|X-Y\|,
    \end{align}
    where we used that $\mu$ is a probability measure and $1-e^{-x}\leq x$ for $x\geq 0$. 
\end{proof}

Using Lemma~\ref{lem:gaussian-shift-continuity}, we prove the following as a corollary of Theorem~\ref{thm:QLAN}.
\begin{corollary}[Restatement of Corollary~\ref{cor:QLAN_with_rate}]
     Fix a rate $r>0$, and set $m_n\coloneqq \floor{rn}$ and $r_n\coloneqq m_n/n$. For $Z\in\mathcal{K}_\hor$, define
    \begin{align}
        \rho_{Z,n}^{(r)}\coloneqq \left(U_{m_n}(\sqrt{r_n}Z)\rho_0 U_{m_n}(\sqrt{r_n}Z)^\dag\right)^{\otimes m_n}.
    \end{align}
    Fix $\beta\in[0,1/12)$, and take any $\kappa$ such that $6\beta<\kappa<1/2$. Then there exist quantum channels
    \begin{align}
        T_{n}^{(r)}:\mathcal{T}_1(\mathcal{H}^{\otimes m_n})\to \mathcal{T}_1(\mathcal{F}_{\hor}),\qquad S_{n}^{(r)}: \mathcal{T}_1(\mathcal{F}_{\hor})\to\mathcal{T}_1(\mathcal{H}^{\otimes m_n})
    \end{align}
    such that
    \begin{align}
        \sup_{\|Z\|\leq n^\beta}\left\|T_{n}^{(r)}(\rho_{Z,n}^{(r)})-\Phi_{\sqrt{r}Z}\right\|_1&=O(n^{-1/2+\kappa}),\qquad
        \sup_{\|Z\|\leq n^\beta}\left\|\rho_{Z,n}^{(r)}-S_{n}^{(r)}(\Phi_{\sqrt{r}Z})\right\|_1=O(n^{-1/2+\kappa}).
    \end{align}
\end{corollary}

\begin{proof}
    Since $K$ is real-linear, we get
    \begin{align}
        m_n^{-1/2}K(\sqrt{r_n}Z)=n^{-1/2}K(Z),\qquad \rho_{Z,n}^{(r)}=\rho_{\sqrt{r_n}Z,m_n}.
    \end{align}
    We choose $\beta'$ such that $\beta<\beta'<\kappa/6$. Then, $\beta'<1/12$, $6\beta'<\kappa$, and
    \begin{align}
        \lim_{n\to\infty}\frac{\sqrt{r_n}n^\beta}{m_n^{\beta'}}=\lim_{n\to\infty}\left(\frac{m_n}{n}\right)^{1/2-\beta'}n^{\beta-\beta'}=0.
    \end{align}
    Therefore, for all sufficiently large $n$, the condition $\|Z\|\leq n^\beta$ implies $\|\sqrt{r_n}Z\|\leq m_n^{\beta'}$. Thus, applying Theorem~\ref{thm:QLAN} with $m_n$ and the exponent $\beta'$, there exist quantum channels $\tilde{T}_{m_n}$ and $\tilde{S}_{m_n}$ such that
    \begin{align}
        \sup_{\|Z\|\leq n^{\beta}}\left\|\tilde{T}_{m_n}(\rho_{Z,n}^{(r)})-\Phi_{\sqrt{r_n}Z}\right\|_1&=O(m_n^{-1/2+\kappa}),\\
         \sup_{\|Z\|\leq n^{\beta}}\left\|\rho_{Z,n}^{(r)}-\tilde{S}_{m_n}(\Phi_{\sqrt{r_n}Z})\right\|_1&=O(m_n^{-1/2+\kappa}).
    \end{align}
    Since $m_n/n\to r>0$ and $\kappa<1/2$, the right-hand sides are $O(n^{-1/2+\kappa})$. 

    It remains to replace $\sqrt{r_n}$ by $\sqrt{r}$ in the Gaussian limit model.
    Since $0\leq r-r_n<1/n$, we have
    \begin{align}
        |\sqrt{r_n}-\sqrt{r}|=\frac{|r-r_n|}{\sqrt{r_n}+\sqrt{r}}\leq \frac{1}{2\sqrt{r_n}n}.
    \end{align}
    By Lemma~\ref{lem:gaussian-shift-continuity}, $\|\Phi_X-\Phi_Y\|_1\leq 2\|X-Y\|$ for any $X,Y\in\mathcal{K}_{\hor}$. 
    Consequently,
    \begin{align}
        \sup_{\|Z\|\leq n^\beta}\|\Phi_{\sqrt{r_n}Z}-\Phi_{\sqrt{r}Z}\|_1\leq \frac{1}{\sqrt{r_n}}n^{-1+\beta}=o(n^{-1/2+\kappa})
    \end{align}

    Define $T_n^{(r)}\coloneqq \tilde{T}_{m_n}$ and $S_n^{(r)}\coloneqq \tilde{S}_{m_n}$. By the triangle inequality, we obtain
    \begin{align}
        \sup_{\|Z\|\leq n^{\beta}}\left\|T_n^{(r)}(\rho_{Z,n}^{(r)})-\Phi_{\sqrt{r}Z}\right\|_1\leq \sup_{\|Z\|\leq n^{\beta}}\left(\left\|\tilde{T}_{m_n}(\rho_{Z,n}^{(r)})-\Phi_{\sqrt{r_n}Z}\right\|_1+\left\|\Phi_{\sqrt{r_n}Z}-\Phi_{\sqrt{r}Z}\right\|_1\right)=O(n^{-1/2+\kappa}).
    \end{align}
    Using the contractivity of trace distance under the channel $S_n^{(r)}$, we also get
    \begin{align}
        \sup_{\|Z\|\leq n^\beta}\left\|\rho_{Z,n}^{(r)}-S_{n}^{(r)}(\Phi_{\sqrt{r}Z})\right\|_1&\leq \sup_{\|Z\|\leq n^\beta}\left(\left\|\rho_{Z,n}^{(r)}-\tilde{S}_{m_n}(\Phi_{\sqrt{r_n}Z})\right\|_1+\left\|\Phi_{\sqrt{r_n}Z}-\Phi_{\sqrt{r}Z}\right\|_1\right)=O(n^{-1/2+\kappa}).
    \end{align}
\end{proof}

For ease of reference, we conclude by recording a concrete coordinate form of QLAN in terms of the eigenvalues and eigenvectors of the reference state and a given family of Hermitian generators.

\begin{theorem}\label{thm:QLAN_explicit_expression}
    Let $\rho$ be a quantum state on a finite-dimensional Hilbert space $\mathcal{H}$ of dimension $d$, and let $\{X_i\}_{i=1}^p$ be a set of Hermitian operators on $\mathcal{H}$. Consider a unitary model at rate $r>0$
    \begin{align}
         \rho_{u,n}^{(r)}\coloneqq \left(\exp\left(\ii n^{-1/2}\sum_{i=1}^p u_i X_i \right)\rho \exp\left(-\ii n^{-1/2}\sum_{i=1}^pu_i X_i \right)\right)^{\otimes \floor{rn}},\qquad u\in\mathbb{R}^p.
    \end{align}
    Using the eigenvalue decomposition $\rho=\sum_{k=1}^d\mu_k\ket{k}\bra{k}$, we define a set of pairs of labels 
    \begin{align}
        \mathcal{J}\coloneqq \{(k,l)\colon \mu_k>\mu_l\}.
    \end{align}
    If $\rho=I/d$, the unitary orbit of $\rho$ consists of a single point, so the QLAN statement is trivial. We therefore restrict attention to $\rho\neq I/d$; equivalently, $\mathcal{J}\neq\emptyset$. 
    Let $\mathcal{F}$ be the $|\mathcal{J}|$-mode bosonic Fock space, whose creation and annihilation operators satisfy the canonical commutation relation
    \begin{align}
        [a_{kl},a_{k'l'}^\dag]=\delta_{kk'}\delta_{ll'}I,\qquad [a_{kl},a_{k'l'}]=0,\qquad [a_{kl}^\dag ,a_{k'l'}^\dag]=0.
    \end{align}
    We define the Gaussian shift family consisting of $|\mathcal{J}|$ modes by
    \begin{align}
        \Phi_{z}^\rho\coloneqq D(z) \Phi^\rho D(z)^\dag, 
    \end{align}
    where the reference Gaussian state is given by
    \begin{align}
        \Phi^\rho \coloneqq \bigotimes_{(k,l)\in\mathcal{J}}\phi_{\beta_{kl}}^{(kl)},\qquad \phi_{\beta_{kl}}^{(kl)}\coloneqq \frac{e^{-\beta_{kl}a^{\dag}_{kl}a_{kl}}}{\Tr e^{-\beta_{kl}a^{\dag}_{kl}a_{kl}}},\qquad \beta_{kl}\coloneqq -\ln\left(\frac{\mu_l}{\mu_k}\right),
    \end{align}
    with the conventions $-\ln 0\coloneqq \infty$ and $\phi_{\infty}^{(kl)}=\ket{0}\bra{0}$, and the displacement operator is defined by
    \begin{align}
        D(z)\coloneqq \exp\left(\sum_{(k,l)\in\mathcal{J}}\left(z_{kl}a^\dag_{kl}-\overline{z_{kl}}a_{kl}\right)\right),\qquad z=(z_{kl})_{(k,l)\in\mathcal{J}}\in\mathbb{C}^{|\mathcal{J}|}.
    \end{align}
    We define a matrix $C\in\mathbb{C}^{|\mathcal{J}|\times p}$ whose matrix elements are given by
    \begin{align}
         C_{(k,l),i}\coloneqq \ii\sqrt{\mu_k-\mu_l}\braket{l|X_i|k}.
    \end{align}
    
    Fix $\epsilon\in[0,1/12)$, and take any $\kappa$ such that $6\epsilon<\kappa<1/2$. Then there exist quantum channels
    \begin{align}
        T_{n}^{(r)}:\mathcal{T}_1(\mathcal{H}^{\otimes \floor{rn}})\to \mathcal{T}_1(\mathcal{F}),\qquad S_{n}^{(r)}: \mathcal{T}_1(\mathcal{F})\to\mathcal{T}_1(\mathcal{H}^{\otimes \floor{rn}})
    \end{align}
    such that
    \begin{align}
        \sup_{\|u\|\leq n^\epsilon}\left\|T_{n}^{(r)}\left( \rho_{u,n}^{(r)}\right)-\Phi_{\sqrt{r}Cu}^\rho\right\|_1&=O(n^{-1/2+\kappa}),\qquad
        \sup_{\|u\|\leq n^\epsilon}\left\|\rho_{u,n}^{(r)}-S_{n}^{(r)}\left(\Phi_{\sqrt{r}Cu}^\rho\right)\right\|_1=O(n^{-1/2+\kappa}).
    \end{align}
\end{theorem}

\begin{proof}
    We define
    \begin{align}
        z_{kl}\coloneqq \sum_{i=1}^p C_{(k,l),i}u_i=\sqrt{\mu_k-\mu_l}\Braket{l|\ii\sum_{i=1}^pX_iu_i|k}.
    \end{align}
    Recalling the notations in Eqs.~\eqref{eq:hor_ver_decomposition}, \eqref{eq:Z_expansion} and~\eqref{eq:expansion_basis_KZ}, we have
    \begin{align}
        Z(u)&=\sum_{(a,b);\,\zeta_a>\zeta_b}\sum_{k\in I_a}\sum_{l\in I_b}z_{kl}\overline{E_{kl}}=\sum_{(a,b);\,\zeta_a>\zeta_b}\sum_{k\in I_a}\sum_{l\in I_b}\sqrt{\mu_k-\mu_l}\Braket{l|\ii\sum_{i=1}^pX_iu_i|k}\overline{E_{kl}}\\
        K(Z(u))&\coloneqq \sum_{(a,b);\,\zeta_a>\zeta_b}\sum_{k\in I_a}\sum_{l\in I_b}\left(\Braket{l|\ii\sum_{i=1}^pX_iu_i|k}\ket{l}\bra{k}-\overline{\Braket{l|\ii\sum_{i=1}^pX_iu_i|k}}\ket{k}\bra{l}\right)\\
        &=\sum_{(a,b);\,\zeta_a>\zeta_b}\sum_{k\in I_a}\sum_{l\in I_b}\left(\Braket{l|\ii\sum_{i=1}^pX_iu_i|k}\ket{l}\bra{k}+\Braket{k|\ii\sum_{i=1}^pX_iu_i|l}\ket{k}\bra{l}\right)\\
        &=\sum_{(a,b);\,\zeta_a>\zeta_b}\left(P_b\left(\ii\sum_{i=1}^pX_iu_i\right)P_a+P_a\left(\ii\sum_{i=1}^pX_iu_i\right)P_b\right)\\
        &=\ii\sum_{i=1}^pX_iu_i-\sum_{a\in\mathcal{A}}P_a\left(\ii\sum_{i=1}^pX_iu_i\right)P_a= K(u)-K^\ver(u)\eqqcolon K^\hor(u),
    \end{align}
    where $ K(u)\coloneqq \ii\sum_{i=1}^pX_iu_i,\, K^\ver(u)\coloneqq \sum_{a\in\mathcal{A}}P_aK(u)P_a$, and $P_a\coloneqq \sum_{k\in I_a}\ket{k}\bra{k}$ is the projector onto the eigenspace with eigenvalue $\zeta_a$ for $a\in\mathcal{A}$. For $m_n\coloneqq \floor{rn}$ and $r_n\coloneqq m_n/n$, we have
    \begin{align}
        \exp\left(\ii n^{-1/2}\sum_{i=1}^p u_i X_i\right)=\exp\left(m_n^{-1/2}K(\sqrt{r_n}u)\right)=\exp\left(m_n^{-1/2}\left(K^\hor(\sqrt{r_n}u)+K^\ver(\sqrt{r_n}u)\right)\right).
    \end{align}
    Defining
    \begin{align}
        \tilde{\rho}_{u,n}^{(r)}\coloneqq \left(\exp\left(m_n^{-1/2}\left(K^\hor(\sqrt{r_n}u)\right)\right) \rho  \exp\left(-m_n^{-1/2}\left(K^\hor(\sqrt{r_n}u)\right)\right)\right)^{\otimes m_n}.
    \end{align}
    by Proposition~\ref{prop:horizontal_reduction}, we have
    \begin{align}
        \left\| \rho_{u,n}^{(r)}-\tilde{\rho}^{(r)}_{u,n}\right\|_1\leq C_K\frac{\|\sqrt{r_n}u\|^2}{\sqrt{m_n}}\leq C_K\frac{r_n}{\sqrt{m_n}}\|u\|^2\leq C_K\frac{\sqrt{r_n}}{\sqrt{n}}\|u\|^2
    \end{align}
    and hence
    \begin{align}
        \sup_{\|u\|\leq n^\epsilon} \left\| \rho_{u,n}^{(r)}-\tilde{\rho}^{(r)}_{u,n}\right\|_1=O(n^{-1/2+\kappa}),\label{eq:hor_red_her}
    \end{align}
    where we have used $2\epsilon<\kappa$ due to $6\epsilon<\kappa$.

    Since
    \begin{align}
        \|Z(u)\|^2&=\sum_{(a,b);\,\zeta_a>\zeta_b}\sum_{k\in I_a}\sum_{l\in I_b}(\mu_k-\mu_l)\left|\Braket{l|\ii\sum_{i=1}^pX_iu_i|k}\right|^2\\
        &\leq \sum_{(a,b);\,\zeta_a>\zeta_b}\sum_{k\in I_a}\sum_{l\in I_b}\left|\Braket{l|\ii\sum_{i=1}^pX_iu_i|k}\right|^2\leq \left\|\ii\sum_{i=1}^pX_iu_i\right\|_\HS^2\leq \left(\sum_{i=1}^pu_i^2\right)\left(\sum_{i=1}^p\|X_i\|_{\HS}^2\right),
    \end{align}
    we get $\|Z(u)\|\leq C_X \|u\|$ with $C_X\coloneqq \left(\sum_{i=1}^p\|X_i\|_{\HS}^2\right)^{1/2}<\infty$. Taking $\beta$ such that $6\epsilon<6\beta<\kappa$, 
    \begin{align}
        \frac{C_Xn^{\epsilon}}{n^{\beta}}\to0.
    \end{align}
    Therefore,
    \begin{align}
        \|u\|\leq n^\epsilon\implies \|Z(u)\|\leq n^{\beta}
    \end{align}
    for all sufficiently large $n$. Moreover, by the real-linearity of $Z$,
    \begin{align}
        K^\hor(\sqrt{r_n}u)=K(Z(\sqrt{r_n}u))=K(\sqrt{r_n}Z(u)).
    \end{align}
    Thus, $\tilde{\rho}^{(r)}_{u,n}$ coincides with the rate-scaled model in Corollary~\ref{cor:QLAN_with_rate} with parameter $Z(u)=Cu$. Therefore, by Corollary~\ref{cor:QLAN_with_rate}, there are quantum channels $T_{n}^{(r)},S_{n}^{(r)}$ such that
    \begin{align}
         \sup_{\|u\|\leq n^{\epsilon}}\left\|T_{n}^{(r)}\left(\tilde{\rho}^{(r)}_{u,n}\right)-\Phi_{\sqrt{r}Cu}^\rho\right\|_1&=O(n^{-1/2+\kappa}),\qquad
        \sup_{\|u\|\leq n^{\epsilon}}\left\|\tilde{\rho}^{(r)}_{u,n}-S_{n}^{(r)}\left(\Phi_{\sqrt{r}Cu}^\rho\right)\right\|_1=O(n^{-1/2+\kappa}).
    \end{align}
    By the triangle inequality and the contractivity of the trace norm under quantum channels, Eq.~\eqref{eq:hor_red_her} therefore yields
    \begin{align}
         \sup_{\|u\|\leq n^{\epsilon}}\left\|T_{n}^{(r)}\left( \rho_{u,n}^{(r)}\right)-\Phi_{\sqrt{r}Cu}^\rho\right\|_1&=O(n^{-1/2+\kappa}),\qquad
        \sup_{\|u\|\leq n^{\epsilon}}\left\|\rho_{u,n}^{(r)}-S_{n}^{(r)}\left(\Phi_{\sqrt{r}Cu}^\rho\right)\right\|_1=O(n^{-1/2+\kappa}).
    \end{align}
\end{proof}

\section{Symmetric subspace, bosonic Fock space, creation and annihilation operators}\label{app:Fock_space_review}

We briefly review the notation and elementary facts about bosonic Fock space used in Part~III. For more details, see, e.g., Refs.~\cite{derezinskiMathematicsQuantizationQuantum2023,parthasarathy_IntroductionQuantumStochasticCalculus_1992}. We use the convention that the Hilbert-space inner product is antilinear in its first argument and linear in its second argument.

Let $\mathcal{K}$ be a finite-dimensional complex Hilbert space. For $k\geq 1$, let $U_\pi$ denote the unitary action of a permutation $\pi\in\mathfrak{S}_k$ that permutes the tensor factors of $\mathcal{K}^{\otimes k}$. We define the symmetrization projection and the $k$-particle symmetric tensor space by
\begin{align}
    P_{\mathrm{s}}^{(k)}\coloneqq \frac{1}{k!}\sum_{\pi\in\mathfrak{S}_k}U_\pi,\qquad \Sym^k\mathcal{K}\coloneqq P_{\mathrm{s}}^{(k)}\mathcal{K}^{\otimes k}.
\end{align}
We set $\Sym^0\mathcal{K}=\mathbb{C}$. The bosonic Fock space over $\mathcal{K}$ is 
\begin{align}
    \Gamma_{\mathrm{s}}(\mathcal{K})\coloneqq \bigoplus_{k=0}^\infty \Sym^k\mathcal{K}.
\end{align}
The vector $\ket{0}\coloneqq (1,0,0,\ldots)$ is the vacuum. 
We denote by $\Pi_k$ the orthogonal projection onto the $k$-particle sector and set
\begin{align}
    \mathcal{F}_{\leq N}\coloneqq  \bigoplus_{k=0}^N \Sym^k\mathcal{K},\qquad P_{\leq N}\coloneqq\sum_{k=0}^N\Pi_k.
\end{align}

For $x\in\mathcal{K}$, the creation and annihilation operators are defined on finite-particle vectors $\psi_k\in \Sym^k\mathcal{K}$ by
\begin{align}
    a^\dag(x)\psi_k&\coloneqq \sqrt{k+1}P_{\mathrm{s}}^{(k+1)}(x\otimes \psi_k) \qquad (k\geq 0)\\
    a(x)\psi_k&\coloneqq \sqrt{k}\left(\bra{x}\otimes I^{\otimes (k-1)}\right)\psi_k\qquad (k\geq 1),
\end{align}
together with $ a(x)\ket{0}=0$.
We remark that $a^\dag$ is linear in $x$, whereas $a$ is antilinear in $x$. 
In particular, for a decomposable tensor $P_{\mathrm{s}}^{(k)}(y_1\otimes\cdots\otimes y_k)\in  \Sym^k\mathcal{K}$, the annihilation operator satisfies
\begin{align}
    a(x)P_{\mathrm{s}}^{(k)}(y_1\otimes\cdots\otimes y_k)\coloneqq \frac{1}{\sqrt{k}}\sum_{i=1}^k\braket{x,y_i}P_{\mathrm{s}}^{(k-1)}(y_1\otimes\cdots\otimes \widehat{y_i}\otimes\cdots\otimes y_k),
\end{align}
where the hat denotes omission of the corresponding tensor factor.
From the above normalization, the creation and annihilation operators satisfy
\begin{align}
    \left\|a^\dag(x)\big|_{\Sym^k\mathcal{K}}\right\|=\sqrt{k+1}\|x\|,\qquad \left\|a(x)\big|_{\Sym^k\mathcal{K}}\right\|=\sqrt{k}\|x\|.\label{eq:cre_ann_sector_bound}
\end{align}
They also satisfy the canonical commutation relations
\begin{align}
    [a(x),a^\dag(y)]=\braket{x,y}I,\qquad [a(x),a(y)]=0,\qquad [a^\dag (x),a^\dag(y)]=0.
\end{align}
For $k=0$, the first relation follows directly from the definitions, since $a(x)a^\dag(y)\ket{0}=\braket{x,y}\ket{0}$ and $a^\dag(y)a(x)\ket{0}=0$. For $k\geq 1$ and $\psi_k\in\Sym^k\mathcal{K}$, we have
\begin{align}
    a(x)a^\dag(y)\psi_k&=(k+1)\left(\bra{x}\otimes I^{\otimes k}\right)P_{\mathrm{s}}^{(k+1)}(y\otimes \psi_k)\\
    &=\braket{x,y}\psi_k+P_{\mathrm{s}}^{(k)} \left(y\otimes\sqrt{k}a(x)\psi_k\right)=\braket{x,y}\psi_k+a^\dag(y)a(x)\psi_k,
\end{align}
which proves the first relation; the other two follow directly from the symmetry of the tensor powers.

The normalized coherent state with amplitude $x\in \mathcal{K}$ is defined by
\begin{align}
    \ket{x}\coloneqq e^{-\|x\|^2/2}\sum_{k=0}^\infty \frac{x^{\otimes k}}{\sqrt{k!}},
\end{align}
which satisfies
\begin{align}
    \left\|\Pi_k\ket{x}\right\|^2=e^{-\|x\|^2}\frac{\|x\|^{2k}}{k!}.
\end{align}
Defining the displacement operators by $D(x)\coloneqq \exp\left(a^\dag (x)-a(x)\right)$, we have, for $t\in\mathbb{R}$, 
\begin{align}
    \ket{x}=D(x)\ket{0},\qquad \frac{\dd}{\dd t}\ket{tx}=\left(a^\dag (x)-a(x)\right)\ket{tx}.
\end{align}

\clearpage

% \bibliography{references,sm}


\end{document}